\documentclass[letterpaper,twocolumn,10pt]{article}
\usepackage{usenix}

\usepackage{tikz}
\usetikzlibrary{arrows.meta,positioning,calc,fit}
\usetikzlibrary{patterns}
\usepackage{amsmath}

\usepackage{etoc}
\usepackage[header, title, toc, page]{appendix}

\usepackage{filecontents}

\usepackage{xcolor}
\usepackage{bm}
\usepackage{microtype}
\usepackage{wrapfig}
\usepackage{tikz}
\usetikzlibrary{positioning, arrows.meta}
\usepackage{tikz-cd}
\usepackage{amsthm}

\usepackage{amsfonts}
\usepackage{amssymb}
\usepackage{enumitem}
\usepackage{multicol}
\usepackage{multirow}
\usepackage{algorithm}
\usepackage{algorithmic}
\usepackage{mathtools}

\usepackage{graphicx}
\usepackage{caption}  
\usepackage{subcaption}
\usepackage{natbib}
\usepackage{stmaryrd}
\usepackage{xparse} 
\usepackage{mathrsfs}

\definecolor{americanrose}{rgb}{1.0, 0.01, 0.24}

\newtheorem{theorem}{Theorem}[section]
\newtheorem{proposition}{Proposition}[section]
\newtheorem{corollary}{Corollary}[section]
\newtheorem{lemma}{Lemma}[section]
\newtheorem{definition}[theorem]{Definition}
\newtheorem{assumption}[theorem]{Assumption}

\begin{document}

\date{}

\title{Sliced R\'enyi Pufferfish Privacy: Tractable Privatization Mechanism and Private Learning with Gradient Clipping
}


\author{
{\rm Tao Zhang}\\
Washington University in St. Louis
\and
{\rm Yevgeniy Vorobeychik}\\
Washington University in St. Louis
} 

\maketitle

\begin{abstract}

We study the design of a privatization mechanism and privacy accounting in the Pufferfish Privacy (PP) family. 
Specifically, motivated by the curse of dimensionality and lack of practical composition tools for iterative learning in the recent R\'enyi Pufferfish Privacy (RPP) framework, 
we propose \textit{Sliced R\'enyi Pufferfish Privacy (SRPP)}. SRPP preserves PP/RPP semantics (customizable secrets with probability-aware secret–dataset relationships) while replacing high-dimensional R\'enyi divergence with projection-based quantification via two sliced measures, \textit{Average SRPP} and \textit{Joint SRPP}. We  develop \textit{sliced Wasserstein mechanisms}, yielding sound SRPP certificates and closed-form Gaussian noise calibration. For iterative learning systems, we introduce an SRPP-SGD scheme with gradient clipping and new accountants based on \textit{History-Uniform Caps (HUC)} and a \textit{subsampling-aware} variant (sa-HUC), enabling decompose-then-compose privatization and additive composition under a common slicing geometry. Experiments on static and iterative privatization show that the proposed framework exhibits favorable privacy–utility trade-offs, as well as practical scalability.

\end{abstract}

\section{Introduction}

Differential Privacy (DP) \cite{dwork2006calibrating} is the de facto standard for measuring data privacy risk. Under DP semantics, the secret is an individual record’s presence or absence, and an adversary with arbitrary side information cannot draw substantially different conclusions from the released output about that inclusion. DP quantifies risk via $\epsilon$-DP or $(\epsilon,\delta)$-DP, i.e., worst-case indistinguishability between neighboring datasets, equivalently expressed through bounds on max divergence and hockey-stick divergence, respectively. R\'enyi DP (RDP) \cite{Mironov2017,papernot2021hyperparameter} preserves DP semantics but measures privacy using $(\alpha,\epsilon)$-R\'enyi divergence.

Pufferfish Privacy (PP) \cite{kifer2014pufferfish} generalizes DP semantics by allowing customizable secrets and modeling probabilistic secret–dataset relationships through \textit{priors} that may represent subjective adversarial beliefs or an objective secret-data-generation process \cite{zhang2025differential}. $\epsilon$-PP and $(\epsilon,\delta)$-PP \cite{Zhang2022} adopt the same indistinguishability measures as DP, while R\'enyi Pufferfish Privacy (RPP) \cite{pierquin2024renyi} shares PP semantics but uses $(\alpha,\epsilon)$-R\'enyi divergence.

Beyond semantics and a quantification measure, a useful privacy framework should provide implementable \textit{privatization mechanisms}. For DP, Laplace \cite{dwork2006calibrating} and Gaussian \cite{dwork2006our} mechanisms calibrate noise to query sensitivity. 
However, computing sensitivity is generally NP-hard \cite{xiao2008output}.
DP and RDP composition properties underpin DP-SGD with gradient clipping \cite{Abadi2016}, a \textit{decompose-then-compose} approach for deep learning without sensitivity computation: each iteration enforces bounded sensitivity via gradient clipping, adds subsampled Gaussian noise, and a moments accountant converts per-iteration RDP costs into an overall $(\epsilon,\delta)$-DP guarantee.

The key technical difference between DP and PP is that PP involves a \textit{probabilistic secret-dataset relationship}.
Because a query only takes the realized dataset as input, PP analysis must account for the induced mapping from secret to output through the prior distributions.
Consequently, practical PP calibration depends on prior-driven distributional quantities and is typically carried out from finite samples and estimators, so error-aware, assumption-dependent guarantees are often unavoidable for useful privacy-utility trade-offs.

For Pufferfish privatization, Song et al. \cite{Song2017} provide the first general-purpose privatization mechanism, the \textit{Wasserstein mechanism}, which uses the $\infty$-Wasserstein distance as an analogue of sensitivity and generalizes the Laplace mechanism. 
Ding \cite{ding2022kantorovich} proposes a Kantorovich-based exponential (and Gaussian) additive-noise mechanism for PP calibrated to 1-Wasserstein sensitivity. 
RPP~\cite{pierquin2024renyi} introduces the General Wasserstein Mechanism (GWM), which uses the shift-reduction lemma~\cite{feldman2018privacy,altschuler2022privacy} to calibrate additive noise via $\infty$-Wasserstein distance. Extensions of GWM include a distribution-aware variant (DAGWM) using $p$-Wasserstein distances to capture average-case data geometry.
While these Wasserstein mechanisms significantly advance the applicability of the PP family frameworks, they inherit the curse of dimensionality from optimal transport (OT): estimating Wasserstein distances and sensitivities remains statistically and computationally expensive in high dimensions, limiting practical applicability.
The tractable general mechanism design principles for PP family remain largely open~\cite{jajodia2025encyclopedia}.

Moreover, PP/RPP lack practical, general-purpose composition theorems for privacy accounting in stochastic optimization algorithms such as SGD. Existing results rely on restrictive structure (e.g., universally composable evolution scenarios \cite{kifer2014pufferfish} or Bayesian networks via Markov quilts \cite{Song2017}). Recent work \cite{zhang2025differential} shows the obstacle is fundamental: probabilistic secret–dataset relationships induce inter-mechanism dependencies that break standard composition; their inverse composition requires solving an infinite-dimensional convex program, making integration with iterative learning nontrivial. RPP’s Privacy Amplification by Iteration (PABI) \cite{pierquin2024renyi} enables privacy guarantees for iterative convex optimization under contractivity, but still depends on Wasserstein distance computation, reintroducing high-dimensional OT costs.

Therefore, making PP-family guarantees operational in practice requires resolving two fundamental obstacles due to the probabilistic secret-dataset relationship captured by priors: (i) \textit{curse of dimensionality in mechanism design}, i.e., tractable calibration without high-dimensional Wasserstein computations; and (ii) \textit{lack of graceful composition and accounting for iterative optimization and learning}, i.e., enabling decompose-then-compose privatization for stochastic iterative methods (e.g., SGD) under PP-family semantics.

\paragraph{Contributions}
To address these two obstacles, we propose a new class of privacy definition termed \textit{Sliced R'enyi Pufferfish Privacy (SRPP)} (Sec.\ref{sec:SRPP}) together with \textit{two complementary, practical privatization mechanisms} equipped with rigorous, provable guarantees:
\begin{enumerate}
\item[(i)] a direct privatization mechanism, termed the \textit{Sliced Wasserstein Mechanism (SWM)} (Sec.\ref{sec:SWM}), calibrated by an efficiently computable \textit{sliced Wasserstein (distance-based) sensitivity} (\textit{SW-sensitivity}) (Sec.~\ref{sec:sliced_Wasserstein}), which circumvents the curse of dimensionality of OT;
\item[(ii)] a decompose-then-compose privatization scheme for iterative learning (e.g., SGD) based on \textit{History-Uniform Caps (HUC)} and their \textit{subsampling-aware} variant \textit{(sa-HUC)} under gradient clipping, together with matching \textit{subsampling-aware (sa-) SRPP} accountants.
This avoids Wasserstein sensitivity computation altogether.
\end{enumerate}

For (i), the key idea is to leverage \textit{slicing} to overcome the scalability bottleneck of OT that arises when calibrating Wasserstein mechanisms via full-dimensional Wasserstein sensitivity (as in RPP’s GWM). Building on the \textit{sliced Wasserstein distance} \cite{rabin2011wasserstein,deshpande2018generative,kolouri2019generalized,nietert2022statistical,bonet2025sliced}, we introduce the SW-sensitivity and design the \textit{SWM} (Sec.\ref{sec:SWM}), which replaces high-dimensional OT sensitivity with a computable sensitivity surrogate based on one-dimensional (1-D) projections.

To align the privacy quantification with SW-sensitivity, we introduce two SRPP notions, \textit{Average SRPP} (Sec.\ref{sec:ave_srpp}) and \textit{Joint SRPP} (Sec.\ref{sec:joint_srpp}), that share the same Pufferfish semantics as PP/RPP, but quantify privacy using \textit{Average Sliced R'enyi Divergence (Ave-SRD)} and \textit{Joint Sliced R'enyi Divergence (Joint-SRD)}, respectively.

For (ii), we develop a general SRPP privatization route for iterative learning via HUC and sa-HUC, enabling DP-SGD-style privatization with gradient clipping \textit{without} requiring contractivity of the update map. 
A central element of the contribution is that HUC/sa-HUC isolate a single scenario-dependent quantity—the \textit{discrepancy cap} (in deterministic form $K_t$ and subsampling-aware, mean-square form $K_t^2$)—that governs the per-iteration secret-induced shift. 
This identification (i) avoids the utility-draining worst-case "group-privacy" \cite{Song2017,pierquin2024renyi} bound (see App.~\ref{app:worst_case_group_privacy} for a formal discussion), and (ii) enables a decompose-then-compose accountant for SGD with gradient clipping, without computing (sliced) Wasserstein sensitivities.
Moreover, when multiple mechanisms are trained independently and privatized under a common slicing geometry, we prove graceful additive composition: the total Ave-/Joint- (sa-Ave/sa-Joint-) SRPP cost is bounded by the sum of the per-mechanism costs, yielding modular, plug-and-play privacy accounting for pipelines and cascades.

Experiments on both static query release and iterative learning (SRPP-SGD) demonstrate favorable privacy–utility trade-offs and practical scalability.
Formal proofs, detailed discussion and analysis are provided in the appendix.

\subsection{Related Work}

\textbf{Privacy Beyond Individual Data Records.}
One-sided DP~\cite{kotsogiannis2020one} protects a designated sensitive subset of the dataset while allowing minimal or no protection for the remaining records.
Fernandes et al.~\cite{fernandes2019generalised} use generalized DP with Earth Mover’s Distance on word embeddings to obfuscate authorship.
Hardt and Roth~\cite{hardt2013beyond} study private singular vector computation under entry-level privacy, where each matrix entry is a secret.
PAC Privacy~\cite{xiao2023pac,sridhar2025pac} and Residual PAC Privacy~\cite{zhang2025breaking} bound the information-theoretic hardness of inference for general secrets, enabling simulation-based measurement and control.
Pufferfish Privacy (PP)~\cite{kifer2014pufferfish,Song2017} and R\'enyi Pufferfish Privacy (RPP)~\cite{pierquin2024renyi} generalize DP and RDP to arbitrary secrets within the same indistinguishability framework.
Related distribution-level notions protect properties of the data distribution rather than individual records~\cite{kawamoto2019local,chen2023protecting}.

\noindent\textbf{Instantiations of Pufferfish Privacy.}
PP has been instantiated for correlated data~\cite{Song2017}, attribute privacy~\cite{Zhang2022}, and confounding settings where secrets are data-derived~\cite{zhang2025differential}.
Related lines include Blowfish, which uses policy constraints to flexibly specify what is protected and what is known~\cite{He2014}, and mutual-information PP, an information-theoretic formulation based on conditional mutual information with improved composition~\cite{nuradha2023pufferfish}.
Applied deployments include smart-meter time-series release and local energy markets~\cite{kessler2015deploying} and temporally correlated time-series trading (HORAE)~\cite{niu2019making}.

\section{Preliminaries: R\'enyi Pufferfish Privacy}

In this section, we review R\'enyi Pufferfish Privacy (RPP) and Wasserstein mechanisms.
Appendix \ref{app:DP} introduces Differential Privacy (DP) and DP-SGD. 

The Pufferfish Privacy (PP) framework~\citep{kifer2014pufferfish} generalizes DP by protecting arbitrary \textit{secrets} linked to the data, rather than only individual records. Let $\mathcal{S}$ be the space of secrets and $\mathcal{Q} \subseteq \mathcal{S}\times\mathcal{S}$ the set of secret pairs to be made indistinguishable. Unlike DP, which conditions on a realized dataset $X$, PP models $X$ as a random variable drawn from a distribution $\theta \in \Theta$, where each $\theta$ is a \textit{prior}, i.e., a probabilistic model capturing either adversarial beliefs about the data or an admissible data-generating scenario (including secret–data relationships).
Crucially, PP is not a prior-dependent relaxation of DP: privacy is defined uniformly over an admissible family of priors $\Theta$ (and over secret pairs in $\mathcal Q$).
Accurate guarantees therefore require specifying (or estimating) $\Theta$.
For each $\theta\in\Theta$, let $P_{\theta}(S,X)$ denote the joint distribution over secrets $S$ and datasets $X$, with marginals $P^{S}_{\theta}$ and $P^{X}_{\theta}$, and conditional $P_{\theta}(x\mid s) = P_{\theta}(s,x)/P^{S}_{\theta}(s)$. We write $\Pr[\cdot|s,\theta]$ for probability $P_{\theta}(\cdot|s)$, and refer to $(\mathcal{S}, \mathcal{Q}, \Theta)$ as the \textit{Pufferfish scenario}.

Recently, Pierquin et al. \cite{pierquin2024renyi} introduced RPP by extending PP using R\'enyi divergence, in the spirit of R\'enyi differential privacy (RDP)~\citep{Mironov2017}, yielding a more flexible analysis of Pufferfish-type guarantees.

\begin{definition}[R\'enyi Pufferfish Privacy (RPP)~\citep{pierquin2024renyi}]
    Let $\alpha > 1$ and $\epsilon\geq 0$. 
    A mechanism $\mathcal{M}$ is $(\alpha,\epsilon)$-RPP in $(\mathcal{S},\mathcal{Q},\Theta)$ if for all $\theta\in\Theta$ and $(s_i,s_j)\in\mathcal{Q}$ with $P^S_{\theta}(s_i),P^S_{\theta}(s_j)>0$,
    \[
        \mathtt{D}_{\alpha}\!\left(\Pr[\mathcal{M}(X)\mid S=s_i, \theta]\,\big\|\,\Pr[\mathcal{M}(X)\mid S=s_j, \theta]\right)
        \;\leq\; \epsilon,
    \]
    where $\mathtt{D}_{\alpha}$ is the R\'enyi divergence of order $\alpha$:
    \begin{equation}\label{eq:renyi_divergence}
        \mathtt{D}_{\alpha}\!\left(\mathrm{P}\,\|\,\mathrm{Q}\right)
        :=
        \tfrac{1}{\alpha-1}\log \mathbb{E}_{Z\sim \mathrm{Q}}\!\left[\left(\tfrac{\mathrm{P}(Z)}{\mathrm{Q}(Z)}\right)^{\alpha}\right],
    \end{equation}
    for distributions $\mathrm{P}$ and $\mathrm{Q}$ on a common measurable space.
\end{definition}

PP~\citep{kifer2014pufferfish,Zhang2022} has the same semantics (the same scenario
$(\mathcal S,\mathcal Q,\Theta)$ and secret indistinguishability requirement), but uses
$(\varepsilon,\delta)$-indistinguishability with $\varepsilon\geq 0$ and $\delta\in[0,1]$ in place of a R\'enyi divergence bound: for all $\theta \in \Theta$, all $(s_i,s_j)\in\mathcal{Q}$ with $P^S_{\theta}(s_i),P^S_{\theta}(s_j)> 0$, and all measurable $\mathcal{T}\subseteq \mathcal{Y}$,
\[
\Pr\bigl(\mathcal{M}(X)\in \mathcal{T} \mid S=s_i, \theta\bigr)\leq e^{\epsilon}\Pr\bigl(\mathcal{M}(X)\in \mathcal{T} \mid S=s_j, \theta\bigr) + \delta.
\]
If $\delta=0$, we say that $\mathcal{M}$ satisfies $\epsilon$-PP.
As $\alpha\to\infty$, $(\alpha,\epsilon)$-RPP recovers $\epsilon$-PP (analogous to RDP recovering DP). 
RPP enjoys post-processing immunity, but it does not provide the graceful composition behavior of DP family, inheriting this limitation from the PP family~\cite{pierquin2024renyi}.

\noindent\textbf{Wasserstein Mechanisms. }
Wasserstein mechanisms are the first general class of privatization mechanisms for PP~\cite{Song2017} and RPP~\cite{pierquin2024renyi}. 
They are calibrated to a Wasserstein-distance-based analogue of DP sensitivity, which we refer to as the \textit{Wasserstein sensitivity}.
We first recall the optimal-transport notions underlying Wasserstein mechanisms: couplings and the $\infty$-Wasserstein distance.

\begin{definition}[Coupling]\label{def:coupling}
Let $\nu, u$ be probability measures on $(\mathbb{R}^d,\mathcal{B}(\mathbb{R}^d))$. 
A \textit{coupling} of $(\nu, \mu)$ is any probability measure $\pi$ on $(\mathbb{R}^d\times\mathbb{R}^d,\mathcal{B}(\mathbb{R}^d)\otimes\mathcal{B}(\mathbb{R}^d))$ whose marginals are $p$ and $q$. 
We denote the set of all couplings by $\Pi(\nu, \mu)$.
\end{definition}

\begin{definition}[$\infty$-Wasserstein Distance]\label{def:WD}
Fix a norm $\|\cdot\|$ on $\mathbb{R}^d$. For probability measures $\nu, u$ on $\mathbb{R}^d$, the $\infty$-Wasserstein distance ($\infty$-WD) is
\[
W_{\infty}(\nu,\mu)=\inf_{\pi\in\Pi(\nu,\mu)} \sup_{(x,y)\in\mathrm{supp}(\pi)}\|x-y\|.
\]
\end{definition}

For a \textit{query function} $f:\mathcal X\to\mathbb R^d$, the \textit{$\infty$-Wasserstein sensitivity} is then defined as
\begin{equation}\label{eq:full_D_WS}
    \Delta_\infty := \max_{(s_i,s_j)\in\mathcal Q,\ \theta\in\Theta}
W_\infty\!\big(\Pr(f(X)\mid s_i, \theta),\,\Pr(f(X)\mid s_j, \theta)\big).
\end{equation}

Let $\zeta$ be the noise distribution.
Pierquin et al. \cite{pierquin2024renyi} provides the General Wasserstein Mechanism (GWM), a general formulation to represent how $\Delta_\infty$, $\zeta$, and $\alpha$ together pin down an uppder bound of R\'enyi divergence of $\mathcal{M}(X)=f(X)+N$, with $N\sim \zeta$, based on the \textit{shift reduction lemma}~\citep{feldman2018privacy}.
When $\zeta=\mathcal{N}(0, \tfrac{\alpha \Delta^{2}_\infty }{2\varepsilon}I_{d})$ with $\Delta_\infty$ computed w.r.t. $\ell_{2}$ norm, $\mathcal{M}(X)$ satisfies $(\alpha, \varepsilon)$-RPP.
The main limitation of Wasserstein mechanisms is the scalability bottleneck of optimal transport (OT), which is exacerbated in high dimensions. 
In particular, computing $W_{\infty}$ (or more generally $W_{p}$; see Definition~\ref{def:p_W} in App.~\ref{app:background_DAGWM_RPP}) between empirical distributions requires solving large-scale OT problems whose complexity grows superlinearly with the sample size. 
This motivates replacing $\Delta_\infty$ by a sliced Wasserstein sensitivity, introduced in Sec. \ref{sec:sliced_Wasserstein}.

\noindent\textbf{Threat model. }
In this work, we consider a standard trusted-curator setting: the data holder runs a randomized mechanism $\mathcal{M}$ on the dataset and releases its output (e.g., a query answer or trained model).
The adversary observes the release and may know $\mathcal{M}$ and the Pufferfish scenario $(\mathcal{S},\mathcal{Q},\Theta)$. 
Potential side information is captured by $\Theta$.
Privacy is required for all $(s_i,s_j)\in\mathcal{Q}$ uniformly over $\theta\in\Theta$.
In the SGD setting, the adversary observes the model/updates released by our algorithm, but not internal randomness beyond what the release reveals.

Beyond DP, several alternative privacy frameworks incorporate prior information to support context-aware analysis and different utility/operational trade-offs, including approaches that exploit intrinsic data-generating structure (e.g., PAC privacy \cite{xiao2023pac,zhang2025breaking}) or explicit assumptions on adversarial beliefs (e.g., Bayesian DP \cite{triastcyn2020bayesian}). While PP \textit{can} interpret $\Theta$ as a set of admissible adversarial priors, in this paper \textit{we treat the Pufferfish scenario $(\mathcal{S},\mathcal{Q},\Theta)$ as the intrinsic modeling setting and focus on the resulting mechanistic and computational challenges of privatization.}

\section{Sliced Wasserstein Sensitivity}\label{sec:sliced_Wasserstein}

In this section, we introduce the \textit{Sliced Wasserstein sensitivity (SW-sensitivity)}, obtained by slicing the $\infty$-WD, to mitigate the computational bottlenecks of Wasserstein mechanisms in high dimensions. 

We first introduce notation.
Let $f:\mathcal{X}\to\mathbb{R}^d$ be a \textit{query} function.
For each $\theta\in\Theta$ and $s\in\mathcal S$, let
$P_{\theta}^{f,s}$ denote the conditional distribution of the (non-perturbed) query output $f(X)$ given $S=s$, i.e.,
$P_{\theta}^{f,s} := P_{\theta}(f(X)\mid s)$.
For a measurable map $g:\mathcal{W}\to\mathcal{C}$ and a Borel measure
$\mu$ on $\mathcal{W}$, the \textit{pushforward} $g_{\#}\mu$ on $\mathcal{C}$ is
defined by
\begin{equation}\label{eq:pushforward}
    g_{\#}\mu(A)\;=\;\mu\bigl(g^{-1}(A)\bigr),
    \qquad \forall A\in\mathcal{B}(\mathcal{C}),
\end{equation}
where $\mathcal{B}(\mathcal{C})$
denotes the Borel $\sigma$-algebra on $\mathcal{C}$.
Let $\mathbb{S}^{d-1}=\{u\in\mathbb{R}^d:\|u\|_2=1\}$ be the unit sphere and
let $\omega$ be a probability measure on $\mathbb{S}^{d-1}$.
For each \textit{direction} $u\in\mathbb{S}^{d-1}$, define the projection
$\Psi^{u}:\mathbb{R}^d\to\mathbb{R}$ by $\Psi^{u}(a)=\langle a,u\rangle$.
Then, for any Borel probability measure $\mathrm{P}$ on $\mathbb{R}^d$,
the pushforward $\Psi^{u}_{\#}\mathrm{P}$ is the 1-D
distribution of $\langle \hat{X},u\rangle$ when $\hat{X}\sim\mathrm{P}$.
For a direction u, we refer to the induced 1-D projection $\Psi^u$ as a \textit{slice}.

Next, we define the \textit{sliced Wasserstein distance}, which defines a metric for comparing probability measures on $\mathbb{R}^d$ by integrating WD over all 1-D projections.

\begin{definition}[Sliced Wasserstein Distance (SWD) \cite{rabin2011wasserstein}]\label{def:SWD}
    Let $\mathcal{P}_p(\mathbb{R}^d)$ denote the set of all Borel probability measures on $\mathbb{R}^d$ with finite $p$-th moment, i.e., $\int_{\mathbb{R}^d} \|x\|^p \,\mathrm{d}\mu(x) < \infty$ for all $\mu\in \mathcal{P}_p(\mathbb{R}^d)$.
    For $p \geq 1$ and $\mu,\nu \in \mathcal{P}_p(\mathbb{R}^d)$, the $p$-SWD is defined by
    \begin{equation}
        \mathrm{SW}_p^p(\mu,\nu) \; := \;
        \int_{\mathbb{S}^{d-1}} 
        W_{p}^{p}\!\left(\Psi^{u}_{\#}\mu,\ \Psi^{u}_{\#}\nu\right)\,\mathrm{d}\lambda(u),
    \end{equation}
    where $\lambda$ is the uniform probability measure on $\mathbb{S}^{d-1}$.
\end{definition}

The SWD is appealing because it replaces a high-dimensional OT comparison by an expectation over 1-D projections, where $W_{p}$ admits a closed form (including $p=\infty$) in terms of quantile functions.
By the Cram\'er-Wold theorem \cite{cramer1936some}, a probability measure on $\mathbb R^d$ is uniquely determined by its
1-D projections. 
Consequently, the SWD is a proper metric.
In practice, $\mathrm{SW}_p$ is typically estimated by Monte Carlo sampling of directions and sorting
the projected samples, reducing the computation to projections and one-dimensional Wasserstein
evaluations. This yields \textit{finite} slicing (or empirical slice distribution) and thus an
approximate SWD estimator (See Sec. \ref{sec:finite_sample_SWM} for formal finite-sample analyses). 
The resulting finite-slice quantity remains a
pseudometric (it is nonnegative, symmetric, and satisfies the triangle inequality). 

In one dimension, $W_p(\hat{\mu},\hat{\nu})$ admits the representation \citep[Remark 2.30]{peyre2019computational}:
\begin{equation}
    W_p^p(\hat{\mu},\hat{\nu})
    =
    \int_0^1 \big|\mathsf{F}_{\hat{\mu}}^{-1}(u)-\mathsf{F}_{\hat{\nu}}^{-1}(u)\big|^p\,\mathrm{d}u,
\end{equation}
where $\mathsf{F}_{\hat{\mu}}^{-1}$ and $\mathsf{F}_{\hat{\nu}}^{-1}$ denote the quantile functions of $\hat{\mu}$ and $\hat{\nu}$.
The map $\hat{\mu}\mapsto \mathsf{F}_{\hat{\mu}}^{-1}$ (the same for $\hat{\nu}\mapsto \mathsf{F}_{\hat{\nu}}^{-1}$) gives an isometric embedding of $(\mathcal{P}_p(\mathbb{R}),W_p)$ into $L_p([0,1])$.
For $p=2$ the metric is induced by the $L_2$ inner product on this embedded subset, yielding a simple, essentially linear geometry compared to the much richer geometry in higher dimensions.

$\mathrm{SW}_{\infty}$ can be viewed as the $p\to\infty$ endpoint of the sliced Wasserstein family $\mathrm{SW}_p$.\footnote{Formally, this corresponds to taking the $p\to\infty$ limit of the per-slice $L_p$ aggregation of 1-D Wasserstein distances.}
At the slice level, $\mathrm{SW}_{\infty}$ provides a natural analogue of the $W_\infty$ (Definition~\ref{def:WD}) for sensitivity analysis in PP/RPP: it captures the largest possible shift in the projected query output induced by changing the secret.

Given a Pufferfish scenario $(\mathcal{S},\mathcal{Q},\Theta)$, for each direction $u\in\mathbb{S}^{d-1}$, we define the \textit{per-slice $\infty$-Wasserstein sensitivity (SW-sensitivity)} as
\begin{equation}\label{eq:per_slice_WS}
    \Delta^{u}_{\infty}
    :=
    \max_{(s_i,s_j)\in\mathcal{Q},\ \theta\in\Theta}
    W_{\infty}\!\Bigl(\Psi^{u}_{\#}P_{\theta}^{f,s_i},\ \Psi^{u}_{\#}P_{\theta}^{f,s_j}\Bigr),
\end{equation}
which is the worst-case $\infty$-WD between the 1-D projections of the conditional data distributions $P_{\theta}^{s_i}$ and $P_{\theta}^{s_j}$ along direction $u$. 
Equivalently, it quantifies the maximal shift in the projected query output when the secret changes from $s_i$ to $s_j$ under any admissible prior in $\Theta$.

Conceptually, we might hope to replace the direction-dependent sensitivity $\Delta_\infty^u$ in~\eqref{eq:per_slice_WS} by a single scalar derived from $\mathrm{SW}_p$. 
However, $\mathrm{SW}_p$ is an \textit{average} over directions of 1-D transport costs and therefore does not provide a uniform, worst-case bound across all $u\in\mathbb{S}^{d-1}$. In contrast, our per-slice R\'enyi envelope~\eqref{eq:sliced_RE} (Sec.~\ref{sec:SWM}) requires controlling each slice, and within each slice taking the worst case over admissible secret pairs and priors.
In particular, an average can be small even when a subset of directions has large sensitivity.

Moreover, an adversary observes the full $d$-dimensional query output, not a one-dimensional projection. 
Thus, calibrating noise using a sliced quantity (e.g., $\Delta_\infty^u$ or an average over $u$) while retaining the original (unsliced) RPP definition introduces a mismatch between the geometry used for calibration and the divergence geometry in the privacy definition, and can therefore lead to unsound privacy bounds.

In particular, for any $u\in\mathbb{S}^{d-1}$, the projection $\Psi^{u}(x)=\langle x,u\rangle$ is $1$-Lipschitz with respect to the Euclidean norm, i.e., $|\Psi^{u}(x)-\Psi^{u}(y)|\le \|x-y\|_2$. Hence, for any probability measures $\mu,\nu$ on $\mathbb{R}^d$,
\[
  W_{\infty}\!\bigl(\Psi^{u}_{\#}\mu,\ \Psi^{u}_{\#}\nu\bigr)
  \;\leq\;
  W_{\infty}(\mu,\nu),
\]
where $W_\infty$ on $\mathbb{R}^d$ is computed with the $\ell_2$ norm.
Indeed, for any coupling $\gamma\in\Pi(\mu,\nu)$ we have $|\Psi^{u}(x)-\Psi^{u}(y)| \leq \|x-y\|_2$ $\gamma$-almost surely. 
Taking the essential supremum over $(x,y)\in\mathrm{supp}(\gamma)$ and then the infimum over $\gamma$ yields the inequality.
Applying this bound to $\mu=P_{\theta}^{s_i}$ and $\nu=P_{\theta}^{s_j}$, and then taking the maximum over $(s_i,s_j)\in\mathcal{Q}$ and $\theta\in\Theta$, gives $\Delta^{u}_{\infty}\leq \Delta_{\infty}$, where $\Delta_{\infty}$ is defined in~\eqref{eq:full_D_WS}.

This raises a natural question: which privacy notion is compatible with sliced Wasserstein sensitivity?
We answer it by defining \textit{Sliced R\'enyi Pufferfish Privacy} in the next section.

\section{Sliced R\'enyi Pufferfish Privacy}\label{sec:SRPP}

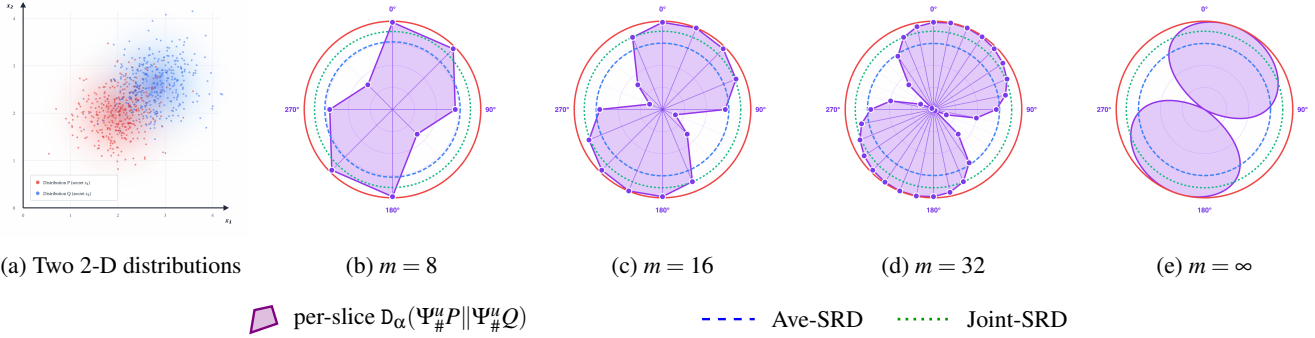
\begin{figure*}[!t]
    \centering

    \begin{subfigure}[b]{0.19\textwidth}
        \centering
        \includegraphics[width=\textwidth]{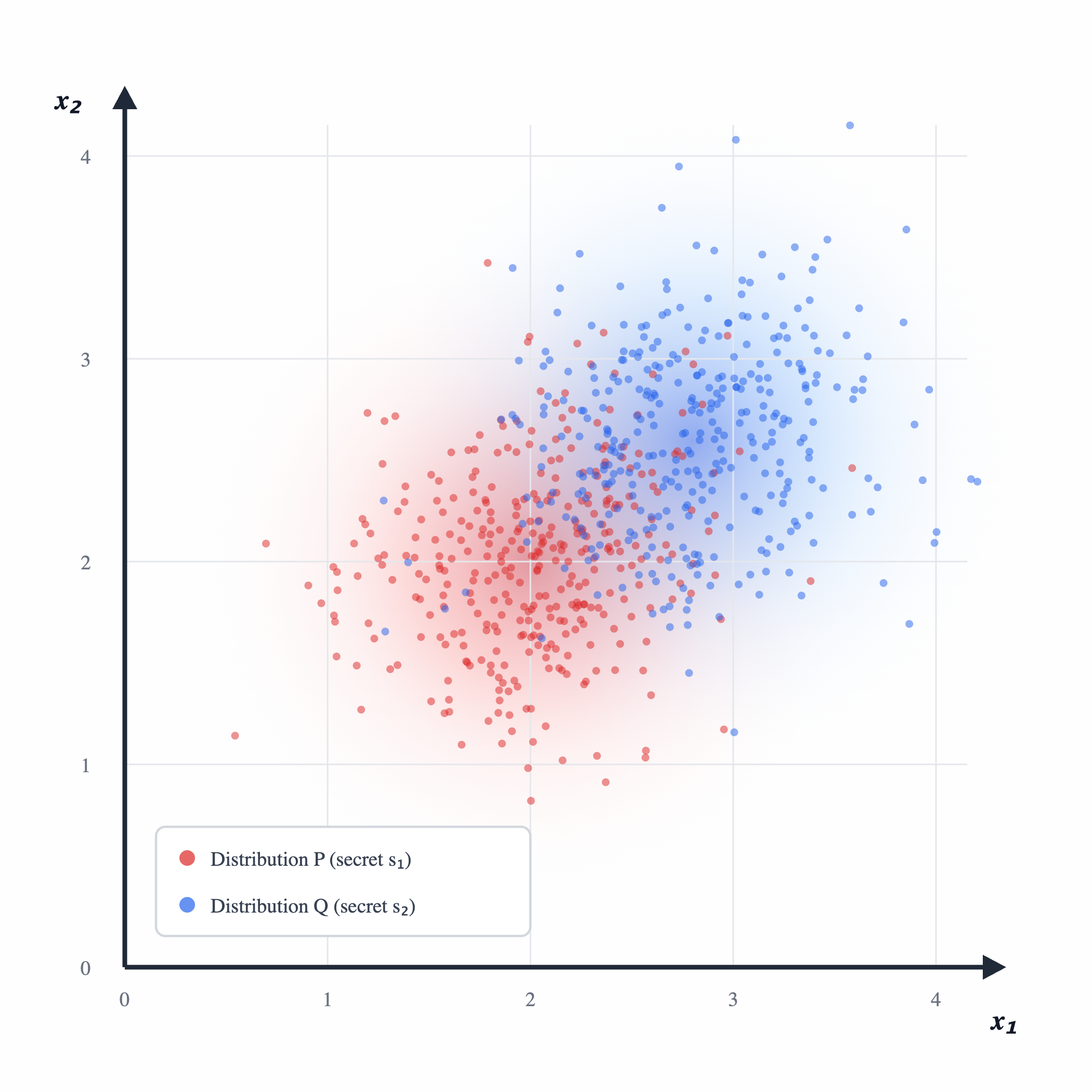}
        \caption{Two 2-D distributions}
        \label{fig:srd_2d}
    \end{subfigure}
    \hfill
    \begin{subfigure}[b]{0.19\textwidth}
        \centering
        \includegraphics[width=\textwidth]{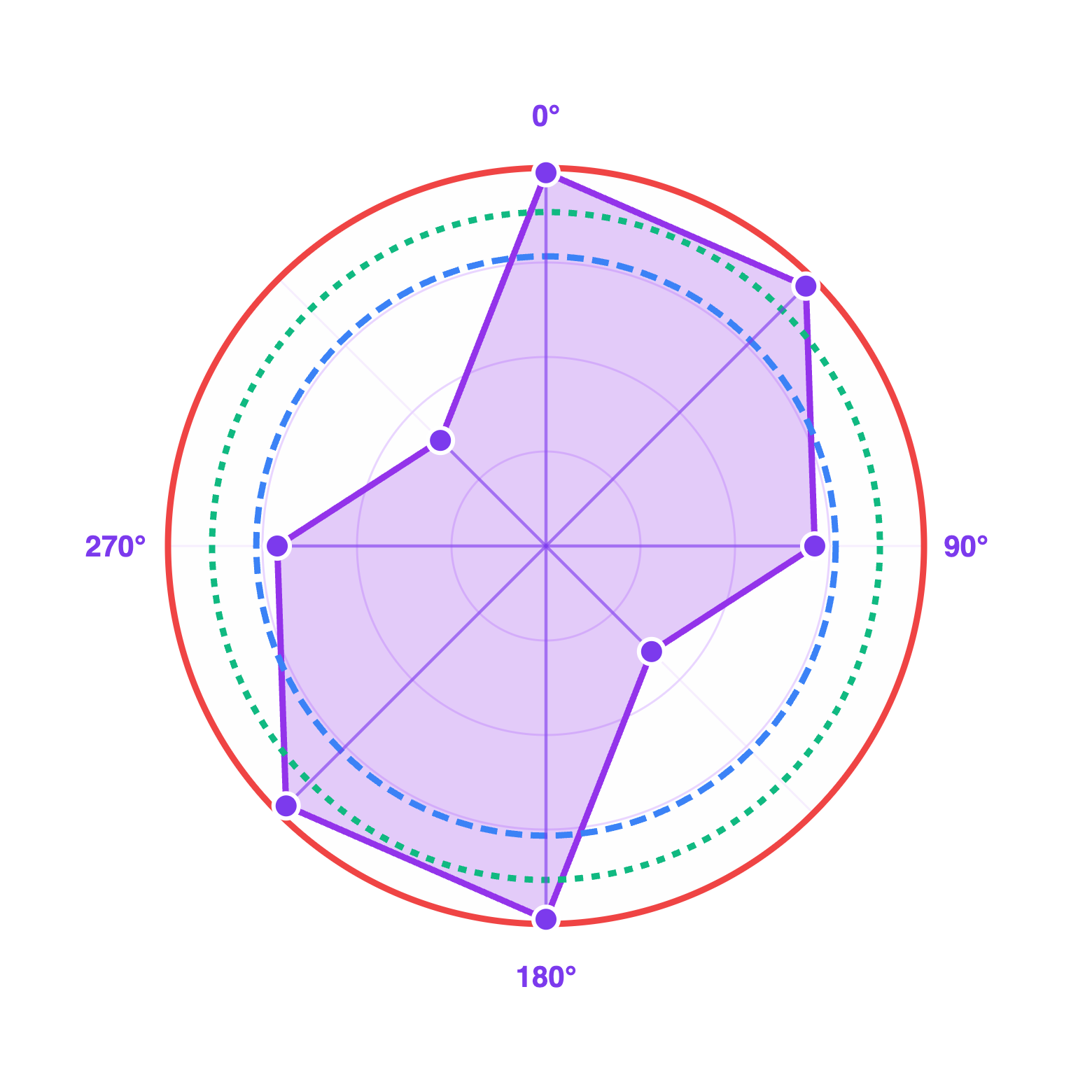}
        \caption{$m=8$}
        \label{fig:srd_m8}
    \end{subfigure}
    \hfill
    \begin{subfigure}[b]{0.19\textwidth}
        \centering
        \includegraphics[width=\textwidth]{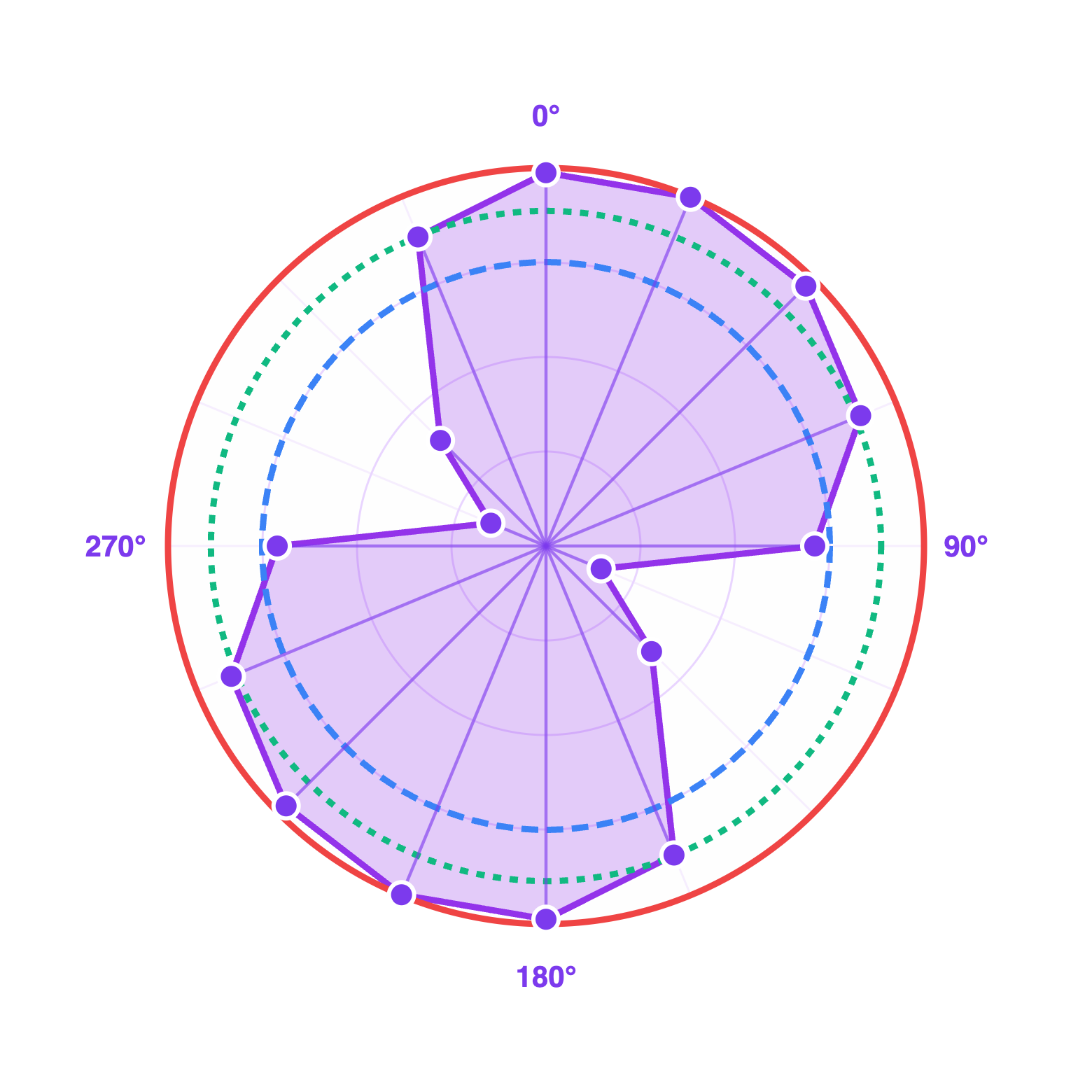}
        \caption{$m=16$}
        \label{fig:srd_m16}
    \end{subfigure}
    \hfill
    \begin{subfigure}[b]{0.19\textwidth}
        \centering
        \includegraphics[width=\textwidth]{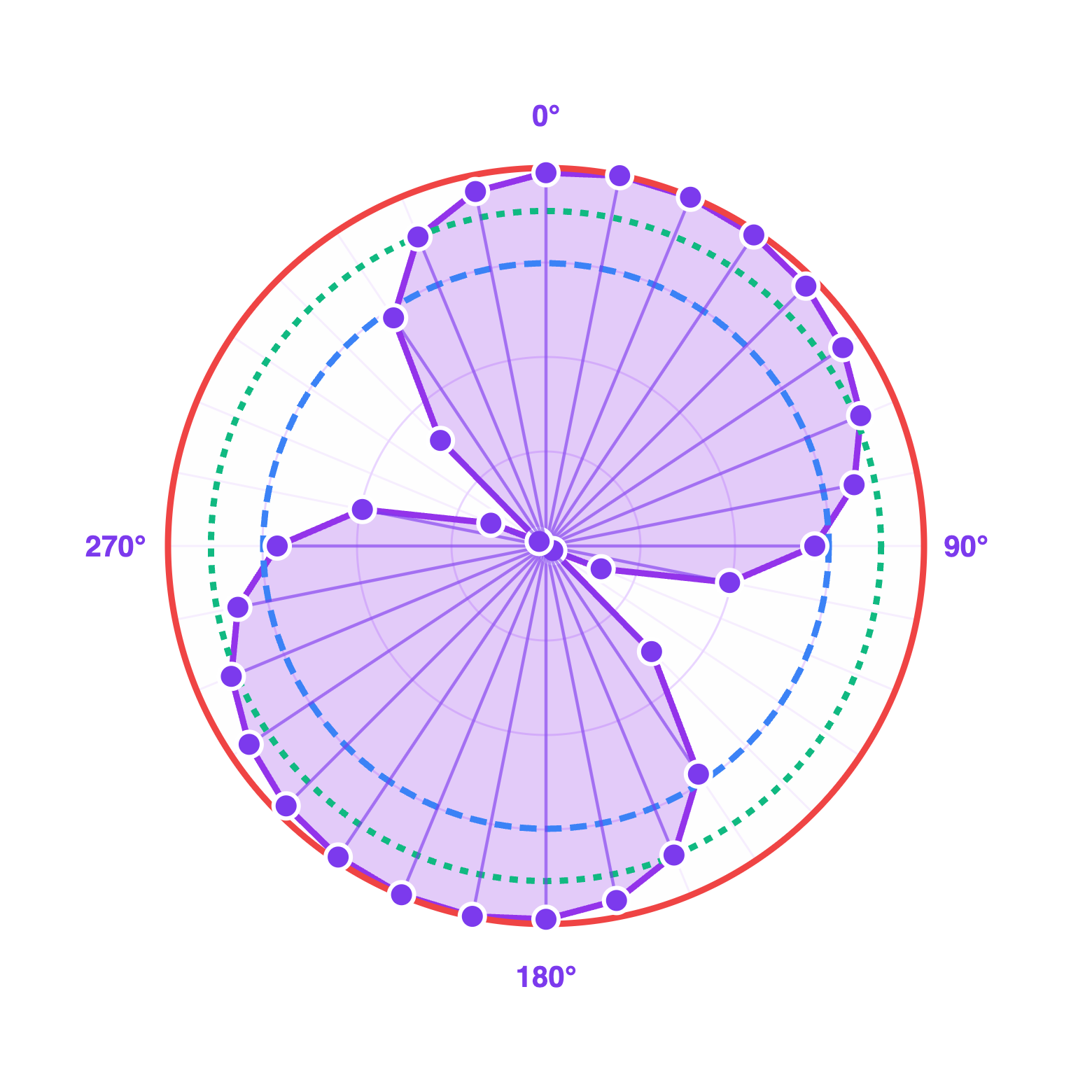}
        \caption{$m=32$}
        \label{fig:srd_m32}
    \end{subfigure}
    \hfill
    \begin{subfigure}[b]{0.19\textwidth}
        \centering
        \includegraphics[width=\textwidth]{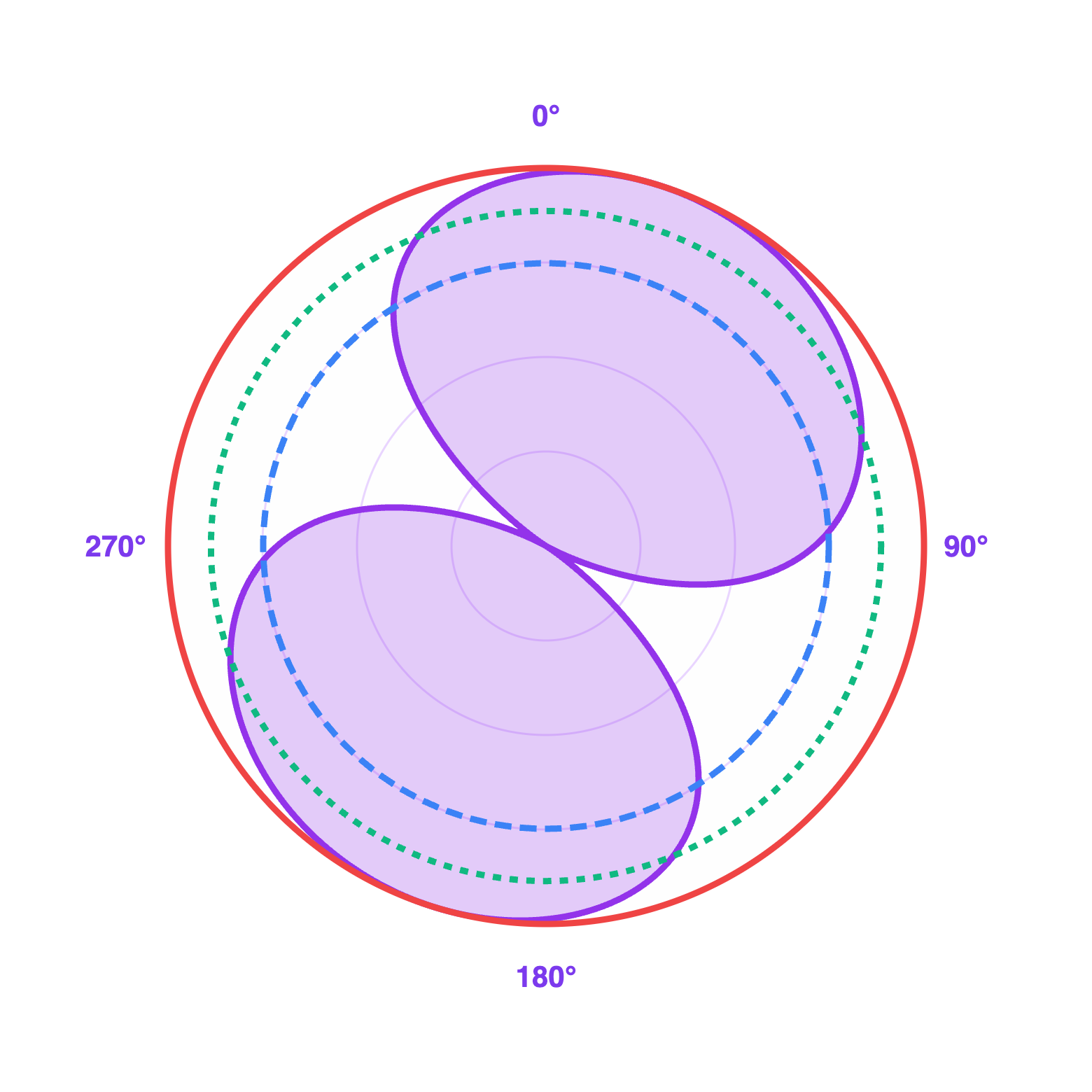}
        \caption{$m=\infty$}
        \label{fig:srd_minf}
    \end{subfigure}

    \vspace{2pt}
    \par\smallskip
    \begin{tikzpicture}[baseline]
        \draw[fill=violet!25, draw=violet, line width=0.9pt]
            (0,0) -- (0.35,0.12) -- (0.28,0.35) -- (0.05,0.30) -- cycle;
        \node[anchor=west] at (0.45,0.18) {\small per-slice $\mathtt{D}_\alpha(\Psi^u_{\#}P\|\Psi^u_{\#}Q)$};

        \draw[blue, dashed, line width=0.9pt] (6.0,0.18) -- (6.7,0.18);
        \node[anchor=west] at (6.8,0.18) {\small Ave-SRD};

        \draw[green!60!black, dotted, line width=1.0pt] (8.6,0.18) -- (9.3,0.18);
        \node[anchor=west] at (9.4,0.18) {\small Joint-SRD};
    \end{tikzpicture}
    \par\smallskip
    \caption{SRD profiles in $\mathbb{R}^2$ at order $\alpha=4$ under a uniform slice distribution $\omega$. (a) Two 2-D Gaussian distributions $P$ (red) and $Q$ (blue). (b)–(d) Discrete slice profiles with $m\in\{8,16,32\}$ directions: the filled purple polygon traces the per-direction divergences $\mathtt{D}_\alpha(\Psi^u_{\#}P\|\Psi^u_{\#}Q)$. The blue dashed circle is Ave-SRD (directional mean), and the green dotted circle is Joint-SRD (log-moment aggregation, emphasizing large divergences). (e) The continuous-direction limit $m=\infty$. 
    }
    \label{fig:SRD}
\end{figure*}

In this section, we introduce \textit{Sliced R\'enyi Pufferfish Privacy} (SRPP), a directional refinement of RPP defined relative to a \textit{slice profile} $(\mathcal U,\omega)$, where $\mathcal U \subseteq \mathbb S^{d-1}$ is a set of directions and $\omega$ is a probability measure on $\mathbb S^{d-1}$ supported on $\mathcal U$. In the finite case, $\mathcal U={u_1,\ldots,u_m}$ and $\omega$ is discrete with weights $(\omega_1,\ldots,\omega_m)$, so that $\int_{\mathbb S^{d-1}} f(u)\,\mathrm d\omega(u)=\sum_{i=1}^m \omega_i f(u_i)$.

We define two variants: Average SRPP (Ave-SRPP; Sec.~\ref{sec:ave_srpp}) and Joint SRPP (Joint-SRPP; Sec.~\ref{sec:joint_srpp}). 
Ave-SRPP bounds the $\omega$-average of per-slice R\'enyi divergences. 
Joint-SRPP aggregates per-slice divergences through a joint log-moment bound, making it more sensitive to rare high-risk directions. 
Sec.~\ref{sec:SWM} formally introduce \textit{Sliced Wasserstein Mechanisms} using SW-sensitivity defined by Sec.~\ref{sec:sliced_Wasserstein}.
We defer additional discussions of SRPP, including the properties such as post-processing immunity and the conversion to \textit{sliced PP}, and the difference between Ave- and Joint-SRPP in App.~\ref{sec:property_srpp}.

\subsection{Average SRPP}\label{sec:ave_srpp}

Our Ave-SRPP aggregates per-slice R\'enyi divergences by averaging, where the quantification measure is \textit{Average Sliced R\'enyi Divergence.}

\begin{definition}[Average Sliced R\'enyi Divergence (Ave-SRD)]\label{def:ave_srd}
For two probability distributions $\mathrm{P}$ and $\mathrm{Q}$ on $\mathbb{R}^{d}$, and a slice profile $(\mathcal{U}, \omega)$, the Ave-SRD of order $\alpha > 1$ with respect to $\omega$ is
\begin{equation}
  \mathtt{AveSD}^{\omega}_{\alpha}\left(\mathrm{P}\;\|\; \mathrm{Q}\right)
  := \int_{\mathbb{S}^{d-1}}
  \mathtt{D}_{\alpha}\!\left(\Psi^{u}_{\#} \mathrm{P} \; \| \; \Psi^{u}_{\#} \mathrm{Q} \right)
  \,\mathrm{d} \omega(u),
\end{equation}
where $\mathtt{D}_{\alpha}$ is the standard R\'enyi divergence.
\end{definition}

Ave-SRD averages the 1-D R\'enyi divergences between the projected measures
$\Psi^{u}_{\#}\mathrm{P}$ and $\Psi^{u}_{\#}\mathrm{Q}$ over directions $u\sim\omega$.
When $d=1$, 
$\mathtt{AveSD}^{\omega}_{\alpha}$ reduces to the standard
$\mathtt{D}_{\alpha}(\mathrm{P}\|\mathrm{Q})$.

\begin{definition}[$(\alpha,\epsilon,\omega)$-Ave-SRPP]\label{def:ave_srpp}
Let $\alpha>1$ and $\epsilon\ge 0$. 
Given a slice profile $(\mathcal{U},\omega)$, a mechanism $\mathcal{M}$ satisfies $(\alpha,\epsilon,\omega)$-Ave-SRPP in $(\mathcal{S},\mathcal{Q},\Theta)$ if for all $\theta\in\Theta$ and all $(s_i,s_j)\in\mathcal{Q}$ with $P^S_{\theta}(s_i),P^S_{\theta}(s_j)>0$,
\[
\mathtt{AveSD}^{\omega}_{\alpha}\!\left(\Pr[\mathcal{M}(X)\mid s_i,\theta]\,\big\|\,\Pr[\mathcal{M}(X)\mid s_j,\theta]\right)\le \epsilon.
\]
\end{definition}

Thus, Ave-SRPP bounds the $\omega$-average of per-slice R\'enyi divergences between the output distributions for any admissible secret pair (uniformly over $\theta\in\Theta$).

\subsection{Joint SRPP}\label{sec:joint_srpp}

Joint-SRPP aggregates per-slice divergences through a joint log-moment (R\'enyi) bound. Its quantification measure is the \textit{Joint Sliced R\'enyi Divergence}.

\begin{definition}[Joint Sliced R\'enyi Divergence (Joint-SRD)]\label{def:joint_srd}
For two probability distributions $\mathrm{P}$ and $\mathrm{Q}$ on $\mathbb{R}^{d}$ and a slice profile $(\mathcal U,\omega)$, the Joint-SRD of order $\alpha>1$ is
\begin{equation}
\begin{aligned}
\mathtt{JSD}^{\omega}_{\alpha}\!\left(\mathrm{P}\,\|\, \mathrm{Q}\right)
&:= \frac{1}{\alpha-1}\log
\int_{\mathbb S^{d-1}} \\
&\quad \exp\Big((\alpha-1)\mathtt{D}_{\alpha}\!\left(\Psi^{u}_{\#}\mathrm{P}\,\|\,\Psi^{u}_{\#}\mathrm{Q}\right)\Big)\,
\mathrm d\omega(u).
\end{aligned}
\end{equation}
\end{definition}

Joint-SRD aggregates per-slice divergences via a log-moment transform, making it more sensitive than Ave-SRD to rare directions with large 1-D divergence (see Figure~\ref{fig:SRD}). 
When $d=1$, $\mathtt{JSD}^{\omega}_{\alpha}(\mathrm{P}\|\mathrm{Q})$ reduces to the standard $\mathtt{D}_{\alpha}(\mathrm{P}\|\mathrm{Q})$.

\begin{definition}[$(\alpha,\epsilon,\omega)$-Joint-SRPP]\label{def:joint_srpp}
Let $\alpha>1$ and $\epsilon\ge 0$. 
Given a slice profile $(\mathcal U,\omega)$, a mechanism $\mathcal{M}$ satisfies $(\alpha,\epsilon,\omega)$-Joint-SRPP in $(\mathcal{S},\mathcal{Q},\Theta)$ if for all $\theta\in\Theta$ and all $(s_i,s_j)\in\mathcal{Q}$ with $P^S_{\theta}(s_i),P^S_{\theta}(s_j)>0$,
\[
\mathtt{JSD}^{\omega}_{\alpha}\!\left(\Pr(\mathcal{M}(X)\mid s_i,\theta)\,\big\|\,\Pr(\mathcal{M}(X)\mid s_j,\theta)\right)\le \epsilon.
\]
\end{definition}


\subsection{Sliced Wasserstein Mechanism}\label{sec:SWM}

In this section, we present the \textit{Sliced Wasserstein Mechanism (SWM)} to obtain SRPP guarantee by leveraging the computational tractability afforded by SW-sensitivity of Sec. \ref{sec:sliced_Wasserstein}.

Let $N\sim\zeta$ denote additive noise with distribution $\zeta$ on $\mathbb{R}^d$. 
For $a\in\mathbb{R}^d$, write $\zeta_{-a}$ for the shifted distribution (i.e., the distribution of $U-a$ for $U\sim\zeta$), where $\zeta_{-a}(w)=\zeta(w-a)$. 
For any $\pi:\mathbb{S}^{d-1}\to \mathbb{R}_{+}$ and direction $u\in \mathbb{S}^{d-1}$, we define the \textit{per-slice R\'enyi envelope} in as
\begin{equation}\label{eq:sliced_RE}
    R_{\alpha}(\zeta,\pi(u))
=
\sup_{\|a\| < \pi(u)} \mathtt{D}_{\alpha}\bigl(\zeta_{-a},\zeta\bigr).
\end{equation}
Recall $P_{\theta}^{f,s} := P_{\theta}(f(X)\mid s)$ for a query function $f$.
When $\mathcal{M}(X) = f(X) + N$ with $N\sim \zeta$ independent of $f(X)$, let $\mathcal{M}^\theta_s = P_{\theta}^{f,s} * \zeta$ be the distribution of $\mathcal{M}(X)$ where $*$ is the convolution operator.

\begin{lemma}\label{lemma:SRPP-envelope}
Fix a Pufferfish scenario $(\mathcal S,\mathcal Q,\Theta)$ and a slice profile $(\mathcal U,\omega)$.
Let $\mathcal M(X)=f(X)+N$ with $N\sim\zeta$ independent of $X$.
For any $\theta\in\Theta$ and any $(s_i,s_j)\in\mathcal Q$ with
$P_\theta^S(s_i)>0$ and $P_\theta^S(s_j)>0$, we have for all $\alpha>1$ and all $u\in\mathcal U$,
\[
\mathtt D_\alpha\!\left(\Psi^u_\#\mathcal M^\theta_{s_i}\;\big\|\;\Psi^u_\#\mathcal M^\theta_{s_j}\right)
\;\le\;
R_\alpha\!\left(\zeta,\, z_{\infty}^u(\theta,s_i,s_j)\right),
\]
where $z_\infty^u(\theta,s_i,s_j)
:=
W_\infty\!\left(\Psi^u_{\#}P_{\theta}^{f,s_i},\ \Psi^u_{\#}P_{\theta}^{f,s_j}\right).$
\end{lemma}

Lemma \ref{lemma:SRPP-envelope} states that, along each direction $u$, the per-slice R\'enyi envelope evaluated at the shift radius $z^{u}_{\infty}$ upper-bounds the actual R\'enyi divergence between the projected output distributions under secrets $s_i$ and $s_j$.

We then define the \textit{averaged R\'enyi envelope} as
\begin{equation}
    \mathtt{AR}^{\infty}_{\alpha, \omega}\bigl(\mathcal{M}^{\theta}_{s_{i}} \,\big\|\, \mathcal{M}^{\theta}_{s_{j}}\bigr)
    :=
    \int_{\mathbb{S}^{d-1}}  R_{\alpha}\bigl(\zeta, z^{u}_{\infty}\bigr)\,\mathrm{d}\omega(u).
\end{equation}
This quantity aggregates the per-slice R\'enyi envelope bounds over slice $u\sim\omega$ and will serve as our sliced R\'enyi envelope surrogate for the true averaged sliced R\'enyi divergence.
Similarly, we define the \textit{joint sliced R\'enyi envelope} by
\begin{equation}
    \begin{aligned}
        & \mathtt{JR}^{\infty}_{\alpha, \omega}\bigl(\mathcal{M}^{\theta}_{s_{i}} \big\| \mathcal{M}^{\theta}_{s_{j}}\bigr)\\
        &:= \tfrac{1}{\alpha-1}\log\left(\int_{\mathbb{S}^{d-1}}\exp\bigl((\alpha-1)\,R_{\alpha}(\zeta, z^{u}_{\infty})\bigr) d\omega(u)\right).
    \end{aligned}
\end{equation}
Here, instead of averaging the per-slice envelope values linearly, we aggregate them via an exponential moment over $V \sim \omega$, followed by a log and $1/(\alpha-1)$ rescaling. This construction mirrors the passage from an average of R\'enyi divergences to a joint R\'enyi divergence for the pair $(V,\Psi^V(\mathcal{M}(X)))$. As a result, $\mathtt{JR}^{\infty}_{\alpha, \omega}$ is more sensitive to rare directions with large envelope values than $\mathtt{AR}^{\infty}_{\alpha, \omega}$.

\begin{corollary}\label{corollary:SRPE_implies_SRPP}
    Fix a continuous or finite slice profile $\{\mathcal{U}, \omega\}$.
    Let $\alpha > 1$ and $\epsilon \geq 0$. 
    For a query function $f$, let $\mathcal{M}(X) = f(X) + N$ with $N\sim \zeta$. 
    Then, (i) $\mathcal{M}(X)$ satisfies $(\alpha,\epsilon,\omega)$-Ave-SRPP if $\mathtt{AR}^{\infty}_{\alpha, \omega}\bigl(\mathcal{M}^{\theta}_{s_{i}} \,\big\|\, \mathcal{M}^{\theta}_{s_{j}}\bigr)
        \leq \epsilon$; and (ii) $\mathcal{M}(X)$ satisfies $(\alpha,\epsilon,\omega)$-Joint-SRPP if $\mathtt{JR}^{\infty}_{\alpha, \omega}\bigl(\mathcal{M}^{\theta}_{s_{i}} \,\big\|\, \mathcal{M}^{\theta}_{s_{j}}\bigr)
        \leq \epsilon$.
\end{corollary}

Corollary~\ref{corollary:SRPE_implies_SRPP} follows immediately from Lemma~\ref{lemma:SRPP-envelope}. In particular, an $\varepsilon$-bounded Rényi envelope under $(\alpha,\omega)$ is sufficient to guarantee $(\alpha,\varepsilon,\omega)$-SRPP, for both the average and joint variants. The envelopes $\mathtt{AR}^{\infty}_{\alpha,\omega}$ and $\mathtt{JR}^{\infty}_{\alpha,\omega}$ are especially useful operational targets: our Sliced Wasserstein Mechanism (SWM), calibrated via SW-sensitivity, provides a direct and tractable way to enforce $\mathtt{AR}^{\infty}_{\alpha,\omega}\!\bigl(\mathcal{M}^{\theta}_{s_i}\,\|\,\mathcal{M}^{\theta}_{s_j}\bigr)\le \varepsilon$ and $\mathtt{JR}^{\infty}_{\alpha,\omega}\!\bigl(\mathcal{M}^{\theta}_{s_i}\,\|\,\mathcal{M}^{\theta}_{s_j}\bigr)\le \varepsilon$ as we detail next.

\subsubsection{Gaussian Sliced Wasserstein Mechanism}\label{sec:GSWM}

We specialize to Gaussian noise, $N\sim\mathcal{N}(0,\sigma^{2} I_{d})$.
Given a slice profile $(\mathcal{U},\omega)$ (continuous or finite) and the per-slice SW-sensitivity $\Delta_{\infty}^{u}$ in~\eqref{eq:per_slice_WS}, define
\[
\overline{\Delta}^2
:=
\int_{\mathbb{S}^{d-1}} \bigl(\Delta_{\infty}^{u}\bigr)^2\,\mathrm{d}\omega(u),
\qquad
\Delta_{\star}^2
:=
\sup_{u\in\mathcal{U}} \bigl(\Delta_{\infty}^{u}\bigr)^2.
\]
Here, $\overline{\Delta}^2$ is the second moment of SW-sensitivity and $\Delta_{\star}^2$ is the largest squared directional sensitivity over $u\in\mathcal{U}$.

\begin{theorem}\label{thm:gaussian-ave-SRPEp}
Fix a continuous or finite slice profile $\{\mathcal{U}, \omega\}$.
    Let $\alpha > 1$ and $\epsilon \geq 0$. 
    A mechanism $\mathcal{M}(X) = f(X) + N$ with $N\sim \zeta$ satisfies $( \alpha, \epsilon,\omega)$-Ave-SRPP in $(\mathcal{S}, \mathcal{Q}, \Theta)$ if $\zeta = \mathcal{N}(0, \tfrac{\alpha \overline{\Delta}^2 }{2\epsilon} I_{d})$.
\end{theorem}

Theorem \ref{thm:gaussian-ave-SRPEp} provides a closed-form Gaussian noise calibration for Ave-SRPP under a fixed slice profile.
Specifically, setting the isotropic variance to $\sigma^2 =\tfrac{\alpha\overline{\Delta}^2}{2\epsilon}$ ensures that the averaged R\'enyi envelope $\mathtt{AR}^{\infty}_{\alpha, \omega}$ is bounded by $\varepsilon$ uniformly over $\theta\in\Theta$ and $(s_{i}, s_{j})\in\mathcal{Q}$.
By Corollary~\ref{corollary:SRPE_implies_SRPP}, this envelop bound is sufficient to guarantee $(\alpha, \epsilon,\omega)$-Ave-SRPP.
Moreover, the per-slice R\'enyi envelope is attained by the Gaussian SWM.

\begin{theorem}\label{thm:gaussian-joint-SRPEp}
Fix a continuous or finite slice profile $\{\mathcal{U}, \omega\}$.
    Let $\alpha > 1$ and $\epsilon \geq 0$. 
    A mechanism $\mathcal{M}(X) = f(X) + N$ with $N\sim \zeta$ satisfies $( \alpha, \epsilon,\omega)$-Joint-SRPP in $(\mathcal{S}, \mathcal{Q}, \Theta)$ if $\zeta = \mathcal{N}(0, \tfrac{\alpha\,\Delta_{\star}^2 }{2\epsilon}I_{d})$.
\end{theorem}

Theorem \ref{thm:gaussian-joint-SRPEp} gives the analogous closed-form noise calibration for Joint-SRPP.
Taking $\sigma^2 = \tfrac{\alpha\Delta_{\star}^2}{2\epsilon}$ ensures that the joint R\'enyi envelope $\mathtt{JR}^{\infty}_{\alpha, \omega}$ is at most $\varepsilon$ uniformly over $\theta\in\Theta$ and $(s_{i}, s_{j})\in\mathcal{Q}$, and hence $\mathcal{M}$ satisfies $( \alpha, \epsilon,\omega)$-Joint-SRPP by Corollary~\ref{corollary:SRPE_implies_SRPP}.
Compared to the Gaussian SWM for the Ave-SRPP, this choice depends on the worst directional sensitivity $\Delta_{\star}^2$, rather than the profile average $\overline{\Delta}^2$, making it more conservative and typically requiring larger noise to protect against rare high-sensitivity directions.

\subsubsection{Computational Properties of SWM}
\label{sec:finite_sample_SWM}

The SRPP guarantees of the SWM in Theorems \ref{thm:gaussian-ave-SRPEp} and \ref{thm:gaussian-joint-SRPEp} hold for both \textit{continuous} slice profiles
$\omega$ on $\mathbb{S}^{d-1}$ and \textit{finite} profiles.
In practice, even when $\omega$ is conceptually continuous, SWM is instantiated by sampling finitely many
directions and the conditional distributions $P^{f,s}_{\theta}$ are typically accessed only through samples. This induces
(i) a \textit{finite-slice} (Monte Carlo) approximation of the profile aggregates and (ii) a \textit{finite-sample}
approximation of the per-slice sensitivities $\Delta_\infty^{u}$.

\noindent\textbf{Notations. }
Fix a query function $f:\mathcal{X}\to\mathbb{R}^{d}$ and a Pufferfish scenario $(\mathcal{S}, \mathcal{Q}, \Theta)$.
Fix finite instantiated families $\widehat{\Theta}\subseteq \Theta$ and $\widehat{\mathcal{Q}}\subseteq \mathcal{Q}$.
For each $(s,\theta)\in \mathcal{S}\times \widehat{\Theta}$, draw $n$ i.i.d.\ samples
$Z^{(\theta, s)}_{1:n} \sim P^{f,s}_{\theta}$.
For a direction $u$, write $V^{\theta, s}_{k}(u):= \langle Z^{(\theta,s)}_k, u\rangle$, and let $\widehat{F}^{u}_{\theta,s}$
denote the empirical CDF of $\{V^{\theta, s}_{k}(u)\}^{n}_{k=1}$, with generalized inverse
$(\widehat F)^{-1}(t):=\inf\{x:\widehat F(x)\ge t\}$.
Let $\omega$ be a continuous slice profile on $\mathbb{S}^{d-1}$ and draw i.i.d.\ directions
$u_1,\ldots,u_{m}\sim \omega$.
Define the empirical profile $\omega_{m} := \tfrac{1}{m}\sum_{\ell=1}^{m} \delta_{u_\ell}$ and the profile aggregates
$\overline{\Delta}^2(\omega):=\mathbb E_{U\sim\omega}\!\big[(\Delta_\infty^{U})^2\big]$ and
$\overline{\Delta}^2(\omega_m)= \tfrac{1}{m}\sum_{\ell=1}^{m} (\Delta_\infty^{u_\ell})^2$.

For $\widetilde{\rho}\in(0,1)$, define
$\mathsf{e}_{n}(\widetilde{\rho}):=\sqrt{\frac{\log(2/\widetilde{\rho})}{2n}}$
\footnote{$\sup_x|\widehat F(x)-F(x)|\leq \smash{\mathsf{e}_n(\widetilde\rho)}$ holds with probability at least $\smash{1-\widetilde\rho}$ due to the Dvoretzky--Kiefer--Wolfowitz inequality.}.
For $u\in\mathbb S^{d-1}$ and $(\theta,(s_i,s_j))\in\widehat{\Theta}\times\widehat{\mathcal{Q}}$, define
\[
\begin{aligned}
\widehat{W}^{\theta,s_i, s_j}_{\infty}\!\big(u, \widetilde{\rho}\big)
&:=
\sup_{t\in[\mathsf{e}_{n}(\widetilde{\rho}/2),\,1-\mathsf{e}_{n}(\widetilde{\rho}/2)]}
\Big|
\big(\widehat F_{\theta,s_i}^{\,u}\big)^{-1}\!\big(t+\mathsf{e}_{n}(\widetilde{\rho}/2)\big)\\&
-
\big(\widehat F_{\theta,s_j}^{\,u}\big)^{-1}\!\big(t-\mathsf{e}_{n}(\widetilde{\rho}/2)\big)
\Big|.
\end{aligned}
\]
For each sampled direction $u_{\ell}$, define
$\widehat{\Delta}_{\infty}(u_{\ell})
:= \max_{(\theta,(s_i,s_j))\in \widehat{\Theta}\times \widehat{\mathcal{Q}}}
\widehat{W}^{\theta,s_i, s_j}_{\infty}\!\big(u_{\ell}, \widetilde{\rho}\big)$.

Theorems~\ref{thm:mc_finite_gaussian_ave_srpp} and~\ref{thm:mc_finite_gaussian_joint_srpp} show that SWM admits
PAC-style calibrations under \textit{both} Monte Carlo slicing and finite-sample estimation of $P^{f,s}_{\theta}$,
yielding $(\alpha,\varepsilon,\omega)$-Ave/Joint-SRPP guarantee with probability at least $1-\gamma$.

\begin{theorem}
\label{thm:mc_finite_gaussian_ave_srpp}
Fix a Pufferfish scenario $(\mathcal S,\mathcal Q,\Theta)$, a query $f:\mathcal X\to\mathbb R^d$, and a continuous slice
distribution $\omega$ on $\mathbb{S}^{d-1}$. Draw i.i.d.\ directions $U_{1:m}\sim\omega$.
Assume $0\leq \Delta_\infty^{u}\le \Delta_0$ for all $u$.
Fix $\alpha>1$, $\varepsilon\ge 0$, and $\gamma\in(0,1)$.
Let $\widehat\Delta_\infty(U_\ell)$ be random quantities satisfying $\Pr\!\Big(\Delta_\infty^{U_\ell}\leq \widehat\Delta_\infty(U_\ell)\ \text{for all }\ell\in[m]\Big)\ge 1-\gamma/2$,
where the probability is over the conditional sampling used to form $\widehat\Delta_\infty(U_\ell)$.
If Gaussian SWM uses $N\sim\mathcal N(0,\sigma^2 I_d)$ with
\[
\sigma^2 \;\geq\; \frac{\alpha}{2\varepsilon}\left(
\frac{1}{m}\sum_{\ell=1}^{m}\big(\widehat\Delta_\infty(U_\ell)\big)^2
+\Delta_0^2\sqrt{\frac{\log(4/\gamma)}{2m}}
\right),
\]
then with probability at least $1-\gamma$ (over both the conditional sampling and the draw of $U_{1:m}$),
$\mathcal M(X)=f(X)+N$ satisfies $(\alpha,\varepsilon,\omega)$-\textit{Ave-SRPP}.
\end{theorem}

\begin{theorem}
\label{thm:mc_finite_gaussian_joint_srpp}
Fix a Pufferfish scenario $(\mathcal S,\mathcal Q,\Theta)$, a query $f:\mathcal X\to\mathbb R^d$, and a continuous slice distribution $\omega$ on $\mathbb S^{d-1}$.
Draw i.i.d.\ directions $U_{1:m}\sim\omega$.
Assume $0\le \Delta_\infty^{u}\le \Delta_0$ for all $u$.
Fix $\alpha>1$, $\varepsilon\ge 0$, and $\gamma\in(0,1)$.
Let $\widehat\Delta_\infty(U_\ell)$ be random quantities satisfying $\Pr\!\Big(\Delta_\infty^{U_\ell}\le \widehat\Delta_\infty(U_\ell)\ \text{for all }\ell\in[m]\Big)\ge 1-\gamma/2$,
where the probability is over the conditional sampling used to form $\widehat\Delta_\infty(U_\ell)$.
For $\sigma^2>0$, define
$c(\sigma):=\frac{\alpha(\alpha-1)}{2\sigma^2}$, $ \widehat\mu_m(\sigma)
:=\frac{1}{m}\sum_{\ell=1}^{m}
\exp\!\Big(c(\sigma)\,\big(\widehat\Delta_\infty(U_\ell)\big)^2\Big)$, and $b(\sigma):=\exp\!\big(c(\sigma)\Delta_0^2\big)$.
Gaussian SWM uses $N\sim\mathcal N(0,\sigma^2 I_d)$ with $\sigma^2$ satisfying
\begin{equation}\label{eq:finite_sample_joint_sigma}
    \frac{1}{\alpha-1}\log\!\left(
\widehat\mu_m(\sigma)
+
\big(b(\sigma)-1\big)\sqrt{\frac{\log(4/\gamma)}{2m}}
\right)
\leq
\varepsilon,
\end{equation}
then with probability at least $1-\gamma$ (over both the conditional sampling and the draw of $U_{1:m}$),
$\mathcal M(X)=f(X)+\mathcal N(0,\sigma^2 I_d)$ satisfies $(\alpha,\varepsilon,\omega)$-\textit{Joint-SRPP}.
\end{theorem}

Since $c(\sigma)=\alpha(\alpha-1)/(2\sigma^2)$ is decreasing in $\sigma$, both $\widehat\mu_m(\sigma)$ and $b(\sigma)$ are
decreasing in $\sigma$, hence the left-hand side of (\ref{eq:finite_sample_joint_sigma}) is decreasing.
Therefore the Joint-SRPP calibration condition $\Phi(\sigma)\le \varepsilon$ can be enforced by a one-dimensional bisection
search over $\sigma$.

\noindent\textbf{Finite $\omega$. }
If $\omega$ is discrete on $\mathcal U=\{u_1,\ldots,u_m\}$ with weights $\{\omega_\ell\}$ (no Monte Carlo directions),
for Ave-SRPP, we can replace $\tfrac{1}{m}\sum_{\ell=1}^m(\cdot)$ by $\sum_{\ell=1}^m \omega_\ell(\cdot)$ and drop the
Hoeffding term $\Delta_0^2\sqrt{\log(4/\gamma)/(2m)}$.
For Joint-SRPP, the same statement applies with $U_\ell=u_\ell$.
Thus, only the finite-sample (distribution-estimation) layer remains.

\begin{proposition}[Computational Advantage of SMD, Informal]\label{prop:comp_advantage_swm}
    With $n$ samples per conditional distribution $P^{f,s}_{\theta}$, query dimension $d$, $\mathsf{M}:=|\widehat\Theta|\,|\widehat{\mathcal Q}|$ instances, and $m$ slice directions, the runtime of computing SW-sensitivity is $O\big(\mathsf{M}\,m\,(dn+n\log n)\big)$.
    When the unsliced WM is computed by the standard discrete OT (Kantorovich) formulation on empirical measures, the runtime is $O\big(\mathsf{M}\,(dn^2+n^3\log n)\big)$.
\end{proposition}

Proposition \ref{prop:comp_advantage_swm} shows that slicing converts d-dimensional OT into $m$ 1-D computations, reducing runtime from quadratic/cubic in n to nearly linear (up to sorting) (see App. \ref{app:computational_advantage} for the formal version). 
Beyond SWD, another mainstream scalable OT approximation is entropic OT via the Sinkhorn algorithm \cite{cuturi2013sinkhorn,luise2018differential,wong2019wasserstein,chizat2020faster}. 
App.~\ref{app:Sinkhorn} describes a Sinkhorn–Wasserstein mechanism for RPP and summarizes why it is less suitable for privacy calibration and privatization.

\section{SRPP-Learning with Gradient Clipping}\label{sec:SRPP-SGD}

For stochastic iterative algorithms such as SGD, directly bounding the end-to-end sensitivity of an entire training run is typically intractable under both DP and Pufferfish scenarios. 
Following the decompose–then–compose paradigm of DP-SGD with gradient clipping (Sec.~\ref{sec:preli_DP_SFO}), we extend SRPP to stochastic learning by introducing SRPP-SGD with gradient clipping, an SRPP analogue of DP-SGD in the Pufferfish setting. We begin by fixing notation.

\noindent\textbf{Dataset. }
Let $x=(x_1,\dots,x_n)\in \mathcal{X} = \bar{\mathcal{X}}^n$ be the finite samples of data points, where each $x_{i} = (a_{i}, b_{i})\in \bar{\mathcal{X}}$ consists of features $a_{i}$ and label $b_{i}$.
In this section, we take $P_{\theta}^{X}$ as the empirical (marginal) distribution of the data samples $x=(x_1,\dots,x_n)\in \mathcal{X} = \bar{\mathcal{X}}^n$ for all $\theta\in\Theta$.
At iteration $t$, we draw a random mini-batch $\mathsf{I}_t \subseteq [n]$ using a subsampling scheme $\{\eta, \rho\}$ characterized by a rate $\eta$ ($\eta=\eta_t$ may vary with $t$) and a scheme type $\rho$, where $\rho \in \{\textsf{WR}, \textsf{WOR}, \textsf{Poisson}\}$ denotes sampling with-replacement (\textsf{WR}), without-replacement (\textsf{WOR}), or via a Poisson process (\textsf{Poisson}), respectively.
Let $X^{t} = (X_{i})_{i\in \mathsf{I}_{t}}$ denote the random mini-batch subsampled from $x$, with $x^{t}$ representing a particular realization.

\noindent\textbf{Sampling randomness. }
Given a subsampling scheme $\{\eta,\rho\}$, let $(\mathcal{R},\mathscr{R},\mathbb{P}_{\eta,\rho})$ denote its probability space, with $\mathbb{P}_{\eta,\rho}$ the corresponding probability measure.
For a sampling draw $r\sim \mathbb{P}_{\eta,\rho}$, let $\mathsf{I}_t(r)\subseteq[n]$ denote the (multi)set of selected indices at iteration $t$ and $B_t(r) := |\mathsf{I}_t(r)|$ its size.
For $\rho \in \{\textsf{WR},\textsf{WOR}\}$, we have $B_t(r) \equiv B$ (deterministic).
For Poisson subsampling ($\rho = \textsf{Poisson}$; see App. \ref{app:poisson_subsampling} for details), $B_t(r)$ is random.
We write $\mathcal{B} : \bar{\mathcal{X}}^{n}\times\mathcal{R}\to \bar{\mathcal{X}}^{B_t(r)}$ and $x^{t} = \mathcal{B}(x;r) := (x_{i})_{i\in \mathsf{I}_{t}(r)}$ for the (mini-batch) subsampling operator.

\noindent\textbf{Empirical Risk Minimization. }
In supervised learning, each example is $x_i = (a_i,b_i)$ with features $a_i$ and label $b_i$.
Let $\ell(\xi;x_i)$ denote the per-sample loss of a model with parameter $\xi$.
The empirical risk minimization (ERM) problem is $\min_{\xi} F(\xi) := \frac{1}{n}\sum_{i=1}^n \ell(\xi;x_i)$.

\noindent\textbf{(Clipped) SGD as a history-dependent query.}
At iteration $t$, a subsampling scheme $\{\eta,\rho\}$ draws randomness
$r_t\sim\mathbb P_{\eta,\rho}$ and selects a minibatch
$\mathcal B_t(x;r_t)= (x_i)_{i\in\mathsf I_t}$ of size $B$.
Let $y_{<t}\in\mathcal Y_{<t}$ denote the history available at iteration $t$;
for standard SGD, $y_{<t}=\xi_{t-1}$.
Define the per-example gradients
\[
g_t(x_i;\xi_{t-1}):=\nabla_\xi \ell(\xi_{t-1};x_i),
\]
and the clipped batch gradient $\bar{g}_t(r_t):=\frac{1}{B}\sum_{i\in\mathsf{I}_t}\tilde{g}_t(x_i;\xi_{t-1})$,
where $\tilde{g}_t(x_i)$ is the per-example clipped gradient (Line~4 of Algorithm~\ref{alg:srpp-sgd}).
The pre-noise update (query) with step size $\kappa_{t}>0$ is
\[
f_t(x,y_{<t};r_t)
:=
\xi_{t-1}-\kappa_t\,\bar g_t(r_t),
\]
and we suppress $r_t$ when unambiguous.




\begin{algorithm}[t]
  \caption{SRPP-SGD with gradient clipping}
  \label{alg:srpp-sgd}
  \begin{algorithmic}[1]
    \REQUIRE Data $x_1,\dots,x_n$; loss $\ell(\xi;x)$; subsampling scheme $\{\eta, \delta\}$;
    clip norm $C$; update maps $(T_t)_{t=1}^T$; noise covariances $(\Sigma_t)_{t=1}^T$.
    \STATE Initialize parameters $\xi_0$; set history $y_{<1} \gets \xi_0$
    \FOR{$t = 1,\dots,T$}
      \STATE Sample minibatch index set $\mathsf{I}_t \subseteq [n]$ via scheme $\{\eta, \delta\}$; let $B_t \gets |\mathsf{I}_t|$
      \STATE For each $i \in \mathsf{I}_t$:
      \[
      \begin{aligned}
      g_t(x_i) &\gets \nabla_{\xi}\,\ell(\xi_{t-1}; x_i),\\
      \tilde g_t(x_i) &\gets g_t(x_i)\cdot
      \min\!\left\{1,\; \frac{C}{\lVert g_t(x_i)\rVert_2}\right\}.
      \end{aligned}
      \]
      \STATE $\bar{g}_t \gets \frac{1}{B_t} \sum_{i \in \mathsf{I}_t} \tilde{g}_t(x_i)$
      \STATE Draw $N_t \sim \mathcal{N}(0, \Sigma_t)$ and set $\hat{g}_t \gets \bar{g}_t + N_t$
      \STATE $\xi_t \gets T_t(\hat{g}_t; y_{<t})$ 
      \STATE Update history $y_{<t+1}$ 
    \ENDFOR
    \STATE \textbf{return} $\xi_T$; compute SRPP parameters from $(\Sigma_t)$ via HUC/sa-HUC accountant.
  \end{algorithmic}
\end{algorithm}

In DP-SGD, composition is key to the effectiveness of the DP privatization \citep{Abadi2016,chen2020understanding}. 
But PP/RPP generally lack graceful composition \citep{kifer2014pufferfish,Song2017,pierquin2024renyi,zhang2025differential} because conditioning on secrets via the prior can induce dependencies between the outputs of mechanisms that are otherwise independent \citep{zhang2025differential}.
The slicing process of SRPP does not circumvent such induced independence, and thus SRPP inherits the RPP's lack of proper composition properties.
To address this and make SRPP applicable to stochastic learning, we introduce the \textit{History-Uniform Cap.} 

\subsection{History-Uniform Cap}\label{sec:HUC}

We define HUC for a fixed Pufferfish scenario $(\mathcal S,\mathcal Q,\Theta)$ and a slice profile $(\mathcal U,\omega)$.
For any $(s_i,s_j)\in\mathcal Q$ and $\theta\in\Theta$, let $\Pi(\mu^\theta_{s_i},\mu^\theta_{s_j})$ denote the set of couplings between $\mu^\theta_{s_i}$ and $\mu^\theta_{s_j}$ (Definition~\ref{def:coupling}), where $\mu^\theta_{s_i}$ and $\mu^\theta_{s_j}$ are the distributions of $X$ given $(S=s_i,\theta)$ and of $X'$ given $(S=s_j,\theta)$, respectively.

\begin{definition}[History-Uniform Cap (HUC)]\label{def:HUC}
A vector $h_{t} = (h_{t,i})_{i=1}^{m} \in \mathbb{R}^{m}_{+}$ is a HUC for $\mathcal{U}=\{u_i\}^{m}_{i=1}$ at iteration $t$ if, for all $\theta\in\Theta$, $(s_{i}, s_{j})\in\mathcal{Q}$, history $y_{<t}$ in the support, and every coupling $\gamma\in \Pi(\mu^\theta_{s_i},\mu^\theta_{s_j})$, we have
$\gamma\times\mathbb{P}_{\eta,\rho}$-a.s. in $((X,X'),R_t)$,
\begin{equation}\label{eq:HUC}
\big| \langle f_t(X,y_{<t}; R_t)-f_t(X',y_{<t}; R_t), u_{i}\rangle  \big|
 \leq \sqrt{h_{t,i}},
 \forall i\in[m].
\end{equation}
\end{definition}

Thus, a HUC $h_t=(h_{t,i})_{i\in[m]}$ provides a \textit{pathwise} bound that is uniform over histories:
for every admissible coupling and every history $y_{<t}$ in support, the directional displacement
along $u_i$ is almost surely bounded by $\sqrt{h_{t,i}}$, and the constants $h_{t,i}$ do not depend
on the particular value of $y_{<t}$.
In matrix form, we may represent these caps by a PSD \textit{HUC matrix} $H_t\succeq 0$ such that,
$\gamma\times\mathbb{P}_{\eta,\rho}$-a.s.\ in $((X,X'),R_t)$,
\[
\smash{\big|\big\langle f_t(X,y_{<t};R_t)-f_t(X',y_{<t};R_t), u\big\rangle\big|^2
\;\leq\;
u^\top H_t u,
\quad \forall u\in\mathcal{U},}
\]
and in particular $h_{t,i}=u_i^\top H_t u_i$ for all $i\in[m]$.
We view $H_t$ as a deterministic, history-uniform sensitivity ellipsoid at iteration $t$, while $h_t$
records its values on the slice set $\mathcal{U}$.

We next show that a finite HUC always exists. 
We start by introducing necessary notions. 

\noindent\textbf{Discrepancy cap. }
A general (probabilistic) secret–dataset relation may alter attributes, individual data points, or global properties of datasets input to algorithms as the secret varies.
Under minibatch subsampling from a fixed finite dataset $x=(x_1,\dots,x_n)$, its effect on the iteration-$t$ update is fully captured by \textit{how many sampled points differ} between two coupled datasets.\footnote{ App.~\ref{app:group-dp-vs-srpp} clarifies how this coupling-based viewpoint differs from \textit{group} DP's worst-case $k$-neighbor model.}
For a sampling draw $r\sim\mathbb{P}_{\eta,\rho}$ with index (multi)set $\mathsf{I}_t(r)$ and size $B_t=|\mathsf{I}_t(r)|$\footnote{We replace $B_t$ by $B^{\mathrm{max}}_{t}:= \smash{\operatorname*{ess\ sup}_{R_{t}\sim \mathbb{P}_{\eta,\rho}}|\mathsf{I}_{t}(R_{t})|}$ if the size is random.}, define the \textit{discrepancy}
\begin{equation}\label{eq:discrepancy_K}
    K_t(x,x';r)
    :=\sum_{j\in\mathsf{I}_t(r)}\mathbf{1}\{x_{j}\neq x'_{j}\}.
\end{equation}
Given a coupling $\gamma \in \Pi(\mu^{\theta}_{s_{i}}, \mu^{\theta}_{s_{j}})$, define the \textit{iteration-$t$ discrepancy cap} as
\begin{equation}\label{eq:discrepancy_cap}
    K_{t}(\gamma)
:=
\operatorname*{ess\ sup}_{((X,X'),R_t) \sim \gamma\times \mathbb{P}_{\eta, \rho}}
K_{t}(X,X'; R_t).
\end{equation}
If a pair $(s_i, s_j)\in \mathcal{Q}$ guarantees that $X$ and $X'$ differ in at most $\overline{K}$ coordinates $\gamma$-a.s., then
$K_{t}(\gamma)\leq \min\{\overline{K}, B_{t}\}$ at iteration $t$ for the draw $r$.
We say that a value $K_t$ is a \textit{feasible discrepancy cap} at iteration $t$ (for the realized $r$) if $K_t \in [K_t(\gamma),\,B_t]$ for all $\theta\in\Theta$ and $(s_i,s_j)\in\mathcal{Q}$.
Intuitively, a feasible discrepancy cap provides a uniform upper bound on the number of differing sampled points that can arise at round $t$ across all priors and admissible secret pairs.
See App.~\ref{app:estimate_caps} for how to estimate $K_t(\gamma)$.

\noindent\textbf{Update map and $L_t$-Lipschitz.}
At iteration $t$, we write the pre-perturbation update as a map
$T_t:\mathbb{R}^d\times \mathcal{Y}_{<t}\to \mathbb{R}^d$ such that, for any realized sampling draw $r$,
\[
f_t(x,y_{<t};r)
=
T_t\bigl(\bar g_t(\mathcal{B}(x;r));\,y_{<t}\bigr),
\]
where $\bar g_t(\mathcal{B}(x;r))\in\mathbb{R}^d$ is the clipped and averaged minibatch gradient (Line~5 of Algorithm~\ref{alg:srpp-sgd}).
We say that $T_t$ is \textit{$L_t$-Lipschitz} in its first argument if there exists $L_t<\infty$ such that, for all $z,z'\in\mathbb{R}^d$ and all histories $y_{<t}$ in the support,
\[
\bigl\|T_t(z; y_{<t}) - T_t(z'; y_{<t})\bigr\|_{2}
\;\leq\;
L_t \,\|z - z'\|_{2}.
\]
Appendix~\ref{app:L_t_Update_example} gives two common examples of $L_t$-Lipschitz update maps.

\begin{assumption}[Slice-wise Lipschitz updates]\label{assp:slicewise_Lipschitz}
For each iteration $t$ and each direction $u_i\in\mathcal{U}$, there exists a finite constant $L_{t,i}<\infty$ such that, for all $z,z'\in\mathbb{R}^d$ and all histories $y_{<t}$ in the support,
\[
\bigl|u_i^{\top}\bigl(T_t(z; y_{<t}) - T_t(z'; y_{<t})\bigr)\bigr|
\;\leq\;
L_{t,i}\,\bigl|u_i^{\top}(z - z')\bigr|.
\]
\end{assumption}

Assumption~\ref{assp:slicewise_Lipschitz} enforces a directional non-expansiveness property: for each $u_i\in\mathcal U$, the $u_i$-component of the update cannot amplify the $u_i$-component of a gradient perturbation by more than $L_{t,i}$.
This holds for standard SGD and for preconditioned updates with bounded linear preconditioners.
Assumption~\ref{assp:slicewise_Lipschitz} is a convenient sufficient condition for a fixed finite direction set $\mathcal U$.
One may alternatively assume global $\ell_2$-Lipschitzness with constant $L_{t}$ of $T_t$, which implies Assumption~\ref{assp:slicewise_Lipschitz} with $L_{t,i}\le L_t$ for all $i$, or impose other geometric constraints. Appendix~\ref{app:perlayer_clipping} discusses structured (e.g., per-layer) clipping/perturbation variants that yield analogous slice-wise bounds with modified constants.

\begin{proposition}\label{prop:exist_HUC}
Fix a slice profile $(\mathcal U,\omega)$.
Suppose per-example gradients are $\ell_2$-clipped at threshold $C>0$.
Let $K_t$ be a feasible discrepancy cap at iteration $t$, and let $B_t\ge 1$ denote the (deterministic) minibatch size.
Suppose Assumption~\ref{assp:slicewise_Lipschitz} holds.
Then, the following is a valid HUC for each $u_i\in\mathcal{U}$,
\begin{equation}\label{eq:feasible_HUC_slicewise}
h_{t,i}
\;:=\;
\Bigl(\frac{2K_t L_{t,i} C}{B_t}\Bigr)^2.
\end{equation}
\end{proposition}


Proposition~\ref{prop:exist_HUC} provides an explicit worst-case, history-uniform bound on the directional displacement of the iteration-$t$ update.
If at most $K_t$ sampled examples differ in a minibatch of size $B_t$, then the averaged clipped gradient can change by at most $\tfrac{2K_{t}C}{B_{t}}$ in any direction.
Assumption~\ref{assp:slicewise_Lipschitz} further implies that the $u_i$-component of the update can expand this change by at most a factor $L_{t,i}$.
Thus, the directional shift along $u_i$ is bounded by $\tfrac{2K_{t} L_{t,i} C}{B_{t}}$, and hence the squared shift is bounded by $(\tfrac{2K_{t} L_{t,i} C}{B_{t}})^{2}$, which matches the cap $h_{t,i}$ in~\eqref{eq:feasible_HUC_slicewise}.
Appendix~\ref{app:attainable_HUC} characterizes HUC's \textit{minimality} and \textit{attainability}.

\subsection{Subsampling-Aware HUC}\label{sec:mean_square_HUC}

Any feasible discrepancy cap $K_t$ used in $h_{t,i}$ is a uniform worst-case upper bound on the random discrepancy $K_t(X,X';R_t)$ under $\gamma\times\mathbb{P}_{\eta,\rho}$.
While this worst-case viewpoint is sound, it can be overly conservative in SGD, where the subsampling randomness $\eta$ is independent of the prior and has a genuine averaging effect.
To better reflect this structure and improve utility, we introduce a \textit{subsampling-aware HUC} (sa-HUC), which uniformly bounds the \textit{expected} squared directional shifts over subsampling, while retaining worst-case uniformity over priors, secrets, and couplings.

\begin{definition}[Subsampling-Aware History-Uniform Cap (sa-HUC)]
\label{def:msHUC}
Fix $t$ and a slicing set $\mathcal{U} = \{u_i\}_{i=1}^m \subset \mathbb{S}^{d-1}$.
A vector $h_t^{\mathsf{sa}} = (h_{t,i}^{\mathsf{sa}})_{i=1}^m \in \mathbb{R}_+^m$ is a sa-HUC for $\mathcal{U}$ at iteration $t$
if, for all $\theta \in \Theta$, all $(s_i,s_j)\in\mathcal{Q}$, all histories $y_{<t}$ in the support, and every coupling
$\gamma \in \Pi(\mu_{s_i}^\theta,\mu_{s_j}^\theta)$, we have, for $\gamma$-almost every $(X,X')$, all $i=1,\dots,m,$
\[
  \smash{\mathbb{E}_{\eta_t}
  \bigl[\big|\langle f_t(X,y_{<t};R_t) - f_t(X',y_{<t};R_t), u_i\rangle\big|^2
  \bigr]\leq
  h_{t,i}^{\mathsf{sa}},}
\]
where the expectation is taken only over the subsampling randomness $R_t \sim \mathbb{P}_{\eta,\rho}$.
\end{definition}

Let $\overline{K}_t^2$ be any scalar such that, for all
$\theta\in\Theta$, $(s_i,s_j)\in\mathcal{Q}$, all histories $y_{<t}$
in the support, and all couplings
$\gamma\in\Pi(\mu^\theta_{s_i},\mu^\theta_{s_j})$, we have, for
$\gamma$-almost every $(X,X')$,
\begin{equation}\label{eq:ms_disc_cap}
\mathbb{E}_{\eta_t}\bigl[K_t(X,X';R_t)^2\bigr]\leq\overline{K}_t^2.
\end{equation}


\begin{proposition}\label{prop:msHUC_from_K2}
Fix an iteration $t$ and a slicing set $\mathcal{U} = \{u_i\}_{i=1}^m \subset \mathbb{S}^{d-1}$.
Assume:
\begin{itemize}[topsep=1pt,itemsep=-2pt]
    \item Per-example gradients are $\ell_2$-clipped at radius $C>0$.
    \item The minibatch size $B_t \ge 1$ is deterministic.
    \item Assumption~\ref{assp:slicewise_Lipschitz} holds with constants $\{L_{t,i}\}_{i=1}^m$.
\end{itemize}
For each $\theta\in\Theta$, $(s_i,s_j)\in\mathcal{Q}$, history $y_{<t}$ in the support, and
$\gamma\in\Pi(\mu^\theta_{s_i},\mu^\theta_{s_j})$, let $\overline{K}_t^2$ satisfy (\ref{eq:ms_disc_cap}).
Then the vector $h_t^{\mathsf{sa}} = (h_{t,i}^{\mathsf{sa}})_{i=1}^m$ with
\begin{equation}\label{eq:msHUC_closed_form_final}
  h_{t,i}^{\mathsf{sa}}
  :=
  \Bigl(\frac{2 L_{t,i} C}{B_t}\Bigr)^2 \overline{K}_t^2,
  \qquad i=1,\dots,m,
\end{equation}
is a valid sa-HUC.
\end{proposition}

Proposition~\ref{prop:msHUC_from_K2} shows that, under slice-wise Lipschitz updates and $\ell_2$-clipping, the mean-square discrepancy $\overline{K}_t^2$ directly yields a closed-form sa-HUC.
In particular, if $\overline{K}_t^2$ uniformly bounds the expected squared discrepancy $K_t^2$ over subsampling, then the mean-square directional shift along each slice $u_i$ is controlled by
$h_{t,i}^{\mathsf{sa}} = \bigl(\tfrac{2L_{t,i}C}{B_t}\bigr)^2 \overline{K}_t^2$.
Thus $h_t^{\mathsf{sa}}$ is a valid sa-HUC.

To calibrate privacy using $\mathbb{E}_{\eta_t}[K_t^2]$ instead of a worst-case $K_t$,
we now introduce mean-square variants of Ave- and Joint-SRPP that match the sa-HUC notion.
Let $\eta$ denote all subsampling randomness used in SRPP-SGD, and for each fixed
$\eta$, let $\mathcal{M}^{\theta,\eta}_s$ be the distribution of the mechanism's
output under prior $\theta$ and secret $s$, conditional on $\eta$.

\begin{definition}[Subsampling-Aware Average SRPP (sa-Ave-SRPP)]
\label{def:ms_Ave_SRPP}
Fix $\alpha>1$ and a slice profile $(\mathcal{U},\omega)$.
A mechanism $\mathcal{M}$ with subsampling randomness $\eta$ satisfies
\textit{$(\alpha,\varepsilon,\omega)$-sa-Ave-SRPP} if, for all
$\theta\in\Theta$ and all $(s_i,s_j)\in\mathcal{Q}$,
\[
  \mathbb{E}_{\eta}\Big[
    \mathtt{ASD}^{\omega}_{\alpha}
    \big(
      \mathcal{M}^{\theta,\eta}_{s_i}
      \,\big\|\,
      \mathcal{M}^{\theta,\eta}_{s_j}
    \big)
  \Big]
  \;\le\;
  \varepsilon.
\]
Here the expectation is taken only over $\eta$.
\end{definition}

Compared to standard Ave-SRPP, sa-Ave-SRPP keeps the same sliced R\'enyi divergence
and the same uniformity over priors and secrets, but relaxes the guarantee from
a worst-case bound for each subsampling pattern $\eta$ to an \textit{on-average}
bound over the subsampling randomness. 

\begin{definition}[Subsampling-Aware Joint SRPP (sa-Joint-SRPP)]
\label{def:ms_Joint_SRPP}
Under the same setup as Definition~\ref{def:ms_Ave_SRPP},
a mechanism $\mathcal{M}$ satisfies
\textit{$(\alpha,\varepsilon,\omega)$-sa-Joint-SRPP}
if, for all $\theta\in\Theta$ and all $(s_i,s_j)\in\mathcal{Q}$,
\[
  \mathbb{E}_{\eta}\Big[
    \mathtt{JSD}^{\omega}_{\alpha}
    \big(
      \mathcal{M}^{\theta,\eta}_{s_i}
      \,\big\|\,
      \mathcal{M}^{\theta,\eta}_{s_j}
    \big)
  \Big]
  \;\leq\;
  \varepsilon.
\]
\end{definition}

Here, ms-Joint-SRPP is the mean-square analogue of Joint-SRPP: it still treats the slice
index and projection jointly, but now requires that the joint sliced R\'enyi
divergence be small \textit{on average over subsampling} rather than for every fixed
realization of $\eta$. As with ms-Ave-SRPP, the only change from the standard
notion is replacing a worst-case guarantee in $\eta$ by an expectation over $\eta$.

\subsection{HUC Accountant for Gaussian Noise}

Similar to DP-SGD, an important issue of SRPP-SGD is computing the aggregated privacy cost of the entire training process, i.e., the \textit{composition part} of SRPP-SGD as a decompose-then-compose privatization scheme.
In this section, we set out to show how HUC can circumvent the RPP/SRPP's lack of graceful composition properties and enable a privacy \textit{accountant} process to compute the privacy costs along the SGD-based training process.
We refer to it as \textit{HUC accountant}.

We start by introducing the Gaussian slice R\'enyi Envelope, which gives a per-iteration, per-slice RPP cost after noise perturbation.

\begin{lemma}\label{lemma:gaussian_SRE}
    Fix $t$ and a slicing direction $u_{i}\in\mathbb{S}^{d-1}$.
    Let $N_{t}\sim \mathcal{N}(0, \Sigma_{t})$ be independent of $((X,X'), R_{t})$
    and set $Y_{t} = f_{t}(X,y_{<t}; R_{t}) + N_{t}$.
    Define $v_{t,i}:= u^{\top}_{i} \Sigma_{t} u_{i}$, and let
    $h_{t}=(h_{t,1}, \dots, h_{t,m})$ be any HUC vector. 
    Assume that $v_{t,i}>0$.
    Then, for any $\alpha > 1$, we have
    \begin{equation}
        \mathtt{D}_{\alpha}\left(
          \Pr\big(\langle Y_{t}, u_{i} \rangle \mid s_{i}, \theta\big)
          \Big\| 
          \Pr\big(\langle Y_{t}, u_{i} \rangle \mid s_{j}, \theta\big)
        \right)
        \leq \frac{\alpha}{2}\,\frac{h_{t,i}}{v_{t,i}}.
    \end{equation}
\end{lemma}

For ease of exposition, we absorb any step-size $\kappa$ and preconditioning into the covariance $\Sigma_t$ of the Gaussian noise. Lemma \ref{lemma:gaussian_SRE} gives a per-slice R\'enyi cost at iteration t directly in terms of the HUC component $h_{t,i}$ and the variance $v_{t,i}$ of the projected noise.

\begin{theorem}\label{thm:HUC_SRPP_SGD}
    Fix a slicing profile $\{\mathcal{U}, \omega\}$.
Let $h = \{h_{1}, \dots, h_{T}\}$ with each $h_{t}=(h_{t,1}, \dots, h_{t,m})$ as a \textup{valid HUC} vector. 
Let $N_{t}\overset{\text{i.i.d.}}{\sim} \mathcal{N}(0, \sigma^{2}I_{d})$.
Then, for any $\alpha>1$, we have:
\begin{itemize}[topsep=1pt,itemsep=-2pt]
    \item[(i)] \textbf{\textup{Ave-SRPP:}} For any $\sigma^2 \geq \tfrac{\alpha}{2\epsilon} \sum\limits^{T}_{t=1}\sum\limits^{m}_{\ell = 1} \omega_{\ell} h_{t,\ell}$, Algorithm \ref{alg:srpp-sgd} is $(\alpha, \epsilon, \omega)$-Ave-SRPP.

    \item[(ii)] \textbf{\textup{Joint-SRPP:}} For any $\sigma^{2}\geq \tfrac{\alpha}{2\epsilon} \sum\limits^{T}_{t=1}\max_{\ell} h_{t,\ell}$, Algorithm \ref{alg:srpp-sgd} is $(\alpha, \epsilon, \omega)$-Joint-SRPP.
\end{itemize}
\end{theorem}

Theorem~\ref{thm:HUC_SRPP_SGD} shows that once a valid sequence of $\{h_t\}_{t=1}^T$ has been constructed for a fixed slicing set $\mathcal{U}$, the overall SRPP guarantee for the entire SGD trajectory reduces to explicit, closed-form calibration conditions on the Gaussian noise. Both Ave-SRPP and Joint-SRPP are ensured by choosing a single variance parameter $\sigma^2$ large enough to offset the cumulative per-step HUC contributions, either in the $\omega$-weighted average form (part (i)) or under the worst-slice aggregation (part (ii)).

\begin{theorem}\label{thm:msHUC_SRPP_SGD}
    Fix a slicing profile $\{\mathcal{U}, \omega\}$ with $\mathcal{U} = \{u_\ell\}_{\ell=1}^m$ and $\omega \in \Delta(\mathcal{U})$.
    Let $h^{\mathsf{sa}} = \{h^{\mathsf{sa}}_{1}, \dots, h^{\mathsf{sa}}_{T}\}$ with each 
    $h^{\mathsf{sa}}_{t} = (h^{\mathsf{sa}}_{t,1}, \dots, h^{\mathsf{sa}}_{t,m})$ a \textup{valid sa-HUC} vector.
    Let $N_{t}\overset{\text{i.i.d.}}{\sim} \mathcal{N}(0, \sigma^{2}I_{d})$.
    Then, for any $\alpha>1$, we have:
    \begin{itemize}[topsep=1pt,itemsep=-2pt]
        \item[(i)] \textbf{\textup{sa-Ave-SRPP:}} If $\sigma^2 \geq \tfrac{\alpha}{2\epsilon} \sum_{t=1}^{T}\sum_{\ell = 1}^{m} \omega_{\ell} h^{\mathsf{sa}}_{t,\ell}$, then lgorithm~\ref{alg:srpp-sgd} satisfies $(\alpha, \epsilon, \omega)$-sa-Ave-SRPP.

        \item[(ii)] \textbf{\textup{sa-Joint-SRPP:}} 
        If $\sigma^{2}\geq\tfrac{\alpha}{2\epsilon} \sum_{t=1}^{T} \max_{\ell} h^{\mathsf{sa}}_{t,\ell}$, then Algorithm~\ref{alg:srpp-sgd} satisfies $(\alpha, \epsilon, \omega)$-sa-Joint-SRPP.
    \end{itemize}
\end{theorem}

Theorem~\ref{thm:msHUC_SRPP_SGD} is the mean-square analogue of Theorem~\ref{thm:HUC_SRPP_SGD}.
Once we have a sequence of sa-HUC caps $h_t^{\mathsf{sa}}$ that bound the \textit{average} (over subsampling) squared directional shifts at each iteration, the overall sa-SRPP guarantee for the full SGD trajectory reduces to the same kind of closed-form calibration on a single noise level $\sigma^2$.
Part (i) ensures an $(\alpha,\epsilon,\omega)$ sa-Ave-SRPP bound by matching the $\omega$-weighted sum of sa-HUCs, while part (ii) does the same for sa-Joint-SRPP using the worst-slice aggregation.
Compared to the standard HUC case, the structure of the accountant is unchanged; the only difference is that the guarantee is now on-average over subsampling randomness.
In Appendix \ref{app:perlayer_clipping}, we characterize SGD with per-layer gradient clipping \cite{mcmahan2017learning}, where we impose an assumption of layer-wise Lipschitz updates instead of Assumption \ref{assp:slicewise_Lipschitz}.

\begin{proposition}[Utility of sa-SRPP, informal]
\label{prop:utility_sa_srpp_informal}
Under mild conditions, subsampling-aware discrepancy caps yield more utility than worst-case caps at the same true SRPP of the mechanism. (see Proposition~\ref{prop:utility_sa_srpp_formal} in App.~\ref{app:utility_sa_srpp} for formal results)
\end{proposition}

Proposition \ref{prop:utility_sa_srpp_informal} confirms that averaging over the subsampling randomness improves utility at the same true SRPP.
App.~\ref{app:utility_sa_srpp} provides formal and detailed analysis.

\subsection{Composition of SRPP-SGD Mechanisms}

In this section, we show that HUC- and sa-HUC-based SRPP/sa-SRPP-SGD mechanisms also enjoy additive composition. For clarity, we present the setup and results in terms of HUC and SRPP; the same arguments apply verbatim to sa-HUC and sa-SRPP by replacing the corresponding notions.

We consider $J$ SRPP-SGD mechanisms $\mathcal{M}_1,\dots,\mathcal{M}_J$ acting on the same dataset $X$.
Without loss of generality, we take the primitive output of each mechanism to be a vector in a common parameter space $\mathbb{R}^d$ (e.g., a final clipped update or parameter vector), and we fix a shared slice profile $(\mathcal{U},\omega)$ with
$\mathcal{U} = \{u_\ell\}_{\ell=1}^m \subset \mathbb{S}^{d-1}$ and $\omega \in \Delta(\mathcal{U})$.
In the remainder of this section, we restrict to the following mechanisms:
for each $j\in\{1,\dots,J\}$, the run $\mathcal{M}_j:\mathcal{X}\to\mathbb{R}^d$ is an SRPP-SGD mechanism with iterates
$\theta^{(j)}_t \in \mathbb{R}^d$ ($t=0,\dots,T_j$), additive Gaussian noises
$N^{(j)}_t \sim \mathcal{N}(0,\Sigma^{(j)}_t)$, and an associated HUC (or sa-HUC) sequence
$h^{(j)} = \{h^{(j)}_t\}_{t=1}^{T_j}$ that satisfies the conditions of
Theorem~\ref{thm:HUC_SRPP_SGD} (resp. Theorem~\ref{thm:msHUC_SRPP_SGD}).
The composed mechanism is
\begin{equation}\label{eq:composition_mechanism}
    \vec{\mathcal{M}}(X)
  :=
  \big(\mathcal{M}_1(X),\dots,\mathcal{M}_J(X)\big)
  \in (\mathbb{R}^d)^J.
\end{equation}

\begin{theorem}\label{thm:composition_srpp_sgd}
    Fix a slice profile $\{\mathcal{U}, \omega\}$ and a Pufferfish scenario $\{\mathcal{S}, \mathcal{Q}, \Theta\}$.
    Let $\vec{\mathcal{M}}$ be the composition mechanism defined in (\ref{eq:composition_mechanism}).
    Fix any $\Xi \in \{\textup{Ave}, \textup{Joint}, \textup{ms-Ave}, \textup{ms-Joint}\}$.
    Suppose that, for each $j\in\{1,\dots,J\}$, the mechanism $\mathcal{M}_j$ is
    $(\alpha,\epsilon_j,\omega)$-$\Xi$-SRPP.
    Then the composed mechanism $\vec{\mathcal{M}}$ is
    $(\alpha,\sum_{j=1}^J \epsilon_j,\omega)$-$\Xi$-SRPP.
\end{theorem}

In other words, for any fixed SRPP notion (Ave, Joint, or their mean-square counterparts),
the privacy costs of $J$ SRPP-SGD mechanisms acting on the same data and slice profile
add up linearly: the composed mechanism has the same $\alpha$ and $\omega$, with
total privacy cost $\sum_j \varepsilon_j$.

\begin{figure*}[t]
    \centering

    \begin{subfigure}[b]{0.24\textwidth}
        \centering
        \includegraphics[width=\textwidth]{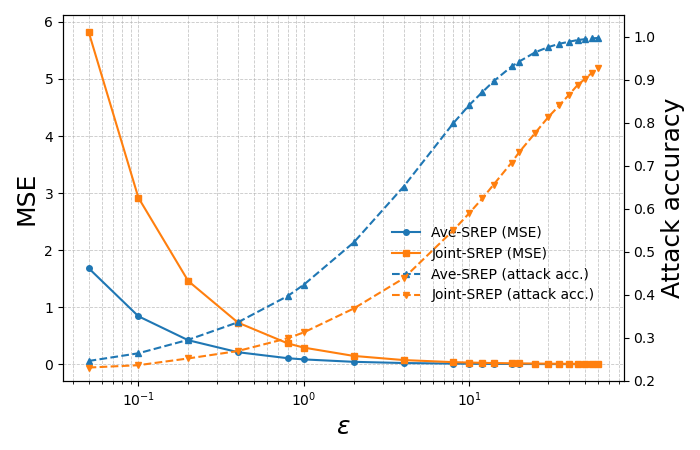}
        \subcaption{Logistic Regression}\label{subfig:logistic02}
    \end{subfigure}
    \begin{subfigure}[b]{0.24\textwidth}
        \centering
        \includegraphics[width=\textwidth]{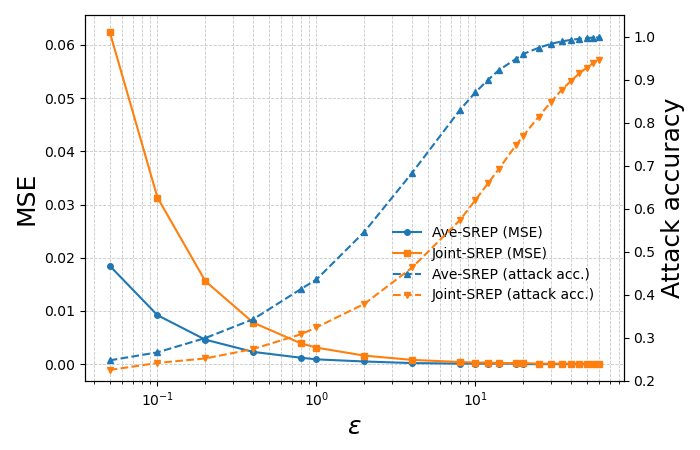}
        \subcaption{Random Forest}\label{subfig:rf02}
    \end{subfigure}
    \begin{subfigure}[b]{0.24\textwidth}
        \centering
        \includegraphics[width=\textwidth]{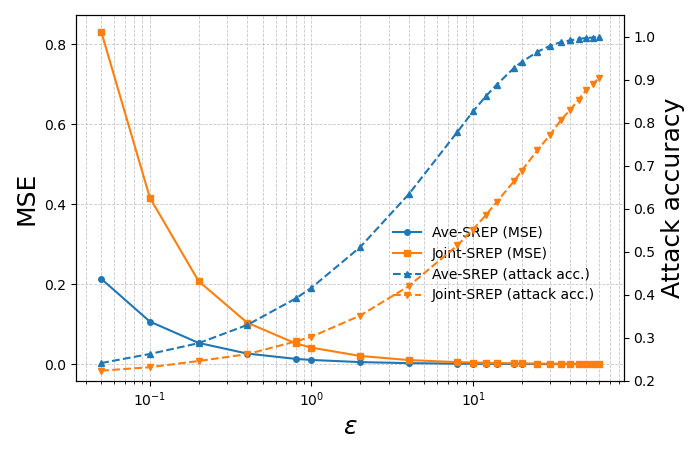}
        \subcaption{SVM}\label{subfig:svm02}
    \end{subfigure}
    \begin{subfigure}[b]{0.24\textwidth}
        \centering
        \includegraphics[width=\textwidth]{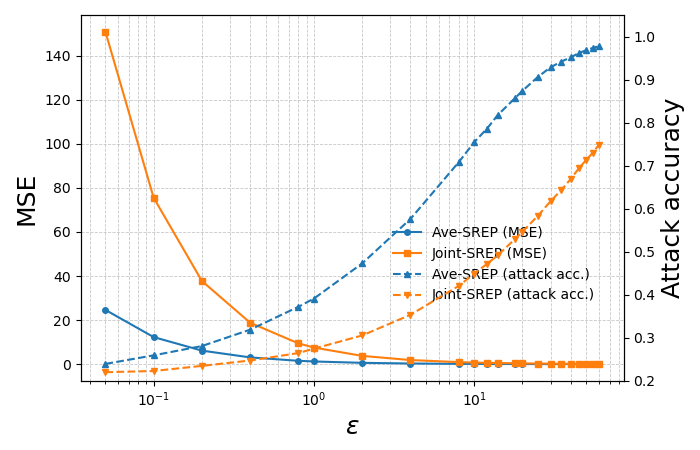}
        \subcaption{Summary statistics}\label{subfig:summary02}
    \end{subfigure}

    \vspace{0.5em}

    \begin{subfigure}[b]{0.24\textwidth}
        \centering
        \includegraphics[width=\textwidth]{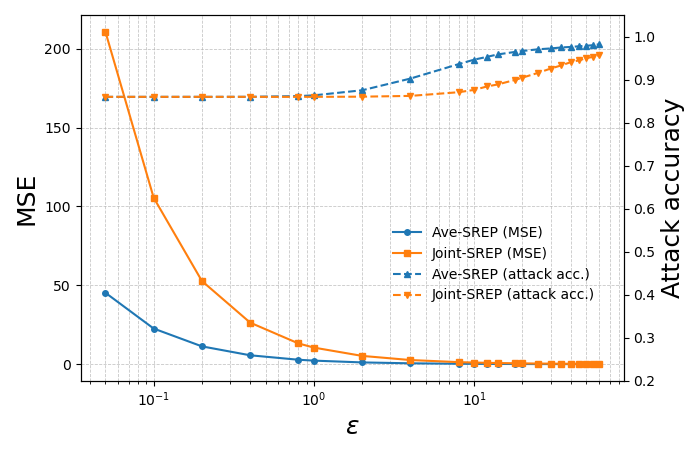}
        \subcaption{Logistic Regression}\label{subfig:logistic08}
    \end{subfigure}
    \begin{subfigure}[b]{0.24\textwidth}
        \centering
        \includegraphics[width=\textwidth]{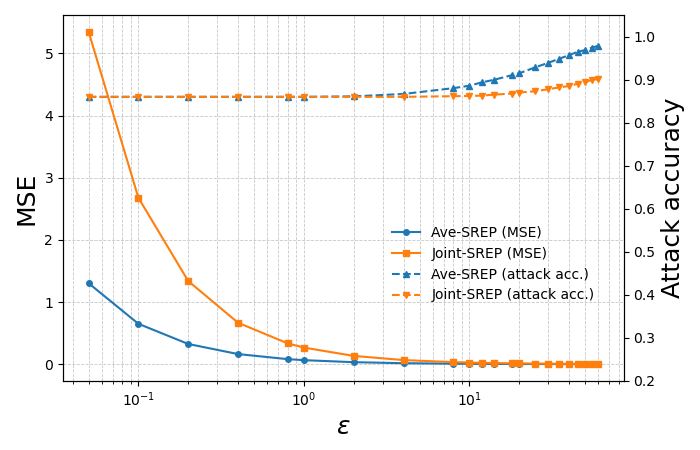}\label{subfig:rf08}
        \subcaption{Random Forest}
    \end{subfigure}
    \begin{subfigure}[b]{0.24\textwidth}
        \centering
        \includegraphics[width=\textwidth]{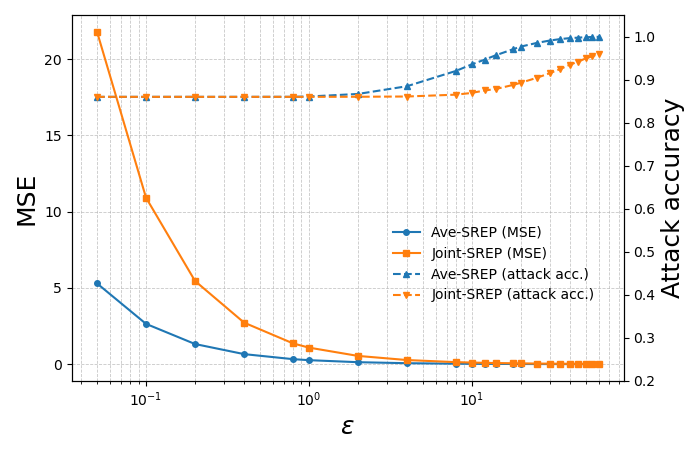}
        \subcaption{SVM}\label{subfig:svm08}
    \end{subfigure}
    \begin{subfigure}[b]{0.24\textwidth}
        \centering
        \includegraphics[width=\textwidth]{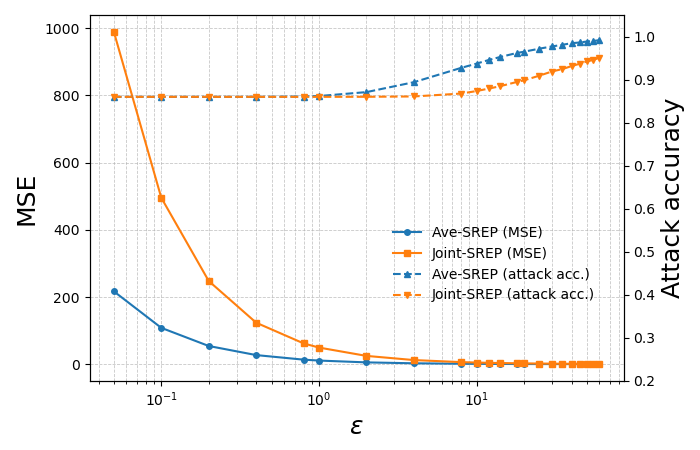}
        \subcaption{Summary statistics}\label{subfig:summary08}
    \end{subfigure}

    \vspace{0.5em}

    \begin{subfigure}[b]{0.24\textwidth}
        \centering
        \includegraphics[width=\textwidth]{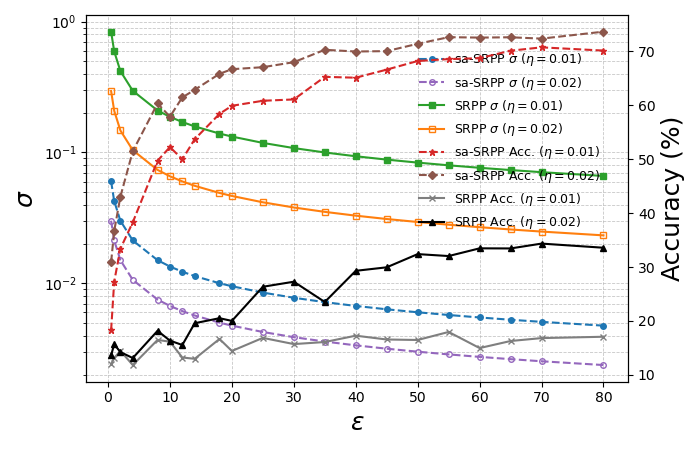}
        \subcaption{ResNet}\label{subfig:gep_resnet}
    \end{subfigure}
    \begin{subfigure}[b]{0.24\textwidth}
        \centering
        \includegraphics[width=\textwidth]{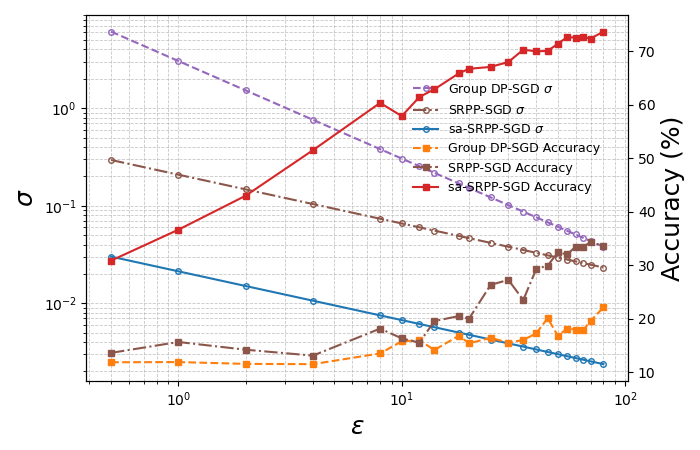}
        \subcaption{Group DP vs. SRPP}\label{subfig:group_dp}
    \end{subfigure}
    \begin{subfigure}[b]{0.24\textwidth}
        \centering
        \includegraphics[width=\textwidth]{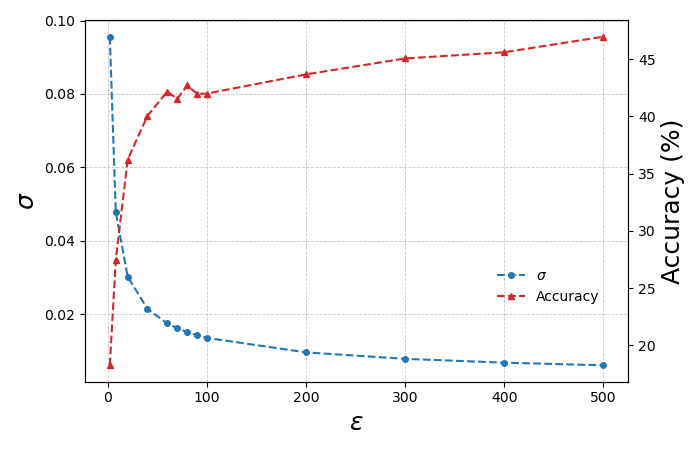}
        \subcaption{Overfitted ResNet}\label{subfig:overfit}
    \end{subfigure}
    \begin{subfigure}[b]{0.24\textwidth}
        \centering
        \includegraphics[width=\textwidth]{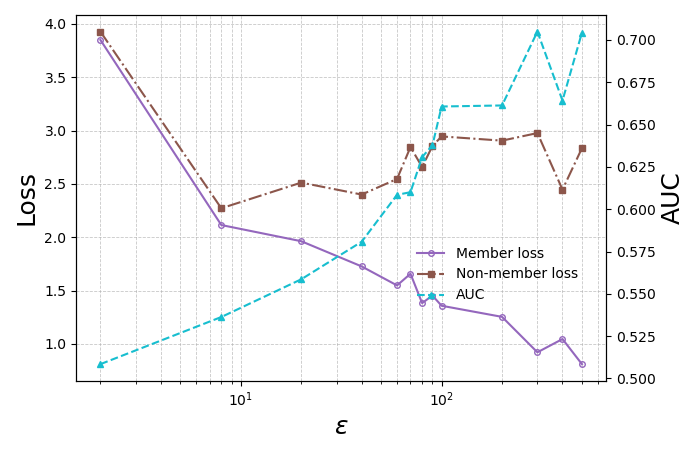}
        \subcaption{Empirical Attack Auditing}\label{subfig:overfit_attack}
    \end{subfigure}
    \caption{(a)-(d) with $\alpha = 4$ compares Ave-/Joint-SRPE on the balanced Adult dataset (uniform $0.2$ prior over races), for four query families: logistic regression, linear SVM, random forest feature importances, and summary statistics. Each panel plots the utility–privacy tradeoff as a function of the SRPP budget $\epsilon$. (e)-(h) with $\alpha=4$ performs the same tasks on the imbalanced Adult dataset (uniform $0.86$ prior over races).
    (i) compares SRPP-SGD and sa-SRPP-SGD on CIFAR-10 for two sampling rates ($\eta=0.01$ and $\eta=0.02$), showing the calibrated per–iteration noise scale $\sigma$ (left axis, log–scale) and test accuracy (right axis) as a function of the privacy budget $\epsilon$.
    (j) shows per-step Gaussian noise scale $\sigma$, and CIFAR-10 test accuracy for group-DP-SGD, SRPP-SGD, and sa-SRPP-SGD under the same clipping, subsampling ($\eta=0.01$ and $\eta=0.02$), and model setup.
    (k) presents overfitted SRPP-SGD on CIFAR-10 for $\eta = 0.02$. 
    (l) illustrates member/non-member losses and attack AUC of the empirical MIA as a function of the sa-SRPP budget $\epsilon$.
    All (i)-(l) use $\alpha = 16$.
    Appendix \ref{app:experiments} provides detailed numerical values.
    }
    \label{fig:SRPP_experiments}
\end{figure*}

\section{Experiments}\label{sec:experiments}

We empirically evaluate our Gaussian sliced Wasserstein mechanisms in both the static and stochastic-learning settings. For static queries, we instantiate SRPP guarantees (Section \ref{sec:SRPP}), and for iterative training we use SRPP/sa-SRPP together with gradient clipping to privatize SGD (Section \ref{sec:SRPP-SGD}). 
Full experimental details and additional plots are in Appendix \ref{app:experiments}.

\subsection{Gaussian Sliced Wasserstein Mechanisms}
\label{subsec:experiment_GSWM}

\textbf{Dataset and Pufferfish scenario.}
We evaluate our sliced Wasserstein mechanisms on three standard tabular benchmarks in an attribute-privacy setting \cite{Zhang2022}: Adult Census \cite{becker1996adult} (race as secret, $k=5$), Cleveland Heart Disease \cite{detrano1989international}, and Student Performance \cite{silva2008using}, with Pufferfish query set $\mathcal{Q} = \{(s_i,s_j): s_i \neq s_j\}$. 
We present results for Adult and Appendix \ref{app:experiments} shows results for the other datasets.
For Adult, we consider both the natural imbalanced race prior (majority race probability $0.86$) and a balanced, stratified subsample with a uniform $0.20$ prior over the five race categories.
We study two query families: (i) low-dimensional summary statistics (means, standard deviations, and rates of selected features) conditioned on each secret value; and (ii) machine learning (logistic regression, random forest, and SVM) queries on Adult, trained separately per secret group.
In all experiments, we use Gaussian mechanisms $Y = \mathcal{M}(X) + N$, where $N\sim \mathcal{N}(0, I_{d}\sigma)$ with $\sigma$ calibrated via our sliced SRPP. 
Utility is measured by mean squared error and related norms between $\mathcal{M}(X)$ and $Y$.
We also use a prior-aware Gaussian MAP attack on the secret attribute to perform empirical privacy auditing, reporting accuracy and advantage over the prior baseline (i.e., $0.86$ for imbalanced Adult, $0.20$ for balanced Adult); an advantage near zero indicates that the mechanism prevents meaningful secret inference beyond the prior.

Figures \ref{subfig:logistic02}-\ref{subfig:summary02}
and \ref{subfig:logistic08}-\ref{subfig:summary08} show utility-privacy trade-off as a function of the SRPP budget $\epsilon$ on balanced and imbalanced Adult datasets for four query types. 
Across all four queries, the behavior is consistent for both balanced and imbalanced settings. 
For very strong privacy (small $\epsilon$), both Ave- and Joint-SRPP force high noise and the MAP attack accuracy stays close to the prior baselines (i.e., $0.2$ and $0.86$), indicating that the mechanism reveals little additional information about the race. 
As $\epsilon$ increases, MSE decreases toward zero while attack accuracy rises monotonically, reflecting the expected utility-privacy trade-off. 
In addition, at any fixed $\epsilon$, Joint-SRPP is more conservative than Ave-SRPP: it induces larger MSE but yields uniformly lower attacker accuracy. This gap is most visible in the moderate privacy regime, where the Ave-SRPP curves already allow the attacker to reach substantially higher accuracy, while Joint-SRPP keeps the attack much closer to the prior baseline.

\subsection{SRPP-SGD Using Gaussian Noise}
\label{subsec:srpp_sgd_gaussian}

We evaluate our HUC-based SRPP-SGD and sa-SRPP-SGD mechanisms on image
classification with Gaussian noise added to clipped gradients in a
DP-SGD-style training pipeline.
In our implementation, the HUC envelope is instantiated using a
\textit{global} $\ell_2$-Lipschitz bound on the (clipped, averaged)
gradient update map. Specifically, for each iteration $t$ we bound the
change in the update under neighboring worlds by a scalar constant
$L_t$, so the global sensitivity of the update in Euclidean norm is at
most $L_t$. This global envelope automatically upper-bounds all
one-dimensional projections and therefore induces valid SRPP guarantees
for any slicing profile $(\mathcal U,\omega)$ via the linear-update setting.
Hence, Ave- and Joint-SRPP-SGD would have the same HUC. Appendix \ref{app-sub:SGD} provides more details.

\textbf{Dataset and Pufferfish scenario.}
We use CIFAR-10 (50{,}000 training and 10{,}000 test images). 
We create a two-world Pufferfish scenario, where texttt{cat} is designated as the secret class, and construct two worlds $s_0$ and $s_1$ that differ only in the prevalence of this class: in the prior, the \texttt{cat} frequency shifts from $p_{\mathrm{low}} = 0.10000$ to
$p_{\mathrm{high}} = 0.10004$.
The realized world $s_0$ is then fixed
and encoded by overwriting the CIFAR-10 training labels with the
realized label vector; all training procedures (private and non-private)
operate only on this realized world.

\textbf{Models.}
Our main model is a ResNet-22 similar to \cite{xiao2023geometry}, implemented with ELU
activations, GroupNorm-based layers, a $3\times 3$
adaptive average pooling, and a linear classifier on a standardized
$64\times 3\times 3$ feature vector. 
The architecture is made
differentially private–compatible (ready for Opacus \cite{yousefpour2021opacus}) by replacing non–DP-friendly
components (such as non-picklable shortcuts) with equivalent modules.

\textbf{SRPP-SGD and sa-SRPP-SGD training.}
Training follows the standard DP-SGD template: at each iteration,
per-sample gradients are clipped to an $\ell_2$ bound $C$ and then
perturbed by isotropic Gaussian noise. 
We use a worst-case discrepancy cap $K_{\mathrm{cap}}$ for SRPP-SGD and use $\mathbb{E}_{\eta}[K^2_t]$ for sa-SRPP-SGD.
We pre-compute, for each target SRPP and sa-SRPP budget $\epsilon$ at a fixed R\'enyi order $\alpha = 16$, a per-step noise scale $\sigma(\epsilon)$ that ensures the resulting Gaussian mechanisms satisfies the SRPP- and sa-SRPP-SGD under the given Pufferfish scenario.
In addition, we also involve group-DP-SGD, under the same setup except the noise calibration: for group-DP-SGD, we following the standard DP-SGD \cite{Abadi2016}, and apply convert it to group DP using $\mathbb{E}_{\eta}[K^2_t]$ ($\mathbb{E}_{\eta}[K^2_t]\leq K_{\mathrm{cap}}$) as the group size. 
As a baseline, we also include group-DP-SGD under the same training setup but with noise calibrated via the standard DP-SGD accountant \cite{Abadi2016}, and then converted to group DP by using $\mathbb{E}_{\eta}[K_t^2]$ (with $\mathbb{E}_{\eta}[K_t^2] \le K_{\mathrm{cap}}$) as the effective group size.

\textbf{Overfitted regime and empirical privacy attacker.}
In addition to the above accuracy-focused experiments on the full
training set, we run a separate overfitted regime tailored to empirical
privacy evaluation. We restrict the training data to a small random
subset of CIFAR-10 (on the order of $10^3$ images), substantially
increase the number of epochs, and remove or reduce weight decay, so
that the resulting private SRPP-SGD models exhibit markedly higher
training accuracy than test accuracy. On these overfitted models, we
apply the loss-threshold membership inference attack of
Yeom et al.~\cite{yeom2018privacy}: for each trained model, we draw
disjoint sets of member and non-member points, compute per-example
cross-entropy losses, and use the negative loss as a scalar membership
score (lower loss $\Rightarrow$ more likely member). We report the ROC
AUC of this score as a function of the privacy budget $\epsilon$ and
compare against a non-private overfitted baseline trained on the same
subset.

As shown in Figure \ref{subfig:gep_resnet}, across all $\epsilon$ and for both step sizes, sa-SRPP consistently requires much smaller Gaussian noise than worst-case SRPP: the sa-SRPP $\sigma$ curves lie well below their SRPP counterparts, often by a factor of several and approaching an order of magnitude around $\epsilon \approx$ 20–40. This tighter calibration directly translates into higher utility. For small budgets both methods yield low accuracy, but as $\epsilon$ increases the sa-SRPP models continue to improve, reaching roughly 60\%–70\% test accuracy at the largest budgets, whereas SRPP saturates around 20\%–35\%. The same pattern holds for both $\eta=0.01$ and $\eta=0.02$: sa-SRPP dominates SRPP in terms of the accuracy–privacy tradeoff, demonstrating that replacing a worst-case discrepancy cap $K_{\text{cap}}$ by the mean-square cap $\mathbb{E}_\eta[K_t^2]$ yields substantially milder noise for the same $\epsilon$.

Although DP, SRPP, and sa-SRPP are different privacy notions (see Appendix \ref{app:group-dp-vs-srpp}), Figure \ref{subfig:group_dp} illustrates what happens if we try to protect the same Pufferfish secret using a group-DP framework (see Appendix \ref{app:group-dp-vs-srpp} for more details). When group-DP-SGD is calibrated for an effective group size equal to $\mathbb{E}_\eta[K_t^2]$, the required Gaussian noise scale $\sigma$ is substantially larger than for SRPP-SGD, and especially sa-SRPP-SGD, across the entire range of $\epsilon$. At tight budgets (small $\epsilon$), group-DP-SGD operates at the largest $\sigma$ values on the plot and its test accuracy stays close to chance, while SRPP-SGD and sa-SRPP-SGD achieve visibly higher accuracies with much smaller noise. Even at moderate and large $\epsilon$, group-DP-SGD remains dominated in utility by both SRPP variants. This suggests that, in our Pufferfish setting, simply converting DP-SGD to group-DP using $\mathbb{E}_\eta[K_t^2]$ leads to overly conservative noise calibration and cannot practically substitute for SRPP/sa-SRPP when the goal is to protect the secret-level prevalence.

In the overfitted setting (Figure \ref{subfig:overfit}),  empirical behavior of sa-SRPP-SGD matches the privacy–utility trade-off seen in Figure \ref{subfig:gep_resnet}.
Figure \ref{subfig:overfit_attack} shows the stress-test performance of the empirical membership-inference attacks (at inference time of the model).
Here, tighter sa-SRPP budgets (smaller $\epsilon$) correspond to larger calibrated noise, lower test accuracy, and weaker membership inference, while looser budgets reduce the noise and improve accuracy at the cost of stronger membership-inference signals.

\section{Conclusion}

We introduced \textit{Sliced R\'enyi Pufferfish Privacy} (SRPP), a PP/RPP-semantics-preserving framework that replaces high-dimensional R\'enyi divergence with projection-based quantification via \textit{Average}- and \textit{Joint}-SRPP. To mitigate the curse of dimensionality in PP-family mechanism design, we developed \textit{sliced Wasserstein mechanisms} calibrated by a computable sliced Wasserstein sensitivity, yielding provable SRPP guarantees and closed-form Gaussian noise calibration without incurring the computational burden of full optimal transport. We further provided finite-sample characterizations, including conditions under which Gaussian noise achieves SRPP with high confidence.

To enable practical privacy accounting for iterative learning under probabilistic secret--dataset relationships, we proposed SRPP-SGD with gradient clipping and new accountants based on \textit{History-Uniform Caps} (HUC) and a subsampling-aware mean-square variant (sa-HUC) that improves utility. This supports a decompose-then-compose privatization strategy without requiring Wasserstein sensitivity computation or contractivity assumptions. Under a shared slicing geometry, we proved graceful additive composition for Average/Joint SRPP and their subsampling-aware counterparts, enabling modular accounting for multi-stage pipelines and cascades. 
Empirically, our experiments on both static privatization and SRPP-SGD demonstrate favorable privacy--utility trade-offs and practical scalability.


\cleardoublepage



\bibliographystyle{plain}
\bibliography{reference}

@String{Computing = "Computing" }

@String{Computer = "{IEEE} Computer" }

@String{Springer = "Springer-Verlag" }

@ArtifactSoftware{R,
    title = {R: A Language and Environment for Statistical Computing},
    author = {{R Core Team}},
    organization = {R Foundation for Statistical Computing},
    address = {Vienna, Austria},
    year = {2019},
    url = {https://www.R-project.org/},
}

@inproceedings{Song2017,
author = {Song, Shuang and Wang, Yizhen and Chaudhuri, Kamalika},
title = {Pufferfish Privacy Mechanisms for Correlated Data},
year = {2017},
isbn = {9781450341974},
url = {https://doi.org/10.1145/3035918.3064025},
doi = {10.1145/3035918.3064025},
booktitle = {Proceedings of the 2017 ACM International Conference on Management of Data},
pages = {1291–1306},
numpages = {16},
series = {SIGMOD '17}
}

@article{Dwork2014,
author = {Dwork, Cynthia and Roth, Aaron},
title = {The Algorithmic Foundations of Differential Privacy},
year = {2014},
volume = {9},
number = {3–4},
issn = {1551-305X},
url = {https://doi.org/10.1561/0400000042},
doi = {10.1561/0400000042},
journal = {Foundations and Trends in Theoretical Computer Science},
pages = {211–407},
numpages = {197}
}

@INPROCEEDINGS{Mironov2017,
  author={Mironov, Ilya},
  booktitle={2017 IEEE 30th Computer Security Foundations Symposium (CSF)}, 
  title={Rényi Differential Privacy}, 
  year={2017},
  volume={},
  number={},
  pages={263-275},
  doi={10.1109/CSF.2017.11}}

@inproceedings{Abadi2016,
author = {Abadi, Martin and Chu, Andy and Goodfellow, Ian and McMahan, H. Brendan and Mironov, Ilya and Talwar, Kunal and Zhang, Li},
title = {Deep Learning with Differential Privacy},
booktitle = {Proceedings of the 2016 ACM SIGSAC Conference on Computer and Communications Security},
year = {2016},
isbn = {9781450341394},
url = {https://doi.org/10.1145/2976749.2978318},
doi = {10.1145/2976749.2978318},
pages = {308–318},
numpages = {11},
series = {CCS '16}
}

@inproceedings{He2014,
author = {He, Xi and Machanavajjhala, Ashwin and Ding, Bolin},
title = {Blowfish Privacy: Tuning Privacy-Utility Trade-Offs Using Policies},
booktitle = {Proceedings of the 2014 ACM SIGMOD International Conference on Management of Data},
year = {2014},
isbn = {9781450323765},
publisher = {Association for Computing Machinery},
url = {https://doi.org/10.1145/2588555.2588581},
doi = {10.1145/2588555.2588581},
pages = {1447–1458},
numpages = {12},
series = {SIGMOD '14}
}

@inproceedings{Zhang2022,
author = {Zhang, Wanrong and Ohrimenko, Olga and Cummings, Rachel},
title = {Attribute Privacy: Framework and Mechanisms},
year = {2022},
isbn = {9781450393522},
publisher = {Association for Computing Machinery},
address = {New York, NY, USA},
url = {https://doi.org/10.1145/3531146.3533139},
doi = {10.1145/3531146.3533139},
booktitle = {Proceedings of the 2022 ACM Conference on Fairness, Accountability, and Transparency},
pages = {757–766},
numpages = {10},
series = {FAccT '22}
}

@inproceedings{pierquin2024renyi,
  title={R{\'e}nyi pufferfish privacy: General additive noise mechanisms and privacy amplification by iteration via shift reduction lemmas},
  author={Pierquin, Cl{\'e}ment and Bellet, Aur{\'e}lien and Tommasi, Marc and Boussard, Matthieu},
  booktitle={International Conference on Machine Learning (ICML 2024)},
  year={2024}
}

@article{kifer2014pufferfish,
  title={Pufferfish: A framework for mathematical privacy definitions},
  author={Kifer, Daniel and Machanavajjhala, Ashwin},
  journal={ACM Transactions on Database Systems (TODS)},
  volume={39},
  number={1},
  pages={1--36},
  year={2014},
  publisher={ACM New York, NY, USA}
}

@inproceedings{feldman2018privacy,
  title={Privacy amplification by iteration},
  author={Feldman, Vitaly and Mironov, Ilya and Talwar, Kunal and Thakurta, Abhradeep},
  booktitle={2018 IEEE 59th Annual Symposium on Foundations of Computer Science (FOCS)},
  pages={521--532},
  year={2018},
  organization={IEEE}
}

@article{peyre2019computational,
  title={Computational optimal transport: With applications to data science},
  author={Peyr{\'e}, Gabriel and Cuturi, Marco and others},
  journal={Foundations and Trends{\textregistered} in Machine Learning},
  volume={11},
  number={5-6},
  pages={355--607},
  year={2019},
  publisher={Now Publishers, Inc.}
}

@inproceedings{rabin2011wasserstein,
  title={Wasserstein barycenter and its application to texture mixing},
  author={Rabin, Julien and Peyr{\'e}, Gabriel and Delon, Julie and Bernot, Marc},
  booktitle={International Conference on Scale Space and Variational Methods in Computer Vision},
  pages={435--446},
  year={2011},
  organization={Springer}
}

@inproceedings{song2013stochastic,
  title={Stochastic gradient descent with differentially private updates},
  author={Song, Shuang and Chaudhuri, Kamalika and Sarwate, Anand D},
  booktitle={2013 IEEE global conference on signal and information processing},
  pages={245--248},
  year={2013},
  organization={IEEE}
}

@article{chen2020understanding,
  title={Understanding gradient clipping in private sgd: A geometric perspective},
  author={Chen, Xiangyi and Wu, Steven Z and Hong, Mingyi},
  journal={Advances in Neural Information Processing Systems},
  volume={33},
  pages={13773--13782},
  year={2020}
}

@inproceedings{dwork2006calibrating,
  title={Calibrating noise to sensitivity in private data analysis},
  author={Dwork, Cynthia and McSherry, Frank and Nissim, Kobbi and Smith, Adam},
  booktitle={Theory of cryptography conference},
  pages={265--284},
  year={2006},
  organization={Springer}
}

@article{altschuler2022privacy,
  title={Privacy of noisy stochastic gradient descent: More iterations without more privacy loss},
  author={Altschuler, Jason and Talwar, Kunal},
  journal={Advances in Neural Information Processing Systems},
  volume={35},
  pages={3788--3800},
  year={2022}
}

@book{jajodia2025encyclopedia,
  title={Encyclopedia of Cryptography, Security and Privacy},
  author={Jajodia, Sushil and Samarati, Pierangela and Yung, Moti},
  year={2025},
  publisher={Springer}
}

@inproceedings{hardt2013beyond,
  title={Beyond worst-case analysis in private singular vector computation},
  author={Hardt, Moritz and Roth, Aaron},
  booktitle={Proceedings of the forty-fifth annual ACM symposium on Theory of computing},
  pages={331--340},
  year={2013}
}

@inproceedings{sridhar2025pac,
  title={Pac-private algorithms},
  author={Sridhar, Mayuri and Xiao, Hanshen and Devadas, Srinivas},
  booktitle={2025 IEEE Symposium on Security and Privacy (SP)},
  pages={3839--3857},
  year={2025},
  organization={IEEE}
}

@inproceedings{xiao2023pac,
  title={Pac privacy: Automatic privacy measurement and control of data processing},
  author={Xiao, Hanshen and Devadas, Srinivas},
  booktitle={Annual International Cryptology Conference},
  pages={611--644},
  year={2023},
  organization={Springer}
}

@article{zhang2025breaking,
  title={Breaking the Gaussian Barrier: Residual-PAC Privacy for Automatic Privatization},
  author={Zhang, Tao and Vorobeychik, Yevgeniy},
  journal={arXiv preprint arXiv:2506.06530},
  year={2025}
}

@inproceedings{fernandes2019generalised,
  title={Generalised differential privacy for text document processing},
  author={Fernandes, Natasha and Dras, Mark and McIver, Annabelle},
  booktitle={International Conference on Principles of Security and Trust},
  pages={123--148},
  year={2019},
  organization={Springer International Publishing Cham}
}

@inproceedings{kotsogiannis2020one,
  title={One-sided differential privacy},
  author={Kotsogiannis, Ios and Doudalis, Stelios and Haney, Sam and Machanavajjhala, Ashwin and Mehrotra, Sharad},
  booktitle={2020 IEEE 36th International Conference on Data Engineering (ICDE)},
  pages={493--504},
  year={2020},
  organization={IEEE}
}

@article{nuradha2023pufferfish,
  title={Pufferfish privacy: An information-theoretic study},
  author={Nuradha, Theshani and Goldfeld, Ziv},
  journal={IEEE Transactions on Information Theory},
  volume={69},
  number={11},
  pages={7336--7356},
  year={2023},
  publisher={IEEE}
}

@inproceedings{yeom2018privacy,
  title={Privacy risk in machine learning: Analyzing the connection to overfitting},
  author={Yeom, Samuel and Giacomelli, Irene and Fredrikson, Matt and Jha, Somesh},
  booktitle={2018 IEEE 31st computer security foundations symposium (CSF)},
  pages={268--282},
  year={2018},
  organization={IEEE}
}

@article{becker1996adult,
  title={Adult},
  author={Becker, Barry and Kohavi, Ronny},
  journal={UCI Machine Learning Repository},
  volume={10},
  pages={C5XW20},
  year={1996}
}

@inproceedings{xiao2023geometry,
  title={Geometry of sensitivity: Twice sampling and hybrid clipping in differential privacy with optimal gaussian noise and application to deep learning},
  author={Xiao, Hanshen and Wan, Jun and Devadas, Srinivas},
  booktitle={Proceedings of the 2023 ACM SIGSAC Conference on Computer and Communications Security},
  pages={2636--2650},
  year={2023}
}

@article{yousefpour2021opacus,
  title={Opacus: User-friendly differential privacy library in PyTorch},
  author={Yousefpour, Ashkan and Shilov, Igor and Sablayrolles, Alexandre and Testuggine, Davide and Prasad, Karthik and Malek, Mani and Nguyen, John and Ghosh, Sayan and Bharadwaj, Akash and Zhao, Jessica and others},
  journal={arXiv preprint arXiv:2109.12298},
  year={2021}
}

@article{silva2008using,
  title={Using data mining to predict secondary school student performance},
  author={Silva, Alice},
  year={2008}
}

@article{detrano1989international,
  title={International application of a new probability algorithm for the diagnosis of coronary artery disease},
  author={Detrano, Robert and Janosi, Andras and Steinbrunn, Walter and Pfisterer, Matthias and Schmid, Johann-Jakob and Sandhu, Sarbjit and Guppy, Kern H and Lee, Stella and Froelicher, Victor},
  journal={The American journal of cardiology},
  volume={64},
  number={5},
  pages={304--310},
  year={1989},
  publisher={Elsevier}
}

@inproceedings{chen2023protecting,
  title={Protecting global properties of datasets with distribution privacy mechanisms},
  author={Chen, Michelle and Ohrimenko, Olga},
  booktitle={International Conference on Artificial Intelligence and Statistics},
  pages={7472--7491},
  year={2023},
  organization={PMLR}
}

@inproceedings{ding2022kantorovich,
  title={Kantorovich mechanism for pufferfish privacy},
  author={Ding, Ni},
  booktitle={International Conference on Artificial Intelligence and Statistics},
  pages={5084--5103},
  year={2022},
  organization={PMLR}
}

@inproceedings{kawamoto2019local,
  title={Local obfuscation mechanisms for hiding probability distributions},
  author={Kawamoto, Yusuke and Murakami, Takao},
  booktitle={European Symposium on Research in Computer Security},
  pages={128--148},
  year={2019},
  organization={Springer}
}

@inproceedings{kessler2015deploying,
  title={Deploying and evaluating pufferfish privacy for smart meter data},
  author={Kessler, Stephan and Buchmann, Erik and B{\"o}hm, Klemens},
  booktitle={2015 IEEE 12th Intl Conf on Ubiquitous Intelligence and Computing and 2015 IEEE 12th Intl Conf on Autonomic and Trusted Computing and 2015 IEEE 15th Intl Conf on Scalable Computing and Communications and Its Associated Workshops (UIC-ATC-ScalCom)},
  pages={229--238},
  year={2015},
  organization={IEEE}
}

@inproceedings{niu2019making,
  title={Making big money from small sensors: Trading time-series data under pufferfish privacy},
  author={Niu, Chaoyue and Zheng, Zhenzhe and Tang, Shaojie and Gao, Xiaofeng and Wu, Fan},
  booktitle={IEEE INFOCOM 2019-IEEE Conference on Computer Communications},
  pages={568--576},
  year={2019},
  organization={IEEE}
}

@article{papernot2021hyperparameter,
  title={Hyperparameter tuning with renyi differential privacy},
  author={Papernot, Nicolas and Steinke, Thomas},
  journal={arXiv preprint arXiv:2110.03620},
  year={2021}
}

@inproceedings{zhang2025differential,
  title={Differential Confounding Privacy and Inverse Composition},
  author={Zhang, Tao and Malin, Bradley A and Raviv, Netanel and Vorobeychik, Yevgeniy},
  booktitle={2025 IEEE International Symposium on Information Theory (ISIT)},
  pages={1--6},
  year={2025},
  organization={IEEE}
}

@inproceedings{dwork2006our,
  title={Our data, ourselves: Privacy via distributed noise generation},
  author={Dwork, Cynthia and Kenthapadi, Krishnaram and McSherry, Frank and Mironov, Ilya and Naor, Moni},
  booktitle={Annual international conference on the theory and applications of cryptographic techniques},
  pages={486--503},
  year={2006},
  organization={Springer}
}

@article{mcmahan2017learning,
  title={Learning differentially private recurrent language models},
  author={McMahan, H Brendan and Ramage, Daniel and Talwar, Kunal and Zhang, Li},
  journal={arXiv preprint arXiv:1710.06963},
  year={2017}
}

@article{bonet2025sliced,
    title={{Sliced-Wasserstein Distances and Flows on Cartan-Hadamard Manifolds}},
    author={Clément Bonet and Lucas Drumetz and Nicolas Courty},
    year={2025},
    journal={Journal of Machine Learning Research},
    volume={26},
    number={32},
    pages={1--76}
}

@article{kolouri2019generalized,
  title={Generalized sliced wasserstein distances},
  author={Kolouri, Soheil and Nadjahi, Kimia and Simsekli, Umut and Badeau, Roland and Rohde, Gustavo},
  journal={Advances in neural information processing systems},
  volume={32},
  year={2019}
}

@article{nietert2022statistical,
  title={Statistical, robustness, and computational guarantees for sliced Wasserstein distances},
  author={Nietert, Sloan and Goldfeld, Ziv and Sadhu, Ritwik and Kato, Kengo},
  journal={Advances in Neural Information Processing Systems},
  volume={35},
  pages={28179--28193},
  year={2022}
}

@inproceedings{deshpande2018generative,
  title={Generative modeling using the sliced wasserstein distance},
  author={Deshpande, Ishan and Zhang, Ziyu and Schwing, Alexander G},
  booktitle={Proceedings of the IEEE conference on computer vision and pattern recognition},
  pages={3483--3491},
  year={2018}
}

@article{chizat2020faster,
  title={Faster Wasserstein distance estimation with the Sinkhorn divergence},
  author={Chizat, Lenaic and Roussillon, Pierre and L{\'e}ger, Flavien and Vialard, Fran{\c{c}}ois-Xavier and Peyr{\'e}, Gabriel},
  journal={Advances in neural information processing systems},
  volume={33},
  pages={2257--2269},
  year={2020}
}

@article{cuturi2013sinkhorn,
  title={Sinkhorn distances: Lightspeed computation of optimal transport},
  author={Cuturi, Marco},
  journal={Advances in neural information processing systems},
  volume={26},
  year={2013}
}

@article{luise2018differential,
  title={Differential properties of sinkhorn approximation for learning with wasserstein distance},
  author={Luise, Giulia and Rudi, Alessandro and Pontil, Massimiliano and Ciliberto, Carlo},
  journal={Advances in Neural Information Processing Systems},
  volume={31},
  year={2018}
}

@inproceedings{wong2019wasserstein,
  title={Wasserstein adversarial examples via projected sinkhorn iterations},
  author={Wong, Eric and Schmidt, Frank and Kolter, Zico},
  booktitle={International conference on machine learning},
  pages={6808--6817},
  year={2019},
  organization={PMLR}
}

@article{knight2008sinkhorn,
  title={The Sinkhorn--Knopp algorithm: convergence and applications},
  author={Knight, Philip A},
  journal={SIAM Journal on Matrix Analysis and Applications},
  volume={30},
  number={1},
  pages={261--275},
  year={2008},
  publisher={SIAM}
}

@inproceedings{triastcyn2020bayesian,
  title={Bayesian differential privacy for machine learning},
  author={Triastcyn, Aleksei and Faltings, Boi},
  booktitle={International Conference on Machine Learning},
  pages={9583--9592},
  year={2020},
  organization={PMLR}
}

@article{cramer1936some,
  title={Some theorems on distribution functions},
  author={Cram{\'e}r, Harald and Wold, Herman},
  journal={Journal of the London Mathematical Society},
  volume={1},
  number={4},
  pages={290--294},
  year={1936},
  publisher={Wiley Online Library}
}

@article{xiao2008output,
  title={Output perturbation with query relaxation},
  author={Xiao, Xiaokui and Tao, Yufei},
  journal={Proceedings of the VLDB Endowment},
  volume={1},
  number={1},
  pages={857--869},
  year={2008},
  publisher={VLDB Endowment}
}

@article{dvoretzky1956asymptotic,
  title={Asymptotic minimax character of the sample distribution function and of the classical multinomial estimator},
  author={Dvoretzky, Aryeh and Kiefer, Jack and Wolfowitz, Jacob},
  journal={The Annals of Mathematical Statistics},
  pages={642--669},
  year={1956},
  publisher={JSTOR}
}

@article{massart1990tight,
  title={The tight constant in the Dvoretzky-Kiefer-Wolfowitz inequality},
  author={Massart, Pascal},
  journal={The annals of Probability},
  pages={1269--1283},
  year={1990},
  publisher={JSTOR}
}


\appendix

\section*{Outline of Appendix}

\subsubsection*{Background}

Appendix~\ref{app:add_preliminaries} provides additional preliminaries on differential privacy (App.~\ref{app:DP}), DP-SGD (App.~\ref{sec:preli_DP_SFO}), and a relaxation of RPP’s general Wasserstein mechanism (GWM), termed the \textit{Distribution-Aware General Wasserstein Mechanism} (App.~\ref{app:background_DAGWM_RPP}). Appendix~\ref{app:worst_case_group_privacy} introduces \textit{worst-case group privacy}, which captures a robust (prior-free) RPP guarantee that does not depend on the prior family $\Theta$.

Appendix~\ref{app:ER_WD} formalizes an RPP variant in which the corresponding $\infty$-Wasserstein distance is approximated via entropically regularized (Sinkhorn) optimal transport. We refer to the resulting privatization method as the \textit{Sinkhorn--Wasserstein mechanism}, and discuss key limitations for privacy quantification and accounting.

\subsubsection*{SRPP Related Discussion and Analysis}

Appendix~\ref{sec:property_srpp} presents detailed properties of SRPP, including ordering (App.~\ref{app:ordering}), post-processing immunity (App.~\ref{app:post_processing}), a deeper comparison between Ave- and Joint-SRPP (App.~\ref{app:ave-vs-joint-srpp}), and conversions from SRPP to \textit{sliced Pufferfish privacy} (App.~\ref{app:srpp_to_epsdelta}).

Appendix~\ref{app:L_t_Update_example} gives two examples of update maps that satisfy the $L_t$-Lipschitz assumption in Sec.~\ref{sec:HUC}. Appendix~\ref{app:estimate_caps} provides analyses and practical approaches for estimating the discrepancy cap in~\eqref{eq:discrepancy_cap}. Appendix~\ref{app:poisson_subsampling} remarks on Poisson subsampling in HUC/sa-HUC-based accounting.

Appendix~\ref{app:perlayer_clipping} characterizes \textit{per-layer gradient clipping} as an alternative setting to the slice-wise assumptions used in Sec.~\ref{sec:HUC}. Appendix~\ref{app:group-dp-vs-srpp} clarifies the relationship between SRPP/sa-SRPP and group DP-SGD; note that \textit{group DP} here differs from the \textit{worst-case group privacy} notion in App.~\ref{app:worst_case_group_privacy}.

Appendix~\ref{app:attainable_HUC} establishes the \textit{minimality} and \textit{attainability} of HUCs (Definition~\ref{def:HUC} in Sec.~\ref{sec:HUC}). Appendix~\ref{app:sampling_models} compares minibatch sampling with replacement (\textsf{WR}) and without replacement (\textsf{WOR}) in the HUC/sa-HUC framework. Appendix~\ref{app:utility_sa_srpp} provides formal utility characterizations showing why subsampling-aware bounds (sa-HUC/sa-SRPP) can be strictly less conservative.

\subsubsection*{Proofs and Experimental Details}

Appendices~\ref{app:proof_lemma:SRPP-envelope}--\ref{app:proof_thm:composition_srpp_sgd} contain the formal proofs of the main-body results. Appendix~\ref{app:experiments} provides additional experimental details.

\section{Additional Preliminaries}\label{app:add_preliminaries}

\subsection{Differential Privacy}\label{app:DP}

This section revisits the Differential Privacy (DP).
Let $X$ be the input dataset.
Each data point $x_i$ is defined over some measurable domain $\mathcal{X}^{\dagger}$, so that $x = (x_1, x_2, \ldots, x_n) \in \mathcal{X}=(\mathcal{X}^{\dagger})^{n}$.
We say two datasets $x, x' \in \mathcal{X}$ are \textit{adjacent} if they differ in exactly one data point.

\begin{definition}[$(\epsilon, \delta)$-Differential Privacy \cite{dwork2006calibrating}]\label{def:DP}
    A randomized mechanism $\mathcal{M}:\mathcal{X}\to\mathcal{Y}$ is said to be $(\epsilon, \delta)$-differentially private (DP), with $\epsilon\geq 0$ and $\delta\in[0,1]$, if for any pair of adjacent datasets $x, x'$, and for all $\mathcal{T}\subseteq \mathcal{Y}$, we have 
    $\Pr[\mathcal{M}(x)\in \mathcal{T}]\leq e^{\epsilon}\Pr[\mathcal{M}(x')\in\mathcal{T}] + \delta$.
\end{definition}

The parameter $\epsilon$ is usually referred to as the \textit{privacy budget}, and $\delta\in(0,1]$ represents the failure probability. 
DP represents a worst-case, input-agnostic adversarial worst-case framework.

\subsection{DP-SGD with Gradient Clipping}\label{sec:preli_DP_SFO}

In this section, we describe the Differentially private stochastic gradient descent (DP-SGD) \cite{Abadi2016,song2013stochastic} has become the standard approach for training neural networks with formal privacy guarantees.
First, We provide the following basic setup for DP-SGD similar to Sec.~\ref{sec:SRPP-SGD} for completeness.

\paragraph{Dataset.}
Let $x=(x_1,\dots,x_n)\in \mathcal{X} = \bar{\mathcal{X}}^n$ be the finite samples of data points, where each $x_{i} = (a_{i}, b_{i})\in \bar{\mathcal{X}}$ consists of features $a_{i}$ and label $b_{i}$.
In this section, we take $P_{\theta}^{X}$ as the empirical (marginal) distribution of the data samples $x=(x_1,\dots,x_n)\in \mathcal{X} = \bar{\mathcal{X}}^n$ for all $\theta\in\Theta$.
At iteration $t$, we draw a random mini-batch $\mathsf{I}_t \subseteq [n]$ using a subsampling scheme $\{\eta, \rho\}$ characterized by a rate $\eta$ and a scheme type $\rho$, where $\rho \in \{\textsf{WR}, \textsf{WOR}, \textsf{Poisson}\}$ denotes sampling with-replacement (\textsf{WR}), without-replacement (\textsf{WOR}), or via a Poisson process (\textsf{Poisson}), respectively.
Let $X^{t} = (X_{i})_{i\in \mathsf{I}_{t}}$ denote the random mini-batch subsampled from $x$, with $x^{t}$ representing a particular realization.

\begin{algorithm}[t]
  \caption{Noise-perturbed SGD with gradient clipping}
  \label{alg:privsgd_dp}
  \begin{algorithmic}[1]
    \REQUIRE Data $x_1,\dots,x_n$; loss $\ell(\xi;x)$; stepsizes $(\eta_t)_{t=0}^{T-1}$;
    batch size $B$; clip norm $C$; noise covariances $(\Sigma_t)_{t=1}^T$.
    \STATE Initialize parameters $\xi_0$
    \FOR{$t = 0,\dots,T-1$}
      \STATE Sample minibatch $\mathsf{I}_t \subseteq [n]$ with $|\mathsf{I}_t| = B$
      \STATE For each $i\in\mathsf{I}_t$:
      \[
        g_t(x_i) \gets \nabla_{\xi} \ell(\xi_t;x_i),\
        \tilde{g}_t(x_i) \gets g_t(x_i)\min\!\left\{1,\frac{C}{\|g_t(x_i)\|_2}\right\}
      \]
      \STATE $\bar{g}_t \gets \frac{1}{B}\sum_{i\in\mathsf{I}_t}\tilde{g}_t(x_i)$
      \STATE Draw $N_t \sim \mathcal{N}(0,\Sigma_t)$ and set $\hat{g}_t \gets \bar{g}_t + N_t$
      \STATE $\xi_{t+1} \gets \theta_t - \eta_t \hat{g}_t$
    \ENDFOR
    \STATE \textbf{return} $\theta_T$; compute privacy parameters from $(\Sigma_t)$ according to DP guarantee.
  \end{algorithmic}
\end{algorithm}

\paragraph{Empirical Risk Minimization.}
In supervised learning, each example is $x_i = (a_i,b_i)$ with features $a_i$ and label $b_i$.
Let $\ell(\theta;x_i)$ denote the per-sample loss of a model with parameter $\xi$.
The empirical risk minimization (ERM) problem is $\min_{\xi} F(\xi) := \frac{1}{n}\sum_{i=1}^n \ell(\theta;x_i)$.

\paragraph{Stochastic Gradient Descent. }
SGD approximately solves ERM by iteratively sampling minibatches and updating $\theta$ along noisy gradients.
At iteration $t$, a subsampling scheme $\{\eta,\rho\}$ produces a random minibatch index set $\mathsf{I}_t \subseteq [n]$ of size $B$ and corresponding minibatch $x^t = (x_i)_{i\in\mathsf{I}_t}$.
The (unclipped) stochastic gradient is
\[
    g_t(x_i) := \nabla_{\xi} \ell(\xi_t;x_i), 
    \qquad
    g_t := \frac{1}{B}\sum_{i\in\mathsf{I}_t} g_t(x_i),
\]
and the parameters are updated as $\xi_{t+1} = \xi_t - \eta_t g_t$, for stepsize $\eta_t>0$.

\paragraph{Gradient clipping and DP-SGD. }
DP-SGD \cite{song2013stochastic,Abadi2016} modifies SGD in two ways: it clips per-example gradients to control sensitivity and adds noise to the averaged gradient for privacy.
Each per-sample gradient is clipped at $\ell_2$ radius $C$,
\[
  \tilde{g}_t(x_i)
  :=
  g_t(x_i)\,\min\!\Bigl\{1,\frac{C}{\|g_t(x_i)\|_2}\Bigr\},
\]
and the averaged clipped gradient is
\begin{equation*}
    \bar{g}_t := \frac{1}{B}\sum_{i\in\mathsf{I}_t} \tilde{g}_t(x_i).
\end{equation*}
Noise-perturbed SGD with gradient clipping (Algorithm~\ref{alg:privsgd_dp}) then draws Gaussian noise
$N_t \sim \mathcal{N}(0,\Sigma_t)$, forms $\hat{g}_t = \bar{g}_t + N_t$, and updates
$\xi_{t+1} = \xi_t - \eta_t \hat{g}_t$.
In DP-SGD with isotropic Gaussian noise perturbation, the Gaussian noise variance is calibrated so that the resulting algorithm satisfies a target $(\epsilon,\delta)$-DP guarantee via, e.g., moment accountant \cite{Abadi2016}.

Since the Pufferfish scenario considers probabilistic secret-dataset relationship captured by the priors, DP and SRPP are fundamentally different in both semantics as well as how the privacy leakages are aggregated iteratively in SGD algorithms.
Appendix \ref{app:group-dp-vs-srpp} provides detailed discussion between group DP-SGD and our SRPP-SGD.

\subsection{Distribution Aware General Wasserstein Mechanism for RPP}\label{app:background_DAGWM_RPP}

In this section, we recall the Distribution Aware General Wasserstein Mechanism (DAGWM) of \cite{pierquin2024renyi}, which refines the General Wasserstein Mechanism (GWM) for R\'enyi Pufferfish Privacy (RPP) by replacing worst-case $\infty$-Wasserstein sensitivities with $p$-Wasserstein quantities that better capture average-case geometry. 
The DAGWM is proved to offer more utility than the GWM at no additional privacy cost \cite{pierquin2024renyi}.

\begin{definition}[$p$-Wasserstein distance]\label{def:p_W}
Fix a norm $\|\cdot\|$ on $\mathbb{R}^d$ and let $1 \le p < \infty$. 
For probability measures $\nu, \mu$ on $\mathbb{R}^d$ with finite $p$-th moments, the $p$-Wasserstein distance is
\[
W_{p}(\nu, \mu): = \inf_{\pi\in\Pi(\nu, \mu)}
\Bigg(\int_{\mathbb{R}^d\times\mathbb{R}^d} \|x-y\|^{p} d\pi(x,y)\Bigg)^{1/p}
\]
If either $\nu$ or $\mu$ lacks a finite $p$-th moment, set $W_p(\nu, \mu)=+\infty$.
\end{definition}

For $1\leq p\le s\leq\infty$, $W_p(\nu, \mu)\leq W_s(\nu, \mu)$.
Moreover, $W_{p}(\nu, \mu)$ is nondecreasing in $p$, and for measures with bounded support we can relate large-$p$ Wasserstein distances to the $\infty$-Wasserstein distance.

To improve the utility of the GWM, Pierquin et al. \cite{pierquin2024renyi} leverage $p$-Wasserstein distance to constrain the sensitivity instead of $\infty$-Wasserstein distances.
This replaces the worst-case transportation cost between $\Pr(f(X)\mid s_i, \theta)$ and $\Pr(f(X)\mid s_j, \theta)$ by moments of the transportation cost. 
Let $\zeta$ be a noise distribution and, for $(s_i,s_j)\in\mathcal{Q}$ and $\theta\in\Theta$, write $\mu_i^{\theta}=P(f(X)\mid s_i,\theta)$. For $q\ge 1$ and $\alpha>1$, define the $p$-Wasserstein sensitivity as
\begin{equation}\label{eq:Delta_DAGWM}
    \begin{aligned}
    \Delta^{\zeta,q,\alpha}_G: = \max_{(s_i,s_j),\theta}&
\inf_{(X,Y)\sim\Pi(\mu_i^{\theta},\mu_j^{\theta})}
\mathbb{E}\Big[\\
&\exp\big(q(\alpha-1)\mathtt{D}_{q(\alpha-1)+1}(\zeta,\zeta\ast(X-Y))\big)\Big].
\end{aligned}
\end{equation}

The DAGWM calibration term is governed by full-dimensional p-Wasserstein distances between the conditional output distributions $\Pr(f(X)\mid s,\theta)$. When $f(X)\in\mathbb{R}^d$ is high-dimensional, accurately estimating these distances from empirical samples and computing them at scale is known to be difficult in general for optimal transport. Therefore, while DAGWM improves the distribution-awareness of RPP relative to worst-case $W_\infty$-based calibration, \textit{it still inherits the high-dimensional intractability associated with full-dimensional Wasserstein sensitivity.}
This motivates replacing full-dimensional OT calibration by sliced constructions, where the mechanism depends only on one-dimensional projections and remains computable and calibratable in high dimensions.

\subsection{Worst-Case Group Privacy}\label{app:worst_case_group_privacy}

Following~\cite{Song2017,pierquin2024renyi}, in this section, we define a \textit{worst-case group privacy} baseline for a
given Pufferfish scenario $(\mathcal{S},\mathcal{Q},\Theta)$, which is a \textit{prior-free} (group-DP-type)
quantification that does not exploit any distributional restrictions encoded by $\Theta$.

Fix $n\in\mathbb{N}^{*}$ and group dimensions $d_1,\dots,d_n\in\mathbb{N}^{*}$.
Let $\mathcal{X}\subset\mathbb{R}^{d}$ and define the block-structured dataset space
\[
\mathcal{X}^{\mathrm{all}} := \mathcal{X}^{d_1}\times\cdots\times \mathcal{X}^{d_n}.
\]
For $x\in\mathcal{X}^{\mathrm{all}}$, write $x=(x^{(1)},\dots,x^{(n)})$ with
$x^{(k)}\in\mathcal{X}^{d_k}$ the $k$-th group (block).

\paragraph{Group adjacency and worst-case group sensitivity.}
For each $k\in\{1,\dots,n\}$, define the group adjacency relation
\[
D_k
:= \Bigl\{(x,x')\in(\mathcal{X}^{\mathrm{all}})^2 \;:\; x^{(j)}=x'^{(j)}\ \ \forall j\neq k \Bigr\},
\]
i.e., $(x,x')\in D_k$ iff $x$ and $x'$ may differ arbitrarily within the $k$-th group and are identical
outside that group. For any query $f:\mathcal{X}^{\mathrm{all}}\to\mathbb{R}^{m}$, define the worst-case group
sensitivity (with respect to $\|\cdot\|$ on $\mathbb{R}^m$) as
\[
\Delta_{\mathrm{GROUP}}(f)
:= \max_{k\in\{1,\dots,n\}}\ \sup_{(x,x')\in D_k}\ \|f(x)-f(x')\|.
\]
Equivalently, letting $D:=\bigcup_{k=1}^n D_k$, one has
\[
\Delta_{\mathrm{GROUP}}(f)=\sup_{(x,x')\in D}\ \|f(x)-f(x')\|.
\]

\paragraph{Worst-case (prior-free) RPP under group adjacency.}
A prior-free (worst-case) instantiation of RPP requires the certificate to hold uniformly over the chosen
adjacency class; here we take $D=\bigcup_{k=1}^n D_k$. Consider the additive-noise mechanism
$M(x)=f(x)+Z$. We say that $M$ satisfies \textit{worst-case (prior-free) RPP} at order $\alpha>1$ with budget
$\varepsilon$ if
\[
\sup_{(x,x')\in D}\ \mathtt{D}_\alpha\!\bigl(M(x)\,\|\,M(x')\bigr) \;\leq\; \varepsilon.
\]
In the noise-calibration form of \cite{pierquin2024renyi}, a sufficient condition is
\[
R_\alpha\!\bigl(Z,\Delta_{\mathrm{GROUP}}(f)\bigr)\leq \varepsilon
\]
where 
\[
R_\alpha(Z,z):=\sup_{\|u\|\le z} \mathtt{D}_\alpha(Z-u\,\|\,Z)
\]
is the R\'enyi envelop (i.e., the worst-case R\'enyi divergence) between the noise distribution and its shifts by at most $z$.

\paragraph{Validity and Conservativeness.}
By construction, the above worst-case (prior-free) specification yields a \textit{valid} RPP certificate:
the privacy guarantee is quantified by a supremum over the entire adjacency class
$D=\bigcup_{k=1}^n D_k$ and therefore does not rely on any distributional restriction (in particular, it
does not exploit $\Theta$). In particular, if $R_\alpha(Z,\Delta_{\mathrm{GROUP}}(f))\le \varepsilon$, then
the additive-noise mechanism $M(x)=f(x)+Z$ satisfies worst-case RPP at order $\alpha>1$ with budget
$\varepsilon$:
\[
\sup_{(x,x')\in D} D_\alpha\!\bigl(M(x)\,\|\,M(x')\bigr)\leq \varepsilon.
\]

However, this prior-free baseline is typically \textit{utility-draining} because
$\Delta_{\mathrm{GROUP}}(f)$ is an \textit{absolute worst-case} sensitivity over all group-adjacent dataset
pairs; for highly correlated data, such group-level worst-case protections can force noise magnitudes that
scale with the group size and may destroy utility~\cite{Song2017}. Since $R_\alpha(Z,z)$ is monotone
nondecreasing in $z$ for common noise families, achieving the same privacy budget
$\varepsilon$ under the prior-free baseline generally requires larger noise, which in turn degrades utility.

\subsection{Sinkhorn-Wasserstein Mechanism for RPP }\label{app:Sinkhorn}

In this section, we introduce the \textit{Sinkhorn-Wasserstein (Sinkhorn-W) mechanism} as an estimated version of the standard full-dimensional WD, where the corresponding Wasserstein distance is estimated by the Sinkhorn algorithm~\cite{cuturi2013sinkhorn,luise2018differential,wong2019wasserstein,chizat2020faster}.

\subsubsection{Entropy-Regularized Wasserstein Distance}\label{app:ER_WD}

We start by defining the \textit{entropy-regularized} Wasserstein distance (ER-WD). 
The WD in Definition \ref{def:WD} can be equivalently represented as follows. 
Let
\[
J_{t} :=\left\{(x,y)\in \mathbb{R}^{d} \times \mathbb{R}^{d} : \|x-y\|\leq t \right\},
\]
\[
\Pi_{t}(\nu, \mu) := \left\{\pi \in \Pi(\nu, \mu) : \pi(J_{t}) = 1 \right\}.
\]
Then, 
\begin{equation}\label{eq:WD_new_rep}
    W_{\infty}(\nu,\mu) = \inf\left\{t\geq 0 : \Pi_{t}(\nu,\mu) \neq \emptyset \right\}.
\end{equation}
The ER-WD re-formulates the WD in (\ref{eq:WD_new_rep}) by adding a regularization based on Kullback–Leibler (KL) divergence.
For any regularization parameter $\hat{\varepsilon} >0$, the ER-WD is then defined as 
\begin{equation}\label{eq:ER_WD}
    \begin{aligned}
        W^{(\hat{\varepsilon})}_{\infty}(\nu,\mu) := \inf_{t\geq0, \pi\in \Pi(\nu, \mu)} \left\{t +\hat{\varepsilon} \mathrm{KL}\left(\pi \| \nu\otimes \mu\right): \pi(J_{t}) = 1 \right\},
    \end{aligned}
\end{equation}
where 
\[
\mathrm{KL}\left(\pi \| \nu\otimes \mu\right) = \int\log\left(\tfrac{d\pi}{d(\nu\otimes \mu)}\right)d\pi
\]
is the KL divergence with $\mathrm{KL}\left(\pi \| \nu\otimes \mu\right) =+\infty$ if $\pi \not\ll (\nu \otimes \mu)$.
For ER-WD, we still minimizing the smallest radius $t$ for which a coupling exists supported inside $\|x-y\|\leq t$, but we also penalize couplings that have low entropy (i.e., "too concentrated"). 
Equivalently, (\ref{eq:ER_WD}) can be viewed as
\[
W^{(\hat{\varepsilon})}_{\infty}(\nu,\mu) = \inf_{t\geq0}\left\{t +\hat{\varepsilon} \inf_{\pi\in \Pi_{t}(\nu, \mu)}\mathrm{KL}\left(\pi \| \nu\otimes \mu\right)\right\}.
\]
The discrete formulation follows immediately by restricting $\nu, \mu$ to empirical measures and optimizing over transport matrices ($\pi$) with fixed marginals, with the same entropic KL regularizer and an $\infty$-cost implemented through a distance-threshold support constraint.

We can also define the more general entropy-regularized $p$-WD as
\[
W^{p}_{p,\hat{\varepsilon}}(\nu,\mu):= \inf_{\pi\in \Pi(\nu,\mu)}\left\{\int \|x-y\|^{p}d\pi(x,y) + \hat{\varepsilon} \mathrm{KL}(\pi\| \nu\otimes \mu)\right\}.
\]

\subsubsection{Sinkhorn Algorithm}\label{app:Sinkhorn_algorithm}

For each $t\geq 0$, define the KL-minimizing coupling
\[
\pi_t^\star \in \arg\min_{\pi\in \Pi_t(\nu,\mu)} \mathrm{KL}\!\left(\pi \,\middle\|\, \nu\otimes\mu\right),
\]
and the associated profile objective
\[
J_{\hat{\varepsilon}}(t)
\;:=\;
t \;+\; \hat{\varepsilon}\,\mathrm{KL}\!\left(\pi_t^\star \,\middle\|\, \nu\otimes\mu\right).
\]
Then \eqref{eq:ER_WD} can be written compactly as
\[
W^{(\hat{\varepsilon})}_{\infty}(\nu,\mu)=\inf_{t\geq 0} J_{\hat{\varepsilon}}(t).
\]

A standard result for KL projections onto marginal constraints is that the solution (when it exists) admits a multiplicative scaling representation: there exist strictly positive measurable functions
$u_{t}:\mathbb{R}^{d}\to (0,\infty)$ and $v_{t}:\mathbb{R}^{d}\to (0,\infty)$, unique up to the rescaling
$(u_t,v_t)\mapsto (c\,u_t,c^{-1}v_t)$ for $c>0$, such that
\[
d\pi_{t}(x,y)
=
u_{t}(x)\, v_{t}(y)\,\mathbf{1}\{\|x-y\|\leq t\}\, d\nu(x)\, d\mu(y).
\]
Imposing the marginal constraints $\pi_t(\cdot,\mathbb{R}^d)=\nu$ and $\pi_t(\mathbb{R}^d,\cdot)=\mu$ yields the IPFP fixed-point equations
\[
\int_{\mathbb{R}^d} v_t(y)\,\mathbf{1}\{\|x-y\|\leq t\}\, d\mu(y) \;=\; \frac{1}{u_t(x)}
\quad \text{for $\nu$-a.e.\ } x,
\]
\[
\int_{\mathbb{R}^d} u_t(x)\,\mathbf{1}\{\|x-y\|\leq t\}\, d\nu(x) \;=\; \frac{1}{v_t(y)}
\quad \text{for $\mu$-a.e.\ } y.
\]

This motivates the (continuous) Sinkhorn--Knopp algorithm (also known as the iterative proportional fitting procedure, IPFP) \cite{knight2008sinkhorn}:
initialize a measurable function $v_t^{(0)}:\mathbb{R}^d\to(0,\infty)$, and for $k\geq 0$ define measurable functions
$u_t^{(k+1)},v_t^{(k+1)}:\mathbb{R}^d\to(0,\infty)$ by
\[
u_t^{(k+1)}(x)
=
\left(\int_{\mathbb{R}^d} \mathbf{1}\{\|x-y\|\leq t\}\, v_t^{(k)}(y)\, d\mu(y)\right)^{-1},
\]
\[
v_t^{(k+1)}(y)
=
\left(\int_{\mathbb{R}^d} \mathbf{1}\{\|x-y\|\leq t\}\, u_t^{(k+1)}(x)\, d\nu(x)\right)^{-1},
\]
for $\nu$-a.e. $x$ and $\mu$-a.e. $y$, respectively. 
If the iterates $(u_t^{(k)},v_t^{(k)})$ converge to $(u_t,v_t)$, then the coupling
\[
d\pi_t(x,y)=u_t(x)v_t(y)\,\mathbf{1}\{\|x-y\|\leq t\}\,d\nu(x)\,d\mu(y)
\]
has marginals $\nu$ and $\mu$.

We write $\pi_t^{\mathrm{SK}}$ for the coupling induced by the Sinkhorn(-Knopp) iterates after termination, i.e.,
\[
d\pi_t^{\mathrm{SK}}(x,y)=u_t^{(k)}(x)v_t^{(k)}(y)\mathbf{1}\{\|x-y\|\leq t\}\,d\nu(x)\,d\mu(y),
\]
for the final iterate $k$.
Given a candidate set $\mathcal{T}\subseteq [0,\infty)$, we define the Sinkhorn-computed ER-$\infty$-WD by
\begin{equation}\label{eq:Sinkhorn_er_wd}
    W^{(\hat{\varepsilon}),\mathrm{SK}}_{\infty}(\nu,\mu;\mathcal{T})
\;:=\;
\min_{t\in\mathcal{T}}
\left\{
t+\hat{\varepsilon}\,\mathrm{KL}\!\left(\pi_t^{\mathrm{SK}} \,\middle\|\, \nu\otimes\mu\right)
\right\}.
\end{equation}

\paragraph{Discrete Case.}
Let $\nu=\sum_{i=1}^n a_i\delta_{x_i}$ and $\mu=\sum_{j=1}^m b_j\delta_{y_j}$ with $a\in\Delta_n$ and $b\in\Delta_m$.
Define the admissibility mask $M(t)\in\{0,1\}^{n\times m}$ by
$M_{ij}(t)=\mathbf{1}\{\|x_i-y_j\|\leq t\}$, and the reference matrix
$K(t):=(ab^\top)\odot M(t)$.
Then $\pi_t$ admits the scaling form $\pi_t=\mathrm{diag}(u)\,K(t)\,\mathrm{diag}(v)$, where
$u\in\mathbb{R}_+^n$ and $v\in\mathbb{R}_+^m$ are obtained by the Sinkhorn updates
$u\leftarrow a\oslash(K(t)v)$ and $v\leftarrow b\oslash(K(t)^\top u)$.

\paragraph{Exactness w.r.t.\ ER-WD.}
Assume that for every $t\in\mathcal{T}$ the Sinkhorn/IPFP iterate converges to the exact KL projection,
i.e., $\pi_t^{\mathrm{SK}}=\pi_t^\star\in\arg\min_{\pi\in\Pi_t(\nu,\mu)}\mathrm{KL}(\pi\|\nu\otimes\mu)$.
If, in addition, $\mathcal{T}$ contains a minimizer of the outer problem, i.e.,
$\arg\min_{t\geq 0}\{t+\hat{\varepsilon}\inf_{\pi\in\Pi_t(\nu,\mu)}\mathrm{KL}(\pi\|\nu\otimes\mu)\}\cap \mathcal{T}\neq\emptyset$,
then 
\[
W^{(\hat{\varepsilon}),\mathrm{SK}}_{\infty}(\nu,\mu;\mathcal{T})=W^{(\hat{\varepsilon})}_{\infty}(\nu,\mu).
\]

\paragraph{Relation to $W_\infty$.}
For any $\hat{\varepsilon}>0$ and any $\mathcal{T}\subseteq[0,\infty)$, one has
\[
W^{(\hat{\varepsilon}),\mathrm{SK}}_{\infty}(\nu,\mu;\mathcal{T})\geq \min\left\{t\in\mathcal{T}\;:\; \Pi_{t}(\nu,\mu)\neq \emptyset  \right\}.
\]
Moreover, if $\hat{\varepsilon}\downarrow 0$ and the inner KL problem is solved exactly (so that $\pi_t^{\mathrm{SK}}=\pi_t^\star$),
then
\[
\lim_{\hat{\varepsilon}\downarrow 0} W^{(\hat{\varepsilon})}_{\infty}(\nu,\mu)=W_\infty(\nu,\mu),
\]
and, in the discrete setting, taking $\mathcal{T}$ to be the set of distinct distances $\{\|x_i-y_j\|\}_{i,j}$ yields
$\min\{t\in\mathcal{T}:\Pi_t(\nu,\mu)\neq\emptyset\}=W_\infty(\nu,\mu)$.

In the present $\infty$-WD formulation, setting $\hat{\varepsilon}=0$ removes the KL term and reduces the inner problem at a fixed threshold $t$ to the feasibility question $\Pi_t(\nu,\mu)\neq\emptyset$, i.e., the existence of a coupling supported on $J_t=\{(x,y):\|x-y\|\leq t\}$. 
In the discrete setting $\nu=\sum_i a_i\delta_{x_i}$, $\mu=\sum_j b_j\delta_{y_j}$, this feasibility problem is equivalent to the existence of a nonnegative transport matrix with prescribed marginals and support constrained to $\{(i,j):\|x_i-y_j\|\leq t\}$, and can be decided via standard max-flow/feasible-circulation methods (equivalently, Hall-type conditions on the bipartite support graph). 
Consequently, the computational advantage of Sinkhorn/IPFP is tied to $\hat{\varepsilon}>0$: the entropic regularization yields a strictly convex projection problem with a multiplicative scaling structure, whereas the unregularized case $\hat{\varepsilon}=0$ reverts to a (potentially large-scale) feasibility/flow problem and does not admit the same Sinkhorn scaling iterations.

\subsubsection{Limitations of Sinkhorn-Wasserstein Mechanism}

For any query function $f:\mathcal{X}\to\mathbb{R}^d$ and any $\mathcal{T}\subseteq[0,\infty)$, define the \textit{$\infty$-Sinkhorn-Wasserstein sensitivity (Sinkhorn-W sensitivity)} as
\begin{equation}
    \begin{aligned}
        &\Delta^{\mathrm{SK}, \hat{\varepsilon}}_{\infty}(\mathcal{T}) \\
        &:= \max_{(s_i,s_j)\in\mathcal{Q},\theta\in\Theta}
W^{(\hat{\varepsilon}),\mathrm{SK}}_{\infty}\left(P^{f,s_{i}}_{\theta}, P^{f,s_{j}}_{\theta};\mathcal{T}\right).
    \end{aligned}
\end{equation}
Define the full-dimensional R\'enyi envelope:
\[
R_{\alpha}(\zeta,z)= \sup_{\|a\|<z} \mathtt{D}_{\alpha}(\zeta_{-a},\zeta),
\]
where $\|\cdot\|$ is the same norm as in the definition of $W_{\infty}$ in Definition \ref{def:WD}.

The Sinkhorn-Wasserstein (Sinkhorn-W) mechanism replaces the unsliced $W_\infty$ sensitivity by an
entropy-regularized surrogate, Sinkhorn-W sensitivity, computed numerically. While this can reduce optimization difficulty,
it introduces additional approximation layers that complicate privacy certification. 
In particular, if the sensitivity used for calibration is under-estimated (due to regularization bias, grid omission,
or premature solver termination), then the injected noise may be insufficient for the intended
worst-case guarantee, and the resulting privacy statement may no longer apply to the true
$W_\infty$-based notion.

\paragraph{Regularization Changes the Certified Notion.}
For $\hat{\varepsilon}>0$, the entropy-regularized objective $W^{(\hat{\varepsilon})}_{\infty}$ generally
does not equal $W_{\infty}$. Consequently, calibrating noise using
$\Delta^{\mathrm{SK},\hat{\varepsilon}}_{\infty}(\mathcal T)$ yields a certificate tied to the
\textit{regularized} sensitivity. Interpreting this certificate as a guarantee for the true
$W_\infty$ sensitivity requires an additional approximation argument as $\hat{\varepsilon}\downarrow 0$
(and, in the discrete case, a sufficiently rich threshold set $\mathcal T$). Without such an argument,
the mechanism may only be provably private with respect to the surrogate objective rather than the
original worst-case coupling radius $W_\infty$.

\paragraph{Two Additional Approximation Layers Beyond Sampling.}
Even for fixed $\hat{\varepsilon}>0$, the computed value
$W^{(\hat{\varepsilon}),\mathrm{SK}}_{\infty}(\nu,\mu;\mathcal T)$ can differ from
$W^{(\hat{\varepsilon})}_{\infty}(\nu,\mu)$ for two independent reasons.
First, the outer minimization is discretized by restricting to $\min_{t\in\mathcal T}$
instead of $\inf_{t\ge 0}$; if $\mathcal T$ omits near-minimizers, the reported value may deviate
from the true optimum.
Second, the inner KL projection is solved only approximately by terminating Sinkhorn/IPFP after
finitely many iterations. Without a certified primal--dual gap (or other a posteriori bound),
the resulting coupling may violate the marginal constraints beyond the desired tolerance, and the
reported objective value can be inaccurate.
From the standpoint of privacy calibration, the critical failure mode is \textit{under-estimation}:
if the computed sensitivity proxy is smaller than the true one, the selected noise scale may be
insufficient for the intended worst-case guarantee.

\paragraph{Conservativeness Can Be Large and Difficult to Control.}
At the optimal $t$, we pay $t$ plus an entropic penalty that is generally strictly positive for the feasible set $\Pi_{t}$. 
Hence, $W^{(\hat{\varepsilon})}_{\infty}$ unless a KL-minimizing feasible coupling has zero KL (which is rare).
This inflates the sensitivity proxy and can force the mechanism to inject unnecessary noise, degrading utility. Decreasing
$\hat{\varepsilon}$ reduces this regularization bias, but it also exacerbates numerical issues
(see the next paragraph). Thus, there is no monotone regime in which the method is simultaneously
computationally benign and tightly aligned with $W_\infty$.

\paragraph{Ill-Conditioning as $\hat{\varepsilon}\downarrow 0$.}
The computational appeal of Sinkhorn/IPFP relies on $\hat{\varepsilon}>0$, which yields a strictly convex
KL projection with a multiplicative scaling structure. 
As $\hat{\varepsilon}\downarrow 0$, the optimal $t$ approaches the feasibility threshold.
At (or near) this threshold, the mask $M(t)$ becomes sparse or near-disconnected, making scaling interations unstable or slow (and sometimes infeasible without enlargint $t$).

\paragraph{Runtime and Memory Remain Coupling-Scale in General.}
In the discrete setting with $n$ samples per conditional law, each Sinkhorn/IPFP update with a dense
pairwise structure costs $O(n^2)$ arithmetic (and typically $\Omega(n^2)$ memory if the pairwise
distance/kernel is materialized). Moreover, the required number of iterations depends on
$\hat{\varepsilon}$ and on the chosen marginal-violation tolerance. Thus, although Sinkhorn avoids
generic coupling optimization, its worst-case scaling remains coupling-sized and is, in general,
less favorable than sliced calibration pipelines that reduce computation to sorting-based one-dimensional
operations.

\paragraph{Computational Properties. }
Fix an instance with empirical measures supported on $n$ points in $\mathbb{R}^d$.
Let $\mathcal{T}$ be a finite set of distance thresholds and let $K$ be the number of Sinkhorn/IPFP
iterations run per $t\in\mathcal{T}$.
Forming the full pairwise distance table $\{\|x_i-y_j\|_2\}_{i,j\in[n]}$ costs $\Theta(dn^2)$ arithmetic operations.
For each $t\in\mathcal{T}$, one Sinkhorn/IPFP update (two dense matrix--vector products plus entrywise scaling)
costs $O(n^2)$ time, hence $K$ updates cost $O(Kn^2)$.
Therefore, evaluating the grid objective over $\mathcal{T}$ costs
\[
T_{\mathrm{SK}}(1\ \mathrm{instance})
=
O\!\big(dn^2+|\mathcal T|\,K\,n^2\big),
\]
and over $M$ instances,
\[
T_{\mathrm{SK}}
=
O\!\Big(\mathsf{M}\big(dn^2+|\mathcal{T}|\,K\,n^2\big)\Big),
\]
where $\mathsf{M}:=|\widehat{\Theta}||\widehat{\mathcal{Q}}|$,
with $\Omega(n^2)$ memory if the distance table or kernel is materialized.

\paragraph{Practical implication.}
A Sinkhorn-W mechanism can be attractive when a moderate $\hat{\varepsilon}$ and a moderate solver tolerance
yield acceptable utility. However, it introduces additional approximation knobs
$(\hat{\varepsilon},\mathcal T,\text{termination tolerance})$, and the resulting privacy calibration can become
either overly conservative (when $\hat{\varepsilon}$ is large) or numerically burdensome (when $\hat{\varepsilon}$
is small). In contrast, sliced mechanisms avoid full-dimensional couplings entirely and admit exact one-dimensional
subroutines, yielding simpler calibration logic and more predictable computational behavior.

\section{Properties of SRPP}\label{sec:property_srpp}

In this section, we provide detailed discussions of Ave-SRPP and Joint-SRPP.
Ave-SRPP and Joint-SRPP use the same ingredients (per-slice R\'enyi divergences along 1-D projections) but aggregate them in different ways, leading to different strengths. Ave-SRPP averages the per-slice R\'enyi divergences over directions $u \sim \omega$, so it controls the mean distinguishability of secrets across slices (good for mean geometry across slices). Joint-SRPP instead can be interpreted as bounding the R\'enyi divergence of the joint observation of slice index $V\sim \omega$ and projection $\Psi^V_{\#}(\mathcal{M}(X))$, which captures and penalizes rare, high-risk directions that Ave-SRPP alone can hide in the average.
Joint-SRPP is generally a stronger requirement than Ave-SRPP: satisfying Joint-SRPP implies the corresponding averaged bound, but not conversely.

Sec. \ref{app:ordering} shows the ordering of Ave-SRD, Joint-SRD, and the standard R\'enyi divergence.
In Sec. \ref{app:post_processing}, we characterize the post-processing immunity property of SRPP notions.
Sec. \ref{app:ave-vs-joint-srpp} further compares Ave-SRPP and Joint-SRPP.
Finally, we show that our SRPP frameworks imply their corresponding sliced Pufferfish Privacy frameworks in Sec. \ref{app:srpp_to_epsdelta}.

\subsection{Ordering}\label{app:ordering}

Our two SRPP frameworks sit between each other and the standard RPP.

\begin{proposition}
\label{prop:ordering}
Fix a slice profile $\{\mathcal{U}, \omega\}$.
For $d\geq 1$, $\alpha>1$, every secret pair $(s_i,s_j)\in\mathcal{Q}$ and
every prior $\theta\in\Theta$,
\[\mathtt{AveSD}^{\omega}_{\alpha}\big(\mathcal{M}^{\theta}_{s_i} \|\mathcal{M}^{\theta}_{s_j}\big)
  \leq
  \mathtt{JSD}^{\omega}_{\alpha}\big(\mathcal{M}^{\theta}_{s_i}\|\mathcal{M}^{\theta}_{s_j}\big)
  \leq \mathtt{D}_{\alpha}\big(\mathcal{M}^{\theta}_{s_i}\|\mathcal{M}^{\theta}_{s_j}\big).
\]
\end{proposition}

The ordering in Proposition \ref{prop:ordering} only compares the \textit{numerical tightness of upper bounds} and does \textit{not} induce a straightforward ranking of privacy strength between the three notions. 
For example, an $(\alpha,\epsilon)$-RPP guarantee implies the
corresponding $(\alpha,\epsilon,\omega)$-Joint/Ave-SRPP guarantees
by Proposition~\ref{prop:ordering}, but the converse need not hold.
Moreover, different choices of $\omega$ and different calibration
strategies for the noise can lead to mechanisms with the same bound
on $\mathtt{AveSD}^{\omega}_{\alpha}$ or $\mathtt{JSD}^{\omega}_{\alpha}$ but
very different protections at the level of secrets and
priors. For this reason, the inequalities in
Proposition~\ref{prop:ordering} should be interpreted as relations
between divergence-based certificates, rather than as an absolute
ranking of which framework is "more private".

\begin{proof}[Proof of Proposition \ref{prop:ordering}]
    
Fix a slice profile $\{\mathcal{U},\omega\}$, a prior $\theta\in\Theta$, and a secret
pair $(s_i,s_j)\in\mathcal{Q}$ with $P_\theta^S(s_i),P_\theta^S(s_j)>0$.
Write
\[
  \mathrm{P} := \mathcal{M}^\theta_{s_i},
  \textup{ and }
  \mathrm{Q} := \mathcal{M}^\theta_{s_j}
\]
for the output probability measures of the mechanism under $s_i$ and $s_j$, respectively.
Let $\{\mathcal{U}, \omega\}$ with $\mathcal{U} = \{u_{\ell}\}_{\ell=1}^m$ and $\omega = \{\omega_{\ell}\}^{m}_{\ell=1}$ be the
discrete slicing profile, and denote
\[
  D(u) := \mathtt{D}_{\alpha}\bigl(\Psi^{u_{\ell}}_{\#}\mathrm{P} \| \Psi^{u_{\ell}}_{\#}\mathrm{Q}\bigr),
  \qquad \ell=1,\dots,m.
\]

By Definition \ref{def:ave_srd}, the average sliced R\'enyi divergence is
\[
\mathtt{AveSD}^\omega_\alpha(\mathrm{P}\|\mathrm{Q})=\int_{\mathbb{S}^{d-1}}D(u)\ d\omega(u)
\]
By Definition \ref{def:joint_srd}, the joint sliced R\'enyi divergence is
\[
  \mathtt{JSD}^\omega_{\alpha}(\mathrm{P}\|\mathrm{Q})=\frac{1}{\alpha-1}
  \log\Bigl(\int_{\mathbb{S}^{d-1}} \exp\bigl((\alpha-1)D(u)\ d\omega(u)\bigr)
  \Bigr).
\]


\medskip\noindent
\textbf{Part 1: $\mathtt{AveSD}^{\omega}_{\alpha}(\mathrm{P}\|\mathrm{Q})
\leq \mathtt{JSD}^{\omega}_{\alpha}(\mathrm{P}\|\mathrm{Q})$.}

Let $Z$ be the random variable $Z = D(U)$ when $U\sim\omega$.
Then, 
\[
\mathbb{E}[Z] = \int_{\mathbb{S}^{d-1}} D(u) \ d\omega(u)= \mathtt{AveSD}^{\omega}_{\alpha}(\mathrm{P}\|\mathrm{Q}),
\]
and
\[
\mathtt{JSD}^\omega_\alpha(\mathrm{P}\|\mathrm{Q}) = \frac{1}{\alpha-1}\log\mathbb{E}\bigl[\exp\bigl((\alpha-1)Z\bigr)\bigr]
\]
For $\alpha>1$, $\exp((\alpha-1)z)$ is convex. 
Then, Jensen's inequality gives
\[
\mathbb{E}\bigl[\exp\bigl((\alpha-1)Z\bigr)\bigr] \geq \exp\bigl((\alpha-1) \mathbb{E}[Z] \bigr).
\]
Thus, by taking log and dividing by $(\alpha-1)$, we get
\[
\mathtt{JSD}^\omega_\alpha(\mathrm{P}\|\mathrm{Q}) \geq \mathbb{E}[Z] = \mathtt{AveSD}^{\omega}_{\alpha}(\mathrm{P}\|\mathrm{Q}).
\]

\medskip\noindent
\textbf{Part 2: $\mathtt{JSD}^\omega_\alpha(\mathrm{P}\|\mathrm{Q})
\leq \mathtt{D}_\alpha(\mathrm{P}\|\mathrm{Q})$.}

Define the randomized post-processing channel $\mathsf{Slice}$ as follows:
given $Y\in\mathbb{R}^d$, draw $U\sim\omega$ independently of $Y$, then output
$\mathsf{Slice}(Y):=(U,\langle Y,U\rangle)$.
Let $\mathrm{P}^{\mathsf{Slice}}$ and $\mathrm{Q}^{\mathsf{Slice}}$ denote the laws of $\mathsf{Slice}(Y)$ when
$Y\sim \mathrm{P}$ and $Y\sim \mathrm{Q}$, respectively.

For each direction $u$, the conditional law of $\langle Y,u\rangle$ under $Y\sim \mathrm{P}$ is
$\Psi^u_\#\mathrm{P}$, and similarly under $\mathrm{Q}$ it is $\Psi^u_\#\mathrm{Q}$.
By the definition of Rényi divergence for product/mixture measures with the same mixing law $\omega$,
we have
\[
\begin{aligned}
    &\mathtt{D}_\alpha(\mathrm{P}^{\mathsf{Slice}}\|\mathrm{Q}^{\mathsf{Slice}})\\&
=
\frac{1}{\alpha-1}\log\!\int_{\mathbb{S}^{d-1}}
\exp\!\Big((\alpha-1)\mathtt{D}_\alpha(\Psi^u_\#\mathrm{P}\,\|\,\Psi^u_\#\mathrm{Q})\Big)\,d\omega(u)\\&
=
\mathtt{JSD}^\omega_\alpha(\mathrm{P}\|\mathrm{Q}).
\end{aligned}
\]
Since $\mathsf{Slice}$ is a (data-independent) post-processing channel applied to $Y$,
the data-processing inequality for Rényi divergence (an $f$-divergence for $\alpha>1$) yields
\[
\mathtt{D}_\alpha(\mathrm{P}^{\mathsf{Slice}}\|\mathrm{Q}^{\mathsf{Slice}})
\le
\mathtt{D}_\alpha(\mathrm{P}\|\mathrm{Q}),
\]
and hence $\mathtt{JSD}^\omega_\alpha(\mathrm{P}\|\mathrm{Q})\le \mathtt{D}_\alpha(\mathrm{P}\|\mathrm{Q})$.

By combining \textbf{Part 1} and \textbf{Part 2}, we can claim 
\[
\mathtt{AveSD}^{\omega}_{\alpha}\big(\mathcal{M}^{\theta}_{s_i} \|
  \mathcal{M}^{\theta}_{s_j}\big)
  \leq
\mathtt{JSD}^{\omega}_{\alpha}\big(\mathcal{M}^{\theta}_{s_i}\|
  \mathcal{M}^{\theta}_{s_j}\big)
  \leq
\mathtt{D}_{\alpha}\big(\mathcal{M}^{\theta}_{s_i}\|
  \mathcal{M}^{\theta}_{s_j}\big),
\]
\end{proof}

\subsection{Post-Processing Immunity}\label{app:post_processing}

A fundamental closure property of DP and RPP is post-processing immunity (PPI): the privacy loss
of a mechanism cannot be increased by arbitrary data-independent post-processing
\cite{pierquin2024renyi,Dwork2014}. Formally, for any Markov kernel (randomized mapping) $K$ that is
independent of the dataset, post-processing acts on the \textit{endpoint distributions} of a release.
Since the R\'enyi divergence $\mathtt{D}_{\alpha}(\cdot\|\cdot)$ satisfies the data-processing inequality
(DPI),
\[
\mathtt{D}_{\alpha}(PK \,\|\, QK)\ \leq\ \mathtt{D}_{\alpha}(P\,\|\,Q),
\]
any privacy certificate expressed in terms of (upper bounds on) $\mathtt{D}_{\alpha}$ is automatically
closed under post-processing.

Post-processing immunity has two aspects.
(i) \textit{True privacy loss never increases:} for any data-independent $K$,
$\mathtt{D}_{\alpha}(PK \,\|\, QK)\le \mathtt{D}_{\alpha}(P\,\|\,Q)$ by DPI.
(ii) \textit{Safe certificate reuse for a downstream pipeline:} one may upper-bound the privacy loss of a
post-processing pipeline using a certificate for the initial mechanism.
For SRPP, this reuse must be interpreted through the relationship between sliced divergences and the
full-dimensional R\'enyi divergence (\ Prop.~\ref{prop:ordering}): SRPP/SRD is a quantification
language for the full release, and changing the representation along the pipeline can change the
numerical SRPP value even though the underlying distinguishability (and hence privacy loss in the
sense of (i)) cannot increase. Consequently, reusing SRPP parameters for an end-to-end pipeline can be
valid but may be overly conservative when the quantification geometry matters.

In privacy, PPI should be understood primarily as a \textit{semantic closure} property: for any data-independent
post-processing $K$ applied to the true release, the post-processed pipeline cannot be \textit{more}
distinguishable than the original release in the underlying divergence sense (DPI). This is stronger than
(and logically prior to) any particular \textit{parameterization rule} such as "reusing the same numeric
privacy parameters", which is merely a convenient upper-bounding convention and can be very conservative (or not directly meaningful) when the quantification geometry changes (as in SRPP).

SRPP’s post-processing discussion is ultimately rooted in the data-processing inequality (DPI) for the
standard R\'enyi divergence $\mathtt{D}_{\alpha}$.  Concretely, SRPP quantifies privacy by applying a
fixed (profile-determined) slicing/observation channel to the \textit{endpoint distribution} of the
(full-dimensional) released output, and then evaluating $\mathtt{D}_{\alpha}$ on the induced sliced laws:
Ave-SRPP takes the $\omega$-average of the resulting per-direction divergences, while Joint-SRPP applies a
log-moment aggregation that is equivalent to a single $\mathtt{D}_{\alpha}$ on an augmented observation.
Accordingly, for any data-independent post-processing kernel $K$ applied to the true release $Y=\mathcal M(X)$,
one first obtains the post-processed endpoint laws $PK$ and $QK$, and SRPP/SRD is then evaluated on
$PK$ versus $QK$ under the same slice profile.  Any nonconservation of the SRPP numerical value under
post-processing reflects this change in endpoint distributions (and the geometry induced by the chosen
profile), rather than a failure of post-processing immunity: the underlying distinguishability cannot
increase due to DPI.

Let $P=\Pr[\mathcal M(X)\mid s_i,\theta]$ and $Q=\Pr[\mathcal M(X)\mid s_j,\theta]$ be the endpoint
distributions on $\mathbb R^d$, and let $K$ be any data-independent post-processing kernel on $\mathbb R^d$.
Then DPI gives $\mathtt D_\alpha(PK\|QK)\le \mathtt D_\alpha(P\|Q)$ (true privacy loss never increases).
Moreover, by Prop.~\ref{prop:ordering},
\[
\mathtt{AveSD}^\omega_\alpha(PK\|QK)\le \mathtt{JSD}^\omega_\alpha(PK\|QK)\leq \mathtt D_\alpha(PK\|QK),
\]
so any full-dimensional R\'enyi/RPP certificate for $\mathcal M$ remains a valid upper bound on the
SRPP/SRD risk of the post-processed pipeline $K\circ\mathcal M$. In general, however, the SRPP values
$\mathtt{AveSD}^\omega_\alpha(\cdot\|\cdot)$ and $\mathtt{JSD}^\omega_\alpha(\cdot\|\cdot)$ need not be
numerically conserved under arbitrary $K$, because they quantify endpoint distinguishability through a
fixed slice profile.

\subsection{Average vs.\ Joint SRPP}
\label{app:ave-vs-joint-srpp}

Average SRPP (Ave-SRPP) and Joint SRPP (Joint-SRPP) are built from the same ingredients, R\'enyi divergences of one-dimensional projections, but aggregate these per-direction costs in different ways, leading to distinct strengths and use cases.

Recall that for any pair of probability measures $\mathrm{P},\mathrm{Q}$ on $\mathbb{R}^d$ and a slice profile $\{\mathcal{U}, \omega\}$, the average sliced R\'enyi divergence of order $\alpha>1$ is
\[
  \mathtt{AveSD}^{\omega}_\alpha(\mathrm{P}\|\mathrm{Q})
  :=
  \int_{\mathbb{S}^{d-1}}
    \mathtt{D}_{\alpha}\big( \Psi^{u}_{\#} P \,\big\|\, \Psi^{u}_{\#} Q \big) d\omega(u),
\]
while the joint sliced R\'enyi divergence is
\[
\begin{aligned}
    &\mathtt{JSD}^{\omega}_\alpha(\mathrm{P}\|\mathrm{Q})\\
  &:=
  \frac{1}{\alpha-1}
  \log
  \mathbb{E}_{U\sim \omega}
  \Big[
    \exp\big(
      (\alpha-1)
      \mathtt{D}_{\alpha}\big( \Psi^{U}_{\#} \mathrm{P} \big\| \Psi^{U}_{\#} \mathrm{Q} \big)
    \big)
  \Big].
\end{aligned}
\]

Ave-SRPP requires that, for every prior $\theta\in\Theta$ and every secret pair $(s_0,s_1)\in\mathcal{Q}$, the \textit{average} sliced R\'enyi divergence between the corresponding output probability distribution is bounded by $\epsilon$. 
Joint-SRPP instead bounds the log-moment aggregation of per-slice divergences, which can be interpreted as the R\'enyi divergence between the joint distributions ($\widetilde{\mathrm{P}}(\cdot, \cdot)$) of the \textit{slice (index) and the projected output}.

In particular, define the randomized sliced observation channel $\mathsf{Slice}:\mathbb R^d\to \mathcal U\times\mathbb R$
by
\[
\mathsf{Slice}(y):=(U,\Psi^{U}(y)),\qquad U\sim \omega\ \text{independent of }y.
\]
Let $\widetilde{\mathrm P}:=\mathsf{Slice}_{\#}\mathrm P$ and $\widetilde{\mathrm Q}:=\mathsf{Slice}_{\#}\mathrm Q$ denote the
endpoint laws of $(U,\Psi^{U}(Z_{\mathrm P}))$ and $(U,\Psi^{U}(Z_{\mathrm Q}))$ when $Z_{\mathrm P}\sim \mathrm P$ and
$Z_{\mathrm Q}\sim \mathrm Q$, respectively. Then
\[
\mathtt{JSD}^{\omega}_\alpha(\mathrm P\|\mathrm Q)
=
\mathtt D_{\alpha}\!\big(\widetilde{\mathrm P}\,\big\|\,\widetilde{\mathrm Q}\big).
\]
In this sense, Joint-SRPP is exactly RPP applied to the augmented observation $(U,\Psi^{U}(\mathcal M(X)))$,
whereas Ave-SRPP only controls the mean per-direction divergence.

These two aggregation schemes are ordered between each other and the full R\'enyi divergence. For any $\mathrm{P}, \mathrm{Q}$ on $\mathbb{R}^d$ and any slice distribution $\omega$,
\[
  \mathtt{AveSD}^{\omega}_{\alpha}(\mathrm{P}\|\mathrm{Q})\leq
  \mathtt{JSD}^{\omega}_{\alpha}(\mathrm{P}\|\mathrm{Q})\leq 
  \mathtt{D}_{\alpha}(\mathrm{P}\|\mathrm{Q}),
\]
and the same inequalities hold for the corresponding output distributions under any fixed Pufferfish model $(\mathcal{S},\mathcal{Q},\Theta)$ ((Proposition \ref{prop:ordering})).
Thus, at the level of numerical certificates, Joint-SRPP is always at least as "conservative" as Ave-SRPP for the same $(\alpha,\epsilon,\omega)$, and both are upper-bounded by the unsliced RPP divergence.

Although Ave-SRPP and Joint-SRPP are closely related, they serve complementary purposes.

\paragraph{Ave-SRPP for Geometry-Aware \textit{Average} Protection.}
Ave-SRPP directly reflects the geometry induced by the slicing distribution $\omega$ and is naturally aligned with sliced Wasserstein mechanism privatization, where one averages per-direction sensitivities. Formally, it controls the \textit{expected} sliced R\'enyi divergence when the projection direction is drawn from $\omega$. 
Equivalently, it upper-bounds the expected per-direction divergence for a direction drawn from $\omega$;
this is an analytic interpretation only (the mechanism still releases the full $d$-dimensional output).
This "single-slice observer" is purely an analytic device: the actual mechanism always outputs the full $d$-dimensional vector, and the sliced view is used only to define a geometry-aware privacy budget. 
In practice, this makes Ave-SRPP a natural \textit{calibration target}, often yielding tighter, more utility-friendly noise levels than enforcing a worst-case bound simultaneously in every direction.

\paragraph{Joint-SRPP for Robust, Slice-Aware Guarantees.}
Joint-SRPP instead can be interpreted as if it treats the slice index as part of the observation channel and bounds the R\'enyi divergence of the joint pair $(U,\Psi^{U}(\mathcal{M}(X)))$. In terms of risk quantification, this corresponds to accounting for an observer that can exploit knowledge of \textit{which} direction was used, not just the projected value. 
Thus Joint-SRPP protects against worst-case sliced distinguishability over the randomized slicing procedure, rather than its average, and provides a more conservative, slice-aware notion of privacy than Ave-SRPP.



\subsection{From SRPP to Sliced Pufferfish Privacy}
\label{app:srpp_to_epsdelta}

This appendix records standard R\'enyi-to-approximate guarantees for our sliced notions, with notation
consistent with Sec.~\ref{sec:sliced_Wasserstein}--\ref{sec:joint_srpp}.
Fix a Pufferfish scenario $(\mathcal{S},\mathcal{Q},\Theta)$ and a slice profile $(\mathcal{U},\omega)$, where
$\omega$ is a probability measure on $\mathbb{S}^{d-1}$ supported on $\mathcal{U}$.
For a mechanism $\mathcal M$ and $\theta\in\Theta$, $s\in\mathcal{S}$, write
\[
\mathcal{M}^\theta_s \ :=\ \Pr\!\big(\mathcal{M}(X)\mid s,\theta\big),
\]
for the conditional output distribution. For each $u\in\mathbb{S}^{d-1}$, recall the projection
$\Psi^{u}(a)=\langle a,u\rangle$ and its pushforward $\Psi^{u}_{\#}$ as in~\eqref{eq:pushforward}.

\paragraph{Approximate Max-Divergence.}
For distributions $P,Q$ on a common measurable space and $\delta\in(0,1)$, define
\begin{equation}\label{eq:approx_max_divergence}
    \mathtt{D}_\infty^\delta(P\|Q)
:= \inf\Bigl\{\varepsilon\geq 0:\ \forall A,\ P(A)\leq e^\varepsilon Q(A)+\delta\Bigr\}.
\end{equation}

These are the $(\varepsilon,\delta)$ analogues of Definitions~\ref{def:ave_srpp} and~\ref{def:joint_srpp},
obtained by replacing per-slice R\'enyi divergences with approximate max-divergences.

\begin{definition}[$(\varepsilon,\delta,\omega)$-Ave-SPP]\label{def:ave_spp}
A mechanism $\mathcal{M}$ satisfies $(\varepsilon,\delta,\omega)$-\textit{Ave-SPP} in
$(\mathcal{S},\mathcal{Q},\Theta)$ if for all $\theta\in\Theta$ and $(s_i,s_j)\in\mathcal{Q}$ with
$P^S_{\theta}(s_i),P^S_{\theta}(s_j)>0$,
\[
\int_{\mathbb{S}^{d-1}}
\mathtt{D}_\infty^\delta\!\Big(\Psi^u_{\#}\mathcal{M}^\theta_{s_i}\ \big\|\ \Psi^u_{\#}\mathcal{M}^\theta_{s_j}\Big)\,
\mathrm{d}\omega(u)
\ \leq\ \varepsilon.
\]
\end{definition}

\begin{definition}[$(\varepsilon,\delta,\omega)$-Joint-SPP]\label{def:joint_spp}
Let $V\sim\omega$ be independent of $X$ and of the internal randomness of $\mathcal{M}$, and define the
\textit{direction-revealing projected mechanism}
\[
\widetilde{\mathcal{M}}(X) \ :=\ \bigl(V,\ \Psi^{V}(\mathcal{M}(X))\bigr).
\]
We say that $\mathcal{M}$ satisfies $(\varepsilon,\delta,\omega)$-\textit{Joint-SPP} in
$(\mathcal{S},\mathcal{Q},\Theta)$ if $\widetilde{\mathcal{M}}$ satisfies $(\varepsilon,\delta)$-PP in
$(\mathcal{S},\mathcal{Q},\Theta)$, i.e., for all $\theta\in\Theta$ and $(s_i,s_j)\in\mathcal{Q}$ with
$P^S_{\theta}(s_i), P^S_{\theta}(s_j)>0$ and all measurable $A$,
\[
\Pr\!\bigl(\widetilde{\mathcal{M}}(X)\in A\mid s_i,\theta\bigr)
\ \leq\ e^\varepsilon\,\Pr\!\bigl(\widetilde{\mathcal{M}}(X)\in A\mid s_j,\theta\bigr)\ +\ \delta.
\]
\end{definition}

Recall the R\'enyi divergence $\mathtt{D}_\alpha(\cdot\|\cdot)$ given by (\ref{eq:renyi_divergence}).
Lemma \ref{lem:renyi_to_epsdelta_app} shows conversion from R\'enyi divergenece to the approximate max-divegrence $\mathtt{D}_\infty^\delta$ in (\ref{eq:approx_max_divergence}).

\begin{lemma}[R\'enyi to approximate max-divergence]\label{lem:renyi_to_epsdelta_app}
Let $\alpha>1$ and $\delta\in(0,1)$. If $\mathtt{D}_\alpha(P\|Q)\leq \epsilon_\alpha$, then
\[
\mathtt{D}_\infty^\delta(P\|Q)\ \leq\ \epsilon_\alpha\ +\ \frac{\log(1/\delta)}{\alpha-1}.
\]
\end{lemma}

\begin{proof}
Let $r:=\frac{dP}{dQ}$ on $\{dQ>0\}$ (and $r:=+\infty$ on $\{dQ=0,dP>0\}$).
For any measurable $A$ and any $\tau>0$,
\[
\begin{aligned}
    P(A)=\int_A r\,dQ
&=\int_{A\cap\{r\le \tau\}} r\,dQ + \int_{A\cap\{r>\tau\}} r\,dQ
\\&\leq \tau\,Q(A)+P(r>\tau).
\end{aligned}
\]
Let $\tau:=e^\varepsilon$. 
It suffices to bound $P(r>e^\varepsilon)$.
Applying Markov's inequality to $r^{\alpha-1}$ under $P$ gives
\[
P(r>e^\varepsilon)
=
P\big(r^{\alpha-1}>e^{(\alpha-1)\varepsilon}\big)
\leq
\frac{\mathbb{E}_P[r^{\alpha-1}]}{e^{(\alpha-1)\varepsilon}}
=
\frac{\int r^\alpha\,dQ}{e^{(\alpha-1)\varepsilon}}.
\]
By definition of R\'enyi divergence,
\[
\int r^\alpha\,dQ
=
\exp\big((\alpha-1)\mathtt{D}_\alpha(P\|Q)\big)
\leq
\exp\big((\alpha-1)\epsilon_\alpha\big).
\]
Hence
\[
P(r>e^\varepsilon)
\leq
\exp\big((\alpha-1)\epsilon_\alpha-(\alpha-1)\varepsilon\big).
\]
Setting $\varepsilon:=\epsilon_\alpha+\frac{\log(1/\delta)}{\alpha-1}$ yields $P(r>e^\varepsilon)\leq\delta$, and therefore
$P(A)\leq e^\varepsilon Q(A)+\delta$ for all measurable $A$, which is equivalent to
$\mathtt{D}_\infty^\delta(P\|Q)\leq \varepsilon$.
\end{proof}

Applying Lemma~\ref{lem:renyi_to_epsdelta_app} to each sliced pair
$(\Psi^u_{\#}\mathcal M^\theta_{s_i},\Psi^u_{\#}\mathcal M^\theta_{s_j})$ yields the pointwise bound
\[
\mathtt{D}_\infty^\delta\!\Big(\Psi^u_{\#}\mathcal M^\theta_{s_i}\,\big\|\,\Psi^u_{\#}\mathcal M^\theta_{s_j}\Big)
\le
\mathtt{D}_\alpha\!\Big(\Psi^u_{\#}\mathcal M^\theta_{s_i}\,\big\|\,\Psi^u_{\#}\mathcal M^\theta_{s_j}\Big)
+\frac{\log(1/\delta)}{\alpha-1}.
\]
Integrating over $u\sim\omega$ converts $(\alpha,\epsilon,\omega)$-Ave-SRPP into an $(\varepsilon,\delta,\omega)$-Ave-SPP
guarantee with $\varepsilon=\epsilon+\frac{\log(1/\delta)}{\alpha-1}$; the same conversion applies to Joint-SRPP by
applying Lemma~\ref{lem:renyi_to_epsdelta_app} to the joint endpoint laws of $(V,\Psi^V(\mathcal M(X)))$.

\subsubsection{Ave-SRPP Implies Ave-SPP}

\begin{theorem}[Ave-SRPP $\Rightarrow$ Ave-SPP]\label{thm:ave_srpp_to_ave_spp_app}
Fix $\alpha>1$. 
If $\mathcal{M}$ satisfies $(\alpha,\epsilon,\omega)$-\textup{Ave-SRPP} in
$(\mathcal{S},\mathcal{Q},\Theta)$, then for any $\delta\in(0,1)$ it satisfies
$(\epsilon',\delta,\omega)$-\textup{Ave-SPP} with
\[
\epsilon'\ :=\ \epsilon\ +\ \frac{\log(1/\delta)}{\alpha-1}.
\]
\end{theorem}

\begin{proof}
Fix $\theta\in\Theta$ and $(s_i,s_j)\in\mathcal{Q}$ with $P^S_{\theta}(s_i),P^S_{\theta}(s_j)>0$.
For each $u\in\mathcal{U}$, apply Lemma~\ref{lem:renyi_to_epsdelta_app} to
$P^u:=\Psi^u_{\#}\mathcal{M}^\theta_{s_i}$ and $Q^u:=\Psi^u_{\#}\mathcal{M}^\theta_{s_j}$ to obtain
\[
\mathtt{D}_\infty^\delta(P^u\|Q^u)
\leq
\mathtt{D}_\alpha(P^u\|Q^u)+\frac{\log(1/\delta)}{\alpha-1}.
\]
Integrating both sides with respect to $\omega$ yields
\[
\int \mathtt{D}_\infty^\delta(P^u\|Q^u)\,\mathrm d\omega(u)
\leq
\int \mathtt{D}_\alpha(P^u\|Q^u)\,\mathrm d\omega(u)
+\frac{\log(1/\delta)}{\alpha-1}.
\]
By $(\alpha,\epsilon,\omega)$-Ave-SRPP (Definition~\ref{def:ave_srpp}), the second integral satisfies $\int \mathtt{D}_\alpha(P^u\|Q^u)\,d\omega(u)\leq \epsilon$. Therefore,
$\int \mathtt{D}_\infty^\delta(P^u\|Q^u)\,d\omega(u)\leq \epsilon+\frac{\log(1/\delta)}{\alpha-1}$.
\end{proof}

\subsubsection{Joint-SRPP Implies Joint-SPP}

\begin{theorem}[Joint-SRPP $\Rightarrow$ Joint-SPP]\label{thm:joint_srpp_to_joint_spp_app}
Fix $\alpha>1$. If $\mathcal{M}$ satisfies $(\alpha,\epsilon,\omega)$-\textup{Joint-SRPP} in
$(\mathcal{S},\mathcal{Q},\Theta)$, then for any $\delta\in(0,1)$ it satisfies
$(\epsilon',\delta,\omega)$-\textup{Joint-SPP with}\footnote{Equivalently, $\widetilde{\mathcal{M}}$ is $(\epsilon',\delta)$-PP.}
\[
\epsilon'\ :=\ \epsilon\ +\ \frac{\log(1/\delta)}{\alpha-1}.
\]
\end{theorem}

\begin{proof}
Fix $\theta\in\Theta$ and $(s_i,s_j)\in\mathcal{Q}$ with $P^S_{\theta}(s_i),P^S_{\theta}(s_j)>0$.
Let $P:=\mathcal{M}^\theta_{s_i}$ and $Q:=\mathcal{M}^\theta_{s_j}$ be distributions on $\mathbb{R}^d$.
Let $\widetilde{P},\widetilde{Q}$ be the laws of $\widetilde{\mathcal{M}}(X)=(V,\Psi^V(\mathcal{M}(X)))$
under secrets $s_i$ and $s_j$.

\textbf{Step 1.}
We show that Joint-SRD is the R\'enyi divergence of the direction-revealing output.
For each $u$, write $P^u:=\Psi^u_{\#}P$ and $Q^u:=\Psi^u_{\#}Q$.
Since $V\sim\omega$ is the same under both secrets and is independent of $\mathcal{M}(X)$,
the disintegrations of $\widetilde{P}$ and $\widetilde{Q}$ satisfy
\[
\widetilde{P}(du,dy)=\omega(du)\,P^u(dy),
\quad
\widetilde{Q}(du,dy)=\omega(du)\,Q^u(dy).
\]
Assuming $\widetilde{P}\ll \widetilde{Q}$ (otherwise the divergence is $+\infty$), the Radon-Nikodym derivative factorizes:
\[
\frac{d\widetilde{P}}{d\widetilde{Q}}(u,y)=\frac{dP^u}{dQ^u}(y)
\]
for $\widetilde Q$-a.e. $(u,y)$.
Therefore,
\begin{align*}
\mathtt{D}_\alpha(\widetilde{P}\|\widetilde{Q})
&=
\frac{1}{\alpha-1}\log\int\left(\frac{d\widetilde{P}}{d\widetilde{Q}}\right)^\alpha d\widetilde{Q}\\
&=
\frac{1}{\alpha-1}\log\int_{\mathbb S^{d-1}}\omega(du)\int \left(\frac{dP^u}{dQ^u}(y)\right)^\alpha dQ^u(y)\\
&=
\frac{1}{\alpha-1}\log\int_{\mathbb{S}^{d-1}}
\exp\!\Big((\alpha-1)\mathtt{D}_\alpha(P^u\|Q^u)\Big)\,\mathrm{d}\omega(u)\\
&=
\mathtt{JSD}^{\omega}_{\alpha}(P\|Q),
\end{align*}
which matches Definition~\ref{def:joint_srd}.

\textit{Step 2 (Apply the standard conversion).}
By $(\alpha,\epsilon,\omega)$-Joint-SRPP (Definition~\ref{def:joint_srpp}),
\[
\mathtt{D}_\alpha(\widetilde{P}\|\widetilde{Q})
=
\mathtt{JSD}^{\omega}_{\alpha}(P\|Q)
\leq \epsilon.
\]
Applying Lemma~\ref{lem:renyi_to_epsdelta_app} to $(\widetilde{P},\widetilde{Q})$ gives
\[
\mathtt{D}_\infty^\delta(\widetilde{P}\|\widetilde{Q})
\leq
\epsilon+\frac{\log(1/\delta)}{\alpha-1}
=
\epsilon'.
\]
By the definition of $\mathtt{D}_\infty^\delta$, this inequality is equivalent to the $(\epsilon',\delta)$-PP condition
for $\widetilde{\mathcal{M}}$, i.e., $(\epsilon',\delta,\omega)$-Joint-SPP for $\mathcal{M}$.
\end{proof}

\subsubsection{Optimizing Over R\'enyi Orders}

If a mechanism admits Ave/Joint-SRPP bounds $\epsilon_\alpha$ for multiple $\alpha>1$, then we may optimize the conversion as
\[
\epsilon(\delta)
:=\inf_{\alpha>1}\left\{\epsilon_\alpha+\frac{\log(1/\delta)}{\alpha-1}\right\}.
\]

\section{Examples of $L_t$-Lipschitz Update Map}\label{app:L_t_Update_example}

In this section, we provide two common cases of $L_t$-Lipschitz update maps described in Sec.~\ref{sec:HUC}.

\textit{(i) Preconditioned SGD.}
Suppose we use a (possibly history- and randomness-dependent) preconditioner
$A_{t}:\mathcal{Y}_{<t}\times \mathcal{R} \to \mathbb{R}^{d\times d}$, so that
\[
T_{t}(z; y_{<t})
=
\xi_{t-1} - A_{t}(y_{<t}, R_t)\, z .
\]
Then $T_t$ is $L_t$-Lipschitz in its first argument with
\[
L_t
:=
\sup_{y_{<t},\, r}\bigl\|A_{t}(y_{<t}, r)\bigr\|_{\mathrm{op}},
\qquad
\|A\|_{\mathrm{op}} := \sup_{\|v\|_{2} = 1} \|Av\|_{2}.
\]
In standard SGD, $A_{t}(y_{<t}, r) = \kappa I_{d}$ (step size $\kappa$), hence $L_t=\kappa$.

\textit{(ii) Proximal or projected updates.}
Let 
\[
\mathrm{prox}(v)
:=
\arg\min_{w\in\mathbb{R}^{d}}
\Bigl\{\tfrac{1}{2}\|w-v\|^{2}_{2} + \bar{\lambda}\mathrm{Reg}(w)\Bigr\},
\]
where $\bar{\lambda}>0$ and $\mathrm{Reg}:\mathbb{R}^{d}\to(-\infty,+\infty]$ is a proper,
lower-semicontinuous, convex regularizer. Let $\overline{\Pi}_{\mathcal{C}}$ be the Euclidean
projection onto a closed convex set $\mathcal{C}\subset \mathbb{R}^{d}$.
Since both $\mathrm{prox}$ and $\overline{\Pi}_{\mathcal{C}}$ are $1$-Lipschitz, the same constant $L_t$
serves as a Lipschitz bound when
\[
T_{t}(z; y_{<t}) = \mathrm{prox}\bigl(\xi_{t-1} - A_{t}(y_{<t}, R_t)z \bigr)
\]
or
\[
T_{t}(z; y_{<t}) = \overline{\Pi}_{\mathcal{C}}\bigl( \xi_{t-1} - A_{t}(y_{<t}, R_t)z\bigr).
\]

\section{Discrepancy Cap Estimation}
\label{app:estimate_caps}

This appendix describes principled approaches to instantiating the \textit{discrepancy cap} terms
used by HUC/sa-HUC accounting.
Fix a Pufferfish scenario $(\mathcal{S},\mathcal{Q},\Theta)$ and an iteration index $t$.
Let $\eta_t$ denote the (possibly randomized) minibatch selection rule (distribution) at iteration $t$ with batch size $B_t$.
For each $\theta\in\Theta$ and each secret $s\in\mathcal{S}$ with $P_\theta^S(s)>0$, let
$P_\theta(X\mid S=s)$ denote the conditional data distribution.

\paragraph{Notation and Correspondence to the Main Text.}
The main body defines the discrepancy count
$K_t(x,x';r)=\sum_{j\in\mathsf{I}_t(r)}\mathbf{1}\{x_j\neq x'_j\}$
and its coupling-based cap
$K_t(\gamma)=\operatorname*{ess\,sup}_{((X,X'),R_t)\sim \gamma\times\mathbb{P}_{\eta,\rho}}K_t(X,X';R_t)$
(see \eqref{eq:discrepancy_K}--\eqref{eq:discrepancy_cap}).
In this appendix we denote the same random discrepancy count by
$\Delta_t=\Delta_t(X^{(i)},X^{(j)},R_t)$
to emphasize its role as a random variable under the coupled experiment.
Concretely, under the identification $(X^{(i)},X^{(j)})\equiv(X,X')$ and with the convention that
$\mathsf{I}_t(r)$ is interpreted as a set (\textsf{WOR}) or a length-$B_t$ sequence (\textsf{WR}),
we have $\Delta_t \equiv K_t(X,X';R_t)$.
Accordingly, the deterministic HUC cap in \eqref{eq:HUC_cap_def} is a uniform upper bound on $\operatorname*{ess\,sup}\Delta_t$, and the sa-HUC cap in \eqref{eq:msHUC_cap_def} uniformly bounds
$\mathbb{E}[\Delta_t^2]$ over the subsampling randomness.

\subsection{Discrepancy Random Variable and Caps}
\label{app:estimate_caps:def}

Caps are defined relative to a \textit{coupling} between two conditional data distributions.
Fix $\theta\in\Theta$ and $(s_i,s_j)\in\mathcal{Q}$ with $P_\theta^S(s_i)>0$ and $P_\theta^S(s_j)>0$.
Let $\Gamma_{\theta,ij}\in\Pi\!\big(P_\theta(X\mid S=s_i),\,P_\theta(X\mid S=s_j)\big)$ be a coupling, and draw
$(X^{(i)},X^{(j)})\sim\Gamma_{\theta,ij}$.
Draw a minibatch index multiset/vector $R_t=(I_1,\dots,I_{B_t})\sim \eta_t$.
Define the batch discrepancy count
\begin{equation}
\label{eq:def_discrepancy_count}
\Delta_t
\;:=\;
\Delta_t\!\big(X^{(i)},X^{(j)},R_t\big)
\;=\;
\#\Big\{b\in\{1,\dots,B_t\}:\ x^{(i)}_{I_b}\neq x^{(j)}_{I_b}\Big\},
\end{equation}
where the indices $(I_1,\dots,I_{B_t})$ may be sampled with or without replacement according to $\eta_t$.
For structured data, the relation "$\neq$" is interpreted as inequality of the atomic data unit that
determines the per-example gradient contribution (e.g., a record, a user trajectory, or another unit of influence).

By construction,
\begin{equation}
\label{eq:Delta_range}
0 \;\leq\; \Delta_t \;\leq\; B_t,
\qquad\text{and hence}\qquad
0 \;\leq\; \Delta_t^2 \;\leq\; B_t^2.
\end{equation}

\paragraph{Deterministic (HUC) Caps.}
For a fixed choice of coupling family $\{\Gamma_{\theta,ij}\}$, a deterministic HUC cap is any
$K_t\in[0,B_t]$ such that
\begin{equation}\label{eq:HUC_cap_def}
K_t\ \geq\
\sup_{\theta\in\Theta}\ \sup_{(s_i,s_j)\in\mathcal{Q}}
\operatorname*{ess\ sup}_{(X^{(i)},X^{(j)},R_t)\sim \Gamma_{\theta,ij}\times \eta_t}
\Delta_t\!\big(X^{(i)},X^{(j)},R_t\big).
\end{equation}
Equivalently, \eqref{eq:HUC_cap_def} enforces $\Delta_t\leq K_t$ almost surely, uniformly over all
$\theta\in\Theta$ and $(s_i,s_j)\in\mathcal{Q}$.

\paragraph{Second-Moment (sa-HUC) Caps.}
A subsampling-aware (sa-HUC) cap is any scalar $\overline{K}_t^2\in[0,B_t^2]$ such that
\begin{equation}\label{eq:msHUC_cap_def}
\overline{K}_t^2\ \ge\
\sup_{\theta\in\Theta}\ \sup_{(s_i,s_j)\in\mathcal{Q}}
\operatorname*{ess\ sup}_{(X^{(i)},X^{(j)})\sim \Gamma_{\theta,ij}}
\mathbb{E}_{R_t\sim \eta_t}\!\big[\Delta_t\!\big(X^{(i)},X^{(j)},R_t\big)^2\big].
\end{equation}

\paragraph{Tail vs.\ Moment.}
In practice, rather than enforcing the almost-sure cap~\eqref{eq:HUC_cap_def}, one may specify a tail level
$\delta_t\in(0,1)$ and seek a \textit{high-probability} cap $K_t(\delta_t)$ such that
\begin{equation}\label{eq:HUC_tail_cap}
\Pr\bigl(\Delta_t \leq K_t(\delta_t)\bigr)\ \geq\ 1-\delta_t
\end{equation}
uniformly over $\theta\in\Theta$ and $(s_i,s_j)\in\mathcal{Q}$,
where the probability is taken over $(X^{(i)},X^{(j)})\sim\Gamma_{\theta,ij}$ and $R_t\sim\eta_t$.
For sa-HUC, one instead seeks a uniform second-moment cap $\overline{K}_t^2$ satisfying~\eqref{eq:msHUC_cap_def}.

\subsection{Monte Carlo Estimation via a Simulator}
\label{app:estimate_caps:mc}

Assume we have (i) sampling access to $P_\theta(X\mid S=s)$ for any $\theta\in\Theta$ and $s\in\mathcal{S}$, and
(ii) (ii) the ability to sample a minibatch draw $R_t\sim \eta_t$.
Fix $\theta\in\Theta$ and $(s_i,s_j)\in\mathcal{Q}$ with $P_\theta^S(s_i)>0$ and $P_\theta^S(s_j)>0$, and fix a coupling
$\Gamma_{\theta,ij}\in\Pi\!\big(P_\theta(X\mid S=s_i),\,P_\theta(X\mid S=s_j)\big)$.
The goal is to estimate the distribution of the discrepancy count $\Delta_t$ defined in~\eqref{eq:def_discrepancy_count},
under the joint randomness
\[
(X^{(i)},X^{(j)})\sim \Gamma_{\theta,ij},
\quad
R_t\sim \eta_t.
\]

\paragraph{Coupling Constructions.}
Any coupling $\Gamma_{\theta,ij}$ consistent with the scenario may be used. Typical constructions include:
\begin{itemize}
\item \textbf{Independent coupling:} draw $X^{(i)}\sim P_\theta(X\mid S=s_i)$ and $X^{(j)}\sim P_\theta(X\mid S=s_j)$ independently.
\item \textbf{Shared-latent coupling:} when the conditional distribution admits a representation $X=g(Z,S)$ for some latent $Z$,
draw $Z\sim P_\theta(Z)$ and set $X^{(i)}=g(Z,s_i)$, $X^{(j)}=g(Z,s_j)$.
\item \textbf{Recordwise maximal coupling:} for i.i.d.\ record models, couple records coordinatewise using maximal couplings of the
per-record conditionals; see Section~\ref{app:estimate_caps:structure}.
\end{itemize}
All bounds below are stated for an arbitrary fixed coupling $\Gamma_{\theta,ij}$; tighter couplings can yield smaller caps.

\paragraph{Monte Carlo Procedures.}
Draw i.i.d.\ replicates
\[
(X^{(i,m)},X^{(j,m)})\sim \Gamma_{\theta,ij},\qquad R_t^{(m)}\sim \eta_t,\qquad m=1,\dots,M,
\]
compute $\Delta_t^{(m)}$ via~\eqref{eq:def_discrepancy_count}, and construct either (i) a second-moment cap for sa-HUC or
(ii) a tail cap for HUC. The parameter $\gamma_t\in(0,1)$ denotes the \textit{calibration failure probability} associated with
the Monte Carlo estimation, whereas $\delta_t\in(0,1)$ denotes the \textit{tail level} of the desired HUC cap.

\begin{algorithm}[t]
  \caption{Monte Carlo estimation of discrepancy caps at iteration $t$ (fixed $\theta,(s_i,s_j)$ and coupling $\Gamma_{\theta,ij}$)}
  \label{alg:mc_caps}
  \begin{algorithmic}[1]
    \REQUIRE Iteration $t$; batch size $B_t$; minibatch subsampling rule $\eta_t$;
    coupling sampler for $\Gamma_{\theta,ij}$; Monte Carlo size $M$;
    tail level $\delta_t\in(0,1)$; calibration failure $\gamma_t\in(0,1)$.
    \FOR{$m=1,\dots,M$}
      \STATE Sample $(X^{(i,m)},X^{(j,m)})\sim\Gamma_{\theta,ij}$.
      \STATE Sample $R_t^{(m)}=(I^{(m)}_1,\dots,I^{(m)}_{B_t})\sim\eta_t$.
\STATE Compute $\Delta_t^{(m)}=\#\{b\in\{1,\dots,B_t\}: x^{(i,m)}_{I^{(m)}_b}\neq x^{(j,m)}_{I^{(m)}_b}\}$.
    \ENDFOR
    \STATE Set $\varepsilon \gets \sqrt{\frac{1}{2M}\log\frac{2}{\gamma_t}}$.
    \STATE \textbf{sa-HUC output:} $\widehat{\mu}_{2,t}\gets \frac{1}{M}\sum_{m=1}^M(\Delta_t^{(m)})^2$ and
    $\widehat{K}_t^2 \gets \widehat{\mu}_{2,t}+B_t^2\,\varepsilon$.
    \STATE \textbf{HUC output:} if $\delta_t>\varepsilon$, let $\widehat{K}_t(\delta_t)$ be the empirical
    $(1-(\delta_t-\varepsilon))$-quantile of $\{\Delta_t^{(m)}\}_{m=1}^M$; otherwise set $\widehat{K}_t(\delta_t)\gets B_t$.
    \STATE \textbf{return} $(\widehat{K}_t(\delta_t),\,\widehat{K}_t^2)$.
  \end{algorithmic}
\end{algorithm}

\subsubsection{Finite-Sample Bounds for sa-HUC}
\label{app:estimate_caps:mc_ms}

We give a direct concentration bound for $\mathbb{E}[\Delta_t^2]$, using the bounded range
$0\leq \Delta_t^2\leq B_t^2$ from~\eqref{eq:Delta_range}.

\begin{lemma}[Monte Carlo sa-HUC cap via Hoeffding]
\label{lem:mc_ms_cap}
Fix an iteration $t$, a prior $\theta\in\Theta$, a secret pair $(s_i,s_j)\in\mathcal{Q}$ with
$P_\theta^S(s_i)>0$ and $P_\theta^S(s_j)>0$, a coupling $\Gamma_{\theta,ij}$, and a minibatch subsampling rule $\eta_t$.
Let $\Delta_t^{(1)},\dots,\Delta_t^{(M)}$ be i.i.d.\ draws of the discrepancy count $\Delta_t$ under the joint randomness
\[
(X^{(i)},X^{(j)})\sim \Gamma_{\theta,ij},\qquad R_t\sim \eta_t,
\]
and define the empirical second moment
\[
\widehat{\mu}_{2,t}(\theta,ij)\;:=\;\frac{1}{M}\sum_{m=1}^M\big(\Delta_t^{(m)}\big)^2.
\]
Then for any $\gamma_t\in(0,1)$, with probability at least $1-\gamma_t$ over the Monte Carlo sampling,
\begin{equation}
\label{eq:mc_ms_cap_bound}
\mathbb E\big[\Delta_t^2\big]
\;\leq\;
\widehat{\mu}_{2,t}(\theta,ij)\;+\;B_t^2\sqrt{\frac{1}{2M}\log\frac{2}{\gamma_t}}.
\end{equation}
\end{lemma}

\begin{proof}
Let $Z_m := (\Delta_t^{(m)})^2$ for $m=1,\dots,M$. By construction,
$\Delta_t\in\{0,1,\dots,B_t\}$ (\ \eqref{eq:Delta_range}), hence
\[
0 \leq Z_m \leq B_t^2 \qquad \text{a.s.}
\]
Moreover, $\{Z_m\}_{m=1}^M$ are i.i.d.\ because the pairs
$(X^{(i,m)},X^{(j,m)})\sim\Gamma_{\theta,ij}$ and minibatch draws
$R_t^{(m)}\sim \eta_t$ are sampled i.i.d.\ across $m$.

Denote $\mu := \mathbb{E}[Z_1]=\mathbb{E}[\Delta_t^2]$ and
$\bar{Z} := \frac{1}{M}\sum_{m=1}^M Z_m=\widehat{\mu}_{2,t}(\theta,ij)$.
By Hoeffding's inequality for bounded i.i.d.\ random variables, for any $\varepsilon>0$,
\[
\Pr\!\left(\mu - \bar{Z} \geq \varepsilon\right)
\leq \exp\!\left(-\frac{2M\varepsilon^2}{(B_t^2-0)^2}\right)
= \exp\!\left(-\frac{2M\varepsilon^2}{B_t^4}\right).
\]
Set $\varepsilon := B_t^2\sqrt{\frac{1}{2M}\log\frac{2}{\gamma_t}}$.
Then
\[
\Pr\!\left(\mu - \bar{Z} \geq B_t^2\sqrt{\frac{1}{2M}\log\frac{2}{\gamma_t}}\right)
\leq \exp\!\left(-\log\frac{2}{\gamma_t}\right)
= \frac{\gamma_t}{2}
\leq \gamma_t.
\]
Equivalently, with probability at least $1-\gamma_t$,
\[
\mathbb{E}[\Delta_t^2]
=\mu
\leq \bar{Z} + B_t^2\sqrt{\frac{1}{2M}\log\frac{2}{\gamma_t}}
= \widehat{\mu}_{2,t}(\theta,ij) + B_t^2\sqrt{\frac{1}{2M}\log\frac{2}{\gamma_t}},
\]
which is exactly \eqref{eq:mc_ms_cap_bound}.
\end{proof}

\paragraph{Uniform ms-Cap Over a Finite Instantiated Family.}
Let $\widehat{\Theta}\subseteq\Theta$ be a finite set of instantiated priors, and let
$\widehat{\mathcal{Q}}\subseteq\mathcal{Q}$ be a finite set of secret pairs.
Define the index set $\mathcal{I}_t := \widehat{\Theta}\times\widehat{\mathcal{Q}}$ and let $N_t:=|\mathcal{I}_t|$.

\begin{corollary}[Uniform sa-HUC cap over $\widehat{\Theta}\times\widehat{\mathcal{Q}}$]
\label{cor:uniform_ms_cap}
For each $(\theta,(s_i,s_j))\in\mathcal{I}_t$, compute
\[
\widehat{K}_{t,\theta,ij}^2
\;:=\;
\widehat{\mu}_{2,t}(\theta,ij)\;+\;B_t^2\sqrt{\frac{1}{2M}\log\frac{2N_t}{\gamma_t}},
\]
i.e., apply Lemma~\ref{lem:mc_ms_cap} with calibration failure $\gamma_t/N_t$.
Then, with probability at least $1-\gamma_t$, simultaneously for all $(\theta,(s_i,s_j))\in\mathcal{I}_t$,
\[
\mathbb{E}[\Delta_t^2]\ \leq\ \widehat{K}_{t,\theta,ij}^2.
\]
Consequently, the uniform sa-HUC cap
\[
\widehat{K}_t^2\;:=\;\max_{(\theta,(s_i,s_j))\in\mathcal{I}_t}\widehat{K}_{t,\theta,ij}^2
\]
satisfies $\mathbb{E}[\Delta_t^2]\leq \widehat{K}_t^2$ for all $(\theta,(s_i,s_j))\in\mathcal{I}_t$ on the same event.
\end{corollary}

\begin{proof}
For each index $k=(\theta,(s_i,s_j))\in\mathcal{I}_t$, apply Lemma~\ref{lem:mc_ms_cap}
with failure probability $\gamma_t/N_t$. This yields
\[
\Pr\!\left(\mathbb{E}[\Delta_t^2] \leq \widehat{K}_{t,\theta,ij}^2\right)\ \geq\ 1-\frac{\gamma_t}{N_t}.
\]
By the union bound over the $N_t$ indices,
\[
\Pr\!\left(\forall(\theta,(s_i,s_j))\in\mathcal{I}_t:\ \mathbb{E}[\Delta_t^2]\leq \widehat{K}_{t,\theta,ij}^2\right)
\ \geq\ 1-\sum_{k\in\mathcal{I}_t}\frac{\gamma_t}{N_t}
\ =\ 1-\gamma_t.
\]
On this same event, taking $\widehat{K}_t^2:=\max_{(\theta,(s_i,s_j))\in\mathcal{I}_t}\widehat{K}_{t,\theta,ij}^2$
gives $\mathbb{E}[\Delta_t^2]\leq \widehat K_t^2$ simultaneously for all $(\theta,(s_i,s_j))\in\mathcal{I}_t$.
\end{proof}

\subsubsection{Finite-Sample Bounds for HUC}
\label{app:estimate_caps:mc_huc}

We estimate a high-probability HUC cap $K_t(\delta_t)$ satisfying~\eqref{eq:HUC_tail_cap} via an empirical quantile
bound based on the Dvoretzky--Kiefer--Wolfowitz (DKW) inequality.

Fix $\theta\in\Theta$, $(s_i,s_j)\in\mathcal{Q}$ with $P_\theta^S(s_i),P_\theta^S(s_j)>0$, a coupling
$\Gamma_{\theta,ij}\in\Pi(P_\theta(X\mid S=s_i),P_\theta(X\mid S=s_j))$, and a subsampling rule $\eta_t$.
Let $\Delta_t$ be the discrepancy count induced by
$(X^{(i)},X^{(j)})\sim\Gamma_{\theta,ij}$ and $R_t\sim\eta_t$.
Let $F(u):=\Pr(\Delta_t\leq u)$ and $\widehat{F}_M$ be the empirical CDF from i.i.d.\ samples
$\Delta_t^{(1)},\dots,\Delta_t^{(M)}$.

\begin{lemma}
\label{lem:dkw_huc_cap}
Let $\Delta_t^{(1)},\dots,\Delta_t^{(M)}$ be i.i.d.\ samples of $\Delta_t$ with true CDF $F$ and empirical CDF $\widehat{F}_M$.
For any calibration failure $\gamma_t\in(0,1)$, define
\[
\varepsilon \;:=\; \sqrt{\frac{1}{2M}\log\frac{2}{\gamma_t}}.
\]
Then, with probability at least $1-\gamma_t$,
\[
\sup_{u\in\mathbb{R}}\big|\widehat{F}_M(u)-F(u)\big|\ \leq\ \varepsilon.
\]
Consequently, for any $\delta_t\in(0,1)$ with $\delta_t>\varepsilon$, let
$\widehat{q}_{1-(\delta_t-\varepsilon)}$ be an empirical $(1-(\delta_t-\varepsilon))$-quantile of
$\{\Delta_t^{(m)}\}_{m=1}^M$ (e.g., the $\lceil(1-(\delta_t-\varepsilon))M\rceil$-th order statistic).
Then on the same event,
\[
\Pr\big(\Delta_t \leq \widehat{q}_{1-(\delta_t-\varepsilon)}\big)\ \geq\ 1-\delta_t.
\]
In particular, $K_t(\delta_t):=\widehat{q}_{1-(\delta_t-\varepsilon)}$ is a valid tail cap at level $\delta_t$
for this fixed $(\theta,(s_i,s_j))$ and coupling $\Gamma_{\theta,ij}$.
\end{lemma}

\begin{proof}
By the Dvoretzky--Kiefer--Wolfowitz (DKW) inequality \cite{dvoretzky1956asymptotic,massart1990tight}, for any $\varepsilon>0$,
\[
\Pr\!\left(\sup_{u\in\mathbb{R}} \big|\widehat{F}_M(u)-F(u)\big| > \varepsilon\right)
\;\leq\;
2e^{-2M\varepsilon^2}.
\]
With $\varepsilon:=\sqrt{\tfrac{1}{2M}\log\tfrac{2}{\gamma_t}}$, the right-hand side equals $\gamma_t$.
Hence, with probability at least $1-\gamma_t$,
\begin{equation}\label{eq:dkw_event}
\sup_{u\in\mathbb{R}} \big|\widehat{F}_M(u)-F(u)\big|\ \leq\ \varepsilon.
\end{equation}

Assume \eqref{eq:dkw_event} holds and fix any $\delta_t\in(0,1)$ with $\delta_t>\varepsilon$.
Let $\widehat{q}:=\widehat{q}_{1-(\delta_t-\varepsilon)}$ be an empirical $(1-(\delta_t-\varepsilon))$-quantile, i.e.,
$\widehat{F}_M(\widehat q)\geq 1-(\delta_t-\varepsilon)$.
Then, using \eqref{eq:dkw_event},
\[
F(\widehat{q})
\;\geq\;
\widehat F_M(\widehat{q})-\varepsilon
\;\geq\;
\bigl(1-(\delta_t-\varepsilon)\bigr)-\varepsilon
\;=\;
1-\delta_t.
\]
Since $F(\widehat{q})=\Pr(\Delta_t\leq \widehat{q})$, we conclude that on the same event,
\[
\Pr(\Delta_t\leq \widehat{q})\ \geq\ 1-\delta_t.
\]
Therefore, with probability at least $1-\gamma_t$ over the Monte Carlo sampling,
the choice $K_t(\delta_t):=\widehat{q}_{1-(\delta_t-\varepsilon)}$ satisfies
$\Pr(\Delta_t\leq K_t(\delta_t))\geq 1-\delta_t$, as claimed.
\end{proof}

\paragraph{Uniform Tail Cap over a Finite Instantiated Family.}
Let $\widehat{\Theta}\subseteq\Theta$ and $\widehat{\mathcal Q}\subseteq\mathcal{Q}$ be finite sets, and let
$\mathcal{I}_t:=\widehat{\Theta}\times\widehat{\mathcal{Q}}$ with $N_t:=|\mathcal{I}_t|$.

\begin{corollary}
\label{cor:uniform_huc_cap}
For each $(\theta,(s_i,s_j))\in\mathcal{I}_t$, apply Lemma~\ref{lem:dkw_huc_cap} with calibration failure $\gamma_t/N_t$,
i.e., with
\[
\varepsilon \;:=\; \sqrt{\frac{1}{2M}\log\frac{2N_t}{\gamma_t}}.
\]
Assume $\delta_t>\varepsilon$. Define
\[
\widehat{K}_{t,\theta,ij}(\delta_t)\;:=\;\widehat{q}^{(\theta,ij)}_{1-(\delta_t-\varepsilon)},
\qquad
\widehat{K}_t(\delta_t)\;:=\;\max_{(\theta,(s_i,s_j))\in\mathcal{I}_t}\widehat{K}_{t,\theta,ij}(\delta_t),
\]
where $\widehat{q}^{(\theta,ij)}_{1-(\delta_t-\varepsilon)}$ is the empirical $(1-(\delta_t-\varepsilon))$-quantile computed
from the Monte Carlo samples generated for that $(\theta,(s_i,s_j))$.
Then, with probability at least $1-\gamma_t$ over the Monte Carlo sampling (simultaneously across all indices in $\mathcal I_t$),
\[
\Pr\big(\Delta_t \leq \widehat{K}_{t,\theta,ij}(\delta_t)\big)\ \geq\ 1-\delta_t
\quad\forall (\theta,(s_i,s_j))\in\mathcal{I}_t,
\]
and hence $\Pr(\Delta_t \leq \widehat{K}_t(\delta_t))\geq 1-\delta_t$ holds uniformly over $(\theta,(s_i,s_j))\in\mathcal{I}_t$
on the same event.
\end{corollary}

\begin{proof}
Fix $t$ and the finite index set $\mathcal{I}_t=\widehat{\Theta}\times\widehat{\mathcal{Q}}$ with $N_t:=|\mathcal{I}_t|$.
For each index $a=(\theta,(s_i,s_j))\in\mathcal{I}_t$, let
$\{\Delta_t^{(m)}(a)\}_{m=1}^M$ denote the $M$ i.i.d.\ Monte Carlo draws generated under the joint randomness
$(X^{(i)},X^{(j)})\sim\Gamma_{\theta,ij}$ and $R_t\sim\eta_t$, and let $\widehat{K}_{t,a}(\delta_t)$ be the resulting
empirical tail cap defined in the corollary.

For each $a\in\mathcal{I}_t$, define the “calibration-success” event
\[
E_a \;:=\; \Big\{\Pr\big(\Delta_t \leq \widehat{K}_{t,a}(\delta_t)\big)\ \geq\ 1-\delta_t\Big\},
\]
where the probability is with respect to the true joint law of $\Delta_t$ for that fixed $a$ (i.e., over
$(X^{(i)},X^{(j)})$ and $R_t$), and the randomness in $E_a$ is only the Monte Carlo sampling used to construct
$\widehat{K}_{t,a}(\delta_t)$.

Apply Lemma~\ref{lem:dkw_huc_cap} to each fixed $a\in\mathcal{I}_t$ with calibration failure $\gamma_t/N_t$.
With
\[
\varepsilon=\sqrt{\frac{1}{2M}\log\frac{2N_t}{\gamma_t}},
\]
and assuming $\delta_t>\varepsilon$, the lemma yields
\[
\Pr_{\mathrm{MC}}(E_a)\ \geq\ 1-\frac{\gamma_t}{N_t},
\qquad \forall a\in\mathcal{I}_t,
\]
where $\Pr_{\mathrm{MC}}$ denotes probability over the Monte Carlo sampling for index $a$.

By a union bound,
\[
\Pr_{\mathrm{MC}}\Big(\bigcap_{a\in\mathcal{I}_t} E_a\Big)
\;\geq\;
1-\sum_{a\in\mathcal{I}_t}\Pr_{\mathrm{MC}}(E_a^c)
\;\geq\;
1-\sum_{a\in\mathcal{I}_t}\frac{\gamma_t}{N_t}
\;=\;
1-\gamma_t.
\]
Therefore, with probability at least $1-\gamma_t$ (simultaneously over all indices $a\in\mathcal{I}_t$),
we have $\Pr(\Delta_t \leq \widehat{K}_{t,a}(\delta_t))\geq 1-\delta_t$ for every $a=(\theta,(s_i,s_j))\in\mathcal{I}_t$,
which proves the first claim.

Finally, define $\widehat{K}_t(\delta_t):=\max_{a\in\mathcal{I}_t}\widehat{K}_{t,a}(\delta_t)$.
On the event $\bigcap_{a\in\mathcal{I}_t}E_a$, for any fixed $a$ we have the set inclusion
$\{\Delta_t \leq \widehat K_{t,a}(\delta_t)\}\subseteq \{\Delta_t \leq \widehat{K}_t(\delta_t)\}$, hence
\[
\Pr\big(\Delta_t \leq \widehat{K}_t(\delta_t)\big)
\;\geq\;
\Pr\big(\Delta_t \leq \widehat{K}_{t,a}(\delta_t)\big)
\;\geq\;
1-\delta_t.
\]
This holds for every $a\in\mathcal{I}_t$ on the same event, completing the proof.
\end{proof}

Suppose a tail cap $K_t(\delta_t)$ is used in HUC-style accounting at each round $t=1,\dots,T$.
Let $E_t$ denote the runtime cap-validity event for the coupled draw and minibatch draw at round $t$,
\[
E_t \;:=\; \{\Delta_t \leq K_t(\delta_t)\},
\]
where the probability is over $(X^{(i)},X^{(j)})$ and $R_t\sim\eta_t$.
If $\Pr(E_t)\geq 1-\delta_t$ holds for each round, then by a union bound,
\[
\Pr\Big(\bigcap_{t=1}^T E_t\Big)\ \geq\ 1-\sum_{t=1}^T \delta_t.
\]
If, additionally, each $K_t(\delta_t)$ is obtained from Monte Carlo estimation with calibration failure probability $\gamma_t$
(e.g., via Corollary~\ref{cor:uniform_huc_cap}), then the event that \textit{all} per-round tail guarantees hold across all rounds
has probability at least
\begin{equation}
\label{eq:delta_gamma_bookkeeping}
1-\sum_{t=1}^T\delta_t\;-\;\sum_{t=1}^T\gamma_t.
\end{equation}
For sa-HUC, there is no tail parameter $\delta_t$; only the calibration failures $\gamma_t$ appear, and the analogous
across-round union bound yields probability at least $1-\sum_{t=1}^T \gamma_t$.

\subsection{Closed-Form Bounds under Record-Level Structure}
\label{app:estimate_caps:structure}

This section provides analytic bounds for $\mathbb{E}[\Delta_t^2]$ and for the tail behavior of $\Delta_t$
under record-level i.i.d.\ structure.
Assume the dataset consists of $N$ records $X=(X_1,\dots,X_N)$ and, conditional on $(S=s,\theta)$, records are i.i.d.:
\begin{equation}
X_k \mid (S=s,\theta) \;\sim\; P_{\theta,s}^{\mathrm{rec}},
\qquad k=1,\dots,N.
\end{equation}

\paragraph{Recordwise Maximal Coupling and Mismatch Probability.}
Fix $\theta\in\Theta$ and $(s_i,s_j)\in\mathcal{Q}$.
Let $p:=P_{\theta,s_i}^{\mathrm{rec}}$ and $q:=P_{\theta,s_j}^{\mathrm{rec}}$.
There exists a \textit{maximal coupling} of a single record $(X_k^{(i)},X_k^{(j)})$ with marginals $p$ and $q$ such that
\begin{equation}
\Pr\!\big[X_k^{(i)}\neq X_k^{(j)}\big]
\;=\;
\tau_{\theta,ij}
\;:=\;
\mathrm{TV}(p,q),
\end{equation}
where $\mathrm{TV}(\cdot,\cdot)$ denotes total variation distance.
Coupling records independently across $k$ yields mismatch indicators
$M_k:=\mathbf{1}\{X_k^{(i)}\neq X_k^{(j)}\}$ that are i.i.d.\ $\mathrm{Bernoulli}(\tau_{\theta,ij})$.

\paragraph{Minibatch Discrepancy: \textsf{WR} vs. \textsf{WOR}.}
Let $\eta_t$ select a minibatch $\mathsf{I}_t\subseteq\{1,\dots,N\}$ of size $B_t$, and recall
$\Delta_t=\sum_{k\in\mathsf{I}_t} M_k$.

\textit{\textsf{WR} (or i.i.d.\ index draws).}
If $\eta_t$ draws indices $I_1,\dots,I_{B_t}$ independently (with replacement), then
\begin{equation}
\label{eq:Delta_binomial_tau}
\Delta_t \;\sim\; \mathrm{Binomial}\!\big(B_t,\tau_{\theta,ij}\big),
\end{equation}
and hence
\begin{align}
\label{eq:Delta_moments_tau}
\mathbb{E}[\Delta_t] &= B_t\tau_{\theta,ij},\\
\mathbb{E}[\Delta_t^2] &= B_t\tau_{\theta,ij}(1-\tau_{\theta,ij}) + (B_t\tau_{\theta,ij})^2.
\end{align}

\textit{\textsf{WOR}.}
If $\eta_t$ samples $\mathsf{I}_t$ uniformly without replacement, then conditional on the total mismatch count
$S_N:=\sum_{k=1}^N M_k$, the discrepancy $\Delta_t$ has a hypergeometric distribution:
\[
\Delta_t \,\big|\, S_N \;\sim\; \mathrm{Hypergeometric}\big(N,\,S_N,\,B_t\big).
\]
In particular, $\mathbb{E}[\Delta_t]=B_t\tau_{\theta,ij}$ still holds, and standard concentration inequalities for sampling
without replacement (e.g., Serfling/Hoeffding-type bounds) yield binomial-style tails. Consequently, the tail cap derived
below for the with-replacement case remains a valid (conservative) choice under without-replacement sampling.

\paragraph{Resulting ms-Cap from Total-Variation Control.}
In the with-replacement case, \eqref{eq:Delta_moments_tau} implies $\mathbb{E}[\Delta_t^2]=G(B_t,\tau_{\theta,ij})$ with
\[
G(B_t,\tau)\;:=\;B_t\tau(1-\tau) + (B_t\tau)^2.
\]
Thus one may take the sa-HUC cap
\begin{equation}
\label{eq:ms_cap_from_tv}
K_t^2 \;:=\; \sup_{\theta\in\Theta}\sup_{(s_i,s_j)\in\mathcal{Q}} G\!\big(B_t,\tau_{\theta,ij}\big),
\end{equation}
whenever $\tau_{\theta,ij}$ (or an upper bound on it) is available uniformly over $\theta$ and $(s_i,s_j)$.

\paragraph{Tail Bounds and a High-Probability HUC.}
Under \eqref{eq:Delta_binomial_tau}, Hoeffding's inequality gives, for any $\lambda>0$,
\[
\Pr\!\big(\Delta_t \geq \mathbb{E}[\Delta_t] + \lambda\big)\ \leq\ \exp\!\Big(-\frac{2\lambda^2}{B_t}\Big).
\]
Therefore, for any tail level $\delta_t\in(0,1)$, the choice
\begin{equation}
\label{eq:huc_cap_from_tv}
K_t(\delta_t)
\;:=\;
\sup_{\theta\in\Theta}\sup_{(s_i,s_j)\in\mathcal{Q}}
\Big(B_t\tau_{\theta,ij} + \sqrt{\tfrac{B_t}{2}\log\tfrac{1}{\delta_t}}\Big)
\end{equation}
ensures $\Pr(\Delta_t \leq K_t(\delta_t))\geq 1-\delta_t$ for the with-replacement minibatch rule.
Moreover, the same cap remains valid (conservatively) for common without-replacement rules by standard concentration
bounds for sampling without replacement.

\subsection{Localized Influence of Secrets}
\label{app:estimate_caps:subset}

In some scenarios, the secret affects only a limited subset of records.
Fix $\theta\in\Theta$ and $(s_i,s_j)\in\mathcal{Q}$, and suppose there exists a (possibly random) index set
$A\subseteq\{1,\dots,N\}$ and a coupling $\Gamma_{\theta,ij}$ such that, for $(X^{(i)},X^{(j)})\sim \Gamma_{\theta,ij}$,
\[
X_k^{(i)} = X_k^{(j)} \quad \text{for all } k\notin A,
\qquad\text{and}\qquad
|A|\leq D_{\max}\ \text{a.s.}
\]
Then, for any minibatch rule selecting $B_t$ indices, the discrepancy count satisfies the deterministic bound
\[
\Delta_t \;\leq\; \min\{B_t,D_{\max}\}.
\]
Accordingly, one may take the HUC and sa-HUC caps
\begin{equation}
\label{eq:subset_caps}
K_t \;:=\; \min\{B_t,D_{\max}\},
\qquad
K_t^2 \;:=\; \min\{B_t^2,D_{\max}^2\}.
\end{equation}
If the distribution of $|A|$ is available (analytically or via simulation), then sharper moment or tail caps can be obtained by
conditioning on $|A|$ and applying the corresponding sampling-with/without-replacement formulas.

\subsection{Instantiation of HUC and sa-HUC}
\label{app:estimate_caps:summary}

To instantiate the HUC/sa-HUC accountants, it suffices to specify (a) a coupling family
$\{\Gamma_{\theta,ij}\}_{\theta\in\Theta,(s_i,s_j)\in\mathcal{Q}}$ consistent with the scenario assumptions and
(b) the minibatch selection rule $\eta_t$.
Given these, one obtains admissible caps either by analytic control of the induced discrepancy distribution (e.g.,
Section~\ref{app:estimate_caps:structure} or Section~\ref{app:estimate_caps:subset}) or by Monte Carlo estimation
(Section~\ref{app:estimate_caps:mc}), with finite-sample calibration guarantees given by
Lemma~\ref{lem:mc_ms_cap} and Corollary~\ref{cor:uniform_ms_cap} for sa-HUC,
and by Lemma~\ref{lem:dkw_huc_cap} and Corollary~\ref{cor:uniform_huc_cap} for HUC tail caps,
and across-round bookkeeping in~\eqref{eq:delta_gamma_bookkeeping}.
The resulting caps $K_t(\delta_t)$ and/or $K_t^2$ can then be used directly in the corresponding HUC/sa-HUC accounting theorems.

\begin{algorithm}[t]
  \caption{Monte Carlo discrepancy-cap estimation at iteration $t$}
  \label{alg:mc_discrepancy_cap}
  \begin{algorithmic}[1]
    \REQUIRE Prior $\theta\in\Theta$; secret pair $(s_i,s_j)\in\mathcal{Q}$ with $P_\theta^S(s_i),P_\theta^S(s_j)>0$;
    coupling sampler $\mathsf{Couple}_\theta(s_i,s_j)$ returning $(X^{(i)},X^{(j)})\sim\Gamma_{\theta,ij}$ with marginals
    $X^{(i)}\sim P_\theta(\cdot\mid S=s_i)$ and $X^{(j)}\sim P_\theta(\cdot\mid S=s_j)$;
    minibatch subsampling rule $\eta_t$ selecting $B_t$ indices;
    Monte Carlo size $M$;
    HUC tail level $\delta_t\in(0,1)$;
    calibration failure $\gamma_t\in(0,1)$.
    \ENSURE A tail cap $\widehat{K}_t(\delta_t)$ for HUC and a second-moment cap $\widehat{K}_t^2$ for sa-HUC.

    \FOR{$m=1,\dots,M$}
      \STATE Sample $(X^{(i,m)},X^{(j,m)}) \gets \mathsf{Couple}_\theta(s_i,s_j)$.
      \STATE Sample minibatch indices $R_t^{(m)}=(I^{(m)}_1,\dots,I^{(m)}_{B_t})\sim\eta_t$.
      \STATE Compute the discrepancy count
      \[
        \Delta_t^{(m)} \;\gets\; \#\{b\in\{1,\dots,B_t\}:\ x^{(i,m)}_{I^{(m)}_b}\neq x^{(j,m)}_{I^{(m)}_b}\},
        \; 0\leq \Delta_t^{(m)}\leq B_t.
      \]
    \ENDFOR

    \STATE Set $\varepsilon \gets \sqrt{\frac{1}{2M}\log\frac{2}{\gamma_t}}$.

    \STATE \textbf{(sa-HUC cap).}
    Compute $\widehat\mu_{2,t}\gets \frac{1}{M}\sum_{m=1}^M(\Delta_t^{(m)})^2$ and set
    \[
      \widehat K_t^2 \;\gets\; \widehat\mu_{2,t} \;+\; B_t^2\,\varepsilon.
    \]

    \STATE \textbf{(HUC tail cap).}
    If $\delta_t>\varepsilon$, set $\widehat{K}_t(\delta_t)$ to the empirical $(1-(\delta_t-\varepsilon))$-quantile of
    $\{\Delta_t^{(m)}\}_{m=1}^M$; otherwise set $\widehat{K}_t(\delta_t)\gets B_t$.

    \STATE \textbf{return} $(\widehat{K}_t(\delta_t),\,\widehat{K}_t^2)$.
  \end{algorithmic}
\end{algorithm}

The sa-HUC accounting accepts any cap $K_t^2$ such that
$\mathbb{E}[\Delta_t^2]\le K_t^2$ uniformly over $\theta\in\Theta$ and $(s_i,s_j)\in\mathcal{Q}$
(or over a finite instantiated family as in Corollary~\ref{cor:uniform_ms_cap}).
If $K_t^2$ is obtained via Monte Carlo (Algorithm~\ref{alg:mc_discrepancy_cap}),
then Lemma~\ref{lem:mc_ms_cap} (and Corollary~\ref{cor:uniform_ms_cap} with a union bound) implies that
$\mathbb{E}[\Delta_t^2]\leq \widehat K_t^2$ holds with probability at least $1-\gamma_t$
over the cap-estimation randomness.


For HUC accounting, we use tail caps $K_t(\delta_t)$ such that
$\Pr(\Delta_t\leq K_t(\delta_t))\geq 1-\delta_t$ uniformly over $\theta\in\Theta$ and $(s_i,s_j)\in\mathcal{Q}$.
If $K_t(\delta_t)$ is obtained via Monte Carlo (Algorithm~\ref{alg:mc_discrepancy_cap}),
then Lemma~\ref{lem:dkw_huc_cap} (and Corollary~\ref{cor:uniform_huc_cap} with a union bound) implies that,
with probability at least $1-\gamma_t$ over the cap-estimation randomness,
the tail guarantee holds. Across rounds, a union bound yields overall probability at least
$1-\sum_{t=1}^T\delta_t-\sum_{t=1}^T\gamma_t$ as in~\eqref{eq:delta_gamma_bookkeeping}.


\subsection{Example Visualization of HUC and sa-HUC}

\begin{figure*}[t]
    \centering

    \begin{subfigure}[b]{0.30\textwidth}
        \centering
        \includegraphics[width=\textwidth]{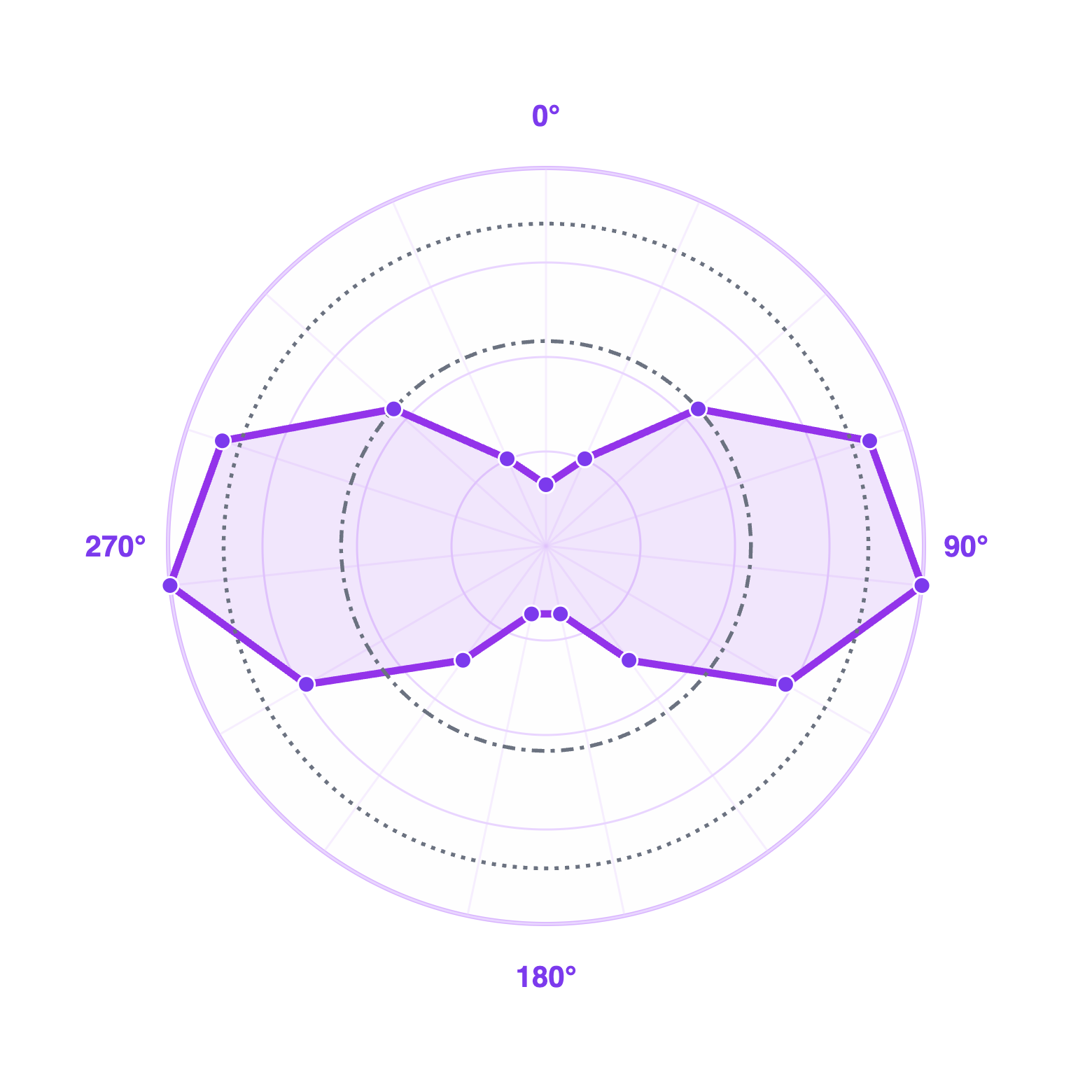}
        \subcaption{HUC (m=15)}
    \end{subfigure}
    \begin{subfigure}[b]{0.30\textwidth}
        \centering
        \includegraphics[width=\textwidth]{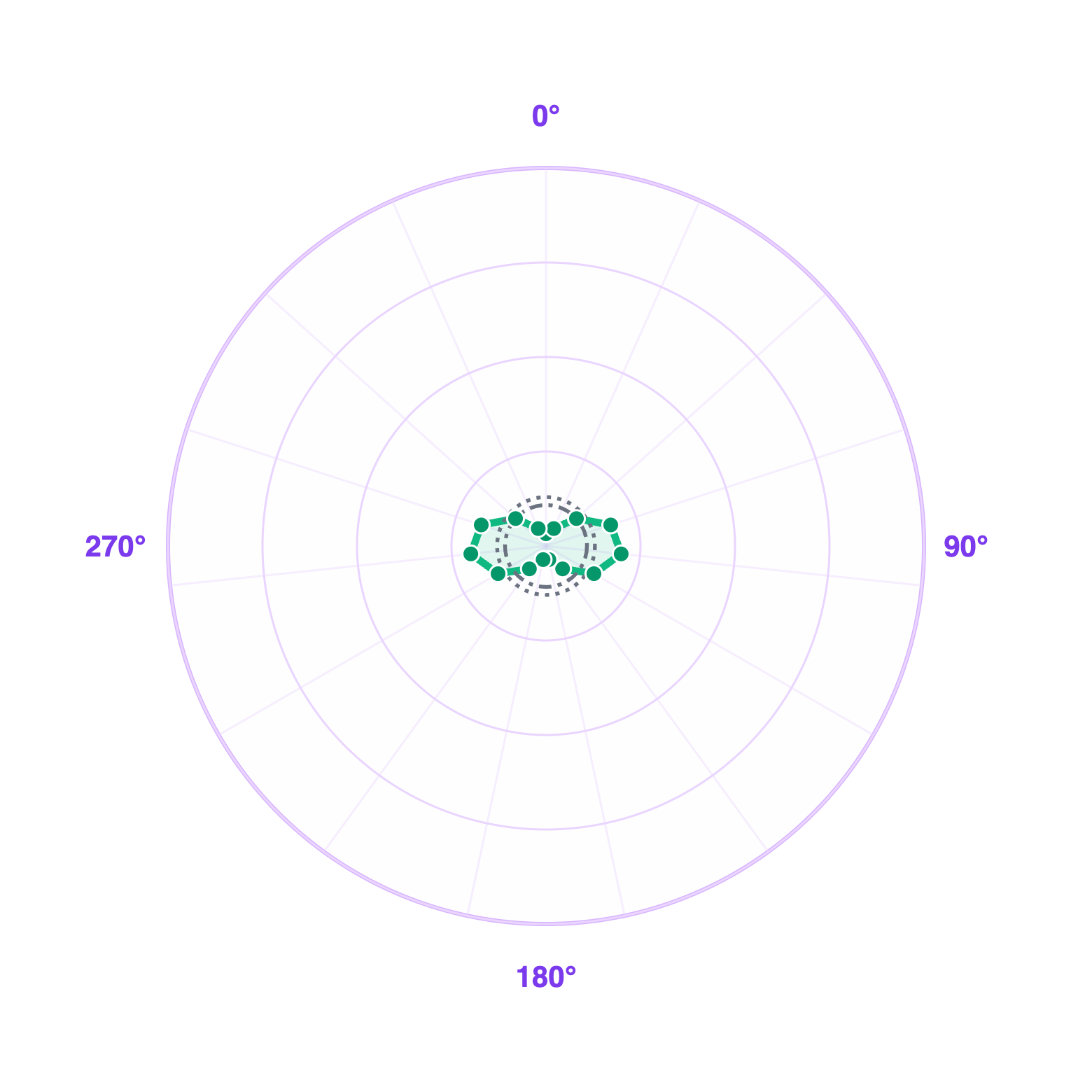}
        \subcaption{sa-HUC (m=15)}
    \end{subfigure}
    \begin{subfigure}[b]{0.30\textwidth}
        \centering
        \includegraphics[width=\textwidth]{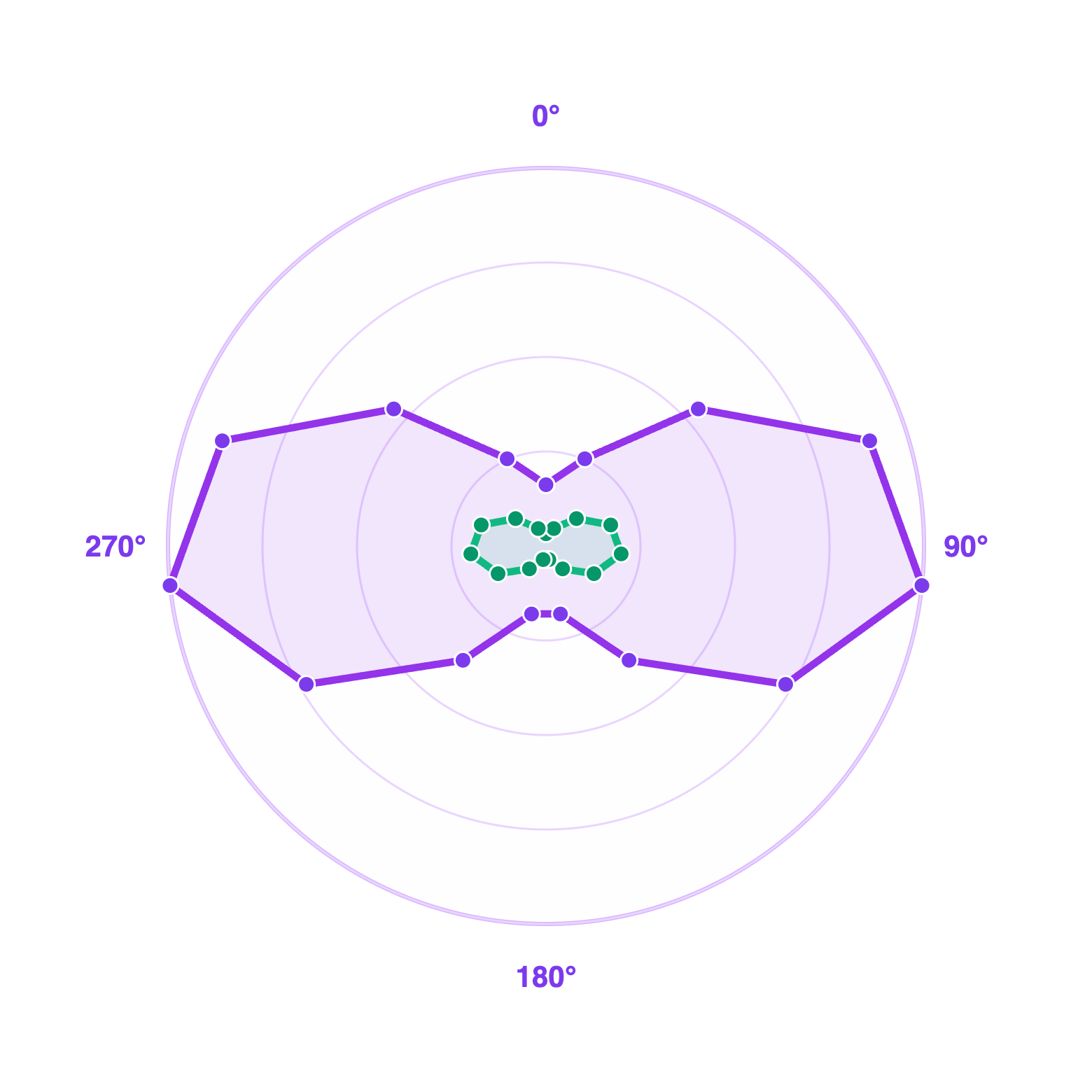}
        \subcaption{Overlay (m=15)}
    \end{subfigure}
    \vspace{0.5em}

    \begin{subfigure}[b]{0.30\textwidth}
        \centering
        \includegraphics[width=\textwidth]{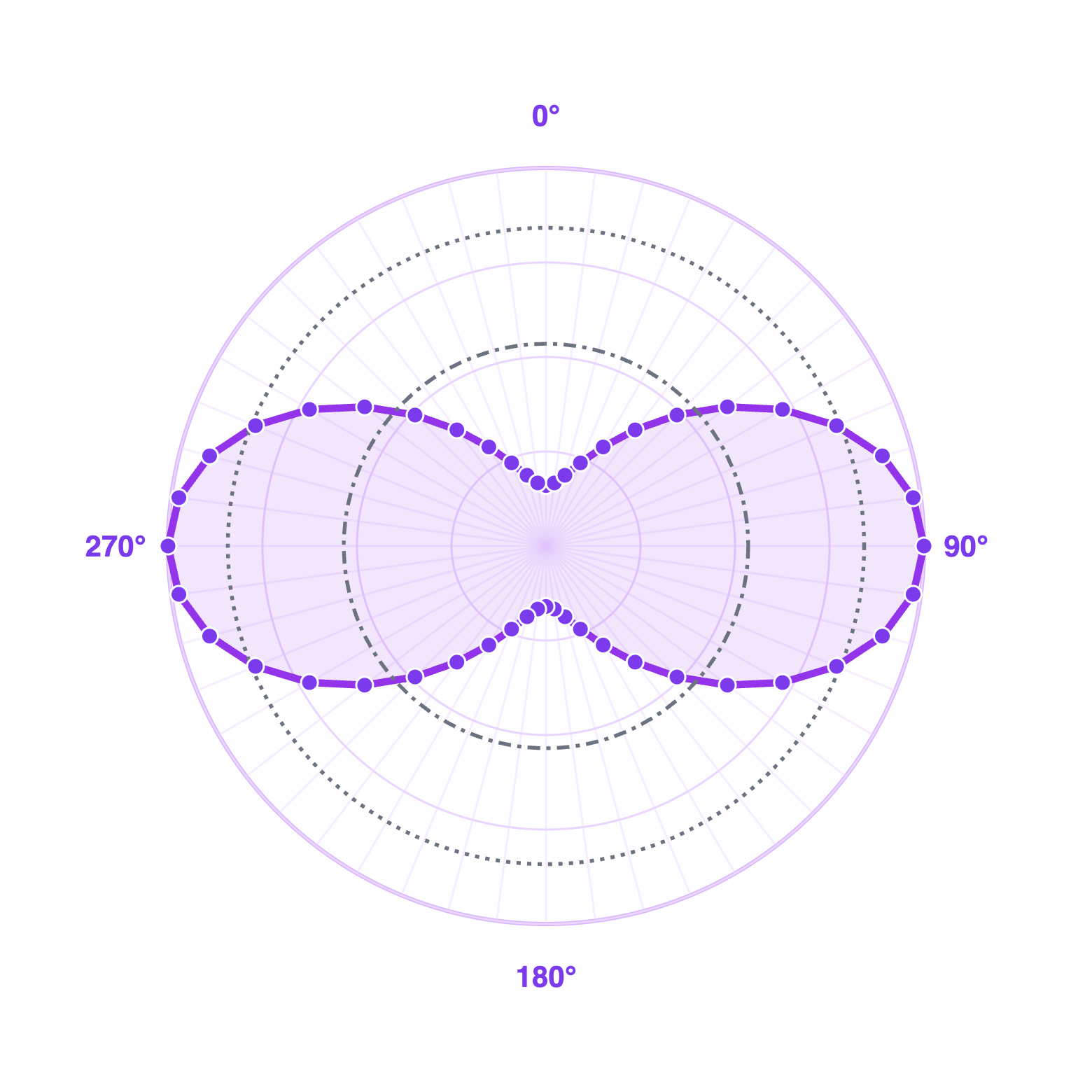}
        \subcaption{HUC (m=48)}
    \end{subfigure}
    \begin{subfigure}[b]{0.30\textwidth}
        \centering
        \includegraphics[width=\textwidth]{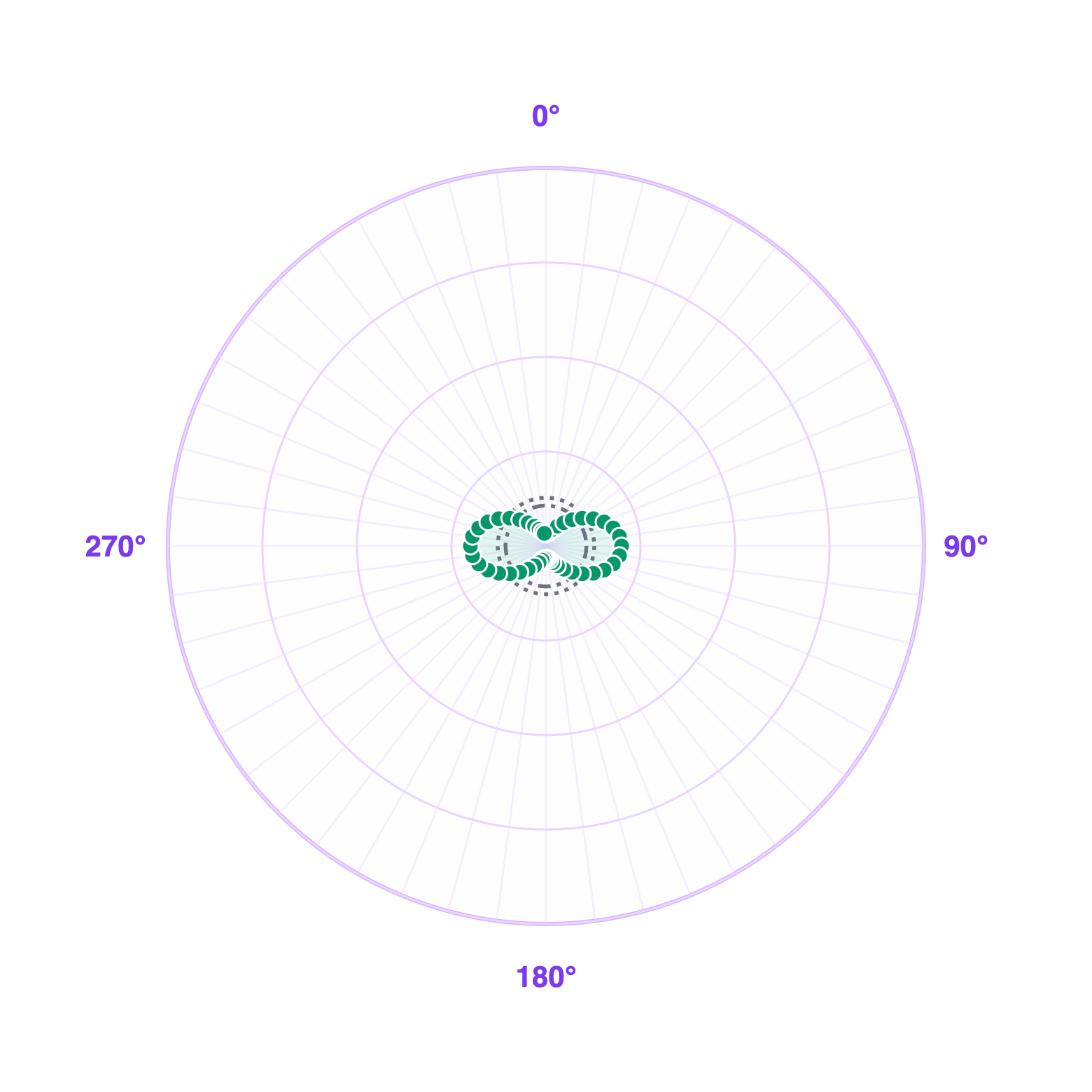}
        \subcaption{sa-HUC (m=48)}
    \end{subfigure}
    \begin{subfigure}[b]{0.30\textwidth}
        \centering
        \includegraphics[width=\textwidth]{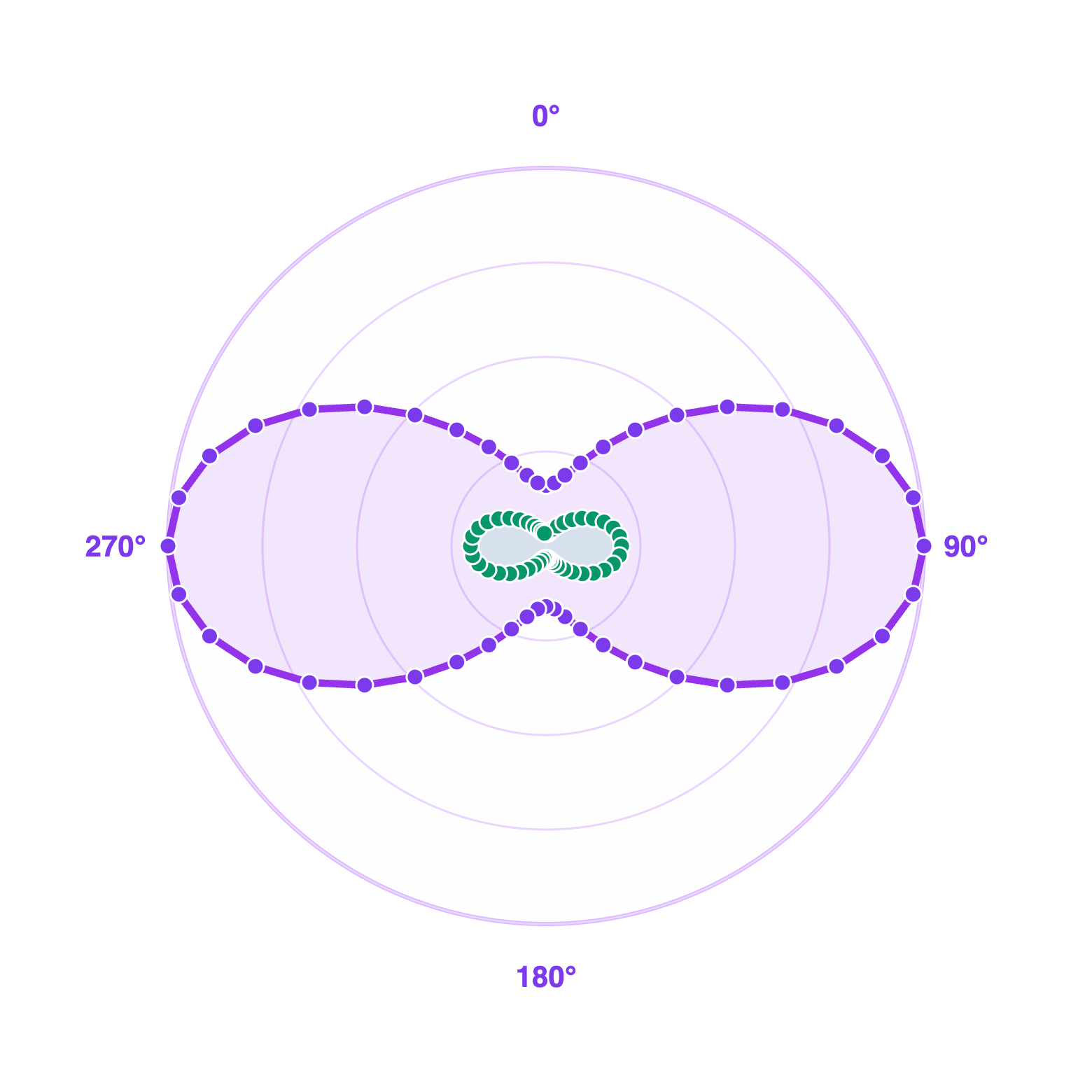}
        \subcaption{Overlay (m=48)}
    \end{subfigure}
    
\centering
\footnotesize
\begin{tabular}{@{}l@{\hspace{1.1em}}l@{\hspace{1.1em}}l@{\hspace{1.1em}}l@{}}
\raisebox{0.15em}{\tikz{\draw[draw=violet, fill=violet!15, line width=1.2pt] (0,0)--(0.55,0)--(0.55,0.16)--(0,0.16)--cycle;}}
& HUC envelope $\{R_i\}$
&
\raisebox{0.15em}{\tikz{\draw[draw=teal, fill=teal!15, line width=1.2pt] (0,0)--(0.55,0)--(0.55,0.16)--(0,0.16)--cycle;}}
& sa-HUC envelope $\{R_i^{\mathrm{sa}}\}$
\\
\raisebox{0.15em}{\tikz{\draw[gray, line width=1.1pt, dash pattern=on 4pt off 2pt on 1pt off 2pt] (0,0)--(0.65,0);}}
& Ave-SRPP (aggregation ring)
&
\raisebox{0.15em}{\tikz{\draw[gray, line width=1.1pt, dotted] (0,0)--(0.65,0);}}
& Joint-SRPP (aggregation ring)
\end{tabular}

    \caption{
Directional envelope profiles for HUC vs.\ sa-HUC under a uniform discrete slice profile $\omega$.
Each panel plots the per-direction R\'enyi envelope values $\{R_i\}_{i=1}^m$ on a polar axis (radius = envelope magnitude, angle = slice direction).
Panels (a,d) show the worst-case HUC envelope; (b,e) show the subsampling-aware envelope (sa-HUC); (c,f) overlay the two profiles on the same radial scale.
The gray rings indicate the corresponding aggregated SRPP costs (Ave-SRPP and Joint-SRPP) computed from the discrete profile.
We use $\alpha=2.5$ and subsampling rate $q=0.20$, and compare $m=15$ versus $m=48$ slice directions.
}
\label{fig:HUC_sa_HUC}
\end{figure*}

Figure \ref{fig:HUC_sa_HUC} visualizes the difference between HUC and sa-HUC.
For each slice direction $u_i$, the plotted radius is the (one-round) per-slice envelope value
$R_i = R_\alpha(\zeta_\sigma, z_i)$ induced by the corresponding directional shift cap $z_i$.
HUC uses a worst-case cap over minibatch randomness, while sa-HUC incorporates subsampling by bounding the envelope
\textit{on average} over minibatch selection; consequently the sa-HUC profile is uniformly smaller and typically
strictly so when the per-round discrepancy varies with the minibatch.
The rings visualize the two SRPP aggregations applied to the same discrete profile:
Ave-SRPP averages $\{R_i\}$ over $\omega$, whereas Joint-SRPP applies a log-moment aggregation that emphasizes larger directions
(and is always at least the Ave value). Increasing $m$ refines the discrete approximation of the directional profile.

\section{SRPP-SGD with Poisson Subsampling}\label{app:poisson_subsampling}

Our definitions allow $\rho\in\{\textsf{WR},\textsf{WOR},\textsf{Poisson}\}$, where under $\rho=\textsf{Poisson}$ the minibatch size
$B_t(R_t)=|\mathsf{I}_t(R_t)|$ is random. The SRPP-SGD accounting theorems in Sec.~\ref{sec:HUC}--\ref{sec:mean_square_HUC}
apply to $\rho=\textsf{Poisson}$ as well, provided the per-round HUC/sa-HUC caps are formed using a \textit{normalized discrepancy}
rather than a fixed $1/B$ factor. Concretely, define
\[
\kappa_t(X,X';R_t):=\frac{K_t(X,X';R_t)}{B_t(R_t)}\,\mathbf{1}\{B_t(R_t)\geq 1\},
\]
and assume either (i) the implementation enforces $B_t(R_t)\geq 1$ a.s.\ (e.g., resampling when $B_t=0$), or (ii) the update is
defined deterministically when $B_t=0$.
Then Proposition~\ref{prop:exist_HUC} and Proposition~\ref{prop:msHUC_from_K2} remain valid after replacing the deterministic bound
$K_t/B_t$ by a uniform cap on $\kappa_t$ (for HUC) or by a uniform bound on $\mathbb{E}_{\eta_t}[\kappa_t^2]$ (for sa-HUC).
All subsequent envelope and composition results are unchanged.

\section{Per-Layer Gradient Clipping }\label{app:perlayer_clipping}

DP-SGD with \textit{per-layer gradient clipping} \cite{mcmahan2017learning} is a structured variant of DP-SGD, where gradient clipping and noise addition are applied independently to disjoint parameter blocks, rather than globally across the entire model.

\subsection{DP-SGD with Per-Layer Gradient Clipping}

We briefly recall the standard DP-SGD \cite{Abadi2016} with per-example gradient clipping.
Given a minibatch $B$ of size $b$, DP-SGD computes per-example gradients $g_{i} = \nabla_{\zeta} \ell(\zeta;x_{i})$, clips each to an $\ell_{2}$ bound $C$, averages the clipped gradients, and adds Gaussian noise calibrated to $C$:
\[
\tilde{g}_{i} \;=\; g_{i} \cdot \min\!\left\{1,\tfrac{C}{\|g_{i}\|_{2}}\right\}, \quad \hat{g}\;=\; \frac{1}{b}\sum_{i\in B}\tilde{g}_{i} \; + \; \mathcal{N}\left(0, \tfrac{\sigma^{2} C^{2}}{b^{2}} I\right),
\]
followed by the SGD update:
\[
\zeta\;\leftarrow \; \zeta-\eta \cdot \hat{g}.
\]

DP-SGD with per-layer (gradient) clipping partitions parameters into $E$ disjoint layers $\zeta = (\zeta^{(1)}, \dots, \zeta^{(E)})$, and each per-example block-gradient is clipped using a block-specific bound $C_{e}$:
\[
\tilde{g}^{(e)}_{i} = g^{(e)}_{i}\cdot \min\left\{1, \tfrac{C_{e}}{\|g^{(e)}_{i}\|_{2}}\right\}, \; \hat{g}^{(e)} = \frac{1}{b}\sum_{i\in B}\tilde{g}^{(e)}_{i} +\mathcal{N}\left(0, \tfrac{\sigma^2C^2_{e}}{b^2}I\right),
\]
and then $\hat{g}=(\hat{g}^{(1)}, \dots, \hat{g}^{(E)})$ is used in the update.

\subsection{Layer-wise Lipschitz Assumption}

We view $\zeta = (\zeta^{(1)}, \dots, \zeta^{(E)})$ as the concatenation of layer parameters, which induces the orthogonal direct-sum decomposition $\mathbb{R}^{d} = \bigoplus^{E}_{e=1} \mathbb{R}^{d_{e}}$.
Let $\{\mathsf{P}_{e}\}^{E}_{e=1}$ denote the canonical orthogonal projector onto the $e$-th block.
Consider a generic (possibly preconditioned) SGD-type update written in the form
\[
f_{t}(x,y_{<t};R_t)\;=\;\xi_{t-1}\;-\;A_{t}\,\bar{g}_{t}(x;R_{t}),
\]
where $\bar{g}_{t}(\cdot)$ is the (clipped) minibatch gradient estimator, built from per-layer clipped per-example gradients, and $A_{t}$ is a deterministic linear map.

Assumption \ref{assp:layerwise_lip} is a layer-wise analogue of Assumption \ref{assp:slicewise_Lipschitz}.

\begin{assumption}[Layer-wise Lipschitz Updates]\label{assp:layerwise_lip}
For each iteration $t$ and each block $e\in\{1,\dots,E\}$, there exists a finite constant $L_{t,e}\geq 0$ such that
\[
\|\mathsf{P}_{e} A^{\top}_{t} u\|_{2} \leq L_{t,e} \|\mathsf{P}_{e}u\|_{2},
\qquad \forall u\in\mathcal{U}.
\]
\end{assumption}


For vanilla SGD with step size $\eta_{t}$, we have $A_{t} = \eta_{t}I$, so the Lipschitz constant is $L_{t,e} = \eta_{t}$ for all $e=1,\dots, E$.
For SGD with per-layer clipping with $\eta_{t,e}$ for each $e$, we have $A_t=\mathsf{blkdiag}(\eta_{t,1}I,\dots,\eta_{t,E}I)$, so the Lipschitz constant for each block $e$ is $L_{t,e} = \eta_{t,e}$.

\subsection{HUC and sa-HUC under per-layer clipping}
Fix time $t$, a slice profile $U=\{u_i\}_{i=1}^m$, and let $B_t$ be the minibatch size and $K_t$ the discrepancy
cap used in our main HUC construction.

A convenient way to express the per-slice cap in the per-layer clipped setting is via the same blockwise
template as Corollary~I.1 (linear, blockwise-clipped case). Specialized to "layers as blocks," this yields
the slice-dependent form
\[
h^{\text{layer}}_{t,i}
\;=\;
\left(\tfrac{2K_t}{B_t}\right)^2
\left(
\sum_{\ell=1}^{E} C_\ell \,\bigl\|P_\ell A_t^\top u_i\bigr\|_2
\right)^2,
\qquad i\in[m].
\]
Using the layer-wise Lipschitz bound $\|P_\ell A_t^\top u_i\|_2\le L_{t,\ell}\|P_\ell u_i\|_2$, we obtain the
fully explicit and easy-to-evaluate envelope
\[
h^{\text{layer}}_{t,i}
\;\leq\;
\left(\tfrac{2K_t}{B_t}\right)^2
\left(
\sum_{\ell=1}^{E} C_\ell \,L_{t,\ell}\,\|P_\ell u_i\|_2
\right)^2.
\]
The corresponding subsampling-aware case (sa-HUC), which underpins sa-SRPP-SGD in Section~5.2, is obtained by
replacing the worst-case $K_t^2$ factor by its mean-square counterpart:
\[
\begin{aligned}
h^{\text{sa,layer}}_{t,i}
&\;=\;
\left(\tfrac{2}{B_t}\right)^2
\overline{K}_t^2\,
\left(
\sum_{\ell=1}^E C_\ell \,\bigl\|\mathsf{P}_\ell A_t^\top u_i\bigr\|_2
\right)^2\\
&\;\leq\;
\left(\tfrac{2}{B_t}\right)^2
\overline{K}_t^2\,
\left(
\sum_{\ell=1}^E C_\ell \,L_{t,\ell}\,\|\mathsf{P}_\ell u_i\|_2
\right)^2.
\end{aligned}
\]

Unlike the "global Lipschitz + global clipping" special case (where the cap can collapse to a slice-independent
expression), per-layer clipping makes $\|P_\ell u_i\|_2$ appear explicitly, so Ave- vs.\ Joint-type aggregation
across slices can differ in SRPP/sa-SRPP-SGD, even though the clipping itself is performed on the true gradients
and is not directional.

\subsection{Numerical Example}
\label{app:perlayer:num_example}

\begin{figure}[t] %
  \centering
  \begin{subfigure}{\columnwidth}
    \includegraphics[width=\linewidth]{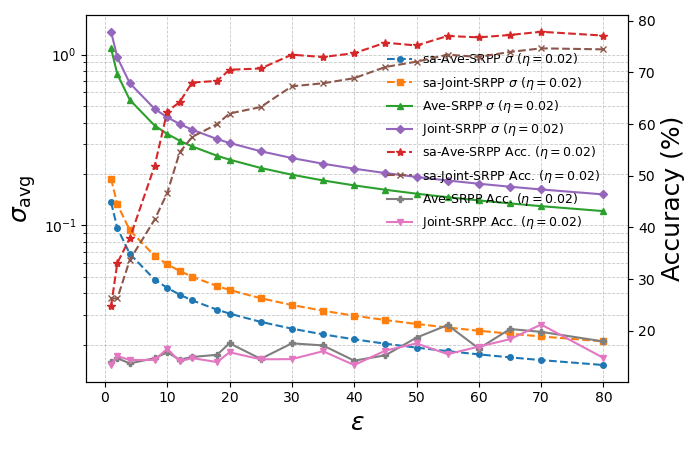}
  \end{subfigure}
  \caption{Per-block ("per-layer") clipping experiment under the same setup as the main-body SGD study (\(\eta=0.02\)).
Left axis: the calibrated reference noise scale \(\sigma_{\mathrm{avg}}\) (at \(C_{\mathrm{base}}\)), where per-block noise obeys
\(\sigma_\ell/C_\ell=\sigma_{\mathrm{avg}}/C_{\mathrm{base}}\). Right axis: test accuracy.}
  \label{Fig:example_SWD}
\end{figure}

We repeat the main-body SGD privatization experiment under the same Pufferfish scenario, training protocol, and subsampling rate
($\eta=0.02$), but replace global clipping by \textit{parameter-tensor blockwise} clipping ("per-layer").
Concretely, we treat each trainable parameter tensor in \texttt{model.named\_parameters()} as one block, yielding $L=23$ blocks
for \texttt{gep\_resnet22} (convolution weights plus the final linear weight and bias).
We use non-uniform clipping thresholds with the \(\sqrt{\mathrm{dim}}\) rule:
\[
C_\ell \;=\; C_{\mathrm{base}}\sqrt{\frac{d_\ell}{\sum_j d_j}},
\quad C_{\mathrm{base}}=5.0,
\]
where $d_\ell$ is the number of parameters in block $\ell$.

For each target SRPP/sa-SRPP budget $\varepsilon$, we first compute the aggregate accountant value $H_{\mathrm{total}}$
(Ave or Joint) and calibrate a \textit{single} reference noise scale
\[
\sigma_{\mathrm{avg}}^{2}(\varepsilon) \;=\; \frac{\alpha}{2\varepsilon}\,H_{\mathrm{total}}.
\]
This $\sigma_{\mathrm{avg}}$ is \textit{not} an average over per-block standard deviations; it is the standard deviation of the
Gaussian noise added to the \textit{batch-averaged clipped gradient}, expressed at the reference clipping scale $C_{\mathrm{base}}$.
Per-block noise is then set by
\[
\sigma_\ell \;=\; \sigma_{\mathrm{avg}}\cdot\frac{C_\ell}{C_{\mathrm{base}}},
\quad\text{so that}\quad
\frac{\sigma_\ell}{C_\ell}=\frac{\sigma_{\mathrm{avg}}}{C_{\mathrm{base}}}\ \ \text{for all }\ell.
\]

Figure~\ref{Fig:example_SWD} reports $\sigma_{\mathrm{avg}}$ and test accuracy versus $\varepsilon$.
Subsampling-aware accounting (sa-Ave/sa-Joint) certifies the same target budget with smaller $\sigma_{\mathrm{avg}}$ than the
corresponding worst-case calibration (Ave/Joint), translating into improved accuracy at matched $\varepsilon$.

In addition, Joint-type accounting is more conservative than Ave-type accounting: because it controls the \textit{joint} leakage across slices
(equivalently, a log-moment / log-sum-exp aggregation over directions), it yields a larger aggregate $H_{\mathrm{total}}$ and therefore
requires a larger calibrated $\sigma_{\mathrm{avg}}$ at the same target budget $\varepsilon$.

\section{Relationship to Group DP-SGD}
\label{app:group-dp-vs-srpp}

In this section, we characterize the relationship and the difference between our SRPP/sa-SRPP and the group DP-SGD. 
Please note that the group DP here is different from the \textit{worst-case group privacy} described in App. \ref{app:worst_case_group_privacy}.

The dependence on the discrepancy cap $K_t$ (worst-case) and its mean-square analogue $\mathbb{E}_{\eta}[K_t^2]$ in our SRPP/sa-SRPP privatization of SGD naturally raises the question: \textit{can group DP-SGD, with an appropriate group size, serve as a reliable proxy for SRPP-SGD when the protected information is a specific Pufferfish secret?}
In this appendix, we clarify the relationship between SRPP-/sa-SRPP-SGD and group DP-SGD, and explain why the latter cannot generally substitute for the former when one cares about secret-level Pufferfish guarantees.

Our SRPP- and sa-SRPP-SGD mechanisms are calibrated directly from the
HUC and sa-HUC envelopes developed in Sections \ref{sec:HUC} and \ref{sec:mean_square_HUC}, respectively.
For a
fixed R\'enyi order $\alpha>1$ and a target sliced privacy budget
$\epsilon$, the corresponding theorems give an upper bound on the
per-run SRPP or sa-SRPP cost in terms of the HUC/sa-HUC sequence
$\{h_t\}$ and the per-step noise covariances $\{\Sigma_t\}$.  In our
experiments we specialize to isotropic Gaussian noise
$\Sigma_t = \sigma^2 I$ and, for each desired $\epsilon$, solve the
resulting bound for a noise scale $\sigma(\epsilon)$ that guarantees the
prescribed SRPP or sa-SRPP budget.  
This calibration is explicitly secret-aware: $H_{\mathrm{total}}$ depends on the Pufferfish scenario $(\mathcal{S},\mathcal{Q},\Theta)$ through the discrepancy caps $K_t$ and mean-square discrepancy caps $\mathbb{E}_\eta[K_t^2]$, which, for each iteration $t$, bound the $\ell_2$–distance between the clipped, averaged gradient updates under any admissible secret pair $(s_i,s_j)\in\mathcal{Q}$ and any coupling of their corresponding data-generating distributions.

In contrast, standard group DP-SGD is calibrated with respect to a
\textit{dataset-level} adjacency relation
\[
  X \sim_k X'
  \;\;\Longleftrightarrow\;\;
  d_H(X,X') \le k,
\]
where $d_H$ is Hamming distance and $k\in\mathbb{N}$ is a chosen group
size.  
A mechanism $M$ is $(\epsilon_g,\delta_g)$-\textit{group-DP} for
group size $k$ if
\begin{equation}
  \Pr[\mathcal{M}(X)\in E]\leq
  e^{\varepsilon_g} \Pr[\mathcal{M}(X')\in E] + \delta_g,
  \;
  \forall E\subseteq\mathcal{Y},\ \forall X\sim_k X'.
  \label{eq:group-dp}
\end{equation}
In DP-SGD, one typically fixes $k$ and uses a moments accountant to
obtain $(\varepsilon_g,\delta_g)$ as a function of $k$, the sampling
probability $q = B_{t}/N$, the number of steps $T$, and the per-step noise
scale $\sigma_{\mathrm{DP}}$.

A natural question is whether group DP-SGD can be used as a proxy for
SRPP-/sa-SRPP-SGD by choosing a group size $k$ that reflects the
secret-induced discrepancy at iteration $t$.
However, $k$ in group DP is a \textit{dataset-level} Hamming distance (an integer),
whereas our sa-HUC uses a \textit{minibatch-level} mean-square discrepancy cap
$\overline K_t^2$ satisfying $\mathbb{E}_{\eta_t}[\Delta_t^2]\le \overline K_t^2$.
A dimensionally compatible proxy would be $k \approx \sqrt{\overline K_t^2}$
(or more conservatively $k\approx \lceil \sqrt{\overline K_t^2}\rceil$),
but even with such a choice, group DP-SGD does not generally yield a certified
SRPP/sa-SRPP guarantee for an arbitrary Pufferfish scenario.

In this section, we clarify two points:

\begin{itemize}
  \item There is no general result that, given a group-DP guarantee
  \eqref{eq:group-dp} for some group size $k$, produces an
  $(\alpha,\varepsilon)$-SRPP or sa-SRPP guarantee for an arbitrary
  Pufferfish scenario $(\mathcal{S},\mathcal{Q},\Theta)$.
  In particular, even if one chooses a \textit{proxy} group size
  $k \approx \lceil \sqrt{\overline{K}_t^2}\rceil$ (or uses a tail cap
  $k=K_t(\delta_t)$) motivated by our HUC/sa-HUC analysis, there is no
  closed-form or universal mapping from $(\varepsilon_g,\delta_g,k)$ to
  an SRPP/sa-SRPP budget $\varepsilon$.

  \item Even when we fix a specific Pufferfish scenario
  $(\mathcal{S},\mathcal{Q},\Theta)$ and choose a proxy group size
  $k \approx \lceil \sqrt{\overline{K}_t^2}\rceil$ (or $k=K_t(\delta_t)$),
  the resulting $(\varepsilon_g,\delta_g)$ group-DP-SGD guarantee is in
  general not a certified conservative bound for our $(\alpha,\varepsilon)$
  SRPP/sa-SRPP guarantees, and there is no theorem ensuring that this
  group-DP guarantee upper-bounds the leakage measured in the SRPP/sa-SRPP sense.
\end{itemize}

\subsection{No Universal Mapping}
\label{app:no-dp-to-srpp-map}

We first formalize the absence of a universal conversion from
group-DP guarantees to SRPP/sa-SRPP guarantees.

\begin{definition}[Hypothetical DP$\to$SRPP conversion]
  \label{def:dp-to-srpp-map}
  A \textit{DP-to-SRPP conversion map} is a function
  \[
    F: (0,\infty)\times[0,1]\times\mathbb{N} \to (0,\infty]
  \]
  that, for given group-DP parameters $(\epsilon_g,\delta_g)$ and
  group size $k$, assigns an SRPP budget $F(\epsilon_g,\delta_g,k)$
  with the following property:
  for every mechanism $\mathcal{M}$ and every Pufferfish scenario
  $(\mathcal{S},\mathcal{Q},\Theta)$, if $\mathcal{M}$ is
  $(\epsilon_g,\delta_g)$-group-DP with group size $k$, then for all
  $\theta\in\Theta$ and all $(s_i,s_j)\in\mathcal{Q}$ the SRPP/sa-SRPP
  divergence between the distributions of $\mathcal{M}$ under $(S=s_i,\theta)$ and
  $(S=s_j,\theta)$ is bounded by $F(\epsilon_g,\delta_g,k)$.
\end{definition}

Our first result is that such a universal conversion map cannot exist.

\begin{proposition}
  \label{prop:no-universal-dp-to-srpp}
  There is no function $F$ satisfying
  Definition \ref{def:dp-to-srpp-map}, even if we restrict to pure
  group-DP ($\delta_g=0$) and to finite secret sets $\mathcal{S}$.
\end{proposition}

\begin{proof}[Proof of Proposition \ref{prop:no-universal-dp-to-srpp}]
  Fix any candidate function $F$. We construct a counterexample.

  Let the dataset space be $\mathcal{X} = \{X_0,X_1\}$, and choose
  $k\in\mathbb{N}$ so that $d_H(X_0,X_1) > k$, where $d_H$ is the Hamming
  distance. 
  Under the group adjacency relation restricted to $\mathcal{X}=\{X_0,X_1\}$,
  there are no distinct adjacent pairs:
  the only pairs with $d_H(X,X')\leq k$ are $(X_0,X_0)$ and $(X_1,X_1)$.
  Hence, for any mechanism $\mathcal{M}$ and any $(\epsilon_g,\delta_g)$, the
  group-DP inequality \eqref{eq:group-dp} is trivially satisfied, so
  $\mathcal{M}$ is $(\epsilon_g,\delta_g)$-group-DP for all choices of
  $(\epsilon_g,\delta_g)$ and $k$.

  Now define a Pufferfish scenario with two secrets
  $\mathcal{S} = \{0,1\}$, a single prior $\Theta = \{\theta\}$, and a
  deterministic data-generating process
  \[
    P_\theta(X = X_0 \mid S = 0) = 1,
    \quad
    P_\theta(X = X_1 \mid S = 1) = 1.
  \]
  In other words, the secret value deterministically selects which of
  the two datasets is realized.

  Consider the mechanism $\mathcal{M}$ that simply reveals the secret:
  \[
    \mathcal{M}(X_0) = 0,
    \quad
    \mathcal{M}(X_1) = 1.
  \]
  As argued above, $\mathcal{M}$ is $(\epsilon_g, \delta_g)$-group-DP for
  \textit{all} $(\epsilon_g,\delta_g)$ and all $k$, because there are
  no nontrivial adjacent dataset pairs.

  However, the SRPP (and RPP/PP) divergence between the output
  distributions under $S=0$ and $S=1$ is infinite: for any measurable
  set $E$ containing $0$ but not $1$ we have
  \[
    \Pr[\mathcal{M}(X)\in E \mid S=0] = 1,
    \qquad
    \Pr[\mathcal{M}(X)\in E \mid S=1] = 0.
  \]
  In particular, the R\'enyi divergence of any order $\alpha>1$ between
  the distributions of $\mathcal{M}(X)$ given $S=0$ and $S=1$ is $+\infty$ (i.e., the two output laws are mutually singular), so no finite
  SRPP budget can hold.

  This contradicts the existence of a finite-valued
  $F(\epsilon_g,\delta_g,k)$ that would upper-bound SRPP leakage
  whenever $\mathcal{M}$ is $(\epsilon_g,\delta_g)$-group-DP with group size
  $k$. Since $(\epsilon_g,\delta_g,k)$ were arbitrary, no such
  universal conversion function $F$ can exist.
\end{proof}

Proposition \ref{prop:no-universal-dp-to-srpp} remains true if, instead
of allowing arbitrary mechanisms, we restrict attention to mechanisms
obtained by running DP-SGD and calibrating noise for a fixed group size
$k$ (including proxy choices such as $k = K_t(\delta_t)$ or
$k \approx \lceil \sqrt{\overline{K}_t^2}\rceil$ motivated by our HUC/sa-HUC analysis).
We can always construct a Pufferfish scenario whose secrets are encoded
in dataset differences at Hamming distance strictly greater than $k$ and
hence outside the scope of the group-DP guarantee.
Therefore, even when the group-DP parameters
$(\epsilon_g,\delta_g,k)$ are known from a moments-accountant
analysis of DP-SGD, there is no closed-form or universal transformation
that converts them into the SRPP/sa-SRPP budget used in our
HUC/sa-HUC–based SRPP-SGD calibration.

\subsection{No Safe Bounds}
\label{app:group-dp-not-conservative}

We now specialize to the prevalence-based Pufferfish scenarios used in our SGD privatization experiments,
where the secret changes only the label configuration on a fixed feature set (e.g., the \texttt{cat} label in CIFAR-10).
Concretely, we consider two secrets $s_0,s_1$ corresponding to two datasets $X^{(0)},X^{(1)}\in\bar{\mathcal{X}}^N$
with different prevalences $p_{\mathrm{low}}\neq p_{\mathrm{high}}$ of the protected class. Let
\[
  \tilde{\Delta} := d_H\bigl(X^{(0)},X^{(1)}\bigr)
\]
be the number of records whose labels differ between the two secrets.
In our HUC/sa-HUC analysis, $\tilde{\Delta}$ influences the secret-induced update discrepancy through the
(minibatch-level) discrepancy cap sequence, e.g., via $K_t$ (or $\mathbb{E}_{\eta_t}[K_t^2]$), which determines the SRPP/sa-SRPP
noise calibration. A natural question is whether one can instead treat such quantities as an \textit{effective group size} and
calibrate standard group DP-SGD with some $k$ derived from $K_t$ or $\mathbb{E}_{\eta_t}[K_t^2]$.
The next proposition shows that, even in prevalence-based scenarios, a group-DP guarantee with any fixed $k$
does not yield a general \textit{conservative} (i.e., universally safe) bound on secret-level SRPP/sa-SRPP leakage.

\begin{proposition}[Group-DP does not upper-bound secret leakage in general]
\label{prop:group-dp-not-conservative}
Fix any group size $k\in\mathbb{N}$ and any group-DP parameters $(\varepsilon_g,\delta_g)$.
There exist
\begin{itemize}
  \item[(i)] a prevalence-based Pufferfish scenario $(\mathcal{S},\mathcal{Q},\Theta)$ with two secrets $s_0,s_1\in\mathcal{S}$
  whose realized datasets $X^{(0)},X^{(1)}$ satisfy $d_H(X^{(0)},X^{(1)})>k$, and
  \item[(ii)] a mechanism $\mathcal{M}$ that is $(\varepsilon_g,\delta_g)$-group-DP with group size $k$
\end{itemize}
such that, for every $\alpha>1$, the secret-level R\'enyi divergence (hence SRPP/sa-SRPP leakage) between the output laws under
$(S=s_0,\theta)$ and $(S=s_1,\theta)$ can be made \textit{arbitrarily large}. In particular, there is no finite function of
$(\varepsilon_g,\delta_g,k)$ that universally upper-bounds this secret-level leakage.
\end{proposition}

\begin{proof}
Fix $k\in\mathbb{N}$ and $(\varepsilon_g,\delta_g)$.
Choose a dataset size $N>k$ and an integer $\tilde{\Delta}\in\{k+1,\dots,N\}$.

\paragraph{Prevalence-based Pufferfish scenario.}
Let the record domain be $\bar{\mathcal{X}}=\mathcal{A}\times\{0,1\}$, where $\mathcal{A}$ denotes non-sensitive features and the
second coordinate is the protected-class label. Let the dataset space be $\mathcal{X}:=\bar{\mathcal{X}}^N$.
Fix a feature vector collection $(a_1,\dots,a_N)\in\mathcal{A}^N$ and two label vectors
$(b^{(0)}_1,\dots,b^{(0)}_N)$ and $(b^{(1)}_1,\dots,b^{(1)}_N)$ such that:
(i) the empirical prevalences differ, i.e.\ $\frac{1}{N}\sum_{i=1}^N b^{(0)}_i=p_{\mathrm{low}}$ and
$\frac{1}{N}\sum_{i=1}^N b^{(1)}_i=p_{\mathrm{high}}$ with $p_{\mathrm{low}}\neq p_{\mathrm{high}}$, and
(ii) exactly $\tilde{\Delta}$ labels differ, i.e.\ $\#\{i: b^{(0)}_i\neq b^{(1)}_i\}=\tilde{\Delta}$.
Define
\[
  X^{(0)} := \bigl((a_i,b^{(0)}_i)\bigr)_{i=1}^N,
  \qquad
  X^{(1)} := \bigl((a_i,b^{(1)}_i)\bigr)_{i=1}^N,
\]
so that $d_H(X^{(0)},X^{(1)})=\tilde{\Delta}>k$.

Define the Pufferfish scenario by
\[
  \mathcal S=\{s_0,s_1\},\quad
  \Theta=\{\theta\},\quad
  \mathcal Q=\{(s_0,s_1),(s_1,s_0)\},
\]
and let the conditional data law be deterministic:
\[
  P_\theta(X=X^{(0)}\mid S=s_0)=1,\qquad
  P_\theta(X=X^{(1)}\mid S=s_1)=1.
\]

\paragraph{Group adjacency.}
Let group adjacency be defined by Hamming distance at most $k$:
\[
  X\sim_k X' \Longleftrightarrow d_H(X,X')\leq k.
\]
Since $d_H(X^{(0)},X^{(1)})=\tilde{\Delta}>k$, the pair $(X^{(0)},X^{(1)})$ is not adjacent, and the group-DP constraint
never compares $\mathcal{M}(X^{(0)})$ to $\mathcal{M}(X^{(1)})$.

\paragraph{Mechanism with group-DP but large secret leakage.}
Fix any statistic $T:\mathcal{X}\to\mathbb{R}$ with bounded group sensitivity at size $k$,
\[
  \Delta_k(T):=\sup_{X\sim_k X'}|T(X)-T(X')|<\infty,
\]
and choose $T$ so that $|T(X^{(0)})-T(X^{(1)})|=L$ for an arbitrary prescribed $L>0$
(e.g., define $T(X)=L\cdot\mathbf{1}\{X=X^{(0)}\}$; then $\Delta_k(T)\leq L$ because any $X'\sim_k X^{(0)}$ must satisfy $X'\neq X^{(1)}$).
Let
\[
  \mathcal{M}(X):=T(X)+Z,\qquad Z\sim\mathcal{N}(0,\sigma^2),
\]
with $\sigma$ calibrated so that $\mathcal{M}$ is $(\varepsilon_g,\delta_g)$-group-DP for group size $k$
under the adjacency $\sim_k$ (this is the standard Gaussian mechanism calibration based on $\Delta_k(T)$).

Under the above Pufferfish scenario, $\mathcal{M}(X)\mid(S=s_0)$ is $\mathcal{N}(T(X^{(0)}),\sigma^2)$ and
$\mathcal{M}(X)\mid(S=s_1)$ is $\mathcal{N}(T(X^{(1)}),\sigma^2)$.
For any $\alpha>1$, the R\'enyi divergence between these two output laws is
\[
 \begin{aligned}
      &\mathtt{D}_\alpha\!\bigl(\mathcal{N}(T(X^{(0)}),\sigma^2)\,\|\,\mathcal{N}(T(X^{(1)}),\sigma^2)\bigr)\\&
  =\frac{\alpha}{2(\alpha-1)}\cdot \frac{|T(X^{(0)})-T(X^{(1)})|^2}{\sigma^2}
  =\frac{\alpha}{2(\alpha-1)}\cdot \frac{L^2}{\sigma^2}.
 \end{aligned}
\]
Since $L>0$ can be chosen arbitrarily large while the group-DP parameters $(\varepsilon_g,\delta_g,k)$ are fixed (by construction),
this R\'enyi divergence (hence SRPP/sa-SRPP leakage) can be made arbitrarily large.
Therefore no finite function of $(\varepsilon_g,\delta_g,k)$ can universally upper-bound secret-level leakage across such
prevalence-based scenarios.
\end{proof}

Proposition~\ref{prop:group-dp-not-conservative} shows that even for prevalence-based secrets as in our CIFAR-10 experiments,
calibrating DP-SGD using a group size chosen from HUC/sa-HUC-derived quantities (e.g., $k=\lceil\sqrt{\overline K_t^2}\rceil$ or
$k=K_{\mathrm{cap}}$) does \textit{not} yield a certified conservative bound in the SRPP/sa-SRPP sense.
Group-DP is indexed by Hamming-adjacent dataset pairs, whereas SRPP/sa-SRPP quantify divergence between the secret-conditioned output
laws in the specified Pufferfish scenario. Our HUC/sa-HUC-based SRPP-SGD calibration is therefore not a reparameterization of group DP-SGD,
but a distinct, secret-aware analysis.

\section{Minimality and Attainability of HUC}\label{app:attainable_HUC}

In this section, we characterize the \textit{minimality} and \textit{attainability} of HUC defined by Definition \ref{def:HUC} in Sec. \ref{sec:HUC}.

\begin{proposition}\label{prop:HUC_tight}
Fix $t$, a slicing set $\mathcal{U}=\{u_i\}_{i=1}^m\subset\mathbb{S}^{d-1}$, a history $y_{<t}\in\mathcal{Y}_{<t}$, and an admissible coupling $\gamma$.
Define the \textup{shift}
\begin{equation}\label{eq:Delta_shift}
    \Delta_t(X,X';R_t) := f_t(X,y_{<t};R_t)-f_t(X',y_{<t};R_t).
\end{equation}
For each direction $u_{i}$, define
\[
\phi_t(u_i)
\;:=\;
\operatorname*{ess\,sup}_{(X,X',R_t)\sim \gamma\times\mathbb P_{\eta,\rho}}
\big|\langle \Delta_t(X,X';R_t),u_i\rangle\big| \textup{ and }
h^\star_{t,i}:=\phi_t(u_i)^2.
\]
In addition, let $\mathsf{P}:=\gamma\times\mathbb{P}_{\eta,\rho}$ and $\mu_t:= (\Delta_t)_{\#} \mathsf{P}$ be the pushforward measure of $\Delta_t$, and define
\[
\mathtt{ELS}_t := \mathrm{supp}(\mu_t) = \big\{z\in\mathbb{R}^d: \mu_t(\mathcal{B}_{\bar{\varepsilon}}(z))>0\ \text{for all }\bar{\varepsilon}>0\big\},
\]
where $\mathcal{B}_{\bar{\varepsilon}}(z)$ is a ball of radius $\bar{\varepsilon}$ centered at $z$.

Then, the following holds.
\begin{itemize}
\item[(i)] \textup{\textbf{Minimality.}} Any HUC $h_t=(h_{t,1},\dots,h_{t,m})$ for $(y_{<t},\gamma)$ satisfies $h_{t,i}\geq h^{\star}_{t,i}$ for all $i\in[m]$.
\item[(ii)] \textup{\textbf{Attainability.}} If $\mathtt{ELS}_t$ is compact, then there exists $\Delta_i^{\star}\in\mathtt{ELS}_t$ such that $|\langle \Delta_i^{\star},u_i\rangle|=\sup_{\Delta\in \mathtt{ELS}_{t}} |\langle \Delta, u_{i} \rangle |$. 
Moreover, by definition of $\phi_t(u_i)$ as an essential supremum, for very $\varepsilon >0$, $\mathsf{P}\!\left(\,|\langle \Delta_t(X,X';R_t),u_i\rangle|\geq \phi_t(u_i)-\varepsilon\,\right)>0$.
\end{itemize}
\end{proposition}

However, the HUC constructed in Proposition~\ref{prop:exist_HUC} is generally not tight; it is a computable conservative bound.
Proposition~\ref{prop:HUC_tight} characterizes the \textit{minimal} HUC and specifies when it is attained.
In particular, any valid HUC $h_t$ must satisfy $h_{t,i}\geq h^{\star}_{t,i}$ for all $i\in[m]$, and the minimal cap $h_t^{\star}$ is attained whenever the essential shift set $\mathtt{ELS}_{t}$ is compact.

In our setting, compactness of $\mathtt{ELS}_t$ follows from three ingredients:
(i) per-example clipping, which uniformly bounds each individual gradient contribution;
(ii) a finite discrepancy cap together with a deterministic (or bounded) batch size; and
(iii) an $L_t$-Lipschitz post-update map.
These yield a uniform almost-sure bound
$\|\Delta_t\|_2\le C_t$ for some finite $C_t$, and hence
$\mathtt{ELS}_t\subset \overline{B(0,C_t)}$, which is closed and bounded in $\mathbb R^d$ and therefore compact.

The following corollary follows Proposition~\ref{prop:HUC_tight}.

\begin{corollary}\label{corollary:vanilla_SGD_HUC}
    Consider the linear, blockwise-clipped case where
    \[
        f_t(x,y_{<t};R_t)
        = \xi_{t-1} - A_t\,\bar g_t(x;R_t),
    \]
    with $A_t$ independent of $((X,X'),R_t)$, and per-block clipping thresholds
    $\{C_{b} \}_{b=1}^{B_{\mathrm{blk}}}$ over an orthogonal block decomposition
    with projectors $\{\mathsf{P}_{b} \}_{b=1}^{B_{\mathrm{blk}}}$.
    Assume that, for each block $b$, the blockwise gradient differences can attain
    the clipping radius in directions aligned with the vectors
    $\mathsf{P}_b A_t^\top u_i$ (for the slices of interest).
    Then the minimal per-slice HUC is
    \begin{equation}\label{eq:vanilla_SGD_HUC}
        h^{\star}_{t,i}
        =
        \Bigl(\frac{2K_t}{B_t}\Bigr)^{ 2}
        \biggl(
            \sum_{b=1}^{B_{\mathrm{blk}}} C_b\,\bigl\|\mathsf{P}_b A_t^\top u_i\bigr\|_2
        \biggr)^2,
        \qquad i\in[m].
    \end{equation}
\end{corollary}

\subsection{Proof of Proposition~\ref{prop:HUC_tight}}

Let $\mathsf{P}:=\gamma\times\mathbb{P}_{\eta,\rho}$ and write
\[
Z_i(X,X',R_t)\;:=\;\big|\langle \Delta_t(X,X';R_t),u_i\rangle\big|,
\qquad i\in[m].
\]
Recall that $\phi_t(u_i)=\operatorname*{ess\,sup}_{\mathsf{P}} Z_i$ and $h^\star_{t,i}=\phi_t(u_i)^2$.

We use the following standard property: if $Z\geq 0$ is measurable and
$Z\leq c$ holds $\mathsf{P}$-a.s., then $\operatorname*{ess\,sup}_{\mathsf{P}} Z\leq c$.
Moreover, by definition of essential supremum, for every $\varepsilon>0$,
\begin{equation}\label{eq:esssup_eps_posprob}
\mathsf{P}\big(Z\geq \operatorname*{ess\,sup}_{\mathsf{P}} Z - \varepsilon\big) > 0.
\end{equation}

\paragraph{(i) Minimality.}
Let $h_t=(h_{t,1},\dots,h_{t,m})$ be any HUC for $(y_{<t},\gamma)$.
By Definition~\ref{def:HUC}, for each $i\in[m]$ we have
\[
Z_i(X,X',R_t)
=\big|\langle \Delta_t(X,X';R_t),u_i\rangle\big|
\leq \sqrt{h_{t,i}}
\quad \text{$\mathsf P$-a.s.}
\]
Applying the essential-supremum property gives
\[
\phi_t(u_i)=\operatorname*{ess\,sup}_{\mathsf{P}} Z_i \leq \sqrt{h_{t,i}}.
\]
Squaring yields $h^\star_{t,i}=\phi_t(u_i)^2\leq h_{t,i}$ for all $i\in[m]$, proving minimality.

\paragraph{(ii) Attainability.}
Define $\mu_t := (\Delta_t)_\# \mathsf{P}$ and $\mathtt{ELS}_t:=\mathrm{supp}(\mu_t)$.
Assume $\mathtt{ELS}_t$ is compact. The map
$\Delta\mapsto |\langle \Delta,u_i\rangle|$ is continuous, hence by the Weierstrass
Extreme Value Theorem there exists $\Delta_i^\star\in \mathtt{ELS}_t$ such that
\[
|\langle \Delta_i^\star,u_i\rangle|
=
\sup_{\Delta\in \mathtt{ELS}_t} |\langle \Delta,u_i\rangle|.
\]
This establishes attainment of the \textit{set supremum} over $\mathtt{ELS}_t$.

Finally, the \textit{essential} attainability statement follows directly from
\eqref{eq:esssup_eps_posprob} applied to $Z_i$: for every $\varepsilon>0$,
\[
\mathsf{P}\big(Z_i \geq \phi_t(u_i)-\varepsilon\big)>0.
\]
Equivalently, in pushforward form,
\[
\mu_t\Big(\big\{\Delta:\ |\langle \Delta,u_i\rangle|\geq \phi_t(u_i)-\varepsilon\big\}\Big)>0.
\]
\qed

\subsection{Proof of Corollary~\ref{corollary:vanilla_SGD_HUC}}

Fix iteration $t$ and a slice $u_i\in\mathcal{U}$.  Let $\mathsf{P}:=\gamma\times \mathbb{P}_{\eta,\rho}$.
Write
\[
\begin{aligned}
    &f_t(x,y_{<t};R_t)=\xi_{t-1}-A_t\,\bar{g}_t(x;R_t),
\\&
\Delta_t(X,X';R_t)=-A_t\big(\bar{g}_t(X;R_t)-\bar{g}_t(X';R_t)\big).
\end{aligned}
\]
Define
\[
v := \bar{g}_t(X;R_t)-\bar{g}_t(X';R_t),\quad w:=A_t^\top u_i.
\]
Then $|\langle \Delta_t,u_i\rangle| = |\langle v,w\rangle|$.

\paragraph{Step 1: a.s.\ feasible set for $v$.}
Under blockwise clipping with thresholds $\{C_b\}_{b=1}^{B_{\mathrm{blk}}}$ and a discrepancy cap $K_t$,
at most $K_t$ per-example contributions in the minibatch can differ between $X$ and $X'$.
Since each per-example block-gradient is clipped to $\ell_2$-norm at most $C_b$ in block $b$,
the blockwise difference in the averaged clipped gradients satisfies, $\mathsf{P}$-a.s.,
\[
\|\mathsf{P}_b v\|_2
\;\leq\;
\frac{2C_b}{B_t}K_t
\;=:\;\alpha_b,
\quad b=1,\dots,B_{\mathrm{blk}}.
\]
Equivalently, $v\in\mathcal{V}$ $\mathsf{P}$-a.s., where
\[
\mathcal{V}
:=
\Big\{
v\in\mathbb{R}^d:\ \|\mathsf{P}_b v\|_2\leq \alpha_b,\ \ b=1,\dots,B_{\mathrm{blk}}
\Big\}.
\]

\paragraph{Step 2: upper bound on the minimal cap.}
By Proposition~\ref{prop:HUC_tight} (minimality with $\phi_t(u_i)=\operatorname*{ess\,sup}_{\mathsf{P}}|\langle \Delta_t,u_i\rangle|$),
\[
\phi_t(u_i)=\operatorname*{ess\,sup}_{\mathsf{P}}|\langle v,w\rangle|
\;\leq\;
\sup_{v\in\mathcal{V}}|\langle v,w\rangle|.
\]
Now decompose $v=\sum_b v_b$ and $w=\sum_b w_b$ with $v_b:=\mathsf{P}_b v$, $w_b:=\mathsf{P}_b w$.
Then, for any $v\in\mathcal{V}$,
\[
|\langle v,w\rangle|
=
\Big|\sum_{b=1}^{B_{\mathrm{blk}}}\langle v_b,w_b\rangle\Big|
\leq
\sum_{b=1}^{B_{\mathrm{blk}}} \|v_b\|_2\,\|w_b\|_2
\leq
\sum_{b=1}^{B_{\mathrm{blk}}} \alpha_b \|w_b\|_2.
\]
Hence
\[
\sup_{v\in\mathcal{V}}|\langle v,w\rangle|
=
\sum_{b=1}^{B_{\mathrm{blk}}} \alpha_b \|w_b\|_2
=
\frac{2K_t}{B_t}\sum_{b=1}^{B_{\mathrm{blk}}} C_b\,\|\mathsf{P}_b A_t^\top u_i\|_2.
\]
Therefore,
\begin{equation}\label{eq:phi_upper}
\phi_t(u_i)
\;\leq\;
\frac{2K_t}{B_t}\sum_{b=1}^{B_{\mathrm{blk}}} C_b\,\|\mathsf{P}_b A_t^\top u_i\|_2.
\end{equation}

\paragraph{Step 3: matching lower bound (attainability in essential-sup sense).}
Consider the optimizer $v^\star\in\mathcal{V}$ defined blockwise by
\[
v_b^\star
=
\begin{cases}
\alpha_b\,\dfrac{w_b}{\|w_b\|_2}, & w_b\neq 0,\\[6pt]
0, & w_b=0,
\end{cases}
\quad b=1,\dots,B_{\mathrm{blk}}.
\]
Then $v^\star\in\mathcal{V}$ and
\[
\langle v^\star,w\rangle
=
\sum_{b=1}^{B_{\mathrm{blk}}} \alpha_b\|w_b\|_2
=
\sup_{v\in\mathcal{V}}\langle v,w\rangle,
\]
so $\sup_{v\in\mathcal{V}}|\langle v,w\rangle|$ equals the right-hand side of \eqref{eq:phi_upper}.

By the corollary's assumption (“the blockwise gradient differences can attain the clipping radius
in directions aligned with $\mathsf{P}_bA_t^\top u_i$”), together with the requirement that the
corresponding discrepancy pattern of size $K_t$ occurs with nonzero probability under the sampling rule,
we have that for every $\varepsilon>0$,
\[
\mathsf P\Big(|\langle v,w\rangle|\ge |\langle v^\star,w\rangle|-\varepsilon\Big)>0.
\]
By the defining property of essential supremum, this implies
\[
\phi_t(u_i)
=
\operatorname*{ess\,sup}_{\mathsf{P}}|\langle v,w\rangle|
=
|\langle v^\star,w\rangle|
=
\frac{2K_t}{B_t}\sum_{b=1}^{B_{\mathrm{blk}}} C_b\,\|\mathsf{P}_b A_t^\top u_i\|_2.
\]
Squaring yields
\[
h^\star_{t,i}=\phi_t(u_i)^2
=
\Bigl(\frac{2K_t}{B_t}\Bigr)^{2}
\biggl(
\sum_{b=1}^{B_{\mathrm{blk}}} C_b\,\|\mathsf{P}_b A_t^\top u_i\|_2
\biggr)^2,
\]
which is exactly \eqref{eq:vanilla_SGD_HUC}.
\qed

\section{Minibatch Subsampling models: \textsf{WR} vs.\ \textsf{WOR}}\label{app:sampling_models}

Our HUC/sa-HUC analysis is compatible with standard minibatch subsampling variants. The proofs follow
the same skeleton. Only the sampling model for the round-$t$ randomness $R_t$ and the interpretation of the
index collection $\mathsf{I}_t$ change (set vs.\ sequence with possible repeats).
In this section, we characterize the comparison between sampling With Replacement (\textsf{WR}) and Without Replacement (\textsf{WOR}).

\paragraph{Convention for \textsf{WOR}.}
Unless stated otherwise, we adopt the following standard model.

\begin{assumption}[Minibatch Subsampling \textsf{WOR}]\label{assp:wor_sampling}
At each iteration $t$, the minibatch indices $\mathsf{I}_t\subseteq[n]$ are sampled uniformly without
replacement with $|\mathsf{I}_t|=B_{t}$. All algorithmic randomness at round $t$ (including subsampling and
any additional randomness) is included in $R_t\sim\mathbb{P}_{\eta,\rho}$.
\end{assumption}

Under Assumption~\ref{assp:wor_sampling}, $\mathsf{I}_t(r)$ is a \textit{set}. The discrepancy
in~\eqref{eq:discrepancy_K} is unambiguous:
\[
K_t(x,x';r)\ :=\ \sum_{j\in\mathsf I_t(r)}\mathbf{1}\{x_j\neq x'_j\}, 
\]
where $\mathsf{I}_t(r)$ is a set.
In particular, if $x$ and $x'$ differ in at most $\overline K$ coordinates, then
$K_t(x,x';r)\le \min\{\overline K,B_{t}\}$ for every draw $r$.

\begin{assumption}[Minibatch Subsampling \textsf{WR}]\label{assp:wr_sampling}
At each iteration $t$, the minibatch indices are formed by $B_{t}$ i.i.d.\ draws from $[n]$.
Equivalently, $\mathsf I_t(r)=(I_{t,1}(r),\ldots,I_{t,B_{t}}(r))$ is a \textit{length-$B$ sequence} in $[n]$
(allowing repetitions). All round-$t$ randomness is included in $R_t\sim\mathbb P_{\eta,\rho}$.
\end{assumption}

We keep the same discrepancy definition as in~\eqref{eq:discrepancy_K}, interpreted \textit{per draw}:
\[
K_t(x,x';r)\ :=\ \sum_{b=1}^{B_{t}}\mathbf{1}\{x_{I_{t,b}(r)}\neq x'_{I_{t,b}(r)}\},
\]
where $\mathsf{I}_t(r)$ is a sequence.
Thus, repeated indices contribute repeatedly; this matches the potential influence entering through the sampled minibatch.

\begin{proposition}[Worst-case discrepancy caps under \textsf{WR} vs.\ \textsf{WOR}]
\label{prop:wr_caps}
Fix any $x,x'\in\bar{\mathcal{X}}^n$.

\textit{\textsf{WR}.}
Under Assumption~\ref{assp:wr_sampling}, for any draw $r$ one has
$0\le K_t(x,x';r)\le B_t$, where
\[
K_t(x,x';r)
:=\sum_{b=1}^{B_t}\mathbf{1}\{x_{I_{t,b}(r)}\neq x'_{I_{t,b}(r)}\}.
\]
If $d_H(x,x')=k$, then $K_t(x,x';R_t)\sim \mathrm{Binomial}(B_t,k/n)$ and
$\mathbb{E}[K_t]=B_t(k/n)$.

Moreover, for any coupling $\gamma$,
\[
K_t(\gamma)
=
\operatorname*{ess\,sup}_{((X,X'),R_t)\sim \gamma\times\mathbb P_{\eta,\rho}}
K_t(X,X';R_t)
=
\begin{cases}
0, & \text{if } X=X'\ \gamma\text{-a.s.},\\
B_t, & \text{otherwise}.
\end{cases}
\]

\textit{\textsf{WOR}.}
Under Assumption~\ref{assp:wor_sampling}, $\mathsf I_t(r)$ is a set and
\[
K_t(x,x';r)
:=\sum_{j\in \mathsf I_t(r)}\mathbf{1}\{x_{j}\neq x'_{j}\}.
\]
If $d_H(x,x')\le \overline K$, then deterministically $K_t(x,x';r)\le \min\{\overline K,B_t\}$
for every draw $r$, and hence $K_t(\gamma)\le \min\{\overline K,B_t\}$ whenever
$d_H(X,X')\le \overline K$ $\gamma$-a.s.
\end{proposition}

\begin{proof}
\textit{(\textsf{WR} upper bound.)}
Under Assumption~\ref{assp:wr_sampling}, $\mathsf I_t(r)=(I_{t,1}(r),\ldots,I_{t,B_t}(r))$ is a length-$B_t$ sequence.
Hence
\[
K_t(x,x';r):=\sum_{b=1}^{B_t}\mathbf{1}\{x_{I_{t,b}(r)}\neq x'_{I_{t,b}(r)}\}
\]
is a sum of $B_t$ indicators, so $0\le K_t(x,x';r)\le B_t$.

\textit{(\textsf{WR} binomial law.)}
If $x$ and $x'$ differ in exactly $k$ coordinates, each draw hits a differing index with probability $k/n$,
independently across draws, so $K_t(x,x';R_t)\sim\mathrm{Binomial}(B_t,k/n)$.

\textit{(\textsf{WR} worst-case cap.)}
Fix a coupling $\gamma$ and consider
\[
K_t(\gamma)
=\operatorname*{ess\,sup}_{((X,X'),R_t)\sim \gamma\times \mathbb P_{\eta,\rho}} K_t(X,X';R_t).
\]
If $X=X'$ $\gamma$-a.s., then $K_t(X,X';R_t)=0$ a.s., so $K_t(\gamma)=0$.
Otherwise, $\Pr_\gamma(X\neq X')>0$. On the event $\{X\neq X'\}$, let
$D:=\{j\in[n]: X_j\neq X'_j\}$, so $D\neq\emptyset$.
Under \textsf{WR}, for any fixed $j\in D$, the event $\{I_{t,1}=\cdots=I_{t,B_t}=j\}$ has probability $(1/n)^{B_t}>0$,
and on this event $K_t(X,X';R_t)=B_t$.
Therefore $\Pr\big(K_t(X,X';R_t)=B_t\big)>0$ on $\{X\neq X'\}$, which implies $K_t(\gamma)=B_t$.

\textit{(\textsf{WOR} comparison.)}
Under Assumption~\ref{assp:wor_sampling}, $\mathsf I_t(r)$ is a set and
$K_t(x,x';r)=\sum_{j\in\mathsf I_t(r)}\mathbf{1}\{x_j\neq x'_j\}$.
If $x$ and $x'$ differ in at most $\overline K$ coordinates then deterministically
$K_t(x,x';r)\le \min\{\overline K,B_t\}$ for every draw $r$, and the same holds $\gamma$-a.s.\ when
$d_H(X,X')\le \overline K$ $\gamma$-a.s.
\end{proof}

\section{Utility of Subsampling-Aware SRPP Accounting}\label{app:utility_sa_srpp}

This appendix formalizes why incorporating minibatch randomness into the accounting can improve utility.
We compare two \textit{sufficient calibrations} for the \textit{same} SRPP guarantee for the \textit{same} SGD mechanism:
a worst-case (HUC) envelope bound versus a subsampling-aware (sa-HUC) envelope bound.

\subsection{A R\'enyi bound for shared-mixtures}\label{app:mixture_renyi}

Let $\alpha>1$. Let $(\Omega,\mathcal{F},\nu)$ be a probability space and $\eta\sim \nu$.
For each $\eta\in\Omega$, let $P_\eta$ and $Q_\eta$ be distributions on a common measurable space $(\mathcal{Y}, \mathcal{A})$,
with $P_\eta\ll Q_\eta$ for $\nu$-a.e.\ $\eta$.
Define the shared-mixtures $P:=\int P_\eta\,d\nu(\eta)$ and $Q:=\int Q_\eta\,d\nu(\eta)$.

\begin{lemma}[R\'enyi divergence of shared-mixtures]\label{lem:renyi_mixture_shared}
Assume $P_\eta\ll Q_\eta$ for $\nu$-a.e.\ $\eta$.
Assume further that $\eta\mapsto D_\alpha(P_\eta\|Q_\eta)$ is measurable and that
$\mathbb{E}_{\eta\sim\nu}\!\left[\exp\!\big((\alpha-1)D_\alpha(P_\eta\|Q_\eta)\big)\right]<\infty$
(otherwise the right-hand side is $+\infty$ and the inequality below is trivial).
Then
\[
\mathtt{D}_\alpha(P\|Q)
\;\leq\;
\frac{1}{\alpha-1}\log \mathbb{E}_{\eta\sim \nu}\!\left[\exp\!\big((\alpha-1)\mathtt{D}_\alpha(P_\eta\|Q_\eta)\big)\right].
\]
Moreover, the right-hand side is strictly smaller than $\operatorname*{ess\,sup}_{\eta} \mathtt{D}_\alpha(P_\eta\|Q_\eta)$
whenever $\mathtt{D}_\alpha(P_\eta\|Q_\eta)$ is not $\nu$-a.s.\ constant.
\end{lemma}

\begin{proof}
Fix a $\sigma$-finite dominating measure $\mu$ on $(\mathcal{Y},\mathcal{A})$ such that
$P_\eta,Q_\eta\ll \mu$ for $\nu$-a.e.\ $\eta$ (e.g.\ take $\mu=Q+\sum_{j\ge1}2^{-j}P_{\eta_j}$ for a dense countable subfamily,
or simply assume existence as standard in this setting).
Write $p_\eta=\frac{dP_\eta}{d\mu}$ and $q_\eta=\frac{dQ_\eta}{d\mu}$.
Then $p=\mathbb{E}_\eta[p_\eta]$ and $q=\mathbb{E}_\eta[q_\eta]$ are Radon--Nikodym derivatives of
$P=\int P_\eta d\nu(\eta)$ and $Q=\int Q_\eta d\nu(\eta)$ with respect to $\mu$.

By definition of R\'enyi divergence,
\[
\exp\!\big((\alpha-1)\mathtt{D}_\alpha(P\|Q)\big)
= \int \Big(\frac{p}{q}\Big)^\alpha\, q\, d\mu
= \int p^\alpha q^{1-\alpha}\, d\mu .
\]
For $\alpha>1$, the map
\[
f:(0,\infty)\times(0,\infty)\to(0,\infty),\quad f(a,b):=a^\alpha b^{1-\alpha}
= b\Big(\frac{a}{b}\Big)^\alpha
\]
is jointly convex (it is the perspective of the convex function $x\mapsto x^\alpha$).
Therefore, for each $y\in\mathcal{Y}$, applying Jensen's inequality to the random pair
$(p_\eta(y),q_\eta(y))$ gives
\[
\begin{aligned}
    &(\mathbb{E}_\eta p_\eta(y))^\alpha(\mathbb{E}_\eta q_\eta(y))^{1-\alpha}= f\big(\mathbb{E}_\eta p_\eta(y),\,\mathbb{E}_\eta q_\eta(y)\big)
\\&\leq
\mathbb{E}_\eta f\big(p_\eta(y),q_\eta(y)\big)
= \mathbb{E}_\eta\!\big[p_\eta(y)^\alpha q_\eta(y)^{1-\alpha}\big].
\end{aligned}
\]
Integrating over $\mu$ and using Tonelli/Fubini (the integrand is nonnegative) yields
\[
\begin{aligned}
    &\exp\!\big((\alpha-1)D_\alpha(P\|Q)\big)
= \int p^\alpha q^{1-\alpha}\,d\mu
\\&\leq
\mathbb{E}_\eta \int p_\eta^\alpha q_\eta^{1-\alpha}\,d\mu
=
\mathbb{E}_\eta \exp\!\big((\alpha-1)D_\alpha(P_\eta\|Q_\eta)\big).
\end{aligned}
\]
Taking $\frac{1}{\alpha-1}\log$ proves the stated bound.

For the strictness claim, note that for any random variable $Z$,
\[
\frac{1}{\alpha-1}\log \mathbb{E}[e^{(\alpha-1)Z}] \leq \operatorname*{ess\,sup} Z,
\]
with strict inequality whenever $Z$ is not $\nu$-a.s.\ constant (since then
$\mathbb{P}(Z \leq \operatorname*{ess\,sup} Z - \varepsilon)>0$ for some $\varepsilon>0$ and thus the exponential moment is
strictly smaller than $e^{(\alpha-1)\operatorname*{ess\,sup} Z}$).
Apply this to $Z=\mathtt{D}_\alpha(P_\eta\|Q_\eta)$.
\end{proof}

\subsection{Subsampling-aware envelopes at one iteration}\label{app:one_round_sa}

Fix an iteration $t$, a Pufferfish scenario $(\mathcal{S},\mathcal{Q},\Theta)$, and a slice profile $\{\mathcal{U},\omega\}$.
Let $\eta_t$ denote the (random) minibatch selector at round $t$.
Condition on a fixed history up to round $t$ (as in the main accountant).
Fix $(\theta,(s_i,s_j))\in\Theta\times\mathcal{Q}$ and a coupling $\Gamma_{\theta,ij}$ for $(X^{(i)},X^{(j)})$.

Let $R_t$ denote all remaining internal randomness used by the round-$t$ update,
and view $(\eta_t,R_t)$ as jointly generated by the algorithm at round $t$.
(Equivalently, one may treat minibatch selection as part of the round randomness; the conditioning below is the same.)

For each slice $u\in\mathcal{U}$ and each minibatch outcome $\eta$ in the support of $\eta_t$,
let $\mathbb{P}(\cdot\mid \eta_t=\eta)$ denote a regular conditional distribution
for $(X^{(i)},X^{(j)},R_t)$ given $\eta_t=\eta$ (under the coupled experiment).
Define the conditional shift bound
\begin{equation}\label{eq:z_t_u_eta_def}
z^{u}_{t}(\eta)
\;:=\;
\operatorname*{ess\,sup}_{\mathbb{P}(\cdot\mid \eta_t=\eta)}
\Big|\big\langle f_t(X^{(i)},y_{<t};R_t)-f_t(X^{(j)},y_{<t};R_t),\,u\big\rangle\Big|.
\end{equation}
Let $z^{u}_{t}(\eta_t)$ be the resulting random variable when $\eta_t$ is drawn by the algorithm.
Define the worst-case (HUC-style) per-slice cap
\[
\bar{z}^{u}_{t}
\;:=\;
\operatorname*{ess\,sup}_{\eta_t}\, z^{u}_{t}(\eta_t).
\]

\paragraph{Worst-case envelope bounds (HUC).}
Define the per-round worst-case averaged and joint envelopes by
\begin{align}
\mathtt{AR}^{\mathrm{wc}}_{t}(\sigma)
&:=\int R_\alpha\!\big(\zeta_\sigma,\bar{z}^{u}_{t}\big)\,d\omega(u),
\label{eq:AR_wc_round}\\
\mathtt{JR}^{\mathrm{wc}}_{t}(\sigma)
&:=\frac{1}{\alpha-1}\log\!\left(\int
\exp\!\big((\alpha-1)R_\alpha(\zeta_\sigma,\bar{z}^{u}_{t})\big)\,d\omega(u)\right),
\label{eq:JR_wc_round}
\end{align}
where $\zeta_\sigma$ denotes the Gaussian noise distribution with variance $\sigma^2$.

\paragraph{Subsampling-aware envelope bounds (sa-HUC).}
Define the per-round subsampling-aware envelopes by
\begin{equation}\label{eq:AR_sa_round}
    \begin{aligned}
    \mathtt{AR}^{\mathrm{sa}}_{t}(\sigma)
:=\int&
\frac{1}{\alpha-1}\log \mathbb{E}_{\eta_t}\!\big[\\
&\exp\!\big((\alpha-1)R_\alpha(\zeta_\sigma,z^{u}_{t}(\eta_t))\big)\big]\,d\omega(u),
    \end{aligned}
\end{equation}
\begin{equation}\label{eq:JR_sa_round}
    \begin{aligned}
        \mathtt{JR}^{\mathrm{sa}}_{t}(\sigma)
:=\frac{1}{\alpha-1}&\log \mathbb{E}_{\eta_t}\!\Big[\\
&\int
\exp\!\big((\alpha-1)R_\alpha(\zeta_\sigma,z^{u}_{t}(\eta_t))\big)\,d\omega(u)\Big].
    \end{aligned}
\end{equation}

\begin{lemma}[Subsampling-aware envelopes are never looser]
\label{lem:sa_never_looser_ARJR}
Fix $\alpha>1$ and $\sigma>0$. Assume that for $\omega$-a.e.\ $u\in\mathcal{U}$,
the map $\eta\mapsto z_t^u(\eta)$ is measurable and
\[
\mathbb{E}_{\eta_t}\!\left[\exp\!\big((\alpha-1)R_\alpha(\zeta_\sigma,z_t^u(\eta_t))\big)\right]<\infty,
\]
(otherwise the corresponding $\mathtt{AR}^{\mathrm{sa}}_t(\sigma)$ term equals $+\infty$ and the inequality is trivial).
Then
\[
\mathtt{AR}^{\mathrm{sa}}_{t}(\sigma)\;\leq\;\mathtt{AR}^{\mathrm{wc}}_{t}(\sigma),
\qquad
\mathtt{JR}^{\mathrm{sa}}_{t}(\sigma)\;\leq\;\mathtt{JR}^{\mathrm{wc}}_{t}(\sigma).
\]
Moreover, the inequality is strict for $\mathtt{AR}$ if there exists a measurable set $U_0\subseteq\mathcal{U}$
with $\omega(U_0)>0$ such that for every $u\in U_0$,
\[
\Pr\!\Big(R_\alpha(\zeta_\sigma,z^{u}_{t}(\eta_t)) < R_\alpha(\zeta_\sigma,\bar{z}^{u}_{t})\Big)>0.
\]
Similarly, the inequality is strict for $\mathtt{JR}$ if
\[
\begin{aligned}
\Pr\!\Big(
\int \exp\!\big((\alpha-1)R_\alpha(&\zeta_\sigma,z^{u}_{t}(\eta_t))\big)\,d\omega(u)\\&
<
\int \exp\!\big((\alpha-1)R_\alpha(\zeta_\sigma,\bar{z}^{u}_{t})\big)\,d\omega(u)
\Big)\\
&>0.
\end{aligned}
\]
\end{lemma}

\begin{proof}
Fix $u\in\mathcal{U}$. By definition, $z_t^u(\eta_t)\leq \bar{z}_t^u$ almost surely.
Since $R_\alpha(\zeta_\sigma,\cdot)$ is nondecreasing, we have
\[
R_\alpha(\zeta_\sigma,z_t^u(\eta_t))\leq R_\alpha(\zeta_\sigma,\bar{z}_t^u)
\quad\text{a.s.}
\]
Therefore,
\[
\mathbb{E}_{\eta_t}\exp\!\big((\alpha-1)R_\alpha(\zeta_\sigma,z_t^u(\eta_t))\big)
\leq
\exp\!\big((\alpha-1)R_\alpha(\zeta_\sigma,\bar{z}_t^u)\big).
\]
Taking $\frac{1}{\alpha-1}\log$ yields
\[
\frac{1}{\alpha-1}\log \mathbb{E}_{\eta_t}\!\left[\exp\!\big((\alpha-1)R_\alpha(\zeta_\sigma,z_t^u(\eta_t))\big)\right]
\leq
R_\alpha(\zeta_\sigma,\bar{z}_t^u).
\]
Integrating both sides over $u\sim\omega$ gives
$\mathtt{AR}^{\mathrm{sa}}_{t}(\sigma)\leq\mathtt{AR}^{\mathrm{wc}}_{t}(\sigma)$.

For strictness of $\mathtt{AR}$, fix any $u$ such that
$\Pr(R_\alpha(\zeta_\sigma,z_t^u(\eta_t)) < R_\alpha(\zeta_\sigma,\bar{z}_t^u))>0$.
Then the random variable $D_u(\eta_t):=R_\alpha(\zeta_\sigma,z_t^u(\eta_t))$ is not $\nu$-a.s.\ constant and satisfies
$D_u(\eta_t)\leq \operatorname*{ess\,sup}_{\eta_t} D_u(\eta_t)=R_\alpha(\zeta_\sigma,\bar{z}_t^u)$.
Hence,
\[
\frac{1}{\alpha-1}\log \mathbb{E}_{\eta_t}\!\left[e^{(\alpha-1)D_u(\eta_t)}\right]
<
R_\alpha(\zeta_\sigma,\bar{z}_t^u).
\]
Integrating this strict inequality over a set of $u$ of positive $\omega$-mass yields
$\mathtt{AR}^{\mathrm{sa}}_{t}(\sigma)< \mathtt{AR}^{\mathrm{wc}}_{t}(\sigma)$.

For $\mathtt{JR}$, the same pointwise monotonicity implies that almost surely in $\eta_t$,
\[
\begin{aligned}
    &\int \exp\!\big((\alpha-1)R_\alpha(\zeta_\sigma,z^{u}_{t}(\eta_t))\big)\,d\omega(u)\\&
\leq
\int \exp\!\big((\alpha-1)R_\alpha(\zeta_\sigma,\bar{z}^{u}_{t})\big)\,d\omega(u).
\end{aligned}
\]
Taking $\mathbb{E}_{\eta_t}$ and then $\frac{1}{\alpha-1}\log$ yields
$\mathtt{JR}^{\mathrm{sa}}_{t}(\sigma)\leq \mathtt{JR}^{\mathrm{wc}}_{t}(\sigma)$.
Strictness follows immediately under the stated strict-inequality event (since $\log$ is strictly increasing).
\end{proof}

\subsection{Full Version of Proposition \ref{prop:utility_sa_srpp_informal}}\label{app:sgd_lift}

Let $\mathcal{M}_{\mathrm{SGD}}$ be the full $T$-round mechanism in Algorithm~\ref{alg:srpp-sgd}.
Our accountant theorems (e.g., Theorem~\ref{thm:HUC_SRPP_SGD} for HUC and Theorem~\ref{thm:msHUC_SRPP_SGD} for sa-HUC)
aggregate valid per-round envelope bounds to obtain a total SRPP guarantee.

\begin{proposition}[Full Version of Proposition \ref{prop:utility_sa_srpp_informal}]
\label{prop:utility_sa_srpp_formal}
Let $\mathcal{M}_{\mathrm{SGD}}$ be the $T$-round mechanism in Algorithm~\ref{alg:srpp-sgd}.
Assume the hypotheses of Theorem~\ref{thm:HUC_SRPP_SGD} (HUC accountant) and
Theorem~\ref{thm:msHUC_SRPP_SGD} (sa-HUC accountant) hold.
Assume additionally that for each fixed $z$, the envelope $R_\alpha(\zeta_\sigma,z)$ is nonincreasing in $\sigma^2$
(e.g., for Gaussian noise $\zeta_\sigma$).

For each round $t$ and each $\sigma>0$, suppose the per-round envelope contribution in the averaged case admits
either a worst-case bound $\mathtt{AR}^{\mathrm{wc}}_{t}(\sigma)$ or a subsampling-aware bound
$\mathtt{AR}^{\mathrm{sa}}_{t}(\sigma)$, and similarly in the joint case admits either
$\mathtt{JR}^{\mathrm{wc}}_{t}(\sigma)$ or $\mathtt{JR}^{\mathrm{sa}}_{t}(\sigma)$, as defined in
Section~\ref{app:one_round_sa}.

Let $\mathsf{Agg}:\mathbb R_+^T\to \mathbb{R}_+$ denote the aggregation rule used by the corresponding accountant theorem,
and assume $\mathsf{Agg}$ is coordinatewise nondecreasing, i.e.,
if $a_t\leq b_t$ for all $t$ then $\mathsf{Agg}(a_{1:T})\leq \mathsf{Agg}(b_{1:T})$.

Then for every $\sigma>0$,
\[
\mathsf{Agg}\!\big(\mathtt{AR}^{\mathrm{sa}}_{1:T}(\sigma)\big)
\;\leq\;
\mathsf{Agg}\!\big(\mathtt{AR}^{\mathrm{wc}}_{1:T}(\sigma)\big),
\]
and
\[
\mathsf{Agg}\!\big(\mathtt{JR}^{\mathrm{sa}}_{1:T}(\sigma)\big)
\;\leq\;
\mathsf{Agg}\!\big(\mathtt{JR}^{\mathrm{wc}}_{1:T}(\sigma)\big),
\]
where $\mathtt{AR}^{\bullet}_{1:T}(\sigma):=(\mathtt{AR}^{\bullet}_{t}(\sigma))_{t=1}^T$
and $\mathtt{JR}^{\bullet}_{1:T}(\sigma):=(\mathtt{JR}^{\bullet}_{t}(\sigma))_{t=1}^T$, for $\bullet\in\{\mathrm{wc},\mathrm{sa}\}$.
Moreover, if $\mathsf{Agg}$ is strictly increasing in at least one coordinate (e.g.\ $\mathsf{Agg}(a_{1:T})=\sum_{t=1}^T a_t$)
and the strictness condition of Lemma~\ref{lem:sa_never_looser_ARJR} holds for at least one round $t$
(for the corresponding Ave/JR case), then the above inequality is strict.

Consequently, for any target SRPP (or sa-SRPP) budget level $\varepsilon$, the minimal noise variance
\[
\sigma^2_{\mathrm{sa}}(\varepsilon):=\inf\{\sigma^2:\ \mathsf{Agg}(\mathtt{AR}^{\mathrm{sa}}_{1:T}(\sigma))\leq \varepsilon\}
\quad(\text{or with }\mathtt{JR})
\]
satisfies $\sigma^2_{\mathrm{sa}}(\varepsilon)\leq \sigma^2_{\mathrm{wc}}(\varepsilon)$, and typically
$\sigma^2_{\mathrm{sa}}(\varepsilon)< \sigma^2_{\mathrm{wc}}(\varepsilon)$.
\end{proposition}

\begin{proof}
Lemma~\ref{lem:sa_never_looser_ARJR} gives $\mathtt{AR}^{\mathrm{sa}}_{t}(\sigma)\leq \mathtt{AR}^{\mathrm{wc}}_{t}(\sigma)$
and $\mathtt{JR}^{\mathrm{sa}}_{t}(\sigma)\leq \mathtt{JR}^{\mathrm{wc}}_{t}(\sigma)$ for each $t$ and each $\sigma$.
Applying the coordinatewise nondecreasing aggregation rule $\mathsf{Agg}$ yields the claimed inequalities for the
full $T$-round mechanism.

If $\mathsf{Agg}$ is strictly increasing in some coordinate and at least one round admits a strict inequality
(as characterized in Lemma~\ref{lem:sa_never_looser_ARJR}), then the aggregated inequality is strict.

Finally, since the aggregated envelope bound is nonincreasing in $\sigma^2$, the set of $\sigma^2$ values sufficient
to certify a fixed target budget under the subsampling-aware bound contains the corresponding set under the worst-case bound.
Thus the minimal sufficient $\sigma^2$ under the subsampling-aware bound cannot exceed that under the worst-case bound.
\end{proof}

\section{Proof of Lemma \ref{lemma:SRPP-envelope}}\label{app:proof_lemma:SRPP-envelope} 

Fix $\theta\in\Theta$, $(s_i,s_j)\in\mathcal{Q}$ with $P_\theta^S(s_i),P_\theta^S(s_j)>0$,
a direction $u\in\mathcal{U}$, and $\alpha>1$.
Let $P^{f,s}:=P_\theta(f(X)\mid s)$ and define the one-dimensional projected pre-noise laws
\[
P_u := \Psi^u_{\#}P_{\theta}^{f,s_i},
\quad
Q_u := \Psi^u_{\#}P_{\theta}^{f,s_j}.
\]
Let $z_\infty^u := W_\infty(P_u,Q_u)$.

By definition of $W_\infty$, for every $\varepsilon>0$ there exists a coupling
$\pi_\varepsilon\in\Pi(P_u,Q_u)$ such that
\begin{equation}\label{eq:W_infty_coupling_eps_app}
|Z-Z'|\leq z_\infty^u+\varepsilon
\quad\text{for $\pi_\varepsilon$-a.e.\ $(Z,Z')$.}
\end{equation}

\medskip\noindent
\textbf{Step 1 (Projected additive-noise form).}
Let $\widetilde{N}:=\Psi^u(N)$. Then $\widetilde{N}$ is independent of $f(X)$ and has distribution
$\widetilde{\zeta}:=\Psi^u_{\#}\zeta$ on $\mathbb{R}$.
For $z\in\mathbb{R}$, write $\widetilde{\zeta}_z$ for the law of $\widetilde{N}+z$.
Then the projected endpoint laws are mixtures of shifted noises:
\[
\Psi^u_{\#}\mathcal{M}^\theta_{s_i}=\int \widetilde{\zeta}_{z}\,P_u(dz),
\quad
\Psi^u_{\#}\mathcal{M}^\theta_{s_j}=\int \widetilde{\zeta}_{z'}\,Q_u(dz').
\]
Using the coupling $\pi_\varepsilon$ of $(P_u,Q_u)$, we can write both mixtures over the common index
$(z,z')\sim\pi_\varepsilon$:
\[
\Psi^u_{\#}\mathcal{M}^\theta_{s_i}=\int \widetilde{\zeta}_{z}\,\pi_\varepsilon(dz,dz'),
\quad
\Psi^u_{\#}\mathcal{M}^\theta_{s_j}=\int \widetilde{\zeta}_{z'}\,\pi_\varepsilon(dz,dz').
\]

\medskip\noindent
\textbf{Step 2 (Pointwise bound for each shift pair).}
For any $z,z'\in\mathbb{R}$, Rényi divergence is invariant under common translations, hence
\[
\mathtt{D}_\alpha(\widetilde{\zeta}_z\|\widetilde{\zeta}_{z'})
=
\mathtt{D}_\alpha(\widetilde{\zeta}_{z-z'}\|\widetilde{\zeta}).
\]
By definition of the (1-D) shift-R\'enyi envelope,
\[
R_\alpha(\widetilde{\zeta},r):=\sup_{|a|\leq r}\mathtt{D}_\alpha(\widetilde{\zeta}_{a}\|\widetilde{\zeta}),
\]
we obtain
\[
\mathtt{D}_\alpha(\widetilde{\zeta}_z\|\widetilde{\zeta}_{z'})
\leq
R_\alpha(\widetilde{\zeta},|z-z'|).
\]
Combining with \eqref{eq:W_infty_coupling_eps_app} yields, for $\pi_\varepsilon$-a.e.\ $(z,z')$,
\[
\mathtt{D}_\alpha(\widetilde{\zeta}_z\|\widetilde{\zeta}_{z'})
\leq
R_\alpha(\widetilde{\zeta},z_\infty^u+\varepsilon).
\]

\medskip\noindent
\textbf{Step 3 (From pointwise to mixtures via DPI).}
Define probability measures on $\mathbb{R}^2\times\mathbb{R}$ by
\[
\begin{aligned}
    &\overline{P}_\varepsilon(dz,dz',dy):=\pi_\varepsilon(dz,dz')\,\widetilde{\zeta}_z(dy),\\
&\overline{Q}_\varepsilon(dz,dz',dy):=\pi_\varepsilon(dz,dz')\,\widetilde{\zeta}_{z'}(dy).
\end{aligned}
\]
Their $y$-marginals are exactly the desired mixtures:
\[
(\overline{P}_\varepsilon)_{Y}=\Psi^u_{\#}\mathcal{M}^\theta_{s_i},
\quad
(\overline{Q}_\varepsilon)_{Y}=\Psi^u_{\#}\mathcal{M}^\theta_{s_j}.
\]
Since marginalization is a data-independent post-processing, DPI implies
\[
\mathtt{D}_\alpha\!\left(\Psi^u_{\#}\mathcal{M}^\theta_{s_i}\ \big\|\ \Psi^u_{\#}\mathcal{M}^\theta_{s_j}\right)
\leq
\mathtt{D}_\alpha(\overline{P}_\varepsilon\|\overline{Q}_\varepsilon).
\]
Moreover, conditioning on $(z,z')$ and using the standard expression for Rényi divergence of a mixture with a common mixing law,
\[
\mathtt{D}_\alpha(\overline{P}_\varepsilon\|\overline{Q}_\varepsilon)
=
\frac{1}{\alpha-1}\log
\int \exp\!\Big((\alpha-1)\mathtt{D}_\alpha(\widetilde{\zeta}_z\|\widetilde{\zeta}_{z'})\Big)\,
\pi_\varepsilon(dz,dz').
\]
Using the pointwise bound from Step 2 and the monotonicity of $r\mapsto R_\alpha(\widetilde\zeta,r)$ gives
\[
\mathtt{D}_\alpha(\overline{P}_\varepsilon\|\overline{Q}_\varepsilon)
\leq
R_\alpha(\widetilde{\zeta},z_\infty^u+\varepsilon).
\]
Therefore,
\[
\mathtt{D}_\alpha\!\left(\Psi^u_{\#}\mathcal{M}^\theta_{s_i}\ \big\|\ \Psi^u_{\#}\mathcal{M}^\theta_{s_j}\right)
\leq
R_\alpha(\widetilde{\zeta},z_\infty^u+\varepsilon).
\]
Finally, letting $\varepsilon\downarrow 0$ yields
\[
\mathtt{D}_\alpha\!\left(\Psi^u_{\#}\mathcal{M}^\theta_{s_i}\ \big\|\ \Psi^u_{\#}\mathcal{M}^\theta_{s_j}\right)
\le
R_\alpha(\widetilde{\zeta},z_\infty^u).
\]
Renaming $\widetilde{\zeta}=\Psi^u_{\#}\zeta$ (as in the lemma statement) proves the claim.
\qed

\section{Proof of Theorem~\ref{thm:gaussian-ave-SRPEp}}

Fix $\theta\in\Theta$ and $(s_i, s_j)\in \mathcal{Q}$ with $P^S_{\theta}(s_i),P^S_{\theta}(s_j)>0$.
Let $\zeta=\mathcal{N}(0,\sigma^2 I_d)$ and define
\[
\mathrm{P}:=\mathcal{M}^{\theta}_{s_i}=P_{\theta}^{f, s_i}*\zeta,
\quad
\mathrm{Q}:=\mathcal{M}^{\theta}_{s_j}=P_{\theta}^{f,s_j}*\zeta.
\]

\medskip\noindent
\textbf{Step 1 (per-slice envelope bound).}
Fix any direction $u\in\mathcal{U}$.
Applying Lemma~\ref{lemma:SRPP-envelope} to the 1-D projection $\Psi^u(y)=\langle y,u\rangle$ yields
\[
\mathtt{D}_{\alpha}\!\left(\Psi^{u}_{\#}\mathrm{P}\ \big\|\ \Psi^{u}_{\#}\mathrm{Q}\right)
\leq
R_{\alpha}\!\left(\zeta,\ z^{u}_{\infty}(\theta,s_i,s_j)\right),
\]
where
\[
z^{u}_{\infty}(\theta,s_i,s_j)
:=
W_{\infty}\!\Bigl(\Psi^{u}_{\#}P_{\theta}^{f,s_i},\ \Psi^{u}_{\#}P_{\theta}^{f,s_j}\Bigr).
\]
By the definition of the per-slice $\infty$-Wasserstein sensitivity $\Delta_\infty^u$
in \eqref{eq:per_slice_WS}, for every $u\in\mathcal{U}$ we have
\[
z^{u}_{\infty}(\theta,s_i,s_j)\leq \Delta_\infty^u,
\]
and since $R_\alpha(\zeta,\cdot)$ is non-decreasing,
\begin{equation}\label{eq:per_slice_bound_final}
\mathtt{D}_{\alpha}\!\left(\Psi^{u}_{\#}\mathrm{P}\ \big\|\ \Psi^{u}_{\#}\mathrm{Q}\right)
\leq
R_{\alpha}\!\left(\zeta,\ \Delta_\infty^u\right).
\end{equation}

\medskip\noindent
\textbf{Step 2 (Gaussian envelope).}
For $\zeta=\mathcal{N}(0,\sigma^2 I_d)$, the projected noise $\Psi^u(N)$ is $\mathcal{N}(0,\sigma^2)$.
For any shift $a\in\mathbb{R}$,
\[
\mathtt{D}_{\alpha}\!\left(\mathcal{N}(a,\sigma^2)\ \big\|\ \mathcal{N}(0,\sigma^2)\right)
=\frac{\alpha a^2}{2\sigma^2}.
\]
Hence the 1-D Gaussian shift envelope satisfies
\[
R_{\alpha}(\zeta,z)
=\sup_{|a|\leq z}\mathtt{D}_{\alpha}\!\left(\mathcal{N}(a,\sigma^2)\ \big\|\ \mathcal{N}(0,\sigma^2)\right)
=\frac{\alpha}{2\sigma^2}z^2.
\]
Substituting this into \eqref{eq:per_slice_bound_final} gives, for all $u\in\mathcal{U}$,
\[
\mathtt{D}_{\alpha}\!\left(\Psi^{u}_{\#}\mathrm{P}\ \big\|\ \Psi^{u}_{\#}\mathrm{Q}\right)
\leq
\frac{\alpha}{2\sigma^2}\,(\Delta_\infty^u)^2.
\]

\medskip\noindent
\textbf{Step 3 (average over $\omega$ and choose $\sigma^2$).}
By Definition~\ref{def:ave_srd},
\[
\mathtt{AveSD}^{\omega}_{\alpha}(\mathrm{P}\|\mathrm{Q})
=
\int_{\mathbb S^{d-1}}
\mathtt{D}_{\alpha}\!\left(\Psi^{u}_{\#}\mathrm{P}\ \big\|\ \Psi^{u}_{\#}\mathrm{Q}\right)\, d\omega(u).
\]
Therefore,
\[
\mathtt{AveSD}^{\omega}_{\alpha}(\mathrm{P}\|\mathrm{Q})
\leq
\int_{\mathbb S^{d-1}} \frac{\alpha}{2\sigma^2}(\Delta_\infty^u)^2\, d\omega(u)
=
\frac{\alpha}{2\sigma^2}\,\overline{\Delta}^2.
\]
Choosing $\sigma^2=\frac{\alpha\,\overline{\Delta}^2}{2\epsilon}$ yields
$\mathtt{AveSD}^{\omega}_{\alpha}(\mathcal M^\theta_{s_i}\|\mathcal{M}^\theta_{s_j})\leq \epsilon$
for all $\theta$ and $(s_i,s_j)$, i.e., $\mathcal{M}$ satisfies
$(\alpha,\epsilon,\omega)$-Ave-SRPP.
\qed

\section{Proof of Theorem \ref{thm:gaussian-joint-SRPEp}}

Fix $\theta\in\Theta$ and $(s_i, s_j)\in \mathcal{Q}$ with $P^S_{\theta}(s_i),P^S_{\theta}(s_j)>0$, and fix a slice profile $\{\mathcal{U}, \omega\}$.
Let 
\[
  \mathrm{P} := \mathcal{M}^{\theta}_{s_i}
  = P_{\theta}^{f,s_i} * \zeta,
  \quad \textup{and} \quad
  \mathrm{Q} := \mathcal{M}^{\theta}_{s_j}
  = P_{\theta}^{f,s_j} * \zeta.
\]

\medskip\noindent
\textbf{Step 1.}
Fix a slice $u\in\mathcal{U}$.
Lemma \ref{lemma:SRPP-envelope} applied to the 1-D
projection $\Psi^u(y) = \langle y,u\rangle$ states that, for all
$\alpha>1$,
\[
  \mathtt{D}_{\alpha}\bigl(\Psi^{u}_{\#}\mathrm{P} \,\big\|\, \Psi^{u}_{\#}Q\bigr)
  \leq
  R_{\alpha}\bigl(\zeta, z^{u}_{\infty}(\theta,s_i,s_j)\bigr),
\]
where $z^{u}_{\infty}(\theta,s_i,s_j)$ upper-bounds the 1-D
shift (in $W_{\infty}$) between the pre-noise probability measures
$\Psi^u_{\#}P_\theta^{f,s_i}$ and $\Psi^u_{\#}P_\theta^{f,s_j}$.
By definition,
\[
  z^{u}_{\infty}(\theta,s_i,s_j)\leq 
  \tilde{\Delta}^u_{\infty}(\theta,s_i,s_j)
  :=
  W_{\infty}\bigl(\Psi^{u}_{\#}P_{\theta}^{f,s_i}\,\big\|\,
                  \Psi^{u}_{\#}P_{\theta}^{f,s_j}\bigr),
\]
and the envelope $R_\alpha(\zeta,\cdot)$ is non-decreasing in its second argument. 
Hence
\begin{equation}\label{eq:per_slice_env_bound}
  \mathtt{D}_{\alpha}\bigl(\Psi^{u}_{\#}\mathrm{P} \,\big\|\, \Psi^{u}_{\#}\mathrm{Q}\bigr)\leq
  R_{\alpha}\bigl(\zeta, \tilde{\Delta}_{\infty}^u(\theta,s_i,s_j)\bigr)
\end{equation}
for all $u\in\mathcal{U}$.

\medskip\noindent
\textbf{Step 2.}
Now, we specialize to Gaussian noise $\zeta = \mathcal{N}(0, \sigma^2 I_{d})$.
For any fixed slice $u\in\mathcal{U}$, the projected noise
$\Psi^u(N)$ is 1-D Gaussian $\mathcal{N}(0,\sigma^2)$.
If we shift by $a\in\mathbb{R}$, the noise becomes
$\mathcal{N}(a,\sigma^2)$, and the order-$\alpha$ R\'enyi divergence
between $\mathcal{N}(a,\sigma^2)$ and $\mathcal{N}(0,\sigma^2)$ is
\[
  \mathtt{D}_{\alpha}\bigl(\mathcal{N}(a,\sigma^2)\,\big\|\,
                         \mathcal{N}(0,\sigma^2)\bigr)
  = \frac{\alpha a^2}{2\sigma^2}.
\]
Therefore the 1-D Gaussian shift envelope is
\[
  R_{\alpha}(\zeta,z):=
  \sup_{|a|\leq z}
  \mathtt{D}_\alpha\bigl(\zeta-a \,\big\|\, \zeta\bigr)
  = \sup_{|a|\le z} \frac{\alpha a^2}{2\sigma^2}
  = \frac{\alpha}{2\sigma^2} z^2.
\]
Substituting $z = \tilde{\Delta}^u_{\infty}(\theta,s_i,s_j)$ into
\eqref{eq:per_slice_env_bound}, we obtain
\begin{equation}\label{eq:slice_div_bound}
  \mathtt{D}_{\alpha}\bigl(\Psi^{u}_{\#}\mathrm{P} \,\big\|\, \Psi^{u}_{\#}\mathrm{Q}\bigr)
  \leq
  \frac{\alpha}{2\sigma^2}\,
  \bigl(\tilde{\Delta}^u_{\infty}(\theta,s_i,s_j)\bigr)^2.
\end{equation}

By Definition \ref{sec:joint_srpp}, the joint sliced R\'enyi divergence of order
$\alpha>1$ with slicing distribution $\omega$ is
\[
\begin{aligned}
    &\mathtt{JSD}^\omega_\alpha(\mathrm{P}\|\mathrm{Q})\\
    &:=
  \frac{1}{\alpha - 1}
  \log \int_{\mathbb{S}^{d-1} }
    \exp\Bigl(
      (\alpha-1)\,
      \mathtt{D}_{\alpha}\bigl(\Psi^{u}_{\#}\mathrm{P} \,\big\|\, \Psi^{u}_{\#}\mathrm{Q}\bigr)
    \Bigr) d\omega(u).
\end{aligned}
\]
Using the bound (\ref{eq:slice_div_bound}), we get for each $u$
\[
  \mathtt{D}_{\alpha}\bigl(\Psi^{u}_{\#}\mathrm{P} \,\big\|\, \Psi^{u}_{\#}\mathrm{Q}\bigr)\leq 
  \frac{\alpha}{2\sigma^2}\,
  \bigl(\tilde{\Delta}_{\infty}^u(\theta,s_i,s_j)\bigr)^2 \leq 
  \frac{\alpha}{2\sigma^2}\,
  \Delta_\star^2(\theta,s_i,s_j),
\]
where
\[
  \Delta_\star^2(\theta,s_i,s_j)
  :=
  \sup_{u\in\mathcal{U}}
  \bigl(\tilde{\Delta}^u_\infty(\theta,s_i,s_j)\bigr)^2.
\]
Therefore, we have
\begin{align*}
  &\mathtt{JSD}^\omega_\alpha(\mathrm{P}\|\mathrm{Q})\\
  &= \frac{1}{\alpha - 1}
     \log \int_{\mathbb{S}^{d-1} }
       \exp\Bigl(
         (\alpha-1)\,
         \mathtt{D}_{\alpha}\bigl(\Psi^{u}_{\#} \mathrm{P} \,\big\|\, \Psi^{u}_{\#} \mathrm{Q}\bigr)
       \Bigr)
   d\omega(u) \\
  &\leq
  \frac{1}{\alpha - 1}
  \log \int_{\mathbb{S}^{d-1}}
       \exp\Bigl(
         (\alpha-1)\,
         \frac{\alpha}{2\sigma^2}\,
         \Delta_\star^2(\theta,s_i,s_j)
       \Bigr) d\omega(u).
\end{align*}
Here, the integrand is independent of $u$, so the integral is just the constant itself:
\[
\begin{aligned}
    &\int_{\mathbb{S}^{d-1}}
    \exp\Bigl(
      (\alpha-1)\,
      \tfrac{\alpha}{2\sigma^2}\,
      \Delta_\star^2(\theta,s_i,s_j)
    \Bigr)
   d\omega(u)\\
  &=
  \exp\Bigl(
    (\alpha-1)\,
    \tfrac{\alpha}{2\sigma^2}\,
    \Delta_\star^2(\theta,s_i,s_j)
  \Bigr).
\end{aligned}
\]
Hence
\[
  \begin{aligned}
      \mathtt{JSD}^\omega_\alpha(\mathrm{P}\|\mathrm{Q})&\leq 
  \frac{1}{\alpha-1}
  \log \exp\Bigl(
    (\alpha-1)\,
    \tfrac{\alpha}{2\sigma^2}
    \Delta_\star^2(\theta,s_i,s_j)
  \Bigr)\\
  &=
  \frac{\alpha}{2\sigma^2}
  \Delta_\star^2(\theta,s_i,s_j).
  \end{aligned}
\]

\medskip\noindent
\textbf{Step 3.}
By definition of the global constant $\Delta_\star^2$, we have
\[
  \Delta_\star^2(\theta,s_i,s_j)\leq 
  \Delta_\star^2
\]
for all $\theta\in\Theta,\ (s_i,s_j)\in\mathcal{Q}$.
Thus
\[
  \mathtt{JSD}^{\omega}_{\alpha}\bigl(
    \mathcal{M}^{\theta}_{s_i}\,\big\|\,
    \mathcal{M}^{\theta}_{s_j}
  \bigr)\leq 
  \frac{\alpha}{2\sigma^2}\,\Delta_{\star}^2.
\]

With the choice
\[
  \sigma^2 = \frac{\alpha\,\Delta_{\star}^2}{2\epsilon},
\]
we obtain
\[
  \mathtt{JSD}^{\omega}_{\alpha}\bigl(
    \mathcal{M}^{\theta}_{s_i}\,\big\|\,
    \mathcal{M}^{\theta}_{s_j}
  \bigr)\leq \epsilon
\]
for every $\theta\in\Theta$ and $(s_i,s_j)\in\mathcal{Q}$ with positive prior mass. This is exactly the $(\alpha,\epsilon,\omega)$-Joint-SRPE
guarantee in $(\mathcal{S},\mathcal{Q},\Theta)$.

\section{Proof of Theorem~\ref{thm:mc_finite_gaussian_ave_srpp}}
\label{app:proof_mc_finite_gaussian_ave_srpp}

Let $U_{1:m}\overset{\mathrm{iid}}{\sim}\omega$. Define the random variables
\[
X_\ell := \big(\Delta_\infty^{U_\ell}\big)^2 \in [0,\Delta_0^2],
\qquad \ell=1,\ldots,m,
\]
and recall
\[
\overline{\Delta}^2(\omega) := \mathbb{E}_{U\sim\omega}\!\big[(\Delta_\infty^{U})^2\big]
= \mathbb{E}[X_1].
\]

\paragraph{Step 1.}
Define the event
\[
\mathcal{E}_{\mathrm{dist}}
:=
\Big\{\Delta_\infty^{U_\ell}\leq \widehat\Delta_\infty(U_\ell)\ \text{for all }\ell\in[m]\Big\}.
\]
By assumption,
\begin{equation}
\label{eq:Edist_prob}
\Pr(\mathcal{E}_{\mathrm{dist}})\geq 1-\gamma/2.
\end{equation}
On $\mathcal{E}_{\mathrm{dist}}$,
\begin{equation}
\label{eq:samplemean_domination}
\frac{1}{m}\sum_{\ell=1}^m X_\ell
=
\frac{1}{m}\sum_{\ell=1}^m \big(\Delta_\infty^{U_\ell}\big)^2
\leq
\frac{1}{m}\sum_{\ell=1}^m \big(\widehat\Delta_\infty(U_\ell)\big)^2.
\end{equation}

\paragraph{Step 2.}
Since $X_1,\ldots,X_m$ are i.i.d.\ and bounded in $[0,\Delta_0^2]$, Hoeffding's inequality yields that for any
$\eta\in(0,1)$,
\[
\Pr\!\left(
\mathbb{E}[X_1]
\leq
\frac{1}{m}\sum_{\ell=1}^m X_\ell
+
\Delta_0^2\sqrt{\frac{\log(2/\eta)}{2m}}
\right)\geq 1-\eta.
\]
Setting $\eta=\gamma/2$ and defining the event
\[
\mathcal{E}_{\mathrm{mc}}
:=
\left\{
\overline{\Delta}^2(\omega)
\leq
\frac{1}{m}\sum_{\ell=1}^m X_\ell
+
\Delta_0^2\sqrt{\frac{\log(4/\gamma)}{2m}}
\right\},
\]
we obtain
\begin{equation}
\label{eq:Emc_prob}
\Pr(\mathcal{E}_{\mathrm{mc}})\geq 1-\gamma/2.
\end{equation}

\paragraph{Step 3. }
On the intersection event $\mathcal{E}_{\mathrm{dist}}\cap\mathcal{E}_{\mathrm{mc}}$, combining the defining inequality of
$\mathcal{E}_{\mathrm{mc}}$ with~\eqref{eq:samplemean_domination} gives
\begin{equation}
\label{eq:barDelta_ucb_final}
\overline{\Delta}^2(\omega)
\leq
\frac{1}{m}\sum_{\ell=1}^m \big(\widehat\Delta_\infty(U_\ell)\big)^2
+
\Delta_0^2\sqrt{\frac{\log(4/\gamma)}{2m}}.
\end{equation}
By the union bound together with~\eqref{eq:Edist_prob} and~\eqref{eq:Emc_prob},
\begin{equation}
\label{eq:intersection_prob}
\Pr(\mathcal{E}_{\mathrm{dist}}\cap\mathcal{E}_{\mathrm{mc}})
\geq
1-\gamma.
\end{equation}

\paragraph{Step 4.}
Assume the mechanism uses $N\sim\mathcal{N}(0,\sigma^2 I_d)$ with
\[
\sigma^2 \;\geq\; \frac{\alpha}{2\varepsilon}\left(
\frac{1}{m}\sum_{\ell=1}^{m}\big(\widehat\Delta_\infty(U_\ell)\big)^2
+\Delta_0^2\sqrt{\frac{\log(4/\gamma)}{2m}}
\right).
\]
On $\mathcal{E}_{\mathrm{dist}}\cap\mathcal{E}_{\mathrm{mc}}$, inequality~\eqref{eq:barDelta_ucb_final} implies
\[
\sigma^2 \;\geq\; \frac{\alpha}{2\varepsilon}\,\overline{\Delta}^2(\omega).
\]
Therefore, by Theorem~\ref{thm:gaussian-ave-SRPEp} (Gaussian calibration for $(\alpha,\varepsilon,\omega)$-Ave-SRPP),
the mechanism $\mathcal{M}(X)=f(X)+N$ satisfies $(\alpha,\varepsilon,\omega)$-\textit{Ave-SRPP} on this event.
Finally,~\eqref{eq:intersection_prob} gives that this conclusion holds with probability at least $1-\gamma$
over both the conditional sampling (used to form $\widehat{\Delta}_\infty$) and the draw of~$U_{1:m}$.
\qed

\section{Proof of Theorem~\ref{thm:mc_finite_gaussian_joint_srpp}}
\label{app:proof_mc_finite_gaussian_joint_srpp}

Fix $\sigma^2>0$ and abbreviate
$c:=c(\sigma)=\frac{\alpha(\alpha-1)}{2\sigma^2}$ and $b:=b(\sigma)=\exp(c\Delta_0^2)$.
Let $U_{1:m}\overset{\mathrm{iid}}{\sim}\omega$ and define the random variables
\[
Y_\ell \;:=\; \exp\!\big(c\,(\Delta_\infty^{U_\ell})^2\big),
\qquad \ell=1,\ldots,m.
\]
Since $0\leq \Delta_\infty^{u}\leq \Delta_0$ for all $u$, we have $Y_\ell\in[1,b]$ almost surely, and
\[
\mu(\sigma)
\;:=\;
\mathbb{E}_{U\sim\omega}\!\Big[\exp\!\big(c\,(\Delta_\infty^{U})^2\big)\Big]
\;=\;
\mathbb{E}[Y_1].
\]

\paragraph{Step 1. }
Define the event
\[
\mathcal{E}_{\mathrm{dist}}
:=
\Big\{\Delta_\infty^{U_\ell}\leq \widehat\Delta_\infty(U_\ell)\ \text{for all }\ell\in[m]\Big\}.
\]
By assumption, $\Pr(\mathcal{E}_{\mathrm{dist}})\ge 1-\gamma/2$.
On $\mathcal{E}_{\mathrm{dist}}$, for each $\ell$,
\[
Y_\ell
=
\exp\!\big(c\,(\Delta_\infty^{U_\ell})^2\big)
\;\leq\;
\exp\!\big(c\,(\widehat\Delta_\infty(U_\ell))^2\big),
\]
and hence
\begin{equation}
\label{eq:mc_joint_domination}
\frac{1}{m}\sum_{\ell=1}^m Y_\ell
\;\leq\;
\widehat{\mu}_m(\sigma)
:=
\frac{1}{m}\sum_{\ell=1}^{m}
\exp\!\Big(c\,\big(\widehat\Delta_\infty(U_\ell)\big)^2\Big).
\end{equation}

\paragraph{Step 2.}
Since $Y_1,\ldots,Y_m$ are i.i.d.\ and bounded in $[1,b]$, by Hoeffding's inequality, we have, for any 
$\eta\in(0,1)$,
\[
\Pr\!\left(\mu(\sigma)\leq \frac1m\sum_{\ell=1}^m Y_\ell + (b-1)\sqrt{\frac{\log(2/\eta)}{2m}}\right)\geq 1-\eta.
\]
Setting $\eta=\gamma/2$ gives the stated $\log(4/\gamma)$ term.
Set $\eta=\gamma/2$ and define the event
\[
\mathcal{E}_{\mathrm{mc}}
:=
\left\{
\mu(\sigma)
\leq
\frac{1}{m}\sum_{\ell=1}^m Y_\ell
+
(b-1)\sqrt{\frac{\log(4/\gamma)}{2m}}
\right\}.
\]
Then $\Pr(\mathcal{E}_{\mathrm{mc}})\geq 1-\gamma/2$.

\paragraph{Step 3.}
On $\mathcal{E}_{\mathrm{dist}}\cap \mathcal{E}_{\mathrm{mc}}$, combining the defining inequality of
$\mathcal{E}_{\mathrm{mc}}$ with~\eqref{eq:mc_joint_domination} yields
\begin{equation}
\label{eq:mu_ucb_final}
\mu(\sigma)
\leq
\widehat{\mu}_m(\sigma)
+
(b-1)\sqrt{\frac{\log(4/\gamma)}{2m}}.
\end{equation}
Moreover, by the union bound,
\[
\Pr(\mathcal{E}_{\mathrm{dist}}\cap \mathcal{E}_{\mathrm{mc}})
\;\geq\;
1-\gamma.
\]

\paragraph{Step 4.}
For Gaussian SWM with noise $N\sim \mathcal{N}(0,\sigma^2 I_d)$, the joint envelope equals
\[
\mathtt{JR}^{\infty}_{\alpha,\omega}
=
\frac{1}{\alpha-1}\log \mu(\sigma).
\]
On $\mathcal{E}_{\mathrm{dist}}\cap \mathcal{E}_{\mathrm{mc}}$, inequality~\eqref{eq:mu_ucb_final} implies
\[
\mathtt{JR}^{\infty}_{\alpha,\omega}
\leq
\frac{1}{\alpha-1}\log\!\left(
\widehat\mu_m(\sigma)
+
(b-1)\sqrt{\frac{\log(4/\gamma)}{2m}}
\right).
\]
Hence, if $\sigma^2$ satisfies the stated condition
\[
\frac{1}{\alpha-1}\log\!\left(
\widehat{\mu}_m(\sigma)
+
(b-1)\sqrt{\frac{\log(4/\gamma)}{2m}}
\right)
\leq
\varepsilon,
\]
then on $\mathcal{E}_{\mathrm{dist}}\cap \mathcal{E}_{\mathrm{mc}}$ we have
$\mathtt{JR}^{\infty}_{\alpha,\omega}\leq \varepsilon$.
Therefore, by the Gaussian Joint-SRPP calibration result (Theorem~\ref{thm:gaussian-joint-SRPEp} stated for SRPP),
the mechanism $\mathcal{M}(X)=f(X)+\mathcal{N}(0,\sigma^2 I_d)$ satisfies $(\alpha,\varepsilon,\omega)$-\textit{Joint-SRPP}
on this event. Since $\Pr(\mathcal{E}_{\mathrm{dist}}\cap \mathcal{E}_{\mathrm{mc}})\geq 1-\gamma$, the claim follows.
\qed

\section{Formal Version and Proof of Proposition~\ref{prop:comp_advantage_swm}}
\label{app:computational_advantage}

We compare the cost of evaluating (i) sliced sensitivities $\Delta_\infty^u$ over finitely many directions
versus (ii) the unsliced $d$-dimensional empirical $W_\infty$ sensitivity.

Assume that for each $(\theta,s)$ we have $n$ samples from $P^{f,s}_\theta$, and let
$\mathsf{M}:=|\widehat\Theta|\,|\widehat{\mathcal Q}|$ denote the number of $(\theta,(s_i,s_j))$ instances.
Fix $L$ directions $U_1,\ldots,U_L\in\mathbb S^{d-1}$.

\begin{lemma}[Exact empirical $W_\infty$ on $\mathbb R$]
\label{lem:empirical_Winfty_1d}
Let $a_1,\ldots,a_n\in\mathbb R$ and $b_1,\ldots,b_n\in\mathbb R$ and define
$\widehat P:=\frac1n\sum_{k=1}^n\delta_{a_k}$ and $\widehat Q:=\frac1n\sum_{k=1}^n\delta_{b_k}$.
Let $a_{(1)}\le\cdots\le a_{(n)}$ and $b_{(1)}\le\cdots\le b_{(n)}$ be the order statistics.
Then
\[
W_\infty(\widehat P,\widehat Q)=\max_{1\le k\le n}\,|a_{(k)}-b_{(k)}|.
\]
\end{lemma}

\begin{proof}
Define the coupling $\pi:=\frac1n\sum_{k=1}^n\delta_{(a_{(k)},b_{(k)})}\in\Pi(\widehat P,\widehat Q)$.
Then $\operatorname{ess\,sup}_{(A,B)\sim\pi}|A-B|=\max_k|a_{(k)}-b_{(k)}|$, hence
$W_\infty(\widehat P,\widehat Q)\le \max_k|a_{(k)}-b_{(k)}|$.

For the reverse inequality, let $t^\star:=\max_k|a_{(k)}-b_{(k)}|$ and fix $t<t^\star$.
Choose $k$ such that $|a_{(k)}-b_{(k)}|>t$.
Assume $a_{(k)}<b_{(k)}-t$ (the other case is symmetric) and set
$A:=(-\infty,a_{(k)}]$ and $B:=[a_{(k)}+t,\infty)$.
Then $\widehat P(A)=k/n$ and $\widehat Q(B)\ge (n-k+1)/n$ since $b_{(k)}\in B$ implies
$b_{(k)},\ldots,b_{(n)}\in B$.
Therefore, under any coupling $\pi'\in\Pi(\widehat P,\widehat Q)$, at least $1/n$ mass must be transported
from $A$ to $B$, so $\pi'(\{|x-y|>t\})>0$.
Hence $\operatorname{ess\,sup}_{(X,Y)\sim\pi'}|X-Y|\ge t^\star$.
Taking the infimum over $\pi'$ yields $W_\infty(\widehat P,\widehat Q)\ge t^\star$.
\end{proof}

\begin{lemma}[Uniform empirical $W_\infty$ reduces to a bottleneck matching problem]
\label{lem:Winfty_bottleneck_matching}
Let $\widehat P=\frac1n\sum_{i=1}^n\delta_{x_i}$ and $\widehat Q=\frac1n\sum_{j=1}^n\delta_{y_j}$ on $\mathbb R^d$.
For $t\ge 0$, define the bipartite graph $G_t$ on $(\{1,\ldots,n\},\{1,\ldots,n\})$ with an edge $(i,j)$ iff
$\|x_i-y_j\|_2\le t$.
Then
\[
W_\infty(\widehat P,\widehat Q)\le t
\quad\Longleftrightarrow\quad
G_t \text{ has a perfect matching}.
\]
Consequently,
\[
W_\infty(\widehat P,\widehat Q)=\min\{t\ge 0:\ G_t \text{ has a perfect matching}\}.
\]
\end{lemma}

\begin{proof}
($\Leftarrow$) If $G_t$ has a perfect matching, there exists a permutation $\sigma$ such that
$\|x_i-y_{\sigma(i)}\|_2\le t$ for all $i$. Then
$\pi:=\frac1n\sum_{i=1}^n\delta_{(x_i,y_{\sigma(i)})}\in\Pi(\widehat P,\widehat Q)$
is supported on $\{(x,y):\|x-y\|_2\le t\}$, so $W_\infty(\widehat P,\widehat Q)\le t$.

($\Rightarrow$) If $W_\infty(\widehat P,\widehat Q)\le t$, there exists $\pi\in\Pi(\widehat P,\widehat Q)$
supported on $\{(x,y):\|x-y\|_2\le t\}$.
Write $\pi=\frac1n\sum_{i,j}\pi_{ij}\,\delta_{(x_i,y_j)}$ where $\Pi:=(\pi_{ij})$ is doubly stochastic.
Support implies $\pi_{ij}=0$ whenever $(i,j)$ is not an edge of $G_t$.
By the Birkhoff--von Neumann theorem, $\Pi$ is a convex combination of permutation matrices,
so there exists a permutation matrix supported on edges of $G_t$, i.e.\ $G_t$ contains a perfect matching.
\end{proof}

\begin{lemma}[Elementary domination for asymptotic sums]
\label{lem:dominate_sum}
For all $n\ge 3$, one has $n\le n\log n$ (with $\log$ the natural logarithm). Consequently,
for any constants $c_1,c_2,c_3\ge 0$ and all $n\ge 3$,
\[
c_1 dn+c_2 n\log n + c_3 n \;\le\; (c_1+c_2+c_3)\,(dn+n\log n).
\]
\end{lemma}

\begin{proof}
For $n\ge 3$, $\log n\ge 1$, hence $n\le n\log n$.
Thus $c_3 n\le c_3 n\log n$ and the inequality follows by grouping terms.
\end{proof}

\begin{proposition}[Formal Version of Proposition~\ref{prop:comp_advantage_swm}]
\label{prop:slice_complexity_comp}
Fix a Pufferfish scenario $(\mathcal{S}, \mathcal{Q},\Theta)$.
\begin{enumerate}
\item \textup{\textbf{Per-slice computation is exact in 1-D.}}
For a fixed direction $u$, the empirical $W_\infty$ between the one-dimensional projected samples
$\{\langle x_k,u\rangle\}_{k=1}^n$ and $\{\langle y_k,u\rangle\}_{k=1}^n$
equals $\max_{k}|a_{(k)}-b_{(k)}|$ after sorting (Lemma~\ref{lem:empirical_Winfty_1d}).
Per instance, projection costs $O(dn)$ and sorting costs $O(n\log n)$, hence
\[
T_{\mathrm{1D}}(u)=O\!\big(\mathsf{M}\,(dn+n\log n)\big).
\]

\item \textup{\textbf{Finite slicing over $L$ directions.}}
Repeating (1) for $u=U_\ell$ over all $\ell\in[L]$ gives
\[
T_{\mathrm{slice}}=O\!\big(\mathsf{M}\,L\,(dn+n\log n)\big).
\]

\item \textup{\textbf{Unsliced $\infty$-Wasserstein.}}
For each instance, forming all distances $\|x_i-y_j\|_2$ requires $\Theta(dn^2)$ arithmetic and
$\Omega(n^2)$ memory to materialize the distance table.
Moreover, exact evaluation of $W_\infty$ can be performed by searching over thresholds $t$ and checking
perfect-matching feasibility in $G_t$ (Lemma~\ref{lem:Winfty_bottleneck_matching}).
If $\mathcal T:=\{\|x_i-y_j\|_2:\ 1\le i,j\le n\}$, then $|\mathcal T|\le n^2$ and binary search over $\mathcal T$
requires $O(\log n)$ feasibility checks. Using a standard bipartite matching algorithm with $O(n^3)$ time,
the per-instance runtime is
\[
T_{\mathrm{unsliced,exact}}(1\ \text{instance})=O(dn^2+n^3\log n),
\]
and over all $M$ instances,
\[
T_{\mathrm{unsliced,exact}}=O\!\big(\mathsf{M}\,(dn^2+n^3\log n)\big),
\]
with memory $\Omega(n^2)$ per instance if distances are materialized.
\end{enumerate}
\end{proposition}

\begin{proof}
We count arithmetic operations and comparisons in the standard word-RAM model.
Fix any instance $(\theta,(s_i,s_j))\in\widehat\Theta\times\widehat{\mathcal{Q}}$ and let
$x_1,\ldots,x_n\in\mathbb{R}^d$ and $y_1,\ldots,y_n\in\mathbb{R}^d$ be the corresponding samples.

\paragraph{(1) Per-slice computation with one fixed direction $u$.}
Fix $u\in\mathbb{S}^{d-1}$. For each $k\in[n]$, compute $a_k:=\langle x_k,u\rangle$ and $b_k:=\langle y_k,u\rangle$.
A dot product in $\mathbb{R}^d$ can be computed using exactly $d$ multiplications and $d-1$ additions; thus there exists
a constant $c_{\mathrm{dot}}>0$ such that each dot product costs at most $c_{\mathrm{dot}}\,d$ primitive operations.
Hence computing all $2n$ projections costs at most $2c_{\mathrm{dot}}\,dn$ operations.

Next, sort $\{a_k\}_{k=1}^n$ and $\{b_k\}_{k=1}^n$ to obtain order statistics
$a_{(1)}\leq\cdots\leq a_{(n)}$ and $b_{(1)}\leq\cdots\leq b_{(n)}$.
Using a standard comparison-based sorting algorithm (e.g., mergesort), sorting $n$ reals costs at most
$c_{\mathrm{sort}}\,n\log n$ comparisons for some constant $c_{\mathrm{sort}}>0$ (for all $n\geq 2$).
Thus sorting both lists costs at most $2c_{\mathrm{sort}}\,n\log n$ operations.

Define $\widehat{P}_u:=\frac1n\sum_{k=1}^n\delta_{a_k}$ and $\widehat{Q}_u:=\frac1n\sum_{k=1}^n\delta_{b_k}$.
By Lemma~\ref{lem:empirical_Winfty_1d},
\[
W_\infty(\widehat{P}_u,\widehat{Q}_u)=\max_{k\in[n]}|a_{(k)}-b_{(k)}|.
\]
Computing the maximum can be done by one scan through $k=1,\ldots,n$, performing $O(1)$ operations per $k$; hence
there exists $c_{\mathrm{scan}}>0$ such that this step costs at most $c_{\mathrm{scan}}\,n$ operations.

Therefore, for this fixed instance and fixed $u$, the total cost is bounded by
\[
T_{\mathrm{inst}}(u)
\;\leq\;
2c_{\mathrm{dot}}\,dn \;+\; 2c_{\mathrm{sort}}\,n\log n \;+\; c_{\mathrm{scan}}\,n.
\]
For all $n\geq 3$, Lemma~\ref{lem:dominate_sum} implies that
\[
T_{\mathrm{inst}}(u)
\;\leq\;
C_{\mathrm{inst}}\,(dn+n\log n),
\quad
C_{\mathrm{inst}}:=2c_{\mathrm{dot}}+2c_{\mathrm{sort}}+c_{\mathrm{scan}}.
\]
Hence $T_{\mathrm{inst}}(u)=O(dn+n\log n)$.

Now fix $u$ and consider the $\mathsf{M}$ instances. Computing $\Delta_\infty^u$ requires computing $W_\infty(\widehat{P}_u,\widehat{Q}_u)$
once per instance and taking their maximum.
Maintaining a running maximum costs at most $c_{\max}$ operations per instance for some constant $c_{\max}>0$,
thus at most $c_{\max}\mathsf{M}$ operations in total.
Consequently, for all $n\geq 3$,
\[
T_{\mathrm{1D}}(u)
\leq
\sum_{r=1}^\mathsf{M} T_{\mathrm{inst}}^{(r)}(u) \;+\; c_{\max}\mathsf{M}
=
O\!\big(\mathsf{M}(dn+n\log n)\big),
\]
which proves (1).

\paragraph{(2) Finite slicing over $L$ directions.}
For each $\ell\in[L]$, computing $\Delta_\infty^{U_\ell}$ is exactly the computation in (1) with $u=U_\ell$.
Therefore,
\[
T_{\mathrm{slice}}
=
\sum_{\ell=1}^L T_{\mathrm{1D}}(U_\ell)
=
O\!\big(\mathsf{M}L(dn+n\log n)\big),
\]
which proves (2).

\paragraph{(3) Unsliced $\infty$-Wasserstein (exact computation).}
Fix one instance and define $\widehat{P}:=\frac1n\sum_{i=1}^n\delta_{x_i}$ and
$\widehat{Q}:=\frac1n\sum_{j=1}^n\delta_{y_j}$.
Compute the full distance table $D\in\mathbb{R}^{n\times n}$ with entries $D_{ij}:=\|x_i-y_j\|_2$.
Computing one Euclidean distance requires $\Theta(d)$ arithmetic operations; hence computing all $n^2$ distances
requires $\Theta(dn^2)$ arithmetic operations. If the table is materialized, it stores $n^2$ reals and therefore
requires $\Omega(n^2)$ memory.

Let $\mathcal{T}:=\{D_{ij}: i,j\in[n]\}$ be the set of candidate thresholds, with $|\mathcal{T}|\leq n^2$.
By Lemma~\ref{lem:Winfty_bottleneck_matching},
\[
W_\infty(\widehat{P},\widehat{Q})=\min\{t\in\mathcal{T}:\ G_t\ \text{has a perfect matching}\}.
\]
Sorting the multiset of at most $n^2$ values in $\mathcal{T}$ costs $O(n^2\log(n^2))=O(n^2\log n)$, and
binary search over the sorted list uses $O(\log|\mathcal{T}|)=O(\log n)$ feasibility checks.
Given $t$, scanning $D$ to determine edges of $G_t$ costs $O(n^2)$ time, which is dominated by matching.
A perfect matching in a bipartite graph on $2n$ vertices can be found in $O(n^3)$ time by standard algorithms.
Thus each feasibility check costs $O(n^3)$ time, and the total time spent in feasibility checks is
$O(n^3\log n)$.

Combining the distance-table computation, threshold sorting, and feasibility checks yields the per-instance bound
\[
\begin{aligned}
    &T_{\mathrm{unsliced,exact}}(1\ \mathrm{instance})\\
&=
O(dn^2) + O(n^2\log n) + O(n^3\log n)\\&
=
O(dn^2+n^3\log n).
\end{aligned}
\]
Over $\mathsf{M}$ instances this becomes $O(\mathsf{M}(dn^2+n^3\log n))$, with $\Omega(n^2)$ memory per instance if $D$ is materialized.
This proves (3) and completes the proof.
\end{proof}

\section{Proof of Proposition \ref{prop:exist_HUC}}

Fix an iteration $t$ and a slice profile $\{\mathcal{U},\omega\}$ with
$\mathcal{U} = \{u_i\}_{i=1}^m$.
Fix arbitrarily a prior $\theta\in\Theta$, a secret pair
$(s_i,s_j)\in\mathcal{Q}$ with $P^S_\theta(s_i),P^S_\theta(s_j)>0$, a history
$y_{<t}$ in the support, and a coupling
$\gamma \in \Pi(\mu^\theta_{s_i},\mu^\theta_{s_j})$.
Let $(X,X')\sim\gamma$ and $R_t\sim\mathbb{P}_{\eta,\rho}$ denote the
minibatch sampling randomness at iteration $t$.
For a draw $r$ of $R_t$, let $\mathsf{I}_t(r)\subseteq[n]$ be the minibatch
index (multi)set and $B_t(r) := |\mathsf{I}_t(r)|$ its size.

\noindent\textbf{Step 1.}
By assumption, the minibatch size is deterministic, so we write $B_t := B_t(r)\geq 1$.
For a dataset $x=(x_1,\dots,x_n)$, let $g_j(x)$ be the per-example gradient
on $x_j$, and define the $\ell_2$-clipped gradient
\[
\tilde{g}_j(x):=g_j(x) \min \left\{1,\frac{C}{\|g_j(x)\|_2}\right\}.
\]
Then $\|\tilde{g}_j(x)\|_2 \le C$ for all $j$ and all $x$.
The averaged clipped minibatch gradient at iteration $t$ is
\[
\bar{g}_t(\mathcal{B}(x;r)):=\frac{1}{B_t}\sum_{j\in \mathsf{I}_t(r)} \tilde{g}_j(x),
\]
and similarly for $x'$.
The discrepancy cap at iteration $t$ (see \eqref{eq:discrepancy_K}) is
\[
  K_t(x,x'; r):=
  \sum_{j\in \mathsf{I}_t(r)} \mathbf{1}\{x_j \neq x'_j\}.
\]
By definition, at most $K_t(x,x'; r)$ indices in $\mathsf{I}_t(r)$ differ
between $x$ and $x'$, so only those terms can change their contribution to the average of gradients.

Whenever $x_j \neq x'_j$, we have
\[
  \|\tilde{g}_j(x) - \tilde{g}_j(x')\|_2 \leq 
  \|\tilde{g}_j(x)\|_2 + \|\tilde{g}_j(x')\|_2\leq 2C.
\]
Therefore, 
\[
\begin{aligned}
    \bigl\|\bar{g}_t(\mathcal{B}(x;r)) - \bar{g}_t(\mathcal{B}(x';r))\bigr\|_2 &\leq 
  \frac{1}{B_t}\,K_t(x,x'; r)\,(2C)\\
  &=\frac{2C}{B_t}\,K_t(x,x'; r).
\end{aligned}
\]
For any slice $u_{\ell}\in\mathcal{U}$, this implies
\begin{equation}\label{eq:pre_update_directional_huc}
  \bigl|\langle \bar{g}_t(\mathcal{B}(x;r))
           - \bar{g}_t(\mathcal{B}(x';r)), u_{\ell}\rangle\bigr|\leq 
  \frac{2C}{B_t} K_t(x,x'; r).
\end{equation}

\noindent\textbf{Step 2.}
At iteration $t$, the pre-perturbation update is
\[
  f_t(x,y_{<t}; r)
  =
  T_t\bigl(\bar{g}_t(\mathcal{B}(x;r)); y_{<t}\bigr),
\]
and similarly
$f_t(x',y_{<t}; r) = T_t(\bar{g}_t(\mathcal{B}(x';r)); y_{<t})$.
Let
\[
  z   := \bar{g}_t(\mathcal{B}(x;r)),\qquad
  z'  := \bar{g}_t(\mathcal{B}(x';r)).
\]
By Assumption~\ref{assp:slicewise_Lipschitz}, for each $u_{\ell}\in\mathcal{U}$
there exists $L_{t,i}<\infty$ such that, for all $z,z'\in\mathbb{R}^d$
and all histories $y_{<t}$,
\[
  \bigl|u_{\ell}^{\top}\bigl(T_t(z; y_{<t}) - T_t(z'; y_{<t})\bigr)\bigr|\leq L_{t,i} \bigl|u_{\ell}^{\top}(z - z')\bigr|.
\]
Applying this with the $z,z'$ defined above and using
(\ref{eq:pre_update_directional_huc}), we obtain
\[
\begin{aligned}
    \bigl|\langle f_t(x,y_{<t}; r) - f_t(x',y_{<t}; r), u_{\ell}\rangle\bigr|&\leq  L_{t,i}\,\frac{2C}{B_t}\,K_t(x,x'; r)\\
&= \frac{2 C L_{t,i}}{B_t}\,K_t(x,x'; r).
\end{aligned}
\]
Squaring both sides yields
\[
\bigl|\langle f_t(x,y_{<t}; r) - f_t(x',y_{<t}; r), u_{\ell}\rangle\bigr|^2\leq \Bigl(\frac{2 C L_{t,i}}{B_t} K_t(x,x'; r)\Bigr)^2.
\]

\noindent\textbf{Step 3.}
By definition of a feasible discrepancy cap at iteration $t$, $K_t$ satisfies
\[
  K_t(x,x'; r) \leq  K_t
\]
for $\gamma\times\mathbb{P}_{\eta,\rho}$-almost every $((X,X'),R_t)$, uniformly over all $\theta\in\Theta$, $(s_i,s_j)\in\mathcal{Q}$, histories
$y_{<t}$ in support, and couplings $\gamma$.
Thus, for $\gamma\times\mathbb{P}_{\eta,\rho}$-a.e.\ $((X,X'),R_t)$,
\[
  \bigl|\langle f_t(X,y_{<t}; R_t) - f_t(X',y_{<t}; R_t), u_{\ell}\rangle\bigr|^2\leq \Bigl(\frac{2 K_t L_{t,i} C}{B_t}\Bigr)^2
  =: h_{t,i}.
\]
Since the bound holds for all $i=1,\dots,m$ and is uniform over $\theta,(s_i,s_j),y_{<t}$, and $\gamma$, the vector
$h_t = (h_{t,1},\dots,h_{t,m})$ with
\[
  h_{t,i}
  :=
  \Bigl(\frac{2 K_t L_{t,i} C}{B_t}\Bigr)^2
\]
satisfies the HUC condition (\ref{eq:HUC}).
Therefore $h_t$ is a valid history-uniform cap at iteration $t$, as claimed. 
\qed

\section{Proof of Proposition \ref{prop:msHUC_from_K2}}

Fix an iteration $t$ and a slice profile $\{\mathcal{U}, \omega\}$ with $\mathcal{U} = \{u_i\}_{i=1}^m \subset \mathbb{S}^{d-1}$.
Let $\theta\in\Theta$, $(s_i,s_j)\in\mathcal{Q}$, a history $y_{<t}$ in the support, and a coupling
$\gamma\in\Pi(\mu^\theta_{s_i},\mu^\theta_{s_j})$ be arbitrary.
Let $(X,X')\sim\gamma$.
In addition, let $r_t\sim\mathbb{P}_{\eta,\rho}$ denote the subsampling randomness at iteration $t$.
We write $\bar{g}_t(X;r_t)$ and $\bar{g}_t(X';r_t)$ for the clipped, averaged minibatch gradients, and
$f_t(X,y_{<t};r_t)$, $f_t(X',y_{<t};r_t)$ for the corresponding pre-perturbation updates.

\noindent\textbf{Step 1.}
By per-example $\ell_2$-clipping at threshold $C$ and the definition of the discrepancy
$K_t(X,X';r_t)$, the argument in the proof of Proposition \ref{prop:exist_HUC} gives, for every
realization of $r_t$,
\[
  \bigl\|\bar g_t(X;r_t) - \bar{g}_t(X';r_t)\bigr\|_2\leq\frac{2C}{B_t} K_t(X,X';r_t).
\]
Hence, for any slice $u_{\ell}\in\mathcal{U}$,
\begin{equation}\label{eq:ms_pre_update_bound}
  \bigl|\langle \bar{g}_t(X;r_t) - \bar{g}_t(X';r_t), u_{\ell}\rangle\bigr|\leq 
  \frac{2C}{B_t} K_t(X,X';r_t).
\end{equation}
Define the pre-update difference
\[
  \Delta_t^{\mathrm{pre}}(X,X';r_t):=
  \bar{g}_t(X;r_t) - \bar{g}_t(X';r_t),
\]
and the update difference
\[
\Delta_t(X,X';r_t):= f_t(X,y_{<t};r_t) - f_t(X',y_{<t};r_t).
\]
By Assumption~\ref{assp:slicewise_Lipschitz}, for each slice $u_{\ell}\in\mathcal{U}$ and every
realization of $r_t$,
\[
  \bigl|\langle \Delta_t(X,X';r_t), u_{\ell}\rangle\bigr|\leq
  L_{t,i} \bigl|\langle \Delta_t^{\mathrm{pre}}(X,X';r_t), u_{\ell}\rangle\bigr|.
\]

Combining this with~\eqref{eq:ms_pre_update_bound} yields the pathwise bound
\begin{equation}\label{eq:ms_directional_pathwise}
  \bigl|\langle \Delta_t(X,X';r_t), u_{\ell}\rangle\bigr|\leq
  \frac{2 L_{t,i} C}{B_t}\,K_t(X,X';r_t)
\end{equation}
for all $i=1,\dots,m$ and all $\eta_t$.

\noindent\textbf{Step 2.}
Squaring both sides of (\ref{eq:ms_directional_pathwise}) and taking expectations with respect to
$\eta_t\sim\mathbb{P}_{\eta}$, we obtain, for each $i=1,\dots,m$,
\[
  \mathbb{E}_{\eta_t}\bigl[
    \bigl|\langle \Delta_t(X,X';r_t), u_i\rangle\bigr|^2
  \bigr]\leq
  \Bigl(\frac{2 L_{t,i} C}{B_t}\Bigr)^2
  \mathbb{E}_{\eta_t}\bigl[ K_t(X,X';r_t)^2\bigr].
\]
By the definition of the mean-square discrepancy cap $\overline{K}_t^2$ in (\ref{eq:ms_disc_cap}),
we have, for $\gamma$-almost every $(X,X')$,
\[
  \mathbb{E}_{\eta_t}\bigl[ K_t(X,X';r_t)^2\bigr]\leq
  \overline{K}_t^2.
\]
Therefore, for $\gamma$-almost every $(X,X')$,
\[
  \mathbb{E}_{\eta_t}\bigl[
    \bigl|\langle \Delta_t(X,X';r_t), u_{\ell}\rangle\bigr|^2
  \bigr]\leq 
  \Bigl(\frac{2 L_{t,i} C}{B_t}\Bigr)^2 \overline{K}_t^2
  = h_{t,i}^{\mathsf{sa}},
  \quad i=1,\dots,m.
\]

\noindent\textbf{Step 3.}
Since $\theta\in\Theta$, $(s_i,s_j)\in\mathcal{Q}$, $y_{<t}$, and
$\gamma\in\Pi(\mu^{\theta}_{s_i},\mu^{\theta}_{s_j})$ were arbitrary, the above bound holds for all
admissible priors, secret pairs, histories, and couplings, and for $\gamma$-almost every
$(X,X')$.
By Definition \ref{def:msHUC} of a mean-square history-uniform cap, this exactly means that
\[
  h_t^{\mathsf{sa}} = \bigl(h_{t,i}^{\mathsf{sa}}\bigr)_{i=1}^m 
  \quad \textup{and} \quad
  h_{t,i}^{\mathsf{sa}}
  =
  \Bigl(\frac{2 L_{t,i} C}{B_t}\Bigr)^2 \overline{K}_t^2,
\]
is a valid sa-HUC at iteration $t$.
\qed

\section{Proof of Lemma~\ref{lemma:gaussian_SRE}}

Fix $t$, a slice $u_i\in\mathbb{S}^{d-1}$, a prior $\theta\in\Theta$, and a secret pair
$(s_i,s_j)\in\mathcal{Q}$ with $P^S_\theta(s_i),P^S_\theta(s_j)>0$.
Let $\gamma\in\Pi(\mu^\theta_{s_i},\mu^\theta_{s_j})$ and draw
$((X,X'),R_t)\sim \gamma\times \mathbb{P}_{\eta,\rho}$.
Let $N_t\sim\mathcal{N}(0,\Sigma_t)$ be independent of $((X,X'),R_t)$ and set
\[
\begin{aligned}
    &\widetilde{Y}_t := \langle f_t(X,y_{<t};R_t),u_i\rangle + \langle N_t,u_i\rangle,\\&
\widetilde{Y}_t' := \langle f_t(X',y_{<t};R_t),u_i\rangle + \langle N_t,u_i\rangle .
\end{aligned}
\]
Then $\langle N_t,u_i\rangle\sim\mathcal N(0,v_{t,i})$ with
$v_{t,i}=u_i^\top\Sigma_t u_i>0$.

\medskip
\noindent\textbf{Step 1. }
Let $S:=(X,X',R_t)$ and define the (random) mean shift
\[
a(S):=\langle f_t(X,y_{<t};R_t)-f_t(X',y_{<t};R_t),u_i\rangle .
\]
Conditioned on $S$, we have
\[
\widetilde{Y}_t\mid S \sim \mathcal{N}(m(S)+a(S),\,v_{t,i}),
\quad
\widetilde{Y}_t'\mid S \sim \mathcal{N}(m(S),\,v_{t,i}),
\]
for some $m(S)\in\mathbb{R}$ (which cancels in the divergence). Hence for every realization $S$,
\[
\begin{aligned}
    \mathtt{D}_\alpha\!\left(\widetilde{Y}_t\mid S \ \big\|\ \widetilde{Y}_t'\mid S\right)
&=
\mathtt{D}_\alpha\!\left(\mathcal{N}(m+a,v_{t,i})\,\big\|\,\mathcal{N}(m,v_{t,i})\right)\\&
= \frac{\alpha\,a(S)^2}{2v_{t,i}}.
\end{aligned}
\]
By the HUC property (Definition~\ref{def:HUC}), we have
$a(S)^2\le h_{t,i}$ $\gamma\times\mathbb{P}_{\eta,\rho}$-a.s., therefore
\[
\mathtt{D}_\alpha\!\left(\widetilde{Y}_t\mid S \ \big\|\ \widetilde{Y}_t'\mid S\right)
\leq \frac{\alpha}{2v_{t,i}}\,h_{t,i}
\quad \gamma\times\mathbb{P}_{\eta,\rho}\text{-a.s.}
\]

\medskip
\noindent\textbf{Step 2. }
Let $\nu$ be the law of $S$ under $\gamma\times\mathbb{P}_{\eta,\rho}$.
For each $s\in\mathrm{supp}(\nu)$, write $P_s$ as the probability distribution of $\widetilde{Y}_t$ given $s$ and $Q_s$ as the probability distribution of $\widetilde Y'_t$ given $s$.
Then the unconditional one-dimensional output laws can be written as mixtures
with the same weights:
\[
\begin{aligned}
&P := \Pr(\langle Y_t,u_i\rangle\mid s_i,\theta)=\int P_s\,\nu(ds),
\\&
Q := \Pr(\langle Y_t,u_i\rangle\mid s_j,\theta)=\int Q_s\,\nu(ds).
\end{aligned}
\]
By the standard Rényi-divergence bound for mixtures with common weights
(equivalently, by applying data processing to the joint laws
$\nu(ds)P_s(dy)$ and $\nu(ds)Q_s(dy)$ and then marginalizing out $S$),
for any $\alpha>1$,
\[
\mathtt{D}_\alpha(P\|Q)\leq \operatorname*{ess\,sup}_{s} \mathtt{D}_\alpha(P_s\|Q_s).
\]
Combining with the bound from Step 1 yields
\[
\mathtt{D}_\alpha\!\left(\Pr(\langle Y_t,u_i\rangle\mid s_i,\theta)\ \big\|\ 
\Pr(\langle Y_t,u_i\rangle\mid s_j,\theta)\right)
\leq \frac{\alpha}{2}\frac{h_{t,i}}{v_{t,i}}.
\]
Since $\theta$ and $(s_i,s_j)$ were arbitrary, the claim follows.
\qed

\section{Proof of Theorem \ref{thm:HUC_SRPP_SGD}}

Fix a slicing profile $\{\mathcal{U},\omega\}$ with
$\mathcal{U} = \{u_\ell\}_{\ell=1}^m \subset \mathbb{S}^{d-1}$ and $\omega \in \Delta(\mathcal{U})$.
Let $h = \{h_t\}_{t=1}^T$ be a sequence of valid HUC vectors, with $h_t = (h_{t,1},\dots,h_{t,m})$ for $t=1,\dots,T$.
That is, 
for every prior $\theta\in\Theta$,
every secret pair $(s_i,s_j)\in\mathcal{Q}$ with
$P^S_\theta(s_i),P^S_\theta(s_j)>0$, every history $y_{<t}$ in the support, and every coupling
$\gamma\in\Pi(\mu^{\theta}_{s_i},\mu^{\theta}_{s_j})$, we have
$\gamma\times\mathbb{P}_{\eta,\rho}$-almost surely
\begin{equation}\label{eq:HUC_again_here}
  \bigl|\langle f_t(X,y_{<t};R_t)-f_t(X',y_{<t};R_t),u_\ell\rangle\bigr|^2\leq  h_{t,\ell},
\end{equation}
for all $\ell=1,\dots,m$.

Let the SRPP-SGD mechanism use additive Gaussian noise
\[
  N_t \sim \mathcal{N}(0,\sigma^2 I_d), \quad t=1,\dots,T,
\]
independent across $t$ and independent of $(X,X',R_t)$,
and write the noisy update at step $t$ as
\[
  Y_t = f_t(X,y_{<t};R_t) + N_t.
\]

\medskip
\noindent\textbf{Step 1. }
Fix $\alpha>1$, $\theta\in\Theta$, and $(s_i,s_j)\in\mathcal{Q}$ with positive prior mass.
Fix an arbitrary history $y_{<t}$ in the support.
Fix an arbitrary coupling $\gamma\in\Pi(\mu^\theta_{s_i},\mu^\theta_{s_j})$.

Since $h_t$ is a valid HUC vector, for $\gamma\times\mathbb{P}_{\eta,\rho}$-a.s.\ in $((X,X'),R_t)$,
\[
\bigl|\langle f_t(X,y_{<t};R_t)-f_t(X',y_{<t};R_t),u_\ell\rangle\bigr|^2\leq h_{t,\ell},
\quad \forall \ell\in[m].
\]
By Fubini/Tonelli, there exists a measurable set $\mathcal{R}_t$ with $\Pr(R_t\in\mathcal{R}_t)=1$
such that for every $r\in\mathcal{R}_t$ the same bound holds $\gamma$-a.s.\ in $(X,X')$:
\[
\bigl|\langle f_t(X,y_{<t};r)-f_t(X',y_{<t};r),u_\ell\rangle\bigr|^2\leq h_{t,\ell}.
\]

Condition on $R_t=r\in\mathcal{R}_t$. Then $\langle Y_t,u_\ell\rangle$ is Gaussian with variance
$v_{t,\ell}=u_\ell^\top(\sigma^2 I_d)u_\ell=\sigma^2$ under both secrets, and the mean shift is
$\langle f_t(X,y_{<t};r)-f_t(X',y_{<t};r),u_\ell\rangle$.
Hence Lemma~\ref{lemma:gaussian_SRE} yields, for all $r\in\mathcal{R}_t$,
\begin{equation}\label{eq:per_step_slice_bound_cond_r}
\begin{aligned}
    \mathtt{D}_\alpha\!\Big(&
\Pr(\langle Y_t,u_\ell\rangle \mid s_i,\theta,y_{<t},R_t=r)\\&
\ \Big\|\ 
\Pr(\langle Y_t,u_\ell\rangle \mid s_j,\theta,y_{<t},R_t=r)
\Big)
\leq \frac{\alpha}{2\sigma^2}\,h_{t,\ell}.
\end{aligned}
\end{equation}

Now view $(R_t,\langle Y_t,u_\ell\rangle)$ as the released pair at round $t$.
Because $R_t$ is independent of the secret (and of $y_{<t}$), it has the same conditional law under $s_i$ and $s_j$.
Using \eqref{eq:per_step_slice_bound_cond_r} (which holds for $R_t$-a.e.\ $r$) and the definition of R\'enyi divergence,
we obtain
\[
\begin{aligned}
    &\mathtt{D}_\alpha\!\Big(
\Pr(R_t,\langle Y_t,u_\ell\rangle \mid s_i,\theta,y_{<t})
\ \Big\|\ 
\Pr(R_t,\langle Y_t,u_\ell\rangle \mid s_j,\theta,y_{<t})
\Big)\\&
\leq \frac{\alpha}{2\sigma^2}\,h_{t,\ell}.
\end{aligned}
\]
Finally, marginalizing out $R_t$ is post-processing, so by data processing,
\begin{equation}\label{eq:per_step_slice_bound}
\mathtt{D}_\alpha\!\Big(
\Pr(\langle Y_t,u_\ell\rangle \mid s_i,\theta,y_{<t})
\ \Big\|\ 
\Pr(\langle Y_t,u_\ell\rangle \mid s_j,\theta,y_{<t})
\Big)
\leq \frac{\alpha}{2\sigma^2}\,h_{t,\ell}.
\end{equation}
This bound holds uniformly over all histories $y_{<t}$ in the support.

\medskip
\noindent\textbf{Step 2.}
Fix a slice index $\ell\in[m]$ and consider the 1-D projected mechanism
\[
  \mathcal{M}^{(\ell)}: X  \longmapsto  Z^{(\ell)}_{1:T} :=\bigl(\langle Y_1,u_\ell\rangle,\dots,\langle Y_T,u_\ell\rangle\bigr).
\]
For each $t$, denote by $\mathrm{P}^{(\ell)}_{t,y_{<t}}$ and $\mathrm{Q}^{(\ell)}_{t,y_{<t}}$ the conditional probability measure of $\langle Y_t,u_\ell\rangle$ given
$Z^{(\ell)}_{<t}=y_{<t}$ under secrets $s_i$ and $s_j$, respectively.
From the per–step bound \eqref{eq:per_step_slice_bound}, we have, for
every history $y_{<t}$ in the support,
\[
  \mathtt{D}_\alpha\Bigl(\Pr\bigl(\langle Y_t,u_\ell\rangle \mid s_i,\theta,y_{<t}\bigr)\,\Big\|\,\Pr\bigl(\langle Y_t,u_\ell\rangle \mid s_j,\theta,y_{<t}\bigr)
  \Bigr)\leq\frac{\alpha}{2\sigma^2}\,h_{t,\ell}.
\]
Equivalently,
\begin{equation}\label{eq:per_step_conditional_slice}
  \mathtt{D}_\alpha\bigl(P^{(\ell)}_{t,y_{<t}}
    \,\big\|\,Q^{(\ell)}_{t,y_{<t}}\bigr)\leq
  \varepsilon_{t,\ell},
  \quad
  \varepsilon_{t,\ell}
  := \frac{\alpha}{2\sigma^2}\,h_{t,\ell}.
\end{equation}
Let
\[
  \mathrm{P}^{(\ell)} := \Pr\bigl(Z^{(\ell)}_{1:T} \mid s_i,\theta\bigr),
  \quad
  \mathrm{Q}^{(\ell)} := \Pr\bigl(Z^{(\ell)}_{1:T} \mid s_j,\theta\bigr)
\]
denote the trajectory probability measures of the projected mechanism along slice
$u_\ell$.

The family
$\{\mathrm{P}^{(\ell)}_{t,y_{<t}}, \mathrm{Q}^{(\ell)}_{t,y_{<t}}\}_{t,y_{<t}}$
satisfies the assumptions of Lemma \ref{lemma:adaptive_renyi_composition} (stated later)
with per–step bounds $\varepsilon_t = \varepsilon_{t,\ell}$.
Applying Lemma~\ref{lemma:adaptive_renyi_composition} yields
\begin{equation}\label{eq:slice_traj_cost}
  \begin{aligned}
      \mathtt{D}_\alpha\Bigl(\Pr\bigl(Z^{(\ell)}_{1:T} \mid s_i,\theta\bigr)\,\Big\|\,\Pr\bigl(Z^{(\ell)}_{1:T} \mid s_j,\theta\bigr)\Bigr)&=\mathtt{D}_{\alpha}\bigl(\mathrm{P}^{(\ell)} \,\big\|\, \mathrm{Q}^{(\ell)}\bigr)\\
      &\leq
  \sum_{t=1}^T \varepsilon_{t,\ell}=
  \frac{\alpha}{2\sigma^2} \sum_{t=1}^T h_{t,\ell}.
  \end{aligned}
\end{equation}

Note that (\ref{eq:slice_traj_cost}) is derived for the full projected trajectory $Z^{(\ell)}_{1:T}$.
If Algorithm \ref{alg:srpp-sgd} outputs only the final iterate (or any measurable function of the trajectory), the corresponding R\'enyi divergence can only decrease by data processing, so the same bound
(\ref{eq:slice_traj_cost}) applies to the actual output distributions.

\medskip
\noindent\textbf{Step 3.}
By the definition of $\omega$–Average Sliced R\'enyi Divergence (Definition~\ref{def:ave_srd}), applied to the output probability measure of the mechanism
(identified with the full projected trajectories), we have
\[
  \begin{aligned}
      &\mathtt{AveSD}^{\omega}_\alpha\Bigl(\Pr(M(X)\mid s_i,\theta)\,\Big\|\,\Pr(M(X)\mid s_j,\theta)\Bigr)\\
      &=
  \sum_{\ell=1}^m \omega_{\ell}
  \mathtt{D}_\alpha\Bigl(\Pr\bigl(Z^{(\ell)}_{1:T} \mid s_i,\theta\bigr)\,\Big\|\,\Pr\bigl(Z^{(\ell)}_{1:T} \mid s_j,\theta\bigr)
  \Bigr).
  \end{aligned}
\]
Using the per–slice bound (\ref{eq:slice_traj_cost}), we get
\[
\begin{aligned}
    &\mathtt{AveSD}^{\omega}_\alpha\Bigl(\Pr(M(X)\mid s_i,\theta)\,\Big\|\,\Pr(M(X)\mid s_j,\theta)
  \Bigr)\\&\leq\sum_{\ell=1}^m \omega_{\ell}
  \frac{\alpha}{2\sigma^2} \sum_{t=1}^T h_{t,\ell} \\
  &= \frac{\alpha}{2\sigma^2}\sum_{t=1}^T \sum_{\ell=1}^m \omega_{\ell} h_{t,\ell}.
\end{aligned}
\]
Therefore, if
\[
  \sigma^2 \geq \frac{\alpha}{2\epsilon}
  \sum_{t=1}^T \sum_{\ell=1}^m \omega_{\ell} h_{t,\ell},
\]
then for every prior $\theta$ and every $(s_i,s_j)\in\mathcal{Q}$ with
positive prior mass, the $\omega$–Ave-SRD between the corresponding
output distributions is at most $\epsilon$, i.e., Algorithm \ref{alg:srpp-sgd}
is $(\alpha,\epsilon,\omega)$–Ave-SRPP.

\medskip
\noindent\textbf{Step 4.}
For Joint-SRPP, we use the joint sliced R\'enyi divergence
(Definition \ref{def:joint_srd}), which for the two output distributions can be
written as
\[
  \begin{aligned}
      &\mathtt{JSD}^{\omega}_{\alpha}\Bigl(\Pr(M(X)\mid s_i,\theta)\,\Big\|\,\Pr(M(X)\mid s_j,\theta)
  \Bigr)\\
  &=
  \frac{1}{\alpha-1}
  \log \mathbb{E}_{U\sim\omega}\!\Bigl[
    \exp\bigl((\alpha-1)\mathtt{D}_\alpha(\mathrm{P}_U\|\mathrm{Q}_U)\bigr)
  \Bigr],
  \end{aligned}
\]
where, for each $u_\ell\in\mathcal{U}$, $\mathrm{P}_{u_\ell}$ and $\mathrm{Q}_{u_\ell}$
denote the distributions of $Z^{(\ell)}_{1:T}$ under $s_i$ and $s_j$, respectively,
and $\mathrm{P}_U,\mathrm{Q}_U$ are the corresponding random choices when $U\sim\omega$.

From (\ref{eq:slice_traj_cost}), for each fixed slice $u_\ell$ we have
\[
  \mathtt{D}_{\alpha}(\mathrm{P}_{u_\ell}\|\mathrm{Q}_{u_\ell})\leq
  \frac{\alpha}{2\sigma^2}\sum_{t=1}^T h_{t,\ell}.
\]
Thus, for the random $U\sim\omega$, it holds almost surely that
\[
  \mathtt{D}_\alpha(P_U\|Q_U)\leq
  \frac{\alpha}{2\sigma^2}\sum_{t=1}^T h_{t,U}\leq
  \frac{\alpha}{2\sigma^2}\sum_{t=1}^T \max_{\ell} h_{t,\ell}.
\]
Plugging this into the definition of Joint-SRD yields
\[
\begin{aligned}
    &\mathtt{JSD}^{\omega}_{\alpha}\Bigl(
    \Pr(M(X)\mid s_i,\theta)\,\Big\|\,
    \Pr(M(X)\mid s_j,\theta)\Bigr)\\
  &=\frac{1}{\alpha-1}
  \log \mathbb{E}_{U\sim\omega}\Bigl[
    \exp\bigl((\alpha-1)\mathtt{D}_\alpha(\mathrm{P}_U\|\mathrm{Q}_U)\bigr)
  \Bigr] \\
  &\leq
  \frac{1}{\alpha-1}
  \log \exp\Bigl(
    (\alpha-1)\frac{\alpha}{2\sigma^2}
    \sum_{t=1}^T \max_{\ell} h_{t,\ell}
  \Bigr) \\
  &= \frac{\alpha}{2\sigma^2}
     \sum_{t=1}^T \max_{\ell} h_{t,\ell}.
\end{aligned}
\]
Therefore, if
\[
  \sigma^2 \geq \frac{\alpha}{2\epsilon}
  \sum_{t=1}^T \max_{\ell} h_{t,\ell},
\]
then for every prior $\theta$ and every $(s_i,s_j)\in\mathcal Q$ the
Joint-SRD between the corresponding output distributions is at most $\epsilon$,
i.e., Algorithm~\ref{alg:srpp-sgd} is $(\alpha,\epsilon,\omega)$–Joint-SRPP.
\qed

\begin{lemma}[Adaptive composition of R\'enyi divergence]
\label{lemma:adaptive_renyi_composition}
Let $\alpha>1$ and let $\mathrm{P}$ and $\mathrm{Q}$ be probability measures on the
product space $\mathcal{Y}_1 \times \cdots \times \mathcal{Y}_T$ with
coordinate random variables $Y_1,\dots,Y_T$.
For each $t=1,\dots,T$, let
$\mathrm{P}_{t \mid y_{<t}}$ and $\mathrm{Q}_{t \mid y_{<t}}$ denote regular conditional
distributions of $Y_t$ given $Y_{<t}=y_{<t}$ under $P$ and $Q$,
respectively, where $y_{<t}=(y_1,\dots,y_{t-1})$.
Suppose there exist constants $\varepsilon_t \geq 0$ such that
\[
  \mathtt{D}_\alpha\bigl(P_{t \mid y_{<t}} \,\big\|\, Q_{t \mid y_{<t}}\bigr)\leq \varepsilon_t,
\]
for all $t=1,\dots,T$ and all histories $y_{<t}$ in the support.
Then the joint R\'enyi divergence satisfies
\[
  \mathtt{D}_\alpha(\mathrm{P} \,\|\, \mathrm{Q})\leq
  \sum_{t=1}^T \varepsilon_t.
\]
\end{lemma}

\begin{proof}[Proof of Lemma \ref{lemma:adaptive_renyi_composition}]

Let $L := \frac{d\mathrm{P}}{d\mathrm{Q}}$ be the Radon--Nikodym derivative of $\mathrm{P}$ w.r.t.\ $\mathrm{Q}$,
assumed to exist (otherwise $\mathtt{D}_\alpha(\mathrm{P}\|\mathrm{Q})=+\infty$ and the
inequality is straightforward).
For each $t$, define the conditional likelihood ratio
\[
  L_t(y_{\leq t}):=\frac{d\mathrm{P}_{t \mid y_{<t}}}{d\mathrm{Q}_{t \mid y_{<t}}}(y_t), \quad y_{\leq t} = (y_1,\dots,y_t).
\]
By the chain rule for densities,
\[
  L(Y_{1:T})= \prod_{t=1}^T L_t(Y_{\leq t})
\]
$\mathrm{Q}$-almost surely.

By definition of R\'enyi divergence,
\[
  \begin{aligned}
      &\mathtt{D}_\alpha(\mathrm{P} \,\|\, \mathrm{Q})=\frac{1}{\alpha-1}
  \log \mathbb{E}_{\mathrm{Q}} \bigl[ L(Y_{1:T})^{\alpha} \bigr]\\&
  =\frac{1}{\alpha-1}
  \log \mathbb{E}_{\mathrm{Q}} \Bigl[ \prod_{t=1}^T L_t(Y_{\leq t})^{\alpha} \Bigr].
  \end{aligned}
\]
For each $t$, define
\[
  M_t:=\sup_{y_{<t}}\mathbb{E}_{\mathrm{Q}} \bigl[L_t(Y_{\leq t})^\alpha \,\big|\, Y_{<t}=y_{<t}\bigr].
\]
Under our assumption,
\[
  \frac{1}{\alpha-1}\log \mathbb{E}_{\mathrm{Q}}\bigl[L_t(Y_{\leq t})^\alpha \,\big|\, Y_{<t}=y_{<t}\bigr]=
  \mathtt{D}_\alpha\bigl(\mathrm{P}_{t \mid y_{<t}} \,\big\|\, \mathrm{Q}_{t \mid y_{<t}}\bigr)
  \leq \varepsilon_t.
\]
Thus, 
\[
  \mathbb{E}_{\mathrm{Q}}\bigl[
    L_t(Y_{\leq t})^\alpha \,\big|\, Y_{<t}=y_{<t}
  \bigr]\leq
  e^{(\alpha-1)\varepsilon_t},
\]
for all $y_{<t}$, and hence
\[
  M_t \leq e^{(\alpha-1)\varepsilon_t}.
\]
We now bound $\mathbb{E}_{\mathrm{Q}}[\prod_{t=1}^T L_t^{\alpha}]$ iteratively using the tower property:
\[
\begin{aligned}
    \mathbb{E}_{\mathrm{Q}}\Bigl[ \prod_{t=1}^T L_t^\alpha \Bigr]
  &= \mathbb{E}_{\mathrm{Q}}\Bigl[
      \mathbb{E}_{\mathrm{Q}}\bigl[ L_T^{\alpha} \,\big|\, Y_{<T} \bigr]
      \prod_{t=1}^{T-1} L_t^{\alpha}
    \Bigr] \\
  &\leq M_T \mathbb{E}_{\mathrm{Q}}\Bigl[ \prod_{t=1}^{T-1} L_t^{\alpha} \Bigr].
\end{aligned}
\]
Repeating this backward for $t=T-1,\dots,1$ yields
\[
  \mathbb{E}_{\mathrm{Q}}\Bigl[ \prod_{t=1}^T L_t^{\alpha} \Bigr]
  \leq \prod_{t=1}^T M_t\leq \prod_{t=1}^T e^{(\alpha-1)\varepsilon_t}=e^{(\alpha-1)\sum_{t=1}^T \varepsilon_t}.
\]
Taking $\frac{1}{\alpha-1}\log(\cdot)$ on both sides gives
\[
  \mathtt{D}_\alpha(\mathrm{P} \,\|\, \mathrm{Q})
  =
  \frac{1}{\alpha-1}
  \log \mathbb{E}_{\mathrm{Q}}\Bigl[ \prod_{t=1}^T L_t^\alpha \Bigr]\leq \sum_{t=1}^T \varepsilon_t,
\]
as claimed.

\end{proof}

\section{Proof of Theorem~\ref{thm:msHUC_SRPP_SGD}}

Fix a slice profile $(\mathcal{U},\omega)$ with $\mathcal{U}=\{u_\ell\}_{\ell=1}^m$ and
$\omega\in\Delta(\mathcal{U})$. Fix $\alpha>1$. Fix an arbitrary prior $\theta\in\Theta$
and an arbitrary secret pair $(s_i,s_j)\in\mathcal{Q}$.

Let $\eta=(\eta_1,\dots,\eta_T)$ denote the (possibly time-varying) subsampling parameters
used by Algorithm~\ref{alg:srpp-sgd} (e.g., sampling rates).
Let $R_t$ denote the \textit{realized subsampling outcome} at round $t$ (e.g., the sampled index multiset),
and let $\mathbf{R}:=(R_1,\dots,R_T)$.
For each realization $\mathbf{r}$ of $\mathbf{R}$, let $\mathcal{M}^{\theta,\mathbf{r}}_{s}$ denote the
conditional law of the final output of Algorithm~\ref{alg:srpp-sgd} under prior $\theta$ and secret $s$,
conditional on $\mathbf{R}=\mathbf{r}$.

Throughout, the Gaussian noises are i.i.d.\ $N_t\sim\mathcal{N}(0,\sigma^2 I_d)$, independent of all
data and of $\mathbf{R}$. Hence for every slice $u_\ell$,
\[
v_{t,\ell}:=u_\ell^\top(\sigma^2 I_d)u_\ell=\sigma^2.
\]

\begin{lemma}[Expected Gaussian sliced R\'enyi cost]\label{lem:sa_gaussian_slice}
Fix iteration $t$ and direction $u\in\mathbb{S}^{d-1}$.
Let $N_t\sim\mathcal{N}(0,\sigma^2 I_d)$ be independent of all data and of $R_t$.
Fix $\theta\in\Theta$, $(s_i,s_j)\in\mathcal{Q}$, a history $y_{<t}$ in the support, and a coupling
$\gamma\in\Pi(\mu_{s_i}^\theta,\mu_{s_j}^\theta)$.
For $\gamma$-a.e.\ $(X,X')$, define the random directional shift
\[
\Delta_t^u(X,X';R_t):=\big\langle f_t(X,y_{<t};R_t)-f_t(X',y_{<t};R_t),u\big\rangle.
\]
Then, for any $\alpha>1$,
\[
\begin{aligned}
    &\mathbb{E}_{R_t}\!\left[
\mathtt{D}_\alpha\!\left(
\Pr(\langle Y_t,u\rangle \mid s_i,\theta,R_t)\ \Big\|\ 
\Pr(\langle Y_t,u\rangle \mid s_j,\theta,R_t)
\right)\right]\\&
=
\frac{\alpha}{2\sigma^2}\,
\mathbb{E}_{R_t}\!\left[\big|\Delta_t^u(X,X';R_t)\big|^2\right],
\end{aligned}
\]
where $Y_t=f_t(X,y_{<t};R_t)+N_t$.
\end{lemma}

\begin{proof}
Fix $t,u,\theta,(s_i,s_j),y_{<t},\gamma$ as stated and take $\gamma$-a.e.\ $(X,X')$.
Condition on $R_t=r$. Then $\langle Y_t,u\rangle$ is Gaussian with variance $\sigma^2$ and mean
$\langle f_t(X,y_{<t};r),u\rangle$, while the corresponding mean under secret $s_j$ is
$\langle f_t(X',y_{<t};r),u\rangle$. Hence the conditional mean shift equals
$\Delta_t^u(X,X';r)$ and the standard Gaussian identity gives
\[
\begin{aligned}
    &\mathtt{D}_\alpha\!\left(
\Pr(\langle Y_t,u\rangle \mid s_i,\theta,R_t=r)\ \Big\|\ 
\Pr(\langle Y_t,u\rangle \mid s_j,\theta,R_t=r)
\right)\\&
=\frac{\alpha}{2\sigma^2}\,\big|\Delta_t^u(X,X';r)\big|^2.
\end{aligned}
\]
Taking expectation over $R_t$ yields the claim.
\end{proof}

\medskip
\noindent\textbf{Step 1. }
Fix an iteration $t\in\{1,\dots,T\}$ and a slice $u_\ell\in\mathcal{U}$.
Fix any history $y_{<t}$ in the support and any coupling
$\gamma\in\Pi(\mu^\theta_{s_i},\mu^\theta_{s_j})$.
Apply Lemma~\ref{lem:sa_gaussian_slice} with $u=u_\ell$:
for $\gamma$-a.e.\ $(X,X')$,
\[
\begin{aligned}
    &\mathbb{E}_{R_t}\!\left[
\mathtt{D}_\alpha\!\left(
\Pr(\langle Y_t,u_\ell\rangle \mid s_i,\theta,R_t)\ \Big\|\ 
\Pr(\langle Y_t,u_\ell\rangle \mid s_j,\theta,R_t)
\right)\right]\\&
=
\frac{\alpha}{2\sigma^2}\,
\mathbb{E}_{R_t}\!\left[
\big|\langle f_t(X,y_{<t};R_t)-f_t(X',y_{<t};R_t),u_\ell\rangle\big|^2
\right].
\end{aligned}
\]
Since $h_t^{\mathsf{sa}}$ is a valid sa-HUC vector (Definition~\ref{def:msHUC}), the expectation on the
right-hand side is at most $h^{\mathsf{sa}}_{t,\ell}$. Therefore,
\begin{equation}\label{eq:sa-perround-final}
\begin{aligned}
    &\mathbb{E}_{R_t}\!\left[
\mathtt{D}_\alpha\!\left(
\Pr(\langle Y_t,u_\ell\rangle \mid s_i,\theta,R_t)\ \Big\|\ 
\Pr(\langle Y_t,u_\ell\rangle \mid s_j,\theta,R_t)
\right)\right]\\&
\leq
\frac{\alpha}{2\sigma^2}\,h^{\mathsf{sa}}_{t,\ell}.
\end{aligned}
\end{equation}

\medskip
\noindent\textbf{Step 2 (sa-Ave-SRPP).}
Fix a realization $\mathbf{r}=(r_1,\dots,r_T)$ of the subsampling outcomes $\mathbf{R}$.
Apply the same conditional SRPP accountant (composition argument) as used for
Theorem~\ref{thm:HUC_SRPP_SGD} to the conditional mechanism $\mathcal{M}^{\theta,\mathbf{r}}$.
This yields the pointwise (in $\mathbf{r}$) bound
\begin{equation}\label{eq:cond-ave}
\begin{aligned}
    &\mathtt{AveSD}^{\omega}_{\alpha}\!\left(
\mathcal{M}^{\theta,\mathbf{r}}_{s_i}\ \Big\|\ \mathcal{M}^{\theta,\mathbf{r}}_{s_j}
\right)
\leq \sum_{t=1}^{T}\sum_{\ell=1}^{m}\omega_\ell\,\\&
\mathtt{D}_\alpha\!\left(
\Pr(\langle Y_t,u_\ell\rangle \mid s_i,\theta,R_t=r_t)\ \Big\|\ 
\Pr(\langle Y_t,u_\ell\rangle \mid s_j,\theta,R_t=r_t)
\right).
\end{aligned}
\end{equation}
Take expectation over $\eta$ and use linearity and the fact that the $t$-th summand depends only on $\eta_t$:
\[
\begin{aligned}
    &\mathbb{E}_{\eta}\!\left[
\mathtt{AveSD}^{\omega}_{\alpha}\!\left(
\mathcal{M}^{\theta,\eta}_{s_i}\ \Big\|\ \mathcal{M}^{\theta,\eta}_{s_j}
\right)\right]
\leq \sum_{t=1}^{T}\sum_{\ell=1}^{m}\omega_\ell\,
\mathbb{E}_{\eta_t}\!\Big[\\
&
\mathtt{D}_\alpha\!\left(
\Pr(\langle Y_t,u_\ell\rangle \mid s_i,\theta,R_t)\ \Big\|\ 
\Pr(\langle Y_t,u_\ell\rangle \mid s_j,\theta,R_t)
\right)\Big].
\end{aligned}
\]
Applying \eqref{eq:sa-perround-final} yields
\[
\mathbb{E}_{\eta}\!\left[
\mathtt{AveSD}^{\omega}_{\alpha}\!\left(
\mathcal{M}^{\theta,\eta}_{s_i}\ \Big\|\ \mathcal{M}^{\theta,\eta}_{s_j}
\right)\right]
\leq \frac{\alpha}{2\sigma^2}\sum_{t=1}^{T}\sum_{\ell=1}^{m}\omega_\ell\,h^{\mathsf{sa}}_{t,\ell}.
\]
Therefore, if
\[
\sigma^2 \geq \frac{\alpha}{2\epsilon}\sum_{t=1}^{T}\sum_{\ell=1}^{m}\omega_\ell\,h^{\mathsf{sa}}_{t,\ell},
\]
then the right-hand side is at most $\epsilon$. Since $\theta$ and $(s_i,s_j)$ were arbitrary,
Definition~\ref{def:ms_Ave_SRPP} implies that Algorithm~\ref{alg:srpp-sgd} satisfies
$(\alpha,\epsilon,\omega)$-sa-Ave-SRPP, proving (i).

\medskip
\noindent\textbf{Step 3 (sa-Joint-SRPP).}
Fix a realization $\mathbf{r}=(r_1,\dots,r_T)$ of the subsampling outcomes $\mathbf{R}$.
Apply the same conditional accountant as in Step~2 for the joint sliced divergence:
\begin{equation}\label{eq:cond-joint}
\begin{aligned}
    &\mathtt{JSD}^{\omega}_{\alpha}\!\left(
\mathcal{M}^{\theta,\mathbf{r}}_{s_i}\ \Big\|\ \mathcal{M}^{\theta,\mathbf{r}}_{s_j}
\right)
\leq \sum_{t=1}^{T}\max_{\ell\in[m]}\,
\mathtt{D}_\alpha\!\Big(
\Pr(\langle Y_t,u_\ell\rangle \mid s_i,\theta,R_t=r_t)\ \\&\Big\|\ 
\Pr(\langle Y_t,u_\ell\rangle \mid s_j,\theta,R_t=r_t)
\Big),
\end{aligned}
\end{equation}
Take expectation over $r_t$. and use linearity across $t$:
\[
\begin{aligned}
    &\mathbb{E}_{\eta}\!\left[
\mathtt{JSD}^{\omega}_{\alpha}\!\left(
\mathcal{M}^{\theta,\eta}_{s_i}\ \Big\|\ \mathcal{M}^{\theta,\eta}_{s_j}
\right)\right]
\leq \sum_{t=1}^{T}\mathbb{E}_{\eta_t}\!\Big[\\&
\max_{\ell\in[m]}\,
\mathtt{D}_\alpha\!\left(
\Pr(\langle Y_t,u_\ell\rangle \mid s_i,\theta,R_t)\ \Big\|\ 
\Pr(\langle Y_t,u_\ell\rangle \mid s_j,\theta,R_t)
\right)\Big].
\end{aligned}
\]

Now invoke the explicit SRPP--SGD sa-HUC instantiation from Proposition~\ref{prop:msHUC_from_K2}.
In that setting, for each $t$ and $\ell$ there is a deterministic coefficient
\[
a_{t,\ell}:=\Bigl(\frac{2L_{t,\ell}C}{B_t}\Bigr)^2
\]
and a (scalar) discrepancy random variable $K_t(R_t)$ (shared across slices at round $t$) such that
for $\gamma$-a.e.\ $(X,X')$ and all $\ell$,
\[
\big|\langle f_t(X,y_{<t};R_t)-f_t(X',y_{<t};R_t),u_\ell\rangle\big|^2
\leq a_{t,\ell}\,K_t(R_t)^2.
\]
Condition on $R_t=r$ and apply the same Gaussian identity as in Lemma~\ref{lem:sa_gaussian_slice} to get
\[
\begin{aligned}
    &\mathtt{D}_\alpha\!\left(
\Pr(\langle Y_t,u_\ell\rangle \mid s_i,\theta,R_t=r)\ \Big\|\ 
\Pr(\langle Y_t,u_\ell\rangle \mid s_j,\theta,R_t=r)
\right)\\&
=\frac{\alpha}{2\sigma^2}\,
\big|\langle f_t(X,y_{<t};r)-f_t(X',y_{<t};r),u_\ell\rangle\big|^2\\&
\leq \frac{\alpha}{2\sigma^2}\,a_{t,\ell}\,K_t(r)^2.
\end{aligned}
\]
Taking the maximum over $\ell$ yields, for each $r$,
\[
\max_{\ell\in[m]}\mathtt{D}_\alpha(\cdots\mid R_t=r)
\leq \frac{\alpha}{2\sigma^2}\,
\Bigl(\max_{\ell\in[m]}a_{t,\ell}\Bigr)\,K_t(r)^2.
\]
Now take expectation over $R_t$:
\[
\mathbb{E}_{R_t}\!\left[
\max_{\ell\in[m]}\mathtt{D}_\alpha(\cdots)\right]
\leq \frac{\alpha}{2\sigma^2}\,
\Bigl(\max_{\ell\in[m]}a_{t,\ell}\Bigr)\,
\mathbb{E}_{R_t}\!\left[K_t(R_t)^2\right].
\]
By the mean-square discrepancy cap \eqref{eq:ms_disc_cap},
$\mathbb{E}_{R_t}[K_t(R_t)^2]\le \overline{K}_t^2$, hence
\[
\mathbb{E}_{R_t}\!\left[
\max_{\ell\in[m]}\mathtt{D}_\alpha(\cdots)\right]
\leq \frac{\alpha}{2\sigma^2}\,
\Bigl(\max_{\ell\in[m]}a_{t,\ell}\Bigr)\,\overline{K}_t^2
=
\frac{\alpha}{2\sigma^2}\,\max_{\ell\in[m]}h^{\mathsf{sa}}_{t,\ell},
\]
where the last equality uses the closed form
$h^{\mathsf{sa}}_{t,\ell}=a_{t,\ell}\overline K_t^2$ from Proposition~\ref{prop:msHUC_from_K2}.
Substituting back,
\[
\mathbb{E}_{\eta}\!\left[
\mathtt{JSD}^{\omega}_{\alpha}\!\left(
\mathcal M^{\theta,\eta}_{s_i}\ \Big\|\ \mathcal{M}^{\theta,\eta}_{s_j}
\right)\right]
\leq \frac{\alpha}{2\sigma^2}\sum_{t=1}^{T}\max_{\ell\in[m]} h^{\mathsf{sa}}_{t,\ell}.
\]
Therefore, if
\[
\sigma^2 \geq \frac{\alpha}{2\epsilon}\sum_{t=1}^{T}\max_{\ell\in[m]} h^{\mathsf{sa}}_{t,\ell},
\]
then the right-hand side is at most $\epsilon$. Since $\theta$ and $(s_i,s_j)$ were arbitrary,
Definition~\ref{def:ms_Joint_SRPP} implies that Algorithm~\ref{alg:srpp-sgd} satisfies
$(\alpha,\epsilon,\omega)$-sa-Joint-SRPP, proving (ii).
\qed

\section{ Proof of Theorem \ref{thm:composition_srpp_sgd} }\label{app:proof_thm:composition_srpp_sgd}

Fix a slice profile $\{\mathcal{U},\omega\}$ and a Pufferfish scenario
$(\mathcal{S},\mathcal{Q},\Theta)$.
Let $\vec{\mathcal{M}}$ be the composed mechanism in~\eqref{eq:composition_mechanism}:
\[
  \vec{\mathcal{M}}(X)
  :=
  \bigl(\mathcal{M}_1(X),\dots,\mathcal{M}_J(X)\bigr),
\]
where $X$ is the underlying common dataset.
For every realized dataset $x$, we assume that the primitive randomness of
$\mathcal{M}_1,\dots,\mathcal{M}_J$ is independent across $j$ when the mechanisms are run on $x$.

Fix $\theta\in\Theta$ and $(s_i,s_j)\in\mathcal{Q}$ with
$P_\theta^S(s_i),P_\theta^S(s_j)>0$.
Let $\mu_{s_i}^\theta$ and $\mu_{s_j}^\theta$ be the corresponding dataset beliefs.
For any fixed dataset realization $x$ in the support of these beliefs, define the conditional output distribution
\[
\begin{aligned}
    &\mathrm{P}_{j}^{x}:= \Pr\bigl(\mathcal{M}_j(x)\in\cdot \,\big|\, s_i,\theta\bigr),\\
    &\mathrm{Q}_{j}^{x}:= \Pr\bigl(\mathcal{M}_j(x)\in\cdot \,\big|\, s_j,\theta\bigr),
\end{aligned}
\]
and
\[
\begin{aligned}
    &\mathrm{P}^{\mathrm{comp},x}:= \Pr\bigl(\vec{\mathcal{M}}(x)\in\cdot \,\big|\, s_i,\theta\bigr),\\
    &\mathrm{Q}^{\mathrm{comp},x}:= \Pr\bigl(\vec{\mathcal{M}}(x)\in\cdot \,\big|\, s_j,\theta\bigr).
\end{aligned}
\]
By independence of the primitive randomness across $j$ for this fixed $x$,
we have
\[
  \mathrm{P}^{\mathrm{comp},x}
  = \bigotimes_{j=1}^J \mathrm{P}_{j}^{x},
  \quad \textup{and}\quad
  \mathrm{Q}^{\mathrm{comp},x}
  = \bigotimes_{j=1}^J \mathrm{Q}_{j}^{x}.
\]
For any R\'enyi order $\alpha>1$ and product measures
$\bigotimes_{j=1}^J \mathrm{P}_j$ and $\bigotimes_{j=1}^J \mathrm{Q}_j$,
the R\'enyi divergence becomes
\begin{equation}\label{eq:renyi_tensorization}
  \mathtt{D}_{\alpha}\Bigl(
    {\textstyle\bigotimes_{j=1}^J} \mathrm{P}_j
    \,\Big\|\,
    {\textstyle\bigotimes_{j=1}^J} \mathrm{Q}_j
  \Bigr)=\sum_{j=1}^J \mathtt{D}_{\alpha}(\mathrm{P}_j\|\mathrm{Q}_j).
\end{equation}
Applying this with $\mathrm{P}_j = \mathrm{P}_{j}^{x}$ and $\mathrm{Q}_{j} = \mathrm{Q}_{j}^{x}$ yields, for every fixed
dataset realization $x$,
\[
\mathtt{D}_{\alpha}\bigl(\mathrm{P}^{\mathrm{comp},x}\big\|\mathrm{Q}^{\mathrm{comp},x}\bigr)=\sum_{j=1}^J \mathtt{D}_{\alpha}(\mathrm{P}_{j}^{d}\|\mathrm{Q}_{j}^{x}).
\]
We now prove the composition bound for each SRPP flavor $\Xi$.

\medskip\noindent\textbf{Case 1: $\Xi\in\{\textup{Ave},\textup{ms-Ave}\}$.}
Fix $\theta\in\Theta$ and $(s_i,s_j)\in\mathcal{Q}$ with $P_\theta^S(s_i),P_\theta^S(s_j)>0$.
Fix any coupling $\gamma\in\Pi(\mu_{s_i}^\theta,\mu_{s_j}^\theta)$ and take $(X,X')\sim\gamma$.
For each $j$, write $W_j:=\mathcal{M}_j(X)$ and $W'_j:=\mathcal{M}_j(X')$.

Fix a slice $u_\ell\in\mathcal{U}$ and consider the 1-D release sequence
\[
\begin{aligned}
    &Z^{(\ell)}_{1:J}:=\big(\langle W_1,u_\ell\rangle,\dots,\langle W_J,u_\ell\rangle\big),
\\&
{Z'}^{(\ell)}_{1:J}:=\big(\langle W'_1,u_\ell\rangle,\dots,\langle W'_J,u_\ell\rangle\big).
\end{aligned}
\]
Let $\mathrm{P}^{(\ell)}:=\Pr(Z^{(\ell)}_{1:J}\in\cdot\mid s_i,\theta)$ and
$\mathrm{Q}^{(\ell)}:=\Pr({Z'}^{(\ell)}_{1:J}\in\cdot\mid s_j,\theta)$.

\smallskip
\noindent\textit{Per-step conditional slice bounds.}
Because each $\mathcal{M}_j$ is an SRPP-SGD mechanism whose SRPP guarantee is obtained from an HUC/sa-HUC
sequence (Theorem~\ref{thm:HUC_SRPP_SGD} or Theorem~\ref{thm:msHUC_SRPP_SGD}),
its per-slice R\'enyi bound is uniform over all histories and (coupled) dataset pairs.
In particular, for every $j\in[J]$, every $\ell\in[m]$, and every history
$z^{(\ell)}_{<j}$ in the support, we have
\begin{equation}\label{eq:case1_cond_slice_j}
\begin{aligned}
    \mathtt{D}_\alpha\!\Big(&
\Pr(\langle W_j,u_\ell\rangle\in\cdot \mid s_i,\theta, Z^{(\ell)}_{<j}=z^{(\ell)}_{<j})
\\& \big\|\ 
\Pr(\langle W'_j,u_\ell\rangle\in\cdot \mid s_j,\theta, {Z'}^{(\ell)}_{<j}=z^{(\ell)}_{<j})
\Big)
\leq \varepsilon_{j,\ell},
\end{aligned}
\end{equation}
for some nonnegative numbers $\{\varepsilon_{j,\ell}\}_{\ell=1}^m$ satisfying
\begin{equation}\label{eq:case1_eps_avg_j}
\sum_{\ell=1}^m \omega_\ell\,\varepsilon_{j,\ell}\leq \epsilon_j.
\end{equation}
(Here $\varepsilon_{j,\ell}$ is the per-slice budget induced by the HUC/sa-HUC-based bound for $\mathcal M_j$.)

Applying Lemma~\ref{lemma:adaptive_renyi_composition} to the sequence $Z^{(\ell)}_{1:J}$
(using the per-step bounds \eqref{eq:case1_cond_slice_j}) yields, for each fixed $\ell$,
\[
\mathtt{D}_\alpha(\mathrm{P}^{(\ell)}\|\mathrm{Q}^{(\ell)})\leq \sum_{j=1}^J \varepsilon_{j,\ell}.
\]
Now average over slices:
\[
\begin{aligned}
    &\mathtt{AveSD}^{\omega}_{\alpha}\bigl(
\Pr(\vec{\mathcal M}(X)\mid s_i,\theta)\,\big\|\,
\Pr(\vec{\mathcal M}(X)\mid s_j,\theta)
\bigr)\\&
=
\sum_{\ell=1}^m \omega_\ell\,\mathtt{D}_\alpha(\mathrm P^{(\ell)}\|\mathrm Q^{(\ell)})
\leq
\sum_{\ell=1}^m \omega_\ell \sum_{j=1}^J \varepsilon_{j,\ell}\\&
=
\sum_{j=1}^J \sum_{\ell=1}^m \omega_\ell \varepsilon_{j,\ell}
\leq \sum_{j=1}^J \epsilon_j,
\end{aligned}
\]
where the last inequality uses \eqref{eq:case1_eps_avg_j}.
Since $\theta$ and $(s_i,s_j)$ were arbitrary, this proves
$(\alpha,\sum_{j=1}^J\epsilon_j,\omega)$-\textup{Ave}-SRPP for $\vec{\mathcal{M}}$.

For \textup{ms-Ave}-SRPP, the same argument applies verbatim after folding the subsampling randomness
into each mechanism's output (i.e., applying the above proof to the expanded-output channels),
so $\vec{\mathcal{M}}$ is also $(\alpha,\sum_{j=1}^J\epsilon_j,\omega)$-\textup{ms-Ave}-SRPP.

\medskip\noindent\textbf{Case 2: $\Xi\in\{\textup{Joint},\textup{ms-Joint}\}$.}
Fix $x$ in the support of $\mu_{s_i}^\theta$ and $\mu_{s_j}^\theta$.
Recall that Joint-SRD can be represented as the R\'enyi divergence of a randomized sliced channel.
Define the (one-shot) slicing operator
\[
  \mathtt{Slice}(z)
  := (U,\langle z,U\rangle),
  \qquad U\sim\omega,
\]
where $U$ is independent of all other randomness.
For a mechanism $\mathcal{M}$, write $\mathcal{N}:=\mathtt{Slice}\circ \mathcal{M}$.
Then, by the definition/representation of Joint-SRD (\ Definition~\ref{def:joint_srd}
and Appendix~\ref{app:ave-vs-joint-srpp}),
\begin{equation}\label{eq:jsd_as_rpp_slice}
\begin{aligned}
      &\mathtt{JSD}^{\omega}_{\alpha}\bigl(
    \Pr(\mathcal M(X)\mid s_i,\theta)\,\big\|\,
    \Pr(\mathcal M(X)\mid s_j,\theta)
  \bigr)\\&
  =
  \mathtt{D}_{\alpha}\bigl(
    \Pr(\mathcal N(X)\mid s_i,\theta)\,\big\|\,
    \Pr(\mathcal N(X)\mid s_j,\theta)
  \bigr).
\end{aligned}
\end{equation}

Now define, for each $j$, the sliced channel $\mathcal{N}_j:=\mathtt{Slice}\circ \mathcal{M}_j$.
Since $\mathcal{M}_j$ is $(\alpha,\epsilon_j,\omega)$-Joint-SRPP, \eqref{eq:jsd_as_rpp_slice} implies
\[
  \mathtt{D}_{\alpha}\bigl(
    \Pr(\mathcal{N}_j(X)\mid s_i,\theta)\,\big\|\,
    \Pr(\mathcal{N}_j(X)\mid s_j,\theta)
  \bigr)\leq \epsilon_j.
\]

Consider the \textit{componentwise} slicing of the composed mechanism:
sample $U_1,\dots,U_J\overset{\text{i.i.d.}}{\sim}\omega$ independently of everything, and define
\[
  \mathtt{Slice}^{J}\circ \vec{\mathcal M}(X)
  :=
  \bigl( (U_1,\langle \mathcal{M}_1(X),U_1\rangle),\dots,(U_J,\langle \mathcal{M}_J(X),U_J\rangle) \bigr).
\]
Equivalently, $\mathtt{Slice}^J\circ \vec{\mathcal{M}}(X)=(\mathcal{N}_1(X),\dots,\mathcal{N}_J(X))$.
By the assumed independence of the primitive randomness across $j$ (and independence of $U_j$),
the conditional laws under $s_i$ and $s_j$ factorize:
\[
  \Pr\bigl((\mathcal{N}_1(x),\dots,\mathcal{N}_J(x))\mid s_i,\theta\bigr)
  =\bigotimes_{j=1}^J \Pr\bigl(\mathcal{N}_j(x)\mid s_i,\theta\bigr),
\]
and similarly under $s_j$.
Therefore, by R\'enyi tensorization,
\[
\begin{aligned}
    &\mathtt{D}_{\alpha}\Bigl(
  \Pr\bigl(\mathtt{Slice}^J\circ\vec{\mathcal{M}}(x)\mid s_i,\theta\bigr)
  \,\Big\|\,
  \Pr\bigl(\mathtt{Slice}^J\circ\vec{\mathcal{M}}(x)\mid s_j,\theta\bigr)
\Bigr)\\&
=
\sum_{j=1}^J
\mathtt{D}_{\alpha}\Bigl(
  \Pr\bigl(\mathcal{N}_j(x)\mid s_i,\theta\bigr)
  \,\Big\|\,
  \Pr\bigl(\mathcal{N}_j(x)\mid s_j,\theta\bigr)
\Bigr)
\leq \sum_{j=1}^J \epsilon_j.
\end{aligned}
\]

Finally, by applying the representation \eqref{eq:jsd_as_rpp_slice} to the composed mechanism
(with $\mathtt{Slice}^J$ as the slicing operator on $(\mathbb{R}^d)^J$),
the left-hand side is exactly the Joint-SRD of $\vec{\mathcal{M}}$ under the same slice profile $\omega$.
Hence $\vec{\mathcal{M}}$ is $(\alpha,\sum_{j=1}^J\epsilon_j,\omega)$-Joint-SRPP.

For ms-Joint-SRPP, the same argument applies after folding the subsampling randomness into each
$\mathcal{M}_j$ (equivalently, applying the above proof to the corresponding expanded-output channels).

Combining Case 1 and Case 2, we conclude that for any
$\Xi\in\{\textup{Ave},\textup{Joint},\textup{ms-Ave},\textup{ms-Joint}\}$,
if each $\mathcal{M}_j$ is $(\alpha,\epsilon_j,\omega)$-$\Xi$-SRPP,
then the composed mechanism $\vec{\mathcal{M}}$ is
$(\alpha,\sum_{j=1}^J\epsilon_j,\omega)$-$\Xi$-SRPP.
\qed

\begin{figure*}[t]
    \centering

    \begin{subfigure}[b]{0.24\textwidth}
        \centering
        \includegraphics[width=\textwidth]{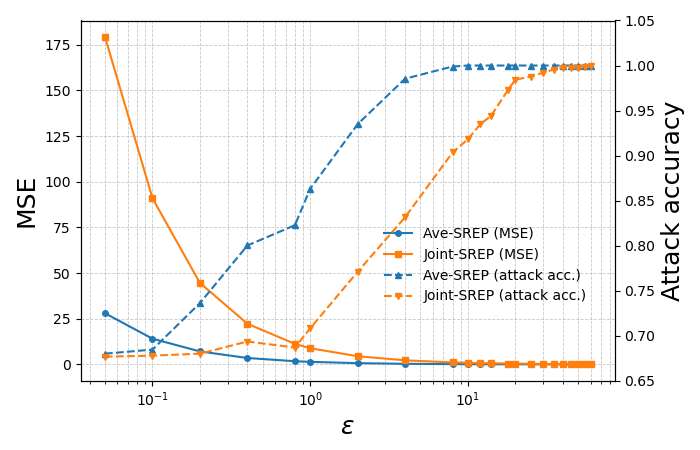}
        \subcaption{Logistic Regression}
    \end{subfigure}
    \begin{subfigure}[b]{0.24\textwidth}
        \centering
        \includegraphics[width=\textwidth]{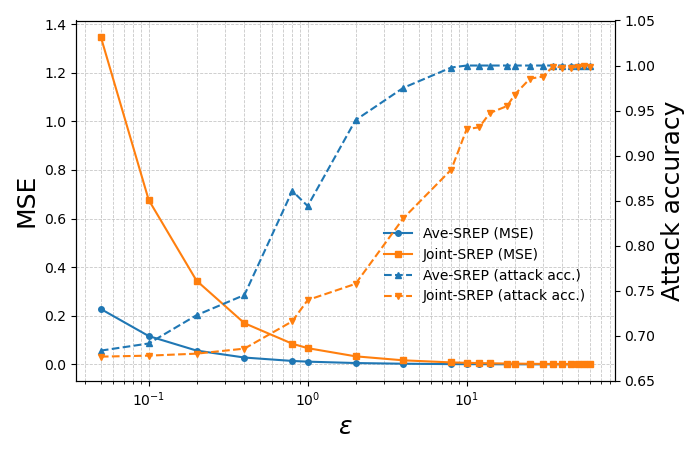}
        \subcaption{Random Forest}
    \end{subfigure}
    \begin{subfigure}[b]{0.24\textwidth}
        \centering
        \includegraphics[width=\textwidth]{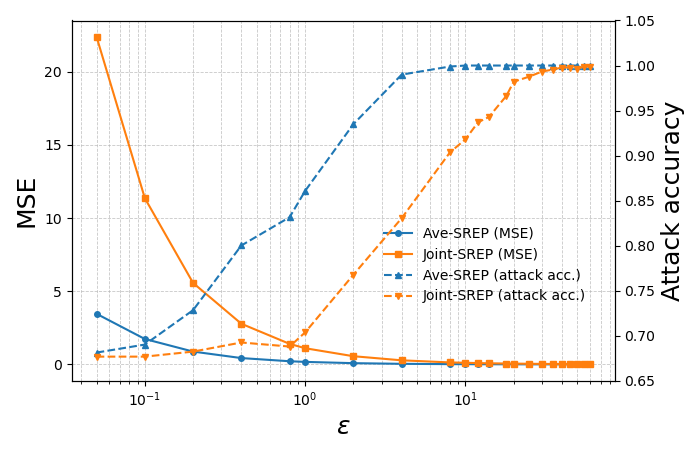}
        \subcaption{SVM}
    \end{subfigure}
    \begin{subfigure}[b]{0.24\textwidth}
        \centering
        \includegraphics[width=\textwidth]{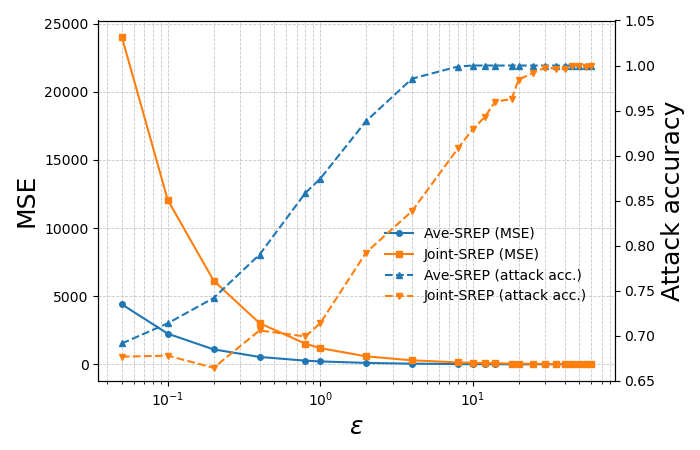}
        \subcaption{Summary statistics}
    \end{subfigure}

    \vspace{0.5em}

    \begin{subfigure}[b]{0.24\textwidth}
        \centering
        \includegraphics[width=\textwidth]{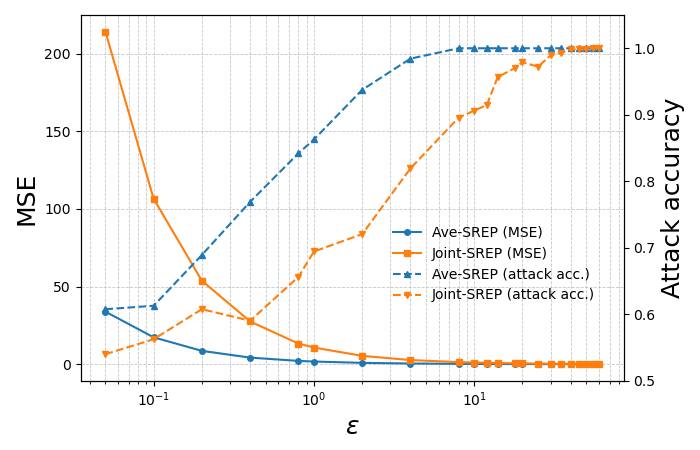}
        \subcaption{Logistic Regression}
    \end{subfigure}
    \begin{subfigure}[b]{0.24\textwidth}
        \centering
        \includegraphics[width=\textwidth]{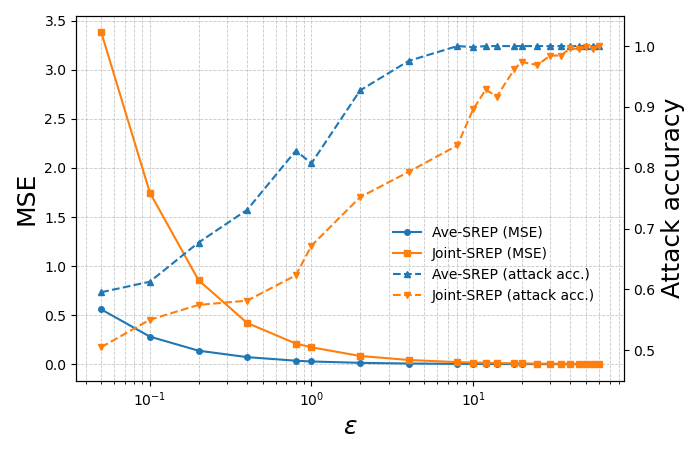}
        \subcaption{Random Forest}
    \end{subfigure}
    \begin{subfigure}[b]{0.24\textwidth}
        \centering
        \includegraphics[width=\textwidth]{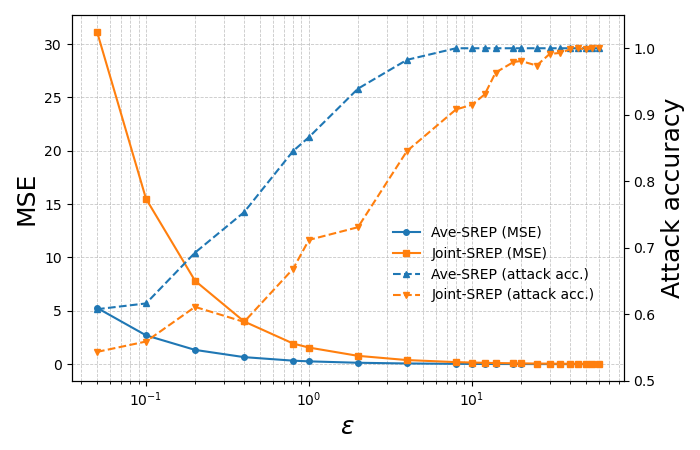}
        \subcaption{SVM}
    \end{subfigure}
    \begin{subfigure}[b]{0.24\textwidth}
        \centering
        \includegraphics[width=\textwidth]{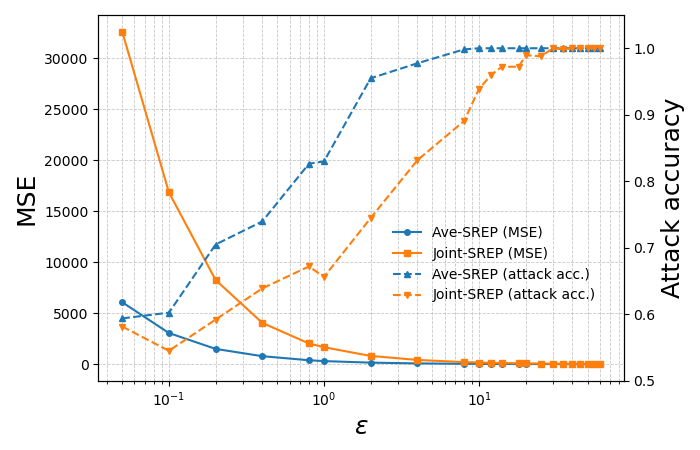}
        \subcaption{Summary statistics}
    \end{subfigure}

    \vspace{0.5em}

    \begin{subfigure}[b]{0.24\textwidth}
        \centering
        \includegraphics[width=\textwidth]{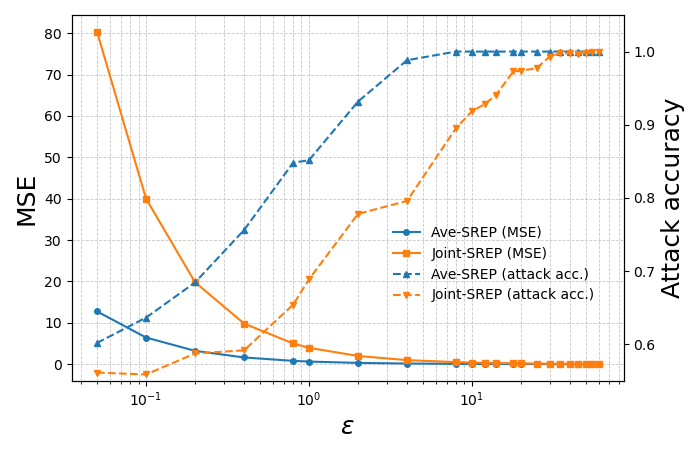}
        \subcaption{Logistic Regression}
    \end{subfigure}
    \begin{subfigure}[b]{0.24\textwidth}
        \centering
        \includegraphics[width=\textwidth]{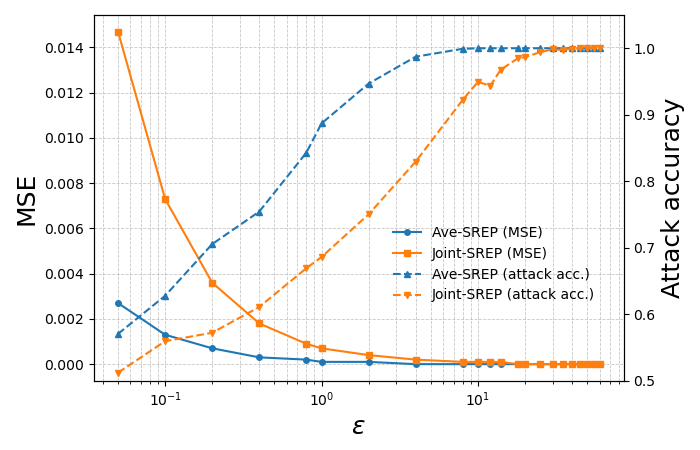}
        \subcaption{Random Forest}
    \end{subfigure}
    \begin{subfigure}[b]{0.24\textwidth}
        \centering
        \includegraphics[width=\textwidth]{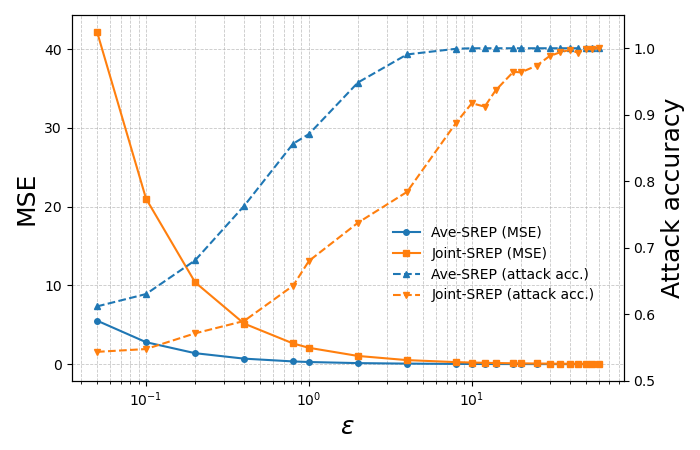}
        \subcaption{SVM}
    \end{subfigure}
    \begin{subfigure}[b]{0.24\textwidth}
        \centering
        \includegraphics[width=\textwidth]{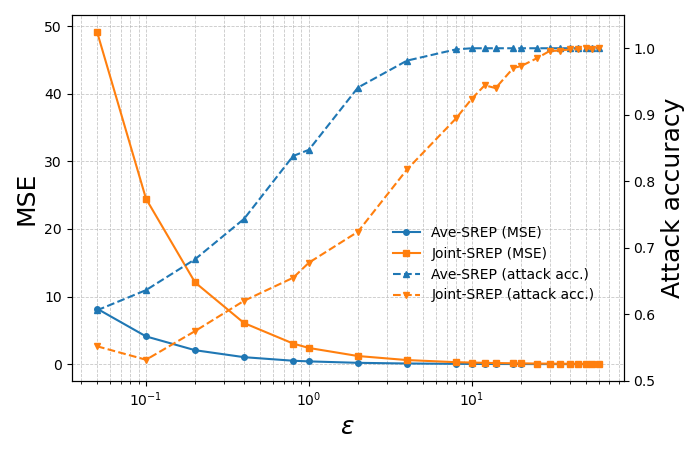}
        \subcaption{Summary Statistics}
    \end{subfigure}

    \begin{subfigure}[b]{0.24\textwidth}
        \centering
        \includegraphics[width=\textwidth]{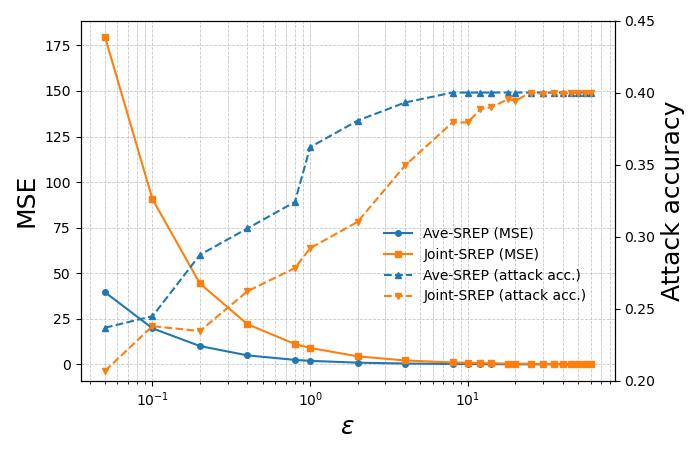}
        \subcaption{Logistic Regression}
    \end{subfigure}
    \begin{subfigure}[b]{0.24\textwidth}
        \centering
        \includegraphics[width=\textwidth]{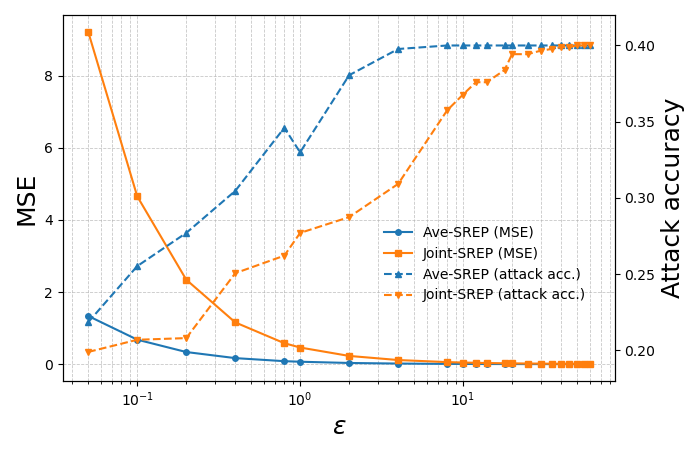}
        \subcaption{Random Forest}
    \end{subfigure}
    \begin{subfigure}[b]{0.24\textwidth}
        \centering
        \includegraphics[width=\textwidth]{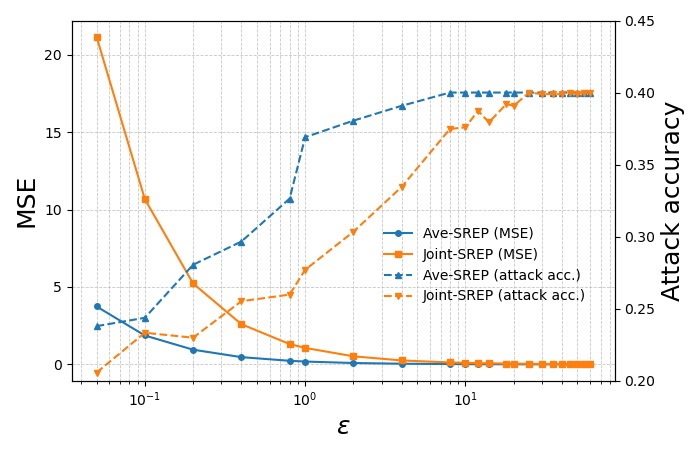}
        \subcaption{SVM}
    \end{subfigure}
    \begin{subfigure}[b]{0.24\textwidth}
        \centering
        \includegraphics[width=\textwidth]{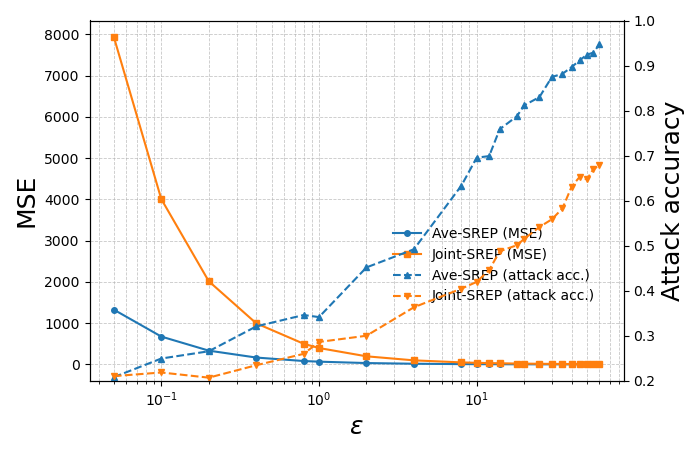}
        \subcaption{Summary Statistics}
    \end{subfigure}
    \caption{\textbf{Privacy-utility tradeoffs on Cleveland Heart Disease and Student Performance datasets.} 
(a)--(d): Cleveland dataset with imbalanced prior (majority-class baseline $\mathrm{Acc}_{\mathrm{maj}}^{0.7} = 0.7$). 
(e)--(h): Cleveland dataset with balanced prior ($\mathrm{Acc}_{\mathrm{maj}}^{0.5} = 0.5$). 
(i)--(l): Student Performance dataset with imbalanced prior (sex as secret, $\mathrm{Acc}_{\mathrm{maj}}^{0.6} = 0.6$). 
(m)--(p): Student Performance dataset with balanced prior (grade group as secret, $\mathrm{Acc}_{\mathrm{maj}}^{0.2} = 0.2$). 
Each configuration evaluates Ave-SRPP and Joint-SRPP mechanisms across privacy budgets $\epsilon \in \{0.05, 0.1, 0.2, 0.4, 0.8, 1.0, 2.0, 4.0, 8.0, 10.0, 12.0, 14.0, 18.0, 20.0, 25.0, 30.0, 35.0, 40.0, 45.0, 50.0, 55.0, 60.0\}$ with R\'enyi order $\alpha = 4$. 
Metrics shown: mean squared error (MSE), attack accuracy, and attack advantage over the priors.
    }
    \label{fig:additional_experiments}
\end{figure*}

\section{Experiment Details}\label{app:experiments}

In this appendix, we provide additional details for the experiments in Section \ref{sec:experiments} and present further empirical results.

\subsection{Experimental Setup: Gaussian Sliced Wasserstein Mechanisms for Static Privatization}

This section provides additional details on our experimental methodology beyond what is presented in Section~\ref{sec:experiments} of the main text.

\paragraph{\textbf{\textup{Dataset configurations and prior distributions.}}}
As described in the main text, we evaluate our sliced SRPP mechanisms on three standard tabular benchmarks: Adult Census (30,162 records), Cleveland Heart Disease (303 records), and Student Performance (395 records).
Each dataset is studied under two prior configurations to examine the impact of class imbalance on both privacy protection and utility.

For \textit{Adult}, the imbalanced configuration uses the natural race distribution from the full dataset, where the majority class (White) accounts for approximately $86\%$ of records, yielding a prior-only baseline accuracy of $\mathrm{Acc}_{\mathrm{maj}}^{0.86} = 0.86$ for attribute inference attacks.
The balanced configuration constructs synthetic groups via deterministic hash-based assignment, partitioning records into five approximately equal groups with uniform prior $\pi(s) \approx 0.20$ for each synthetic secret value $s \in \{G_1, \ldots, G_5\}$, giving $\mathrm{Acc}_{\mathrm{maj}}^{0.20} = 0.20$.

For \textit{Cleveland}, the imbalanced configuration uses the full dataset with its natural sex distribution (majority class probability $\approx 0.70$), while the balanced configuration performs stratified subsampling to obtain equal numbers of male and female records, yielding a uniform prior with $\mathrm{Acc}_{\mathrm{maj}}^{0.50} = 0.50$.

For \textit{Student Performance}, we similarly consider both the full dataset with its empirical sex distribution ($\mathrm{Acc}_{\mathrm{maj}}^{0.60} \approx 0.60$) and a balanced subsample with equal representation of each sex ($\mathrm{Acc}_{\mathrm{maj}}^{0.50} = 0.50$).

Detailed descriptions of each dataset configuration, including feature construction and sampling procedures, are provided in the following subsections.

\paragraph{\textbf{\textup{Query construction.}}}
We consider two families of queries, both of which aggregate information at the group level rather than releasing individual records.

\textit{Summary statistics queries.}
For each secret value $s \in \mathcal{S}$, we compute a $d$-dimensional aggregate statistic over all records with $S = s$.
For Adult ($d=9$), the query includes: means and standard deviations of age, education level, and hours worked per week, together with the fraction of records with positive capital gain, positive capital loss, and income exceeding \$50K.
For Cleveland ($d=9$), we compute means and standard deviations of age, resting blood pressure, and cholesterol level, along with rates of high cholesterol ($\geq 240$ mg/dL), high blood pressure ($\geq 140$ mmHg), and presence of heart disease.
For Student ($d=9$), we use means and standard deviations of grades G1, G2, and G3, together with rates of high study time ($\geq 3$ hours/week), any past failures, and high absences ($> 10$ days).

\textit{Machine learning model queries.}
For the Adult dataset, we also treat trained model parameters as query outputs.
For each race group $s \in \mathcal{S}$, we train a binary classifier to predict high income ($> \$50$K) using a 9-dimensional feature vector constructed from demographic and employment attributes.
The query output $X$ consists of model parameters: for linear SVM and logistic regression, we extract the $d=10$ coefficients (9 feature weights plus intercept); for random forest and gradient boosting, we extract the $d=9$ feature importance scores.
Each model is trained independently per secret group, so the query reflects group-specific predictive patterns.

\paragraph{\textbf{\textup{Privacy mechanism and calibration.}}}
All experiments use isotropic Gaussian noise mechanisms of the form
\[
Y = X + Z, \quad Z \sim \mathcal{N}(0, \sigma^2 I_d),
\]
where the noise scale $\sigma$ is calibrated via our sliced SRPP composition bounds.

For slicing, we generate $m = 200$ directions $\{u_\ell\}_{\ell=1}^m$ sampled uniformly from the unit sphere $\mathbb{S}^{d-1}$ with equal weights $\omega_\ell = 1/m$.
Along each direction $u_\ell$, we compute the one-dimensional Wasserstein-$\infty$ sensitivity
\[
\Delta^{u_\ell}_{\infty} = \max_{(s_i, s_j) \in \mathcal{Q}} W_{\infty}\bigl(P^{\theta}_{u_\ell^\top X \mid S=s_i}, P^{\theta}_{u_\ell^\top X \mid S=s_j}\bigr),
\]
where $P^{\theta}_{u_\ell^\top X \mid S=s}$ is the empirical distribution of the projected query values for secret $s$ under world $\theta$.
These per-slice sensitivities are then aggregated according to the Ave-SRPP and Joint-SRPP definitions to determine the required noise level for target privacy parameters $(\alpha, \epsilon)$.

We fix R\'enyi order $\alpha = 4$ throughout and sweep the privacy budget $\epsilon$ over the range $[0.05, 60]$ to study the privacy-utility tradeoff.
Each configuration is evaluated over $N_{\mathrm{trials}} = 3$ independent trials with different random seeds to account for Monte Carlo variability in the slicing directions and noise realization.

\paragraph{\textbf{\textup{Utility and privacy evaluation.}}}
For utility, we report mean squared error (MSE), mean absolute error (MAE), average $\ell_2$ distance per record, relative Frobenius error, and signal-to-noise ratio (SNR in dB).
The primary utility metric is MSE:
\[
\mathrm{MSE} = \frac{1}{n} \sum_{i=1}^{n} \lVert X_i - Y_i \rVert_2^2.
\]

For empirical privacy auditing, we implement a prior-aware Gaussian maximum a posteriori (MAP) attack.
Given a privatized output $y_i$, the attacker estimates the secret as
\[
\hat{s}(y_i) = \arg\max_{s \in \mathcal{S}} \Bigl\{ \log \pi(s) + \log \phi\bigl(y_i; \mu_s, \sigma^2 I_d\bigr) \Bigr\},
\]
where $\pi(s)$ is the prior distribution over secrets (either the empirical distribution from the dataset or the uniform prior for balanced configurations), $\mu_s = \mathbb{E}[X \mid S=s]$ is the mean query output for secret $s$ computed from the unprivatized data, and $\phi(\cdot; \mu, \Sigma)$ denotes the Gaussian density with mean $\mu$ and covariance $\Sigma = \sigma^2 I_d$.

We report two metrics: (i) overall attack accuracy, and (ii) advantage over the prior-only baseline, defined as $\mathrm{Adv} = \mathrm{Acc} - \mathrm{Acc}_{\mathrm{maj}}$.
For imbalanced Adult, the baseline is $\mathrm{Acc}_{\mathrm{maj}}^{0.86} = 0.86$; for balanced Adult, Cleveland, and Student (balanced), the baseline is $\mathrm{Acc}_{\mathrm{maj}}^{0.20} = 0.20$, $\mathrm{Acc}_{\mathrm{maj}}^{0.50} = 0.50$, and $\mathrm{Acc}_{\mathrm{maj}}^{0.50} = 0.50$, respectively; for imbalanced Cleveland and Student, the baselines are $\mathrm{Acc}_{\mathrm{maj}}^{0.70} = 0.70$ and $\mathrm{Acc}_{\mathrm{maj}}^{0.60} = 0.60$.
An advantage close to zero indicates that the privacy mechanism successfully prevents the attacker from learning substantially more than what the prior distribution already reveals.

The following subsections provide detailed specifications of each dataset configuration.

\subsubsection{Adult Census (race as secret)}
We use the \textit{Adult Census Income} dataset (30,162 records after preprocessing) in an attribute-privacy setting where the sensitive attribute is \textit{race}.
Let
\[
\begin{aligned}
    S \in \mathcal{S} &= \{\text{White}, \text{Black}, \text{Asian-Pac-Islander}, \\&\text{Amer-Indian-Eskimo}, \text{Other}\}, 
\end{aligned}
\]
with $k = |\mathcal{S}| = 5$, and let $X \in \mathbb{R}^{d}$ denote the per-record query vector constructed from demographic and socioeconomic features (see below for the specific query families).
The Pufferfish scenario specifies the discriminatory set
\[
  \mathcal{Q} 
  = \{(s_i, s_j) : s_i, s_j \in \mathcal{S},\, s_i \neq s_j\},
\]
so that every ordered pair of distinct secret values is protected.
We consider two configurations of the dataset that differ only in the prior over $S$.

\paragraph{Adult–0.86 prior (imbalanced).}
We use the full processed dataset with its empirical race distribution. 
Writing
\[
\pi_{0.86}(s) = \frac{1}{n} \sum_{i=1}^{n} \mathbf{1}\{S_i = s\},\quad s \in \mathcal{S},
\]
for the empirical prior over $S$, the corresponding \textit{majority-class baseline} for any attribute-inference attack is
\[
\mathrm{Acc}_{\mathrm{maj}}^{0.86} = \max_{s \in \mathcal{S}} \pi_{0.86}(s) \approx 0.86.
\]
The associated Pufferfish world $\theta_{D}^{0.86}$ is given by the empirical joint distribution of $(X, S)$ over all records in the full dataset.

\paragraph{Adult–0.20 prior (balanced).}
To study a uniform prior, we construct a balanced configuration by creating synthetic balanced groups via hash-based assignment.
We partition the full dataset into $k=5$ groups of approximately equal size using a deterministic hash function applied to each record's feature vector, yielding a subsample with an approximately uniform prior
\[
\pi_{0.20}(s) = \frac{1}{n'} \sum_{i=1}^{n'} \mathbf{1}\{S_i = s\} \approx 0.20,
\quad s \in \mathcal{S}.
\]
In this configuration, the majority-class baseline is
\[
\mathrm{Acc}_{\mathrm{maj}}^{0.20} = \max_{s \in \mathcal{S}} \pi_{0.20}(s) \approx 0.20.
\]
The corresponding Pufferfish world $\theta_{D}^{0.20}$ is the empirical joint distribution of $(X,S)$ restricted to this balanced subsample.

In both configurations, the Pufferfish family is instantiated by the empirical conditionals 
$\{P^{\theta}_{X \mid S = s}\}_{s \in \mathcal{S}}$ under the chosen world 
$\theta \in \{\theta_{D}^{0.86}, \theta_{D}^{0.20}\}$,
together with the discriminatory set $\mathcal{Q}$ defined above.
Our sliced SRPP mechanisms are calibrated to protect all pairs in $\mathcal{Q}$ for the specified R\'enyi order and privacy budget.

\subsubsection{Cleveland Heart Disease (sex as secret)}
We use the \textit{processed Cleveland Heart Disease} dataset (303 records) in an attribute-privacy setting where the sensitive attribute is \textit{sex}.  
Let
\[
S \in \mathcal{S} = \{\text{female}, \text{male}\}, 
\qquad k = |\mathcal{S}| = 2,
\]
and let $X \in \mathbb{R}^{d}$ denote the per-record query vector constructed from clinical features (see below for the specific query families).
The Pufferfish scenario specifies the discriminatory set
\[
  \begin{aligned}
      \mathcal{Q} 
  &= \{(s_i, s_j) : s_i, s_j \in \mathcal{S},\, s_i \neq s_j\}
  \\&= \{(\text{female}, \text{male}), (\text{male}, \text{female})\},
  \end{aligned}
\]
so that every ordered pair of distinct secret values is protected.
We consider two configurations of the dataset that differ only in the prior over $S$.

\paragraph{Cleveland–0.7 prior.}
We use the full processed dataset with its empirical sex distribution. 
Writing
\[
\pi_{0.7}(s) = \frac{1}{n} \sum_{i=1}^{n} \mathbf{1}\{S_i = s\},\quad s \in \mathcal{S},
\]
for the empirical prior over $S$, the corresponding \textit{majority-class baseline} for any attribute-inference attack is
\[
\mathrm{Acc}_{\mathrm{maj}}^{0.7} = \max_{s \in \mathcal{S}} \pi_{0.7}(s) \approx 0.7.
\]
The associated Pufferfish world $\theta_{D}^{0.7}$ is given by the empirical joint distribution of $(X, S)$ over all records in the full dataset.

\paragraph{Cleveland–0.5 prior.}
To study a uniform prior, we construct a balanced configuration by stratified subsampling on sex.
Let $n_{\mathrm{f}}$ and $n_{\mathrm{m}}$ denote the numbers of female and male records in the full dataset.  
We draw $n' = \min(n_{\mathrm{f}}, n_{\mathrm{m}})$ records uniformly at random without replacement from each sex subset, yielding a subsample of size $2n'$ with an approximately uniform prior
\[
\pi_{0.5}(s) = \frac{1}{2n'} \sum_{i=1}^{2n'} \mathbf{1}\{S_i = s\} \approx 0.5,
\quad s \in \mathcal{S}.
\]
In this configuration, the majority-class baseline is
\[
\mathrm{Acc}_{\mathrm{maj}}^{0.5} = \max_{s \in \mathcal{S}} \pi_{0.5}(s) \approx 0.5.
\]
The corresponding Pufferfish world $\theta_{D}^{0.5}$ is the empirical joint distribution of $(X,S)$ restricted to this balanced subsample.

In both configurations, the Pufferfish family is instantiated by the empirical conditionals 
$\{P^{\theta}_{X \mid S = s}\}_{s \in \mathcal{S}}$ under the chosen world 
$\theta \in \{\theta_{D}^{0.7}, \theta_{D}^{0.5}\}$,
together with the discriminatory set $\mathcal{Q}$ defined above.
Our sliced SRPP mechanisms are calibrated to protect all pairs in $\mathcal{Q}$ for the specified R\'enyi order and privacy budget.

\subsubsection{Student Performance (grade group as secret)}
We use the \textit{Student Performance} dataset \cite{silva2008using} (395 records) in an attribute-privacy setting where the sensitive attribute is a discretized \textit{grade group} derived from the final grade $G3 \in \{0,\dots,20\}$.  
We partition students into $k=5$ grade groups using quantile-based binning:
\[
S \in \mathcal{S} = \{\text{G1}, \text{G2}, \text{G3}, \text{G4}, \text{G5}\}, 
\qquad k = |\mathcal{S}| = 5,
\]
where each group corresponds to a quintile of the final grade distribution (G1 = lowest quintile, G5 = highest quintile).
Let $X \in \mathbb{R}^{d}$ denote the per-record query vector constructed from grade and study-related features.
The Pufferfish scenario specifies the discriminatory set
\[
  \mathcal{Q} 
  = \{(s_i, s_j) : s_i, s_j \in \mathcal{S},\, s_i \neq s_j\},
\]
so that every ordered pair of distinct grade groups is protected.
We consider two configurations of the dataset that differ in the prior over $S$.

\paragraph{Student (full dataset).}
We use all 395 records with grade groups assigned via 5-quantile binning.
Since quantile-based binning attempts to create approximately equal-sized groups by construction, the empirical prior is
\[
\pi(s) = \frac{1}{n} \sum_{i=1}^{n} \mathbf{1}\{S_i = s\},
\quad s \in \mathcal{S},
\]
with majority-class baseline
\[
\mathrm{Acc}_{\mathrm{maj}} = \max_{s \in \mathcal{S}} \pi(s).
\]
Due to tied final grade values preventing perfect equal division, the empirical distribution may deviate slightly from uniform.
The Pufferfish world $\theta_{D}^{\text{full}}$ is the empirical joint distribution of $(X, S)$ over the full dataset.

\paragraph{Student (balanced subsample).}
To ensure an exactly uniform prior, we perform stratified subsampling: let $n_s$ denote the number of records in grade group $s$ in the full dataset, and set $n' = \min_{s \in \mathcal{S}} n_s$. 
We draw $n'$ records uniformly at random without replacement from each grade group, yielding a subsample of size $5n'$ with exactly uniform prior
\[
\pi^{\text{bal}}(s) = \frac{1}{5n'} \sum_{i=1}^{5n'} \mathbf{1}\{S_i = s\} = 0.2,
\quad \forall s \in \mathcal{S}.
\]
The corresponding majority-class baseline is
\[
\mathrm{Acc}_{\mathrm{maj}}^{0.2} = \max_{s \in \mathcal{S}} \pi^{\text{bal}}(s) = 0.2.
\]
The Pufferfish world $\theta_{D}^{\text{bal}}$ is the empirical joint distribution of $(X,S)$ restricted to this balanced subsample.

In both configurations, the Pufferfish family is instantiated by the empirical conditionals 
$\{P^{\theta}_{X \mid S = s}\}_{s \in \mathcal{S}}$ under the chosen world 
$\theta \in \{\theta_{D}^{\text{full}}, \theta_{D}^{\text{bal}}\}$,
together with the discriminatory set $\mathcal{Q}$ defined above.
Our sliced SRPP mechanisms are calibrated to protect all pairs in $\mathcal{Q}$ for the specified R\'enyi order and privacy budget.

\subsubsection{Additional Results.}\label{app:cleveland_student_results}

Figure~\ref{fig:additional_experiments} presents analogous privacy-utility tradeoff curves for the Cleveland Heart Disease and Student Performance datasets, evaluating the same four query types (logistic regression, random forest, SVM, and summary statistics) under both balanced and imbalanced prior configurations.

The results exhibit qualitatively similar behavior to the Adult dataset experiments shown in Section \ref{sec:experiments}.
Across all configurations, Ave-SRPP and Joint-SRPP demonstrate the expected monotonic tradeoff: as the privacy budget $\epsilon$ increases, MSE decreases while attack accuracy rises from the prior baseline toward perfect inference.
At small $\epsilon$, both mechanisms successfully limit attack accuracy to near the prior baselines (Cleveland: $0.7$ and $0.5$; Student: $0.6$ and $0.2$), confirming effective privacy protection.
At any fixed $\epsilon$, Joint-SRPP remains more conservative than Ave-SRPP, inducing higher MSE but lower attack accuracy, with the gap most pronounced in the moderate privacy regime.
These additional experiments validate that the comparative performance of Ave-SRPP and Joint-SRPP generalizes across datasets of varying size (Cleveland: 303 records, Student: 395 records versus Adult: 30,162 records), different numbers of secret classes ($k=2$ for Cleveland and Student-sex versus $k=5$ for Adult and Student-grade).

\subsection{SRPP- sa-SRPP-SGD with Gradient Clipping}\label{app-sub:SGD}

For our CIFAR-10 experiments (Figures \ref{subfig:gep_resnet}-\ref{subfig:overfit_attack}) in Section \ref{sec:experiments},  we use the ResNet22, adapted for compatibility with Opacus \cite{yousefpour2021opacus} by following the architecture of \cite{xiao2023geometry},

The network is a 22-layer residual network with three stages of channel sizes $(16,32,64)$ and $(3,4,3)$ residual blocks per stage. The first layer is a $3\times 3$ convolution (16 channels, stride 1, padding 1) followed by group normalization and ELU. Stages 2--3 downsample via stride 2 in their first blocks; shortcuts use spatial subsampling with channel zero-padding.

Each residual block has the form
\begin{align*}
x &\mapsto \mathrm{GN}_1(\mathrm{ELU}(\mathrm{Conv}_1(x))) \\
  &\mapsto \mathrm{GN}_2(\mathrm{Conv}_2(\cdot)) + \text{shortcut}(x) \\
  &\mapsto \mathrm{GN}_3(\mathrm{ELU}(\cdot)),
\end{align*}
where both convolutions are $3\times 3$ with padding 1; $\mathrm{Conv}_1$ uses stride 2 for downsampling blocks and stride 1 otherwise, while $\mathrm{Conv}_2$ always uses stride 1. Group normalization layers $\mathrm{GN}_1,\mathrm{GN}_2,\mathrm{GN}_3$ use at most four groups without affine parameters.

After the final stage, we apply $3\times 3$ adaptive average pooling and flatten to $z\in\mathbb{R}^{576}$. Before classification, we apply per-example standardization $\tilde z = (z - \mu(z))/(\sigma(z) + 10^{-6})$ where $\mu(z)$ and $\sigma(z)$ are the sample mean and standard deviation, then compute logits via $\ell = W \tilde z + b$ with $W\in\mathbb{R}^{10\times 576}$ and $b\in\mathbb{R}^{10}$. Weights use Kaiming-normal initialization.
The model contains approximately 270,000 parameters.
This architecture is fully compatible with per-example gradient computation in Opacus; in particular, we replace non-picklable lambda-based shortcut layers with explicit \texttt{PaddingShortcut} modules.

\paragraph{Standard SGD}
With this ResNet22 model, we consider a \textit{global Lipschitz constant} as a special case of Assumption~\ref{assp:slicewise_Lipschitz}: the update map $T_t(z; y_{<t}) = \xi_{t-1} - \eta_t(\mu v_{t-1} + z)$ for SGD with momentum satisfies
\[
\bigl|u_i^{\top}\bigl(T_t(z; y_{<t}) - T_t(z'; y_{<t})\bigr)\bigr| = \eta_t \bigl|u_i^{\top}(z - z')\bigr|
\]
for all slices $u_i \in \mathcal{U}$ and all histories $y_{<t}$, yielding the \textit{uniform slice-wise Lipschitz constant} $L_{t,i} = \eta_t$ for all $i \in [m]$. This simplifies the HUC from Proposition~\ref{prop:exist_HUC} to
\[
h_{t,i} = \left(\frac{2K_t \eta_t C}{B_t}\right)^2, \quad \forall i \in [m],
\]
which depends only on the learning rate $\eta_t$, clipping threshold $C$, batch size $B_t$, and discrepancy cap $K_t$, but is \textit{independent of the slice direction}.
As a result, Ave-SRPP and Joint-SRPP require the same noise calibration in this uniform case, though their privacy quantification schemes remain different.

\subsection{Experiment Setup}
\label{app:sgd_setup}

Now, we provide detailed specifications for our SRPP-SGD and sa-SRPP-SGD experiments (Figures \ref{subfig:gep_resnet}-\ref{subfig:overfit_attack}) in Section \ref{sec:experiments} of the main text.

\paragraph{\textbf{\textup{Pufferfish two-world construction.}}}
We construct a Pufferfish privacy scenario on CIFAR-10 by creating two aligned worlds that differ minimally in the prevalence of a designated secret class.
Specifically, we select \texttt{cat} as the secret class and define two target prevalences: $p_{\mathrm{low}} = 0.10000$ and $p_{\mathrm{high}} = 0.10004$.

The world construction proceeds in two steps to ensure minimal edit distance:
\begin{enumerate}[leftmargin=*,nosep]
\item \textit{Base labels $\to$ world $s_0$}: Starting from CIFAR-10's original training labels (50,000 examples), we minimally relabel examples to achieve exactly $n_0 = \lfloor p_{\mathrm{low}} \cdot N \rfloor$ \texttt{cat} instances. If the original dataset contains fewer than $n_0$ cats, we randomly select non-cat examples and relabel them as cats; if it contains more, we randomly select cat examples and relabel them to other (non-cat) classes.

\item \textit{World $s_0 \to$ world $s_1$}: From $s_0$, we perform minimal edits to reach exactly $n_1 = \lfloor p_{\mathrm{high}} \cdot N \rfloor$ cat instances, producing world $s_1$. The edit operations are analogous to step 1.
\end{enumerate}

This construction yields two label vectors $y_0$ and $y_1$ (both of length $N = 50{,}000$) that differ on exactly $\Delta = |n_1 - n_0| = 20$ examples, while the image data remains identical across both worlds.
The mismatch rate between worlds is $p_{\mathrm{realized}} = \Delta / N = 0.0004$.

For training, we fix a single \textit{realized world}---either $s_0$ or $s_1$---and overwrite the CIFAR-10 training labels with the corresponding realized label vector.
All subsequent training (private and non-private) operates exclusively on this realized world, ensuring that the privacy guarantee protects against an adversary who might observe the other world.

\paragraph{\textbf{\textup{Privacy accounting: SRPP-SGD, sa-SRPP-SGD, and group-DP-SGD.}}}
We compare three privacy accounting methods.

\textit{(1) SRPP-SGD (deterministic worst-case cap):}
For each iteration $t$, we bound the number of differing examples in a batch by the deterministic cap $K_{\mathrm{cap}} = \min(L, \Delta) = 20$.
Under isotropic Gaussian noise with single-block clipping, the per-step HUC is
\[
h_t^* = \Bigl(\frac{2 \eta_t C K_{\mathrm{cap}}}{L}\Bigr)^2 = \Bigl(\frac{2 \eta_t \cdot 4.0 \cdot 20}{512}\Bigr)^2.
\]
The total HUC over all $T$ steps is $H_{\mathrm{total}} = \sum_{t=1}^{T} h_t^*$, and the noise scale is calibrated via:
\[
\sigma^2(\epsilon) = \frac{\alpha}{2\epsilon} \cdot H_{\mathrm{total}},
\]
where $\alpha = 16$ is the R\'enyi order. The Opacus noise multiplier is $\sigma(\epsilon) \cdot L / C$.

\textit{(2) sa-SRPP-SGD (expected $K^2$ via hypergeometric model):}
We model the number of differing examples $K_t$ in each batch as a random variable following a hypergeometric distribution: $K_t \sim \mathrm{Hypergeometric}(N, \Delta, L)$ under fixed-size sampling without replacement.
For this distribution,
\[
\mathbb{E}[K_t] = L \cdot \frac{\Delta}{N}, \quad \mathrm{Var}[K_t] = L \cdot \frac{\Delta}{N} \cdot \Bigl(1 - \frac{\Delta}{N}\Bigr) \cdot \frac{N - L}{N - 1},
\]
and thus $\mathbb{E}[K_t^2] = \mathrm{Var}[K_t] + \mathbb{E}[K_t]^2$.
We define
\[
\gamma^2 = a^2 \cdot \frac{\mathbb{E}[K_t^2]}{L^2},
\]
where $a = 1$ for add-remove adjacency (our default) or $a = 2$ for replace adjacency.
The total HUC is then
\[
H_{\mathrm{total}} = 4 C^2 \gamma^2 \sum_{t=1}^{T} \eta_t^2,
\]
and the noise scale is calibrated identically to SRPP-SGD via $\sigma^2(\epsilon) = (\alpha / 2\epsilon) \cdot H_{\mathrm{total}}$.

\textit{(3) group-DP-SGD (baseline):}
As a baseline, we also implement group-DP-SGD using the standard DP-SGD accountant \cite{Abadi2016} with group size $G = \lceil \mathbb{E}_{\eta}[K_t^2]^{1/2} \rceil$ (derived from the hypergeometric model).
The DP-SGD accountant computes a per-step privacy loss, which we then convert to group DP by scaling with the group size. Appendix \ref{app:group-dp-vs-srpp} provides detailed discussions of the relationship between SRPP/sa-SRPP-SGD and group DP-SGD.

In the experiments for Figures \ref{subfig:gep_resnet} and \ref{subfig:group_dp}, we sweep the privacy budget $\epsilon$ over the range $\{$ 80, 70, 65, 60, 55, 50, 45, 40, 35, 30, 25, 20, 18, 14, 12, 10, 8, 4, 2, 1, 0.5 $\}$ and report test accuracy at each $\epsilon$.
We evaluate two sampling rates corresponding to batch sizes $L = 512$ (sampling rate $q \approx 0.01$) and $L = 1024$ (sampling rate $q \approx 0.02$) on the full 50,000-example training set.
Training uses 40 epochs for $L = 512$ and 40 epochs for $L = 1024$, with clipping norm $C = 4.0$ (for $L = 512$) or $C = 5.0$ (for $L = 1024$), learning rate $\eta_0 = 0.2$ (for $L = 512$) or $\eta_0 \approx 0.14$ (for $L = 1024$) with cosine decay and 5\% warmup, momentum $0.9$, and weight decay $5 \times 10^{-4}$.
All privacy accounting uses R\'enyi order $\alpha = 16$ and add-remove adjacency.

\paragraph{\textbf{\textup{Evaluation metrics.}}}
For each privacy budget $\epsilon$, we train the model to convergence (50 epochs) and report:
\begin{itemize}[leftmargin=*,nosep]
\item \textit{Test accuracy}: Classification accuracy on the CIFAR-10 test set (10,000 examples).
\item \textit{Training accuracy}: Classification accuracy on the training set (for diagnostic purposes).
\item \textit{Reference $\epsilon$}: The privacy budget computed by Opacus's built-in accountant (at $\delta = 10^{-5}$) for comparison.
\end{itemize}

\paragraph{\textbf{\textup{Overfitted regime and membership inference attacks.}}}
To perform empirical privacy auditing, we conduct additional experiments in an overfitted regime where membership attacks are most effective.
We randomly subsample 1,000 examples from the CIFAR-10 training set and train models for 200 epochs with reduced regularization (weight decay set to zero) to induce overfitting.
We use batch size $L = 256$, clipping norm $C = 5.0$, learning rate $\eta_0 = 0.1$ with cosine decay, and evaluate privacy budgets $\epsilon \in \{2, 8, 20, 50, 70, 100, 200, 300, 500\}$ plus a non-private baseline.

On these overfitted models, we apply the loss-threshold membership inference attack of \cite{yeom2018privacy}.
For each example $x$ with true label $y$, we compute the cross-entropy loss $\ell(x) = -\log p_{\theta}(y \mid x)$ and use the negative loss as a membership score: lower loss suggests the example is more likely a training member.
We compute the ROC AUC for distinguishing the 1,000 training members from an equal-sized set of non-members sampled from the remaining CIFAR-10 training data.
ROC AUC of 0.5 indicates random guessing (perfect privacy), while higher values indicate successful membership inference.
We report ROC AUC as a function of $\epsilon$ to demonstrate how privacy protection degrades membership leakage.

\end{document}